%% file: 0-main.tex
\newif\ifnotes
\notesfalse
\newif\iffocs
\focsfalse
\def\showtableofcontents{1}

\documentclass[11pt,letterpaper]{article}

\input{preamble.tex}

\begin{document}
%%%%%%%%%%%%%%%%%%%%%%%%%%%%%%%%%%%%%%%%%%%%%%%%%%%%%%%%%%%%%%%%%%%%%%%%%%%%%%%%
%%%%%%%%%%%%%%%%%%%%%%%%%%%%%%%%%%%%%%%%%%%%%%%%%%%%%%%%%%%%%%%%%%%%%%%%%%%%%%%%
%%%%%%%%%%%%%%%%%%%%%%%%%%%%%%%%%%%%%%%%%%%%%%%%%%%%%%%%%%%%%%%%%%%%%%%%%%%%%%%%
\title{Post-Quantum Zero Knowledge, Revisited \\ {\large or: How to Do Quantum Rewinding Undetectably}}
\author{
  \begin{tabular}[h!]{ccc}
    \FormatAuthor{Alex Lombardi}{alexjl@mit.edu}{MIT} &
    \FormatAuthor{Fermi Ma}{fermima@alum.mit.edu}{Simons Institute \& UC Berkeley} &
    \FormatAuthor{Nicholas Spooner}{nspooner@bu.edu}{Boston University}
  \end{tabular}
}
\date{November 22, 2021}
%%%%%%%%%%%%%%%%%%%%%%%%%%%%%%%%%%%%%%%%%%%%%%%%%%%%%%%%%%%%%%%%%%%%%%%%%%%%%%%%
%%%%%%%%%%%%%%%%%%%%%%%%%%%%%%%%%%%%%%%%%%%%%%%%%%%%%%%%%%%%%%%%%%%%%%%%%%%%%%%%
%%%%%%%%%%%%%%%%%%%%%%%%%%%%%%%%%%%%%%%%%%%%%%%%%%%%%%%%%%%%%%%%%%%%%%%%%%%%%%%%
\maketitle
%%%%%%%%%%%%%%%%%%%%%%%%%%%%%%%%%%%%%%%%%%%%%%%%%%%%%%%%%%%%%%%%%%%%%%%%%%%%%%%%
%%%%%%%%%%%%%%%%%%%%%%%%%%%%%%%%%%%%%%%%%%%%%%%%%%%%%%%%%%%%%%%%%%%%%%%%%%%%%%%%
%%%%%%%%%%%%%%%%%%%%%%%%%%%%%%%%%%%%%%%%%%%%%%%%%%%%%%%%%%%%%%%%%%%%%%%%%%%%%%%%

\input{1-abstract.tex}

\ifnum\showtableofcontents=1
{
\thispagestyle{empty}
\newpage
\setcounter{tocdepth}{2}
{\thispagestyle{empty}
\tableofcontents}
\thispagestyle{empty}
 }
\clearpage
\setcounter{page}{1}
\fi

\input{1-intro}

\input{2-techoverview}

\input{3-preliminaries}

\input{4-collapse-unique}

\input{5-special-soundness}

\input{6-svt}

\input{7-pseudoinverse}

\input{8-extractor}

\input{9-eqpt}

\input{10-state-preserving-ext}

\input{11-gni}

\input{12-feige-shamir}

\input{13-goldreich-kahan}

\bibliographystyle{alpha}

\input{14-acknowledgments}

{\small\bibliography{references,custom,abbrev3,crypto}}

\appendix

%%%%%%%%%%%%%%%%%%%%%%%%%%%%%%%%%%%%%%%%%%%%%%%%%%%%%%%%%%%%%%%%%%%%%%%%%%%%%%%%

\end{document}

%% file: preamble.tex
\usepackage[T1]{fontenc} % uses good PDF fonts
\usepackage{lmodern} % version of computer modern that supports T1 encoding
\usepackage{hyperref}
\usepackage[utf8]{inputenc} % handles UTF characters in TeX source
\usepackage{cmap} % makes PDFs consistently copy pastable
\usepackage{multirow}
\usepackage{float}
\usepackage{amsmath,amsthm,amssymb,algpseudocode,algorithm,cryptocode}
\usepackage{bm} % fermi added this for making mixed states bold
\usepackage{graphicx}
\usepackage{fullpage}
\usepackage{caption}
\usepackage{tikz}
\usetikzlibrary{angles,quotes,quantikz}
\usepackage{physics}
\usepackage{mathtools}
\usepackage[capitalize,nameinlink]{cleveref}
  \crefname{step}{Step}{Steps}
  \crefname{claim}{Claim}{Claims}
\usepackage{xspace} % define \newcommands that don't eat spaces
\usepackage{paralist}
\usepackage{enumitem}
\usepackage{mdframed}

\allowdisplaybreaks

\newtheorem{theorem}{Theorem}[section]

\newtheorem{claim}[theorem]{Claim}
\newtheorem{lemma}[theorem]{Lemma}

\newtheorem{corollary}[theorem]{Corollary}

\newtheorem{remark}[theorem]{Remark}

\theoremstyle{definition}
\newtheorem{definition}[theorem]{Definition}

\newtheorem{construction}[theorem]{Construction}

\theoremstyle{remark}

\newcommand{\FormatAuthor}[3]{
  \begin{tabular}{c}
  #1 \\ {\small\texttt{#2}} \\ {\small #3}
  \end{tabular}
}

\newcommand{\textdef}[1]{\textnormal{\textsf{#1}}}
\newcommand{\defemph}[1]{\textbf{#1}}

\newcommand{\secp}{\lambda}

\newcommand{\eqdef}{\coloneqq}

\newcommand{\Naturals}{\mathbb{N}}
\newcommand{\C}{\mathbb{C}}
\newcommand{\Bits}{\{0,1\}}
\DeclareMathOperator*{\Expectation}{\mathbb{E}}

\newcommand{\NP}{\mathsf{NP}}

\def\poly{{\rm poly}}
\def\negl{{\rm negl}}

\let\oldket\ket
\renewcommand{\ket}[1]{\oldket*{#1}}
\let\oldketbra\ketbra
\renewcommand{\ketbra}[1]{\oldketbra*{#1}}
\newcommand{\tensor}{\otimes}

\newcommand{\sC}{\mathsf{C}}
\newcommand{\sD}{\mathsf{D}}
\newcommand{\sG}{\mathsf{G}}

\newcommand{\sP}{\mathsf{P}}
\newcommand{\sS}{\mathsf{S}}
\newcommand{\sT}{\mathsf{T}}
\newcommand{\sU}{\mathsf{U}}
\newcommand{\sV}{\mathsf{V}}
\newcommand{\RegA}{\mathcal{A}}
\newcommand{\RegB}{\mathcal{B}}

\newcommand{\RegH}{\mathcal{H}}

\newcommand{\RegK}{\mathcal{K}}

\newcommand{\RegM}{\mathcal{M}}

\newcommand{\RegO}{\mathcal{O}}

\newcommand{\RegQ}{\mathcal{Q}}
\newcommand{\RegR}{\mathcal{R}}
\newcommand{\RegS}{\mathcal{S}}

\newcommand{\RegV}{\mathcal{V}}
\newcommand{\RegW}{\mathcal{W}}
\newcommand{\RegX}{\mathcal{X}}
\newcommand{\RegY}{\mathcal{Y}}
\newcommand{\RegZ}{\mathcal{Z}}

\newcommand{\contra}[1]{#1^{\dagger}}

\newcommand{\Transform}{\mathsf{Transform}}

\newcommand{\Estimate}{\mathsf{Estimate}}
\newcommand{\VarEstimate}{\mathsf{VarEstimate}}

\newcommand{\Flip}{\mathsf{Flip}}
\newcommand{\Threshold}{\mathsf{Threshold}}

\newcommand{\Id}{\mathbf{I}}
\newcommand{\Meas}[1]{\mathsf{M}_{#1}}
\newcommand{\BProj}[1]{\Pi_{#1}}
\newcommand{\BMeas}[1]{\left( #1, \Id - #1 \right)}

\newcommand{\SProj}[2][]{\Pi_{#2}^{#1}}
\newcommand{\Hermitians}[1]{\mathbf{S}(#1)}

\newcommand{\Projector}{\Pi}
\newcommand{\Measurement}{\mathsf{M}}
\newcommand{\ProjMeasurement}{\mathsf{P}}
\newcommand{\MeasA}{\mathsf{A}}
\newcommand{\MeasB}{\mathsf{B}}
\newcommand{\MeasC}{\mathsf{C}}

\newcommand{\DMatrix}{\bm{\rho}}
\newcommand{\brho}{\bm{\rho}}
\newcommand{\dm}{\DMatrix}
\newcommand{\dmr}{\dm}
\newcommand{\DMatrixW}{\bm{\sigma}}
\newcommand{\dmw}{\DMatrixW}
\newcommand{\DMatrixH}{\bm{\phi}}
\newcommand{\dmh}{\DMatrixH}
\newcommand{\CPTPUnitary}{U_{T}}
\newcommand{\DMatrixT}{\bm{\tau}}

\newcommand{\dmwr}{\DMatrixT}

\newcommand{\ProjA}{\BProj{\MeasA}}
\newcommand{\ProjB}{\BProj{\MeasB}}

\newcommand{\ProjP}{\BProj{p,\eps,\delta}}

\newcommand{\eps}{\varepsilon}

\newcommand{\Co}{a}
\newcommand{\Ch}{r}
\newcommand{\Resp}{z}
\newcommand{\vk}{\mathsf{vk}}
\renewcommand{\epsilon}{\varepsilon}

\newcommand{\SHVZK}{\mathsf{SHVZK}}
\newcommand{\Sim}{\mathsf{Sim}}
\newcommand{\Extract}{\mathsf{Extract}}
\newcommand{\FindWitness}{\mathsf{FindWitness}}
\newcommand{\Consistent}{\mathsf{Consistent}}

\newcommand{\Samp}{\mathsf{Samp}}
\renewcommand{\Check}{\mathsf{Check}}
\newcommand{\prefix}{\tau_{\mathrm{pre}}}

\newcommand{\EQPTC}{\mathsf{EQPT}_{c}}
\newcommand{\EQPTM}{\mathsf{EQPT}_{m}}

\newcommand{\Subspace}{\RegS}

\DeclareMathOperator{\spanset}{span}
\DeclareMathOperator{\image}{image}

\newcommand{\Jor}{\mathsf{Jor}}
\newcommand{\dan}{\Jor_r}

\newcommand{\MeasJor}{\Meas{\Jor}}
\newcommand{\JorSymb}[3]{#1_{#2,#3}}
\newcommand{\JorKetA}[2]{\ket{\JorSymb{v}{#1}{#2}}}
\newcommand{\JorKetB}[2]{\ket{\JorSymb{w}{#1}{#2}}}
\newcommand{\JorBraA}[2]{\bra{\JorSymb{v}{#1}{#2}}}
\newcommand{\JorBraB}[2]{\bra{\JorSymb{w}{#1}{#2}}}
\newcommand{\JorBraKetAB}[2]{\braket{\JorSymb{v}{#1}{#2}}{ \JorSymb{w}{#1}{#2}}}
\newcommand{\JorKetBraA}[2]{\ketbra{\JorSymb{v}{#1}{#2}}}
\newcommand{\JorKetBraB}[2]{\ketbra{\JorSymb{w}{#1}{#2}}}

\newcommand{\Adversary}{\mathsf{Adv}}

\newcommand{\Relation}{\mathfrak{R}}

\newcommand{\Alphabet}{\Sigma}

\newcommand{\PSSExtract}{\mathsf{PSSExtract}}
\newcommand{\SSExtract}{\mathsf{SSExtract}}

\newcommand{\Com}{\mathsf{Com}}

\newcommand{\ck}{\mathsf{ck}}
\newcommand{\com}{\mathsf{com}}
\newcommand{\Gen}{\mathsf{Gen}}
\newcommand{\Commit}{\mathsf{Commit}}

\newcommand{\Adv}{\mathsf{Adv}}

\newcommand{\RegState}{\mathcal{Q}}
\newcommand{\RegHead}{\mathcal{I}}
\newcommand{\RegTape}{\mathcal{T}}
\newcommand{\RegInput}{\mathcal{A}}
\newcommand{\RegInpHead}{\mathcal{I}_{\mathsf{in}}}

\newcommand{\Tx}{\delta}
\newcommand{\TStates}{Q}
\newcommand{\TState}{q}
\newcommand{\TStateInit}{q_0}
\newcommand{\TStateFinal}{q_f}

%% file: 1-abstract.tex
\begin{abstract}
    
    \noindent 
    
    A major difficulty in quantum rewinding is the fact that measurement is destructive: extracting information from a quantum state irreversibly changes it. This is especially problematic in the context of zero-knowledge simulation, where preserving the adversary's state is essential.
    
    In this work, we develop new techniques for quantum rewinding in the context of extraction and zero-knowledge simulation:
    \begin{enumerate}
        \item We show how to extract information from a quantum adversary by rewinding it \emph{without disturbing its internal state}. We use this technique to prove that important interactive protocols, such as the Goldreich-Micali-Wigderson protocol for graph non-isomorphism and the Feige-Shamir protocol for $\mathsf{NP}$,  are zero-knowledge against quantum adversaries. 
        \item We prove that the Goldreich-Kahan protocol for $\mathsf{NP}$ is post-quantum zero knowledge using a simulator that can be seen as a natural quantum extension of the classical simulator. 
    \end{enumerate}

Our results achieve (constant-round) black-box zero-knowledge with \emph{negligible} simulation error, appearing to contradict a recent impossibility result due to Chia-Chung-Liu-Yamakawa (FOCS 2021). This brings us to our final contribution:

\begin{enumerate}

\item[3.] We introduce \emph{coherent-runtime} expected quantum polynomial time, a computational model that (1) captures all of our zero-knowledge simulators, (2) cannot break any polynomial hardness assumptions, and (3) is not subject to the CCLY impossibility. In light of our positive results and the CCLY negative results, we propose coherent-runtime simulation to be the right quantum analogue of classical expected polynomial-time simulation.

\end{enumerate}

\end{abstract}

%% file: 1-intro.tex
\section{Introduction}

Zero-knowledge protocols \cite{STOC:GolMicRac85} are a fundamental tool in modern cryptography in which a prover convinces a verifier that some statement is true without revealing any additional information. This security property is formalized via \emph{simulation}: the view of any (even malicious) verifier $V^*$ can be simulated in polynomial time (without access to, e.g., an $\mathsf{NP}$ witness for the statement). 

Although the zero-knowledge property sounds almost paradoxical, it is achieved by designing a simulator $S^{V^*}$ that makes use of $V^*$ in ways that the honest protocol execution cannot, thereby resolving the apparent paradox. In the simplest and most common setting, the key simulation technique is \emph{rewinding}. Given an interactive adversary $A$, an oracle algorithm $S^{A}$ is said to rewind the adversary if it saves the state of $A$ midway through an execution in order to run $A$ \emph{multiple times} on different inputs. Rewinding is ubiquitous in the analysis of interactive proof systems, establishing properties such as zero-knowledge \cite{STOC:GolMicRac85,FOCS:GolMicWig86}, soundness \cite{STOC:Kilian92}, and knowledge-soundness \cite{STOC:GolMicRac85,C:BelGol92}. 

However, since the foundational techniques of interactive proof systems were established, our conception of what constitutes efficient computation has fundamentally changed. Both in theory~\cite{FOCS:Shor94} and in practice~\cite{Google}, quantum computers appear to have capabilities beyond that of any efficient classical computer. Thus, it is imperative to analyze security against quantum adversaries. In this work, we consider this question for zero-knowledge protocols.

\begin{center}
    \emph {When do classical zero-knowledge protocols remain secure against quantum adversaries?}
\end{center}

At a minimum, such protocols must be based on post-quantum cryptographic assumptions. However, since zero-knowledge is typically proved via rewinding, resolving this question also entails understanding \emph{to what extent we can rewind quantum adversaries}. Unfortunately, rewinding quantum adversaries is notoriously difficult because an adversary's internal state may be disturbed if any classical information is recorded about its response, potentially rendering it useless for subsequent executions~\cite{vandeGraaf97,STOC:Watrous06,EC:Unruh12,FOCS:AmbRosUnr14}. 

By now, a few techniques exist to rewind quantum adversaries~\cite{STOC:Watrous06,EC:Unruh12,EC:Unruh16,C:ChiChuYam21,C:AnaChuLap21,FOCS:CMSZ21}, but the range of protocols to which these techniques apply remains quite limited. As a basic example, Watrous's zero-knowledge simulation technique \cite{STOC:Watrous06} applies to the standard \cite{FOCS:GolMicWig86} zero-knowledge proof system for graph isomorphism but (as noted in \cite{EC:Unruh12,FOCS:AmbRosUnr14}) does \emph{not} apply to the related~\cite{FOCS:GolMicWig86} zero-knowledge proof system for graph \emph{non}-isomorphism (GNI). Recall that in the GNI protocol, the prover $P$ wants to convince the verifier $V$ that two graphs $G_0, G_1$ are not isomorphic. To do so, the verifier sends a random isomorphic copy $H$ of $G_b$ for a uniformly random bit $b$, to which the prover returns $b$.\footnote{For this overview, we focus on the soundness $1/2$ case, but appropriate parallel repetition of this step reduces the soundness error.} However, to ensure zero-knowledge, the verifier first gives a proof of knowledge (PoK) that $H$ is isomorphic to either $G_0$ or $G_1$ via a variant of the parallel-repeated graph \emph{isomorphism} $\Sigma$-protocol. Intuitively, this ensures that a malicious verifier $V^*$ already \emph{knows} $b$ and hence does not learn anything new from the interaction.

The classical zero-knowledge simulator for the GNI protocol has two steps:
\begin{enumerate}
    \item \textbf{Extract} an isomorphism $\pi$ satisfying $\pi(H) = G_b$ for some $b$ using \emph{multiple} valid PoK responses from the malicious verifier $V^*$.
    \item \textbf{Simulate} the view of $V^*$ in an real interaction by returning $b$ (computed efficiently from $\pi$).
\end{enumerate}

It turns out that this kind of extract-and-simulate approach is beyond the reach of existing quantum rewinding techniques, because all known techniques for extracting information from \emph{multiple executions} of the adversary~\cite{EC:Unruh12,EC:Unruh16,FOCS:CMSZ21} fundamentally disturb the state. While this particular example just concerns the GNI protocol, this extract-and-simulate approach is very widely applicable, especially in the context of \emph{composition}, including protocols that follow the ``FLS trapdoor paradigm'' \cite{STOC:FeiSha90} or make use of extractable commitments \cite{EC:RicKil99,FOCS:PraRosSah02,TCC:Rosen04}. Given this state of affairs, we ask:

\begin{center}
  \em When is it possible to \emph{undetectably} extract information from a quantum adversary?
\end{center}

As per the above discussion, if it is possible to \emph{undetectably} extract from the proof-of-knowledge subroutine in the \cite{FOCS:GolMicWig86} GNI protocol, then the full protocol is zero-knowledge against quantum adversaries. 

\subsection{This Work}

    In this work, we develop new techniques for quantum rewinding in the context of extraction and zero-knowledge simulation:
    \begin{enumerate}
        \item[(1)] Our first contribution is to give a quantum analogue of the \emph{extract-and-simulate} paradigm used in many classical zero-knowledge protocols, in which a simulator uses information extracted from \emph{multiple protocol transcripts} to simulate the verifier's view. The key difficulty in the quantum setting is to extract this information without causing any noticeable disturbance to the verifier's quantum state --- beyond the disturbance caused by a single protocol execution.
        
        While the recent techniques of~\cite{FOCS:CMSZ21} allow extracting from multiple protocol transcripts, a major problem is that their extractor noticeably disturbs the adversary's state. We revisit the~\cite{FOCS:CMSZ21} approach for extraction and, using several additional ideas, construct an \emph{undetectable} extractor for a broad class of protocols. Using this extraction technique, we prove that the original~\cite{FOCS:GolMicWig86} protocol for graph non-isomorphism and some instantiations of the \cite{STOC:FeiSha90} protocol for $\mathsf{NP}$ are zero-knowledge against quantum adversaries. 
        \item[(2)] We next turn our attention to the Goldreich-Kahan \cite{JC:GolKah96} zero-knowledge proof system for $\mathsf{NP}$. Informally, analyzing the \cite{JC:GolKah96} proof system presents different challenges as compared to \cite{FOCS:GolMicWig86,STOC:FeiSha90} because in the latter protocols, rewinding is used for \emph{extraction} (after which simulation is straight-line), while in the \cite{JC:GolKah96} protocol, rewinding is used for the \emph{simulation} step (while extraction is trivial/straight-line). 
        
        Nevertheless, we show that some of our techniques are also applicable in this setting. We prove that the \cite{JC:GolKah96} protocol is zero-knowledge against quantum adversaries. Our simulator can be viewed as a natural quantum extension of the classical simulator.
        
        Previously,~\cite{C:ChiChuYam21} used different techniques to show that the \cite{JC:GolKah96} protocol is $\eps$-zero-knowledge against quantum adversaries, but their simulation strategy cannot achieve negligible accuracy. 
    \end{enumerate}

\paragraph{Isn't this impossible?}

As stated above, our results (both (1) and (2)) achieve constant-round black-box zero-knowledge with negligible simulation accuracy. Recently,~\cite{FOCS:CCLY21} showed that there do not exist black-box expected quantum polynomial time (EQPT) simulators for constant-round protocols for any language $L \not\in \mathsf{BQP}$. The source of the disconnect between our results and \cite{FOCS:CCLY21} is an ambiguity in the definition of EQPT. This brings us to our final contribution.

\begin{enumerate}
\item[(3)] We formally study the notion of expected runtime for quantum machines and formulate a model of expected quantum polynomial time simulation that avoids the~\cite{FOCS:CCLY21} impossibility result. 
\end{enumerate}

We now discuss these contributions in more detail. To avoid confusion about the formal statements of (1) and (2), we begin by describing (3). 

\subsection{Coherent-Runtime Expected Quantum Polynomial Time}
\label{sec:intro-creqpt}

While~\cite{FOCS:CCLY21} do not formally define EQPT,\footnote{When defining quantum zero-knowledge simulation, \cite[Page~12]{FOCS:CCLY21} requires that the simulator is a quantum Turing machine with \emph{expected} polynomial runtime, and refers to \cite{BBBV97} (which uses the~\cite{BV97} definition of a quantum Turing machine) for the quantum Turing machine model. However, as we discuss in \cref{sec:tech-eqpt},~\cite{BV97} restricts quantum Turing machines to have a \emph{fixed} running time (see~\cite[Def~3.11]{BV97}) in order to avoid difficult-to-resolve subtleties about quantum Turing machines with variable running time~\cite{Myers97,Ozawa98a,LindenP98,Ozawa98b}.} implicit in their result is the following computational model, which we call \emph{measured-runtime} EQPT ($\EQPTM$). In this model, a computation on input $\ket{\psi}_\RegA$ is the following process, for a fixed ``transition'' unitary $U_{\delta}$ (corresponding to a quantum Turing machine transition function $\delta$) and time bound $T = \exp(\secp)$:\footnote{The exponential time bound is simply for convenience; by Markov's inequality, for any expected polynomial time computation truncating the computation after an exponential number of steps has only a negligible effect on the output state.}
\begin{enumerate}[noitemsep]
    \item Initialize a fresh memory/work register $\RegW$ to $\ket{0^{T}}_{\RegW}$ and a designated ``halt'' qubit $\RegQ$ to $\ket{0}$;
    
    \item Repeat for at most $T$ steps:
    \begin{enumerate}[nolistsep]
        \item measure $\RegQ$ and halt if it is $1$;
        \item
        \label[step]{step:apply-unitary} apply $U_{\delta}$ to $\RegA \otimes \RegQ \otimes \RegW$.
    \end{enumerate}
\end{enumerate}
The result of the computation is the residual state on $\RegA$ once the computation has halted. We say that a computation is $\EQPTM$ if for \emph{all} states $\ket{\psi}_{\RegA}$, the expected running time of this computation is $\poly(\secp)$. Using this model, we can give a more precise formulation of the \cite{FOCS:CCLY21} theorem: black-box $\EQPTM$ zero-knowledge simulators for constant-round protocols do not exist. The key feature of the $\EQPTM$ model that enables the \cite{FOCS:CCLY21} result is that the runtime is \emph{measured}.

In this work, we consider a different computational model for EQPT simulation called \emph{coherent-runtime EQPT} ($\EQPTC$). In our model, simulators have the ability to run $\EQPTM$ procedures \emph{coherently} --- which yields a superposition over computations with different runtimes --- and then later \emph{uncompute} the runtime by running the same computation in reverse. 

Our notion of $\EQPTC$ (see \cref{def:cr-eqpt}) captures (as a special case) computations of the form depicted in \cref{fig:cr-eqpt-simple} on an input $\ket{\phi}_{\RegX}$, where the result of the computation is the residual state on $\RegX$. In~\cref{fig:cr-eqpt-simple}, $C_1, C_2, C_3$ are arbitrary polynomial-size quantum circuits and $U$ is a unitary that \emph{coherently implements} an $\EQPTM$ computation.\footnote{Our actual definition is for \emph{uniform} computation, so $(U, C_1, C_2, C_3)$ will have a uniform description.} In slightly more detail, any $\EQPTM$ computation with transition unitary $U_{\delta}$ and runtime bound $T$ can be expressed as a unitary circuit $U = V_T \cdots V_2 V_1$ where each $V_i$ consists of two steps: (1) $\mathsf{CNOT}$ the halt qubit onto a register $\RegB_i$, and then (2) apply $U_{\delta}$ controlled on $\RegB_i = 0$. The unitary $U$ acts on $\RegA \otimes \RegQ \otimes \RegW \otimes \RegB$ where $\RegB \eqdef \RegB_1 \otimes \cdots \otimes \RegB_T$. While our full definition of $\EQPTC$ is more general, all the simulators we give can be written in the form of~\cref{fig:cr-eqpt-simple}. In~\cref{fig:cr-eqpt-simple}, the input register is of the form $\RegX = \RegX_1 \tensor \RegX_2$ where $\RegX_2$ is isomorphic to $\RegA$. 

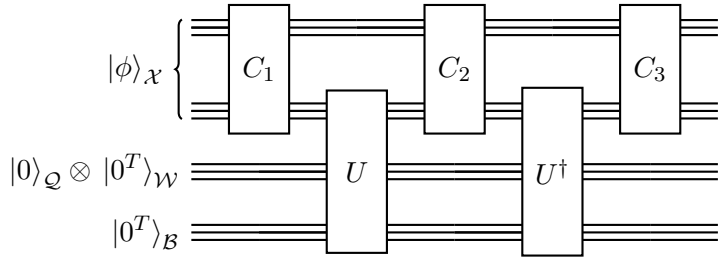
\begin{figure}[h!]
\centering
\begin{quantikz}
\lstick[wires=2]{$\ket{\phi}_{\RegX}$} & \gate[wires=2][0.8cm]{C_1}\qwbundle[alternate]{} & \qwbundle[alternate]{} & \gate[wires=2][0.8cm]{C_2}\qwbundle[alternate]{} & \qwbundle[alternate]{} & \gate[wires=2][0.8cm]{C_3}\qwbundle[alternate]{} & \qwbundle[alternate]{} \\
\lstick[wires=1]{} & \qwbundle[alternate]{} & \gate[wires=3][0.8cm]{U}\qwbundle[alternate]{} & \qwbundle[alternate]{} & \gate[wires=3][0.8cm]{U^\dagger}\qwbundle[alternate]{} & \qwbundle[alternate]{} & \qwbundle[alternate]{} \\
\lstick[wires=1]{$\ket{0}_{\RegQ} \otimes \ket{0^T}_{\RegW}$} & \qwbundle[alternate]{} & \qwbundle[alternate]{} & \qwbundle[alternate]{} & \qwbundle[alternate]{} & \qwbundle[alternate]{} & \qwbundle[alternate]{} \\
 \lstick[wires=1]{$\ket{0^T}_{\RegB}$} \qwbundle[alternate]{} & \qwbundle[alternate]{} &  \qwbundle[alternate]{}  & \qwbundle[alternate]{} & \qwbundle[alternate]{} & \qwbundle[alternate]{} & \qwbundle[alternate]{} 
\end{quantikz}
\caption{An example of an $\EQPTC$ circuit.}\label{fig:cr-eqpt-simple}
\end{figure}

\noindent

\iffalse
\begin{definition}[informal]
   Every computation of the following form is $\EQPTC$: for the coherent implementation $U$ on $\RegA \otimes \RegW \otimes \RegB$ of some $\EQPTM$ computation and polynomial-size quantum circuits\footnote{Our actual definition is for \emph{uniform} computation, so $(U, C_1, C_2, C_3)$ will have a uniform description.} $C_1,C_2,C_3$ on $\RegA \otimes \RegX$ for some register $\RegX$:
   \begin{inparaenum}[(1)]
        \item initialize $\RegW \otimes \RegB$ to $\ket{0}_{\RegW} \ket{0}_{\RegB}$;
        \item run $C_1$;
        \item apply $U$;
        \item run $C_2$;
        \item apply $U^{\dagger}$;
        \item run $C_3$.
   \end{inparaenum}
   The result of the computation is the residual state on $\RegX$.
\end{definition}
\fi 

We discuss and motivate the definition of $\EQPTC$ in detail in \cref{sec:tech-eqpt}. For now, we note two key properties of the model. First, the ability to apply $U^{\dagger}$ is what enables us to circumvent the \cite{FOCS:CCLY21} impossibility for $\EQPTM$. Second, we show a result analogous to the statement that any expected \emph{classical} polynomial time computation can be truncated to \emph{fixed} polynomial-time with small error.
\begin{lemma}[informal, see \cref{lemma:truncation}]\label{lemma:intro-truncation}
    Any $\EQPTC$ computation can be approximated with $\varepsilon$ accuracy by a quantum circuit of size $\poly(\secp,1/\varepsilon)$.
\end{lemma}

Importantly, this lemma ensures that $\EQPTC$ computations cannot break post-quantum polynomial hardness assumptions (unless the assumptions are false). We also note that this lemma implies that black-box zero-knowledge with $\EQPTC$ simulation implies $\varepsilon$-zero-knowledge with strict quantum polynomial time simulation.\footnote{We also note that all of our simulators can be truncated to run in time $\mathsf{poly}(\secp) \cdot 1/\eps$ and achieve $\eps$-ZK, matching the $\eps$-dependence of the classical simulator.}

\paragraph{What does this mean for post-quantum zero knowledge?} Since the introduction of zero-knowledge protocols~\cite{STOC:GolMicRac85}, expected polynomial-time simulation has been the \emph{de facto} model for classical zero knowledge. Although expected polynomial-time simulators cannot actually be run ``in real life'' (to negligible accuracy), the security notion captures negligible accuracy simulation in a computational model that \emph{cannot break polynomial hardness assumptions}. Moreover, since~\cite{STOC:BarLin02} rules out \emph{strict} polynomial-time simulation for constant-round protocols, expected polynomial-time simulation captures the strongest\footnote{Other models of efficient simulation \cite{Levin,eprint:Klooss20} have been proposed (in the classical setting), but they are \emph{relaxations} of expected polynomial time simulation. It should be possible to define similar relaxations in the quantum setting, but we focus on obtaining an analogue to ``standard'' expected polynomial time simulation.} provable zero-knowledge properties of many fundamental protocols such as quadratic non-residuosity~\cite{STOC:GolMicRac85}, graph non-isomorphism~\cite{FOCS:GolMicWig86}, Goldreich-Kahan~\cite{JC:GolKah96}, and Feige-Shamir~\cite{STOC:FeiSha90}.

The combination of our positive results and the~\cite{FOCS:CCLY21} negative results transports this state of affairs entirely to the post-quantum setting. In particular, the conclusion we draw from the~\cite{FOCS:CCLY21} negative result is that one must go beyond the $\EQPTM$ model in order to find the right quantum analogue of classical expected polynomial-time zero knowledge simulation. We propose $\EQPTC$ to be that quantum analogue.

\vspace{10pt}

With this discussion in mind, we proceed to describe our results on post-quantum zero-knowledge and extraction in more detail.

\subsection{Results on Zero Knowledge} 
Our main results regarding post-quantum zero knowledge are as follows. First, we show that the \cite{FOCS:GolMicWig86} graph non-isomorphism protocol is zero knowledge against quantum verifiers. 

\begin{theorem}\label{thm:szk}
    The \cite{FOCS:GolMicWig86} $4$-message proof system for graph non-isomorphism is a post-quantum statistical zero knowledge proof system. The zero-knowledge simulator is black-box and runs in $\EQPTC$.

\end{theorem}

The \cite{FOCS:GolMicWig86} GNI protocol follows a somewhat general template using instance-dependent commitments \cite{BelMicOst90,JC:ItoOhtShi97,C:MicVad03}; we believe \cref{thm:szk} should extend to other instantiations of this paradigm (e.g. for lattice problems). 

With some additional work, we use similar techniques to show how to instantiate the ``extract-and-simulate'' paradigm of Feige-Shamir \cite{STOC:FeiSha90} in the post-quantum setting. 

\begin{theorem}\label{thm:feige-shamir}
    Assuming super-polynomially secure non-interactive commitments, a particular instantiation of the \cite{STOC:FeiSha90} $4$-message argument system for $\mathsf{NP}$ is (sound and) zero-knowledge against quantum verifiers. The zero-knowledge simulator is black-box and runs in $\EQPTC$. 
\end{theorem}

Finally, using a different approach, we show that the Goldreich-Kahan \cite{JC:GolKah96} zero-knowledge proof system remains ZK against quantum adversaries. 

\begin{theorem}\label{thm:gk}
    When instantiated using a collapse-binding and statistically-hiding commitment scheme, the \cite{JC:GolKah96} protocol is zero-knowledge with a black-box $\EQPTC$ simulator. 
\end{theorem}

As a bonus, the simulator we construct in \cref{thm:gk} bears a strong resemblance to the \emph{classical} Goldreich-Kahan simulator, giving a clean conceptual understanding of constant-round zero knowledge in the quantum setting. 

\subsection{Results on Extraction}

As alluded to in the introduction, \cref{thm:szk} and \cref{thm:feige-shamir} are proved using new results on post-quantum \emph{extraction}. We achieve ``undetectable extraction'' under the following definition of a \emph{state-preserving proof of knowledge}.\footnote{This is a quantum analogue of \emph{witness-extended emulation}~\cite{CCC:BarGol02}. Our definition is also similar to a definition appearing in~\cite{C:AnaChuLap21}, although they only consider the setting of \emph{statistical} state preservation.}

\begin{definition}\label{def:state-preserving-extraction}
   An interactive protocol $\Pi$ is defined to be a \textdef{state-preserving argument (resp. proof) of knowledge} if there exists an extractor $\mathsf{Ext}^{(\cdot)}$ with the following properties:
   
   \begin{itemize}
        \item \textbf{Syntax}: For any quantum algorithm $P^*$ and auxiliary state $\ket{\psi}$, $\mathsf{Ext}^{P^*, \ket{\psi}}$ outputs a protocol transcript $\tau$, prover state $\ket{\psi'}$, and witness $w$. 
        \item \textbf{Extraction Efficiency}: If $P^*$ is a QPT algorithm, $E^{P^*(\cdot), \ket{\psi}}$ runs in expected quantum polynomial time ($\EQPTC$).
        \item \textbf{Extraction Correctness}: the probability that $\tau$ is an accepting transcript but $w$ is an invalid $\mathsf{NP}$ witness is negligible.
        \item \textbf{State-Preserving}: the pair $(\tau, \ket{\psi'})$ is computationally (resp. statistically) indistinguishable from a transcript-state pair $(\tau^*, \ket{\psi^*})$ obtained through an honest one-time interaction with $P^*(\cdot, \ket{\psi})$ (where $\ket{\psi^*}$ is the prover's residual state). 
   \end{itemize}
\end{definition}

Proofs/arguments of knowledge are typically used (rather than just sending an $\mathsf{NP}$ witness) to achieve either \emph{succinctness} \cite{STOC:Kilian92} or \emph{security against the verifier} (e.g., witness indistinguishability) \cite{FOCS:GolMicWig86,STOC:FeiSha90}. We show that standard $3$- and $4$-message protocols in both of these settings are state-preserving proofs/arguments of knowledge. 

\begin{theorem}[State-preserving succinct arguments] \label{thm:succinct-state-preserving}
Assuming collapsing hash functions exist, there exists a 4-message public-coin state-preserving succinct argument of knowledge for $\mathsf{NP}$.
\end{theorem}

For witness indistinguishability (WI), we have three related constructions achieving slightly different properties under different computational assumptions. 

\begin{theorem}[State-preserving WI arguments] \label{thm:state-preserving-wi}
Assuming collapsing hash functions or \emph{super-polynomially secure} one-way functions, there exists a 4-message public-coin state-preserving witness-indistinguishable argument (in the case of collapsing)/proof (in the case of OWFs) of knowledge. Assuming \emph{super-polynomially secure} non-interactive commitments, there exists a 3-message PoK achieving the same properties.
\end{theorem}

In fact, as we will explain in the technical overview, we give explicit conditions under which any proof/argument of knowledge is also state-preserving.

One special case of \cref{thm:state-preserving-wi} that we would like to highlight is that of \emph{extractable commitments} \cite{FOCS:PraRosSah02,TCC:PasWee09}. An extractable commitment scheme $\mathsf{ExtCom}$ is a commitment scheme with the property that a committed message $m$ can be extracted given black-box access to an adversarial sender (provided that the adversary is sufficiently convincing). Analogously to the setting of proofs-of-knowledge, we consider ``state-preserving'' extractable commitments (see, e.g., \cite{STOC:BitShm20,EC:GLSV21,C:BCKM21b}), in which the extractor must simulate the entire view of the adversarial committer in addition to extracting the message. This variant of extractable commitments is quite natural; for example, it is exactly the property necessary to prove the post-quantum security of the \cite{TCC:Rosen04} zero-knowledge proof system for $\mathsf{NP}$. An immediate corollary of \cref{thm:state-preserving-wi} is a new construction of state-preserving extractable commitments. 

\begin{corollary}[Extractable commitments]
   Assuming super-polynomially secure non-interactive commitments, there exists a $3$-message public-coin post-quantum statistically-binding \emph{extractable commitment scheme}. Assuming super-polynomially secure one-way functions, there exists a $4$-message scheme with the same properties. Finally, assuming (polynomially secure) collapsing hash functions, there exists a $4$-message public-coin collapse-binding extractable commitment scheme. 
\end{corollary}

We leave open the problem of using these techniques to achieve a statistically-binding extractable commitment scheme from polynomial assumptions.

More generally, we expect our state-preserving extraction results to be useful for future applications, both in the context of zero-knowledge and beyond.

%% file: 2-techoverview.tex
\section{Technical Overview}

In this section, we describe our techniques for proving our results on state-preserving extraction (\cref{thm:succinct-state-preserving,thm:state-preserving-wi}) and post-quantum zero knowledge (\cref{thm:szk}, \cref{thm:feige-shamir}, and \cref{thm:gk}). Finally, we discuss related work in \cref{sec:related-work}.

\subsection{Defining Expected Quantum Polynomial Time Simulation}
\label{sec:tech-eqpt}

In order to clearly present our results on zero knowledge, we begin with a detailed discussion of our model of expected quantum polynomial time simulation and how it relates to the \cite{FOCS:CCLY21} impossibility result.

\paragraph{Why is EQPT simulation hard to define?}
Recall from \cref{sec:intro-creqpt} that \cite{FOCS:CCLY21} rules out zero-knowledge simulators in a class of computations that we formalize as measured-runtime expected quantum polynomial-time ($\EQPTM$). An $\EQPTM$ computation takes as input a state $\ket{\psi}_{\RegA}$, initializes a large ancilla/workspace register $\ket{0}_{\RegW, \RegB}$ and state register $\ket{q_0}_{\RegQ}$ (where $\ket{q_0}$ denotes the initial state of a quantum Turing machine), then repeatedly applies some fixed transition unitary $U_{\delta}$ to $\RegA \tensor \RegW \tensor \RegB \tensor \RegQ$. After each application of $U_{\delta}$, $\RegQ$ is measured (applying some $(\Pi_f, \Id - \Pi_f)$) to determine if the computation is in the ``halt state'' $\ket{q_f}$; the computation halts if the outcome of this measurement is $1$. A computation is $\EQPTM$ if the expected number of steps before halting is polynomial for all inputs $\ket{\psi}_{\RegA}$.

Our $\EQPTM$ definition is based on the definition of a quantum Turing machine (QTM) given in the seminal work of Deutsch~\cite{Deutsch85} (though we use a halt state \cite{Ozawa98a} in place of Deutsch's halt qubit). Note that the operation of a QTM is unitary \emph{except} for the measurement of whether the machine has halted. The validity of this ``halting scheme'' was the subject of some debate in a sequence of later works~\cite{Myers97,Ozawa98a,LindenP98,Ozawa98b}.

While the particulars of this debate are not so important here, there was a clear message: the reversibility of a QTM implies that the runtime of any QTM computation is \emph{always} effectively measured, even if there is no explicit monitoring of the halt state. Intuitively, this is because a QTM that has halted must, when reversed, know when to ``un-halt''; this requires counting the number of computation steps since the machine halted.

It was observed by \cite{LindenP98} that this prevents ``useful interference'' between branches of a QTM computation with different runtimes. That is, each branch of the computation is entangled with a description of its runtime, which prevents the branches from interfering with one another. Because interference is crucial in the design of efficient quantum algorithms, this is considered a major drawback of the QTM model. The now-standard definitions of \emph{efficient} quantum computation \cite{BV97,BBBV97} deliberately avoid this problem by restricting quantum Turing machines to have a \emph{fixed} runtime; these QTMs are effectively uniform quantum circuit families.

This phenomenon underpins the \cite{FOCS:CCLY21} impossibility result. Both in the classical \cite{STOC:BarLin02} and quantum \cite{FOCS:CCLY21} settings, there do not exist \emph{strict} polynomial time black-box simulators for constant-round protocols. It follows that such a simulator must have a variable runtime. By the observation of \cite{LindenP98}, simulation branches with different runtimes do not interfere. \cite{FOCS:CCLY21} leverage this by designing an adversary which can \emph{detect} this absence of interference.

\paragraph{Can we avoid measuring the runtime?}
The above discussion suggests that the $\EQPTM$ model (i.e., quantum Turing machines in Deutsch's model~\cite{Deutsch85} with expected polynomial runtime) may not capture arbitrary efficient quantum computation. In particular, we ask whether it is possible to formalize a model in which the runtime is \emph{not} measured. Such a model could potentially avoid the \cite{FOCS:CCLY21} impossibility result.

Our solution is to formalize computations in which the runtime of an $\EQPTM$ subcomputation is left \emph{in superposition} and can later be \emph{uncomputed}. To describe our formalism in more detail, we first briefly discuss coherent computation.

\paragraph{Coherent computation.}
It is well known that any quantum operation $\Phi$ on a state $\ket{\psi}$ can be realized in three steps:
\begin{inparaenum}[(1)]
	\item prepare some ancilla qubits in a fixed state $\ket{0}$;
	\item apply a unitary operation $U_{\Phi}$ to both $\ket{\psi}$ and the ancilla;
	\item discard (trace out) the ancilla.
\end{inparaenum}
We refer to $U_{\Phi}$ as a \emph{unitary dilation} of $\Phi$. $U_{\Phi}$ is not uniquely determined by $\Phi$, but all such dilations are related by an isometry acting only on the ancilla system. 

Since an $\EQPTM$ computation is a quantum operation, it has a unitary dilation. In fact, we can choose a unitary dilation with a natural explicit form that we call a ``coherent implementation'', as shown in \cref{fig:coherent-qtm}.

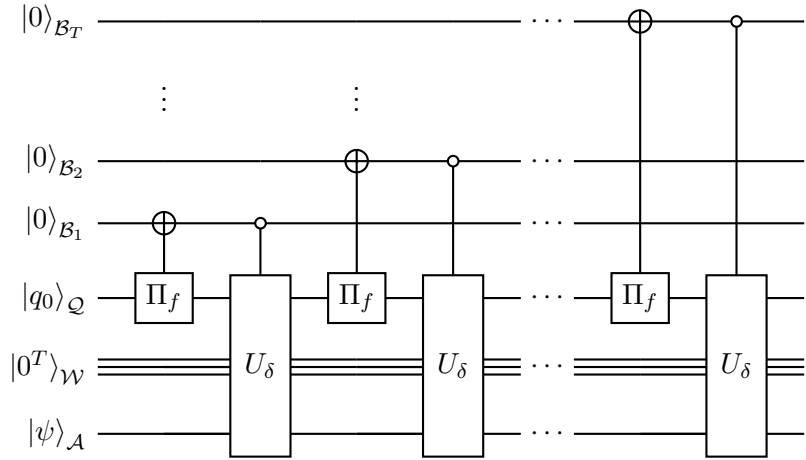
\begin{figure}[tbhp]
\centering
\begin{quantikz}
\lstick[wires=1]{$\ket{0}_{\RegB_T}$} & \qw & \qw & \qw & \qw & \qw\ \cdots\  & \targ{}\vqw{4} & \octrl{4} & \qw \\
& \vdots & & \vdots \\
\lstick[wires=1]{$\ket{0}_{\RegB_2}$} & \qw & \qw & \targ{}\vqw{2} & \octrl{2} & \qw\ \cdots\  & \qw & \qw & \qw \\
\lstick[wires=1]{$\ket{0}_{\RegB_1}$} & \targ{}\vqw{1} & \octrl{1} & \qw & \qw & \qw\ \cdots\  & \qw & \qw & \qw \\
\lstick[wires=1]{$\ket{\TStateInit}_{\RegState}$} & \gate{\Pi_f} & \gate[wires=3][0.8cm]{U_{\delta}} & \gate{\Pi_f} & \gate[wires=3][0.8cm]{U_{\delta}} & \qw \ \cdots\  & \gate{\Pi_f} & \gate[wires=3][0.8cm]{U_{\delta}} & \qw \\
 \lstick[wires=1]{$\ket{0^T}_{\RegW}$} \qwbundle[alternate]{} & \qwbundle[alternate]{} &  \qwbundle[alternate]{}  & \qwbundle[alternate]{} & \qwbundle[alternate]{} &  \qwbundle[alternate]{} \ \cdots\ & \qwbundle[alternate]{} & \qwbundle[alternate]{} & \qwbundle[alternate]{} \\
 \lstick{$\ket{\psi}_{\RegA}$} \qw & \qw & \qw & \qw & \qw & \qw\ \cdots\ & \qw & \qw & \qw &
\end{quantikz}
\caption{A coherent implementation $U$ (unitary dilation) of a quantum Turing machine with transition function $\delta$. The open circles indicate that the $i$th $U_{\delta}$ is applied when $\RegB_i$ contains $\ket{0}$. See \cref{sec:qtms} for more details.}\label{fig:coherent-qtm}
\end{figure}

Note that since we think of $T$ as being exponentially large, the unitary $U$ (corresponding to~\cref{fig:coherent-qtm}) is of exponential size. However, as long as the ancilla $\RegW \otimes \RegB$ (where $\RegB \eqdef \RegB_1 \otimes \cdots \otimes \RegB_T$) is initialized to zero and $\RegQ$ is initialized to $\ket{q_0}$, the effect of $U$ on $\RegA$ is identical to the original $\EQPTM$ computation. Indeed, the only difference from the original computation is that the runtime is written (in unary) on $\RegB$ and left in superposition. This means, in particular, that circuits making a single black-box query to a coherent implementation $U$ of an $\EQPTM$ computation (and that cannot otherwise access $\RegB$) can only perform $\EQPTM$ computations.

\paragraph{Our formalism: coherent-runtime EQPT.}
The advantage of moving to coherent implementations is that, unlike the original computation, $U$ has an \emph{inverse} $U^{\dagger}$. A coherent-runtime EQPT computation is allowed to invoke both $U$ and $U^{\dagger}$ in a restricted way, as we specify next.

\begin{definition}[Coherent-runtime EQPT (informal)]
	\label{def:cr-eqpt-informal}
	A computation on a register $\RegX$ is coherent-runtime EQPT ($\EQPTC$) if it can be implemented by a procedure that can perform the following operations any polynomial number of times:
	\begin{enumerate}[nolistsep,label=(\arabic*)]
		\item \label{cr-eqpt-polyckt} apply any polynomial-size quantum circuit to $\RegX$;
		\item \label{cr-eqpt-forward} initialize fresh ancillas $\RegW^{(i)} \otimes \RegB^{(i)}$ to zero and $\RegQ^{(i)}$ to $\ket{q_0^{(i)}}$ and apply a coherent implementation $U_i$ of an $\EQPTM$ computation to $\RegW^{(i)} \otimes \RegB^{(i)}\tensor \RegQ^{(i)}$ and a subregister of $\RegX$;
		\item \label{cr-eqpt-inverse} apply $U_i^{\dagger}$ to ancillas $\RegW^{(i)} \otimes \RegB^{(i)} \tensor \RegQ^{(i)}$ and a subregister of $\RegX$, then discard $(\RegW^{(i)}, \RegB^{(i)}, \RegQ^{(i)})$. For each $U_i$, this operation may occur at any time after performing \ref{cr-eqpt-forward} with respect to $U_i$.
	\end{enumerate}
	The output of the computation is the residual state on $\RegX$.
\end{definition}
Note that, because all unitary dilations are equivalent up to local isometry, the map computed by an $\EQPTC$ computation is independent of the particular implementation of $U_i$.

\paragraph{What is the runtime of an $\EQPTC$ computation?}
While procedures performing only \ref{cr-eqpt-polyckt} and \ref{cr-eqpt-forward} are clearly efficient (they are $\EQPTM$), the efficiency of performing $U_i^\dagger$ in \ref{cr-eqpt-inverse} is less immediate. We analyze this in two ways: 

\begin{itemize}
    \item We prove (\cref{lemma:truncation}) that any $\EQPTC$ computation has strict polynomial-time approximations (obtained by simultaneously truncating each $U_i$ and $U_i^\dagger$ to the same fixed runtime). This tells us that $\EQPTC$ algorithms do not implement ``inefficient'' computations. 
    \item We give a natural interpretation of ``expected runtime'' under which the expected runtime of $U^\dagger$ (as applied in \ref{cr-eqpt-inverse}) is equal to the expected runtime of $U$.
\end{itemize} 

Together, these give us a motivated definition of the expected runtime of an $\EQPTC$ computation.

\cref{lemma:truncation} is proved in \cref{sec:eqpt}. In this overview, we focus on the expected runtime interpretation. For simplicity, we consider the basic form of $\EQPTC$ computation depicted in \cref{fig:cr-eqpt-simple}; we assume in addition that $C_1 = C_3 = \Id$ and that the computation always halts in at most $T$ steps.

Let $U^{(t)}$ be the truncation of $U$ to just after the $t$-th controlled application of $U_{\delta}$. Observe that after applying $U$ followed by $C_2$ on input state $\ket{\phi}$, the state of the system can be described as a superposition over $t$ of the applications of the $U^{(t)}$:
\[ \ket{\Psi}
= \sum_{t=0}^{T} \alpha_t \ket{0^t 1^{T-t}}_{\RegB} \ket{\TStateFinal}_{\RegQ} \ket{\phi_t}_{\RegW,\RegX}
= \sum_{t=0}^{T} \alpha_t (\Id \otimes C_2) U^{(t)} \ket{0^t 1^{T-t}}_{\RegB} \ket{\TStateInit}_{\RegQ} \ket{0}_{\RegW} \ket{\phi}_{\RegX}
\]
for some states $\ket{\phi_t}$ and $\alpha_t \in \C$ where $|\alpha_t|^2$ is the probability that the $\EQPTM$ computation halts in $t$ steps. The latter equality holds because if the computation halts at step $t$, the effect of the last $T-t$ steps of $U$ is only to flip $\RegB_{t+1},\ldots,\RegB_T$ from $\ket{0}$ to $\ket{1}$ (and $C_2$ does not act on $\RegB$).

We emphasize that since $U$ is a coherent implementation of an $\EQPTM$ computation, we know that
\[ \sum_t |\alpha_t|^2 \cdot t = \poly(\secp),
\]
as the left-hand side is equal to the expected runtime of $U$ as an $\EQPTM$ computation. 

Now, for any state $\ket{\psi}$ on $\RegW \otimes \RegX$, we can also express an application of $U^\dagger$ to $\RegB \tensor \RegQ \tensor \RegW \tensor \RegX$ entirely in terms of the unitary $U^{(t)}$, where $t$ is the contents of the $\RegB$ register. Specifically, we have
\[ U^{\dagger} \ket{0^t 1^{T-t}}_{\RegB} \ket{\TStateFinal}_{\RegQ} \ket{\psi} = (U^{(t)})^{\dagger} \ket{0^T}_{\RegB} \ket{\TStateFinal}_{\RegQ} \ket{\psi} \]
because the effect of the first $T-t$ steps of $U^{\dagger}$ on this state is only to flip $\RegB_{t+1},\ldots,\RegB_T$ from $\ket{1}$ to $\ket{0}$. As a result, the final state of the system (after the entire $\EQPTC$ computation) is
\[ U^{\dagger} \ket{\Psi} = \sum_{t=0}^{T} \alpha_t U^{\dagger} (\Id \otimes C_2) U^{(t)} \ket{0^t 1^{T-t}}_{\RegB} \ket{\phi}
= \sum_{t=0}^{T} \alpha_t (U^{(t)})^{\dagger} (\Id \otimes C_2) U^{(t)} \ket{0^T}_{\RegB} \ket{\phi}. \]
We can interpret this to mean that within the ``branch'' of the superposition where $U$ ran in time $t$, the running time of $U^{\dagger}$ is also $t$, even if an arbitrary computation $C_2$ has been applied to $\RegX$ in between the applications of $U$ and $U^\dagger$. This gives an intuitive explanation for how $\EQPTC$ computations are efficient: they simply compute an $(\alpha_t)_t$-superposition over branches in which $U$ and $U^\dagger$ together ran for $2t$ steps, such that the expectation $\sum_t |\alpha_t|^2 \cdot t$ is polynomial! Curiously, the \cite{LindenP98} reversibility issue indicates that such a computation cannot be implemented by an $\EQPTM$ quantum Turing machine, which is what necessitates our new $\EQPTC$ definition. 

With all of this as motivation, we define the expected running time of an $\EQPTC$ computation of the form $(U, C_1 = \Id, C_2, C_3 = \Id)$ to be the appropriate linear combination of the branch runtimes, which is
\[ \sum_t |\alpha_t|^2 \cdot (2t + \mathsf{time}(C_2)) =  2\cdot \mathsf{time}(U) + \mathsf{time}(C_2),
\]
where $\mathsf{time}(U)$ is the expected running time of $U$ as an $\EQPTM$ computation and $\mathsf{time}(C_2)$ is the (strict) running time of $C_2$. \cref{lemma:truncation} (whose proof makes use of this analysis) provides additional justification for this definition.

\paragraph{Extension to multiple $U_i$.} Everything we have discussed so far extends to the general case of \cref{def:cr-eqpt-informal}. However, we emphasize that the above analysis crucially relies on the ancilla being well-formed. This is the reason that $\EQPTC$ algorithms have restricted access to $U_i,U_i^{\dagger}$: removing any of these restrictions could lead to applying these operations on malformed ancillas. Indeed, one can show that allowing an algorithm to apply $U_i,U_i^{\dagger},U_i$ to the same ancilla register would enable it to perform exponential-time computations.

\vspace{10pt}

Having established our computational model for simulation/extraction, we now give a detailed overview of our simulation and extraction techniques.

\subsection{Post-Quantum ZK for \cite{FOCS:GolMicWig86} and~\cite{STOC:FeiSha90} from Guaranteed Extraction}

The central idea behind our proofs of post-quantum ZK for the \cite{FOCS:GolMicWig86} GNI protocol (\cref{thm:szk}) and a variant of the~\cite{STOC:FeiSha90} protocol for $\NP$ (\cref{thm:feige-shamir}) is \emph{state-preserving extraction} (\cref{def:state-preserving-extraction}). Given a state-preserving extractor of the appropriate ``one-out-of-two graph isomorphism'' subroutine, proving the post-quantum ZK for the \cite{FOCS:GolMicWig86} GNI protocol (\cref{thm:szk}) follows easily, as simulating a cheating verifier immediately reduces to performing a state-preserving extraction of the verifier's (uniquely determined) bit $b$ such that $H\simeq G_b$. Proving post-quantum ZK for the~\cite{STOC:FeiSha90} protocol (\cref{thm:feige-shamir}) is more complicated because the Feige--Shamir protocol is a concurrent composition of two different protocols; we refer the reader to \cref{sec:feige-shamir} for details on its analysis. 

In this subsection, we show that state-preserving extraction reduces to a related task that we call \emph{guaranteed extraction}; achieving the latter will be the focus of \cref{sec:tech-overview-hpe}.

Consider a $3$-message\footnote{Throughout our discussion of proofs of knowledge, we focus on the case of $3$- and $4$-message protocols. We sometimes ignore the first verifier message $\vk$ in a 4-message protocol for notational convenience.} public coin classical proof of knowledge $(P_{\Sigma}, V_{\Sigma})$ satisfying \emph{special soundness}:\footnote{This particular special soundness assumption is also for convenience; we later describe generalizations of special soundness for which we have results.} for any prover first message $a$ and any \emph{pair} of accepting transcripts $(a, r, z), (a, r', z')$ on different challenges $r \neq r'$, it is possible to extract a witness from $(a, r, z, r', z')$. For any such protocol, in the classical setting, it is possible to extract a witness from a cheating prover $P^*$ as follows:

\begin{itemize}
    \item Given a cheating prover $P^*$, the extractor first generates a single transcript $(a,r,z)$ by running $P^*$ to obtain $a$, and then running it on a random $r$ to get $z$. If the transcript is rejecting, the extractor gives up.
    \item If the transcript is accepting, the extractor rewinds $P^*$ to the point after $a$ was sent, and then repeatedly sends i.i.d. challenges $r_1, r_2, \hdots$ until $P^*$ produces \emph{another} accepting transcript. 
\end{itemize}
As long as the prover has significantly greater than $2^{-\secp}$ probability of convincing the verifier, the second accepting transcript $(a, r', z')$ produced will satisfy $r \neq r'$ with all but negligible probability, and thus a witness can be computed. In other words, this extractor \emph{guarantees} (with all but negligible probability) that a witness is extracted conditioned on an initial accepting execution. Moreover, for \emph{any} efficient $P^*$, the expected runtime of this procedure is $\poly(\secp)$, since if $P^*$ (with some fixed random coins) is convincing with probability $p$, the expected number of rewinds in this procedure is $\frac 1 p$ and thus the overall expected number of rewinds is $p \cdot \frac 1 p = 1$.

In the quantum setting, one might hope for a similar ``guaranteed'' extractor, but prior works~\cite{EC:Unruh12,EC:Unruh16,FOCS:CMSZ21} fail to achieve this. Indeed, \cite[Page 32]{EC:Unruh12} explicitly asks whether something of this nature is possible. 

Our first idea is to abstractly define a quantum analogue of this ``guaranteed'' extraction property and show that under certain conditions, it generically implies state-preserving extraction. Since the classical problem can only be solved in \emph{expected} polynomial time, there is again an ambiguity in what the quantum efficiency notion should be. However, it turns out that there is no \cite{FOCS:CCLY21}-type impossibility result for the problem of guaranteed extraction, so we demand the stronger $\EQPTM$ extraction efficiency notion.

\begin{definition}\label{def:high-probability-extraction}
$(P_{\Sigma}, V_{\Sigma})$ is a post-quantum proof of knowledge with \emph{guaranteed extraction} if it has an extractor $\Extract^{P^*}$ of the following form.

\begin{itemize}[noitemsep]
    \item $\Extract^{P^*}$ first runs the cheating prover $P^*$ to generate a (classical) first message $a$.
    \item $\Extract^{P^*}$ runs $P^*$ coherently on the superposition $\sum_{r \in R} \ket{r}$ of all challenges to obtain a superposition $\sum_{r,z} \alpha_{r, z} \ket{r, z}$ over challenge-response pairs.\footnote{In general, the response $z$ will be entangled with the prover's state; here we suppress this dependence.} 
    \item $\Extract^{P^*}$ then computes (in superposition) the verifier's decision $V(x, a, r, z)$ and measures it. If the measurement outcome is $0$, the extractor gives up. 
    \item If the measurement outcome is $1$, run some quantum procedure $\FindWitness^{P^*}$ that outputs a string $w$. 
    \end{itemize}

We require that the following two properties hold. 

\begin{itemize}[noitemsep]
    \item \textbf{Correctness (guaranteed extraction):} The probability that the initial measurement returns $1$ but the output witness $w$ is invalid is $\negl(\secp)$.
    \item \textbf{Efficiency:} For any QPT $P^*$, the procedure $\Extract^{P^*}$ is in $\EQPTM$. 
\end{itemize}
\end{definition}
\noindent We claim that under suitable conditions, this kind of guaranteed extraction \emph{generically} implies state-preserving extraction, where the extractor will be $\EQPTC$ rather than $\EQPTM$. We describe the simplest example of these conditions: when the $\mathsf{NP}$ language itself is in $\mathsf{UP}$ (i.e. witnesses are unique).

\begin{lemma}[see \cref{lemma:state-preserving-high-probability}]\label{lemma:tech-overview-state-preserving-high-probability}
  If $(P_{\Sigma}, V_{\Sigma})$ is a post-quantum proof of knowledge with guaranteed extraction for a language with unique witnesses, then $(P_{\Sigma}, V_{\Sigma})$ is a state-preserving proof of knowledge with $\EQPTC$ extraction. 
\end{lemma}
\noindent \cref{lemma:tech-overview-state-preserving-high-probability} can be extended to higher generality. For example, informally:

\begin{enumerate}
    \item We can also extract ``partial witnesses'' that are uniquely determined by the instance $x$.
    \item We can extract undetectably when the first message $a$ ``binds'' the prover to a single witness in the sense that the guaranteed extractor will only output this one witness (even if many others exist). 
    \item This can also be extended to certain protocols whose first messages are informally ``collapse-binding'' \cite{EC:Unruh16} to the witness.
\end{enumerate} 

These generalizations are formalized in \cref{sec:state-preserving} using the notion of a ``witness-binding protocol'' (\cref{def:witness-binding}). In this overview, we give a proof for the ``unique witness'' setting.

\begin{proof}[Proof sketch]

Let $\Extract^{P^*}$ be a post-quantum guaranteed extractor with associated subroutine $\FindWitness^{P^*}$. We will present an $\EQPTC$ extractor $\overline{\Extract}^{P^*}$ that has the form of an $\EQPTC$ computation (see~\cref{fig:cr-eqpt-simple}) where the unitary $U$ is a coherent implementation of $\FindWitness^{P^*}$.

\begin{remark} This is an oversimplification of our real state-preserving extractor. In particular, $\Extract^{P^*}$ as described in this overview does not fit the $\EQPTC$ model because $\FindWitness^{P^*}$ is not necessarily an $\EQPTM$ computation --- its running time is only expected polynomial when viewed as a subroutine of $\Extract^{P^*}$, which runs $\FindWitness^{P^*}$ with some probability (which may be negligible) and moreover, only runs it on inputs consistent with the verifier decision $V(x,a,r,z) = 1$. In \cref{sec:state-preserving}, we formally demonstrate that our state-preserving extractor is $\EQPTC$ by showing that it can be written in the form of~\cref{fig:cr-eqpt-simple} where the unitary $U$ is a coherent implementation of the $\EQPTM$ procedure
$\Extract^{P^*}$. 
\end{remark}

\noindent Our (simplified) $\EQPTC$ extractor $\overline{\Extract}^{P^*}$ is defined as follows. 

\begin{itemize}
    \item Given $P^*$, generate a first message $a$ and superposition $\sum_{r, z} \alpha_{r, z} \ket{r, z}$ as in $\Extract^{P^*}$. 
    \item Compute the verifier's decision bit $V(x, a, r, z)$ in superposition and then measure it. If the measurement outcome is $0$, measure $r, z$ and terminate, outputting $(a, r, z, w=\bot)$ along with the current prover state.
    \item If the measurement outcome is $1$, let $\ket{\psi}_{\RegH}$ denote the current prover state. For simplicity, assume that $\ket{\psi}_{\RegH}$ includes the superposition over $(r,z)$ and space to write the extracted witness. The next steps are: 
    
    \begin{itemize}
        \item Run $U$ on input $\ket{\psi}_{\RegH} \tensor \ket{0}_{\RegB, \RegW}$.
        \item Measure the sub-register of $\RegH$ containing the witness $w$. 
        \item Run $U^\dagger$.
        \item Measure the sub-register of $\RegH$ containing the current transcript $r, z$.
        \item Return $(a, r, z, w)$ and the residual prover state (i.e., the rest of $\RegH$). 
    \end{itemize}
\end{itemize}

Extraction correctness follows from the correctness of $\FindWitness^{P^*}$. Moreover, one can see that $\overline{\Extract}^{P^*}$ is state-preserving by considering two cases:

\begin{itemize}
    \item \textbf{Case 1:} The initial measurement returns $0$. In this case, the transcript $(r, z)$ is immediately measured, and the resulting (sub-normalized) state exactly matches the component of the post-interaction $P^*$ view corresponding to when the verifier rejects. 
    \item \textbf{Case 2:} The initial measurement returns $1$. In this case, the procedure $\FindWitness^{P^*}$ would output a valid witness with probability $1-\negl$, so the output register of $U (\ket{\psi}_{\RegA} \tensor \ket{0}_{\RegB, \RegW})$ contains a valid witness with probability $1-\negl$. Since we assumed that the language $L$ is in $\mathsf{UP}$, this witness register is actually \emph{deterministic}, so  measuring it is computationally (even statistically!) undetectable, and hence after applying $U^\dagger$ the resulting state $\ket{\psi'}$ is computationally indistinguishable from $\ket{\psi}$. Thus, the output of the extractor the measured witness $w$ along with a view that is computationally indistinguishable from the view of $P^*$ corresponding to when the verifier accepts.
\end{itemize}
This completes the proof sketch. \qedhere
 
\end{proof}

\paragraph{How do we apply \cref{lemma:tech-overview-state-preserving-high-probability}?}  We now describe how to instantiate $\Sigma$-protocols so that the reduction in \cref{lemma:tech-overview-state-preserving-high-probability} applies (see \cref{sec:state-preserving-examples}).

First, we note that the \emph{un-repeated} variants of standard proofs of knowledge \cite{FOCS:GolMicWig86,Blum86} are ``witness-binding'' in the informal sense of the generalization (2); an extractor run on such protocols will only output a witness consistent with the commitment string $a$. However, since the un-repeated protocols only have constant (or worse) soundness error, there is no guaranteed extraction procedure for them (even in the classical setting).

In order to obtain negligible soundness error, these protocols are typically repeated in parallel; in this case, we \emph{do} show guaranteed extraction procedures, but the protocols \emph{lose} the witness-binding property (2). This is because each ``slot'' of the parallel repetition may be consistent with a different witness, and the extractor has no clear way of outputting a canonical one. In this case, measuring the witness potentially disturbs the prover's state by collapsing it to be consistent with the measured witness, which would not happen in the honest execution.

We resolve this issue using \emph{commit-and-prove}. Given a generic $\Sigma$-protocol for which we have a guaranteed extractor, we consider a modified protocol in which the prover sends a (collapsing or statistically binding) commitment $\com = \Com(w)$ to its $\mathsf{NP}$-witness along with a $\Sigma$-protocol proof of knowledge of an opening of $\com$ to a valid $\mathsf{NP}$-witness. When the extractor $\overline{\Extract}^{P^*}$ of \cref{lemma:tech-overview-state-preserving-high-probability} is applied to this protocol composition, the procedure $\FindWitness^{P^*}$ (which is run coherently as $U$) actually obtains \emph{both} an $\NP$ witness $w$ \emph{and} an opening of $\mathsf{com}$ to $w$. Therefore, the collapsing property of $\mathsf{Com}$ says that $w$ can be measured undetectably. In other words, the commit-and-prove compiler enforces a computational uniqueness property sufficient for \cref{lemma:tech-overview-state-preserving-high-probability} to apply. It also turns out that the (original, unmodified) \cite{FOCS:GolMicWig86} graph-nonisomorphism protocol can be viewed as using this commit-and-prove paradigm,\footnote{The verifier sends an instance-dependent commitment \cite{BelMicOst90,JC:ItoOhtShi97,C:MicVad03} of a bit to the prover (which is perfectly binding in the proof of ZK) and demonstrates knowledge of the bit and its opening.} which is one way to understand the proof of \cref{thm:szk}.

Finally, we remark that this commit-and-prove compiler is the cause of the super-polynomial assumptions in \cref{thm:state-preserving-wi,thm:feige-shamir}. This is because in order to show that a commit-and-prove protocol remains witness-indistinguishable, it must be argued that the proof of knowledge does not compromise the hiding of $\mathsf{Com}$, which we only know how to argue by simulating the proof of knowledge in superpolynomial time (and assuming that $\mathsf{Com}$ is superpolynomially secure). This issue does not arise when $\mathsf{Com}$ is statistically hiding and the $\Sigma$-protocol is statistically witness-indistinguishable.

\subsection{Achieving Guaranteed Extraction}
\label{sec:tech-overview-hpe}

So far, we have reduced from state-preserving extraction to the problem of guaranteed extraction. We now describe how we achieve guaranteed extraction for a wide class of $\Sigma$-protocols. Informally, we require that the protocol satisfies two important properties in order to perform guaranteed extraction:

\begin{itemize}
    \item \textbf{Collapsing:} Prover responses can be measured undetectably provided that they are valid.
    \item \textbf{$k$-special soundness:} It is possible to obtain a witness given $k$ accepting protocol transcripts $(a, r_1, z_1, \hdots, r_k, z_k)$ with distinct $r_i$ (for the same first prover message $a$).
\end{itemize}
Both of these restrictions can be relaxed substantially\footnote{We highlight that the PoK subroutine in the \cite{FOCS:GolMicWig86} graph non-isomorphism protocol is \emph{not} collapsing; it is only collapsing onto its responses of $0$ challenge bits; however, it turns out that this property is still sufficient to obtain guaranteed extraction for the subroutine (see \cref{sec:examples,sec:obtaining-guaranteed-extraction}).} (see \cref{subsec:protocol-prelims,sec:gss,sec:obtaining-guaranteed-extraction} for more details), but we focus on this case for the technical overview. 

\begin{theorem}[See \cref{thm:high-probability-extraction}]\label{thm:tech-overview-high-probability}
   Any public-coin interactive argument satisfying collapsing and $k$-special soundness is a post-quantum proof of knowledge with guaranteed extraction (in $\EQPTM$). 
\end{theorem}

We consider \cref{thm:tech-overview-high-probability} to be an interesting result in its own right and expect it to be useful in future work. We now describe our proof of~\cref{thm:tech-overview-high-probability} over the course of several steps:
\begin{itemize}
    \item We begin by describing an abstract template that generalizes the~\cite{FOCS:CMSZ21} extraction procedure in~\cref{sec:overview-cmsz-template}. In this template, the extractor repeatedly (1) queries the adversary on i.i.d. random challenges and then (2) applies a ``repair procedure'' to restore the adversary's success probability.
    
    \item In~\cref{sec:tech-overview-first-attempt}, we describe a natural ``first attempt'' at guaranteed extraction based on the~\cite{FOCS:CMSZ21} template. 
    
    \item We then observe in~\cref{sec:tech-overview-not-eqpt} that the entire template is unlikely to achieve guaranteed extraction in expected polynomial time. Perhaps surprisingly (and unlike the classical setting), querying the adversary on i.i.d. challenges appears \emph{too slow} for this extraction task.

    \item In~\cref{sec:tech-overview-new-template}, we introduce a new extraction template in which the adversary is \emph{entangled} with a superposition of challenges, and the challenge is only measured once the adversary is guaranteed to give an accepting response. 
    
    \item While this new template is a promising idea, we are still far from achieving guaranteed extraction. For the rest of the overview (\cref{sec:tech-overview-no-speedup,sec:tech-overview-speedup,subsubsec:gk-issue}), we outline several technical challenges in instantiating this approach, eventually leading to our final extraction procedure and analysis.
    
    \end{itemize}

\subsubsection{An Abstract \cite{FOCS:CMSZ21} Extraction Template}\label{sec:overview-cmsz-template}
\cite{FOCS:CMSZ21} recently showed that protocols satisfying collapsing and $k$-special soundness are post-quantum proofs of knowledge. Unlike our setting of guaranteed extraction, the \cite{FOCS:CMSZ21} extractor $\Extract^{P^*}(x, \gamma)$ is given \emph{as advice} an error parameter $\gamma$ and and extracts from cheating provers $P^*$ (that may have some initial quantum state) that are convincing with probability $\gamma^* \geq \gamma$. The extractor's success probability is roughly $\frac \gamma 2$. %runtime is too hard to calculate

At a high level, our abstract template makes use of two core subroutines that we call $\Estimate$ and $\Transform$. We describe the correctness properties required of $\Estimate$ and $\Transform$ below, and also describe their particular instantiations in \cite{FOCS:CMSZ21}.

\paragraph{Jordan's lemma and singular vector algorithms.} Let $\BProj{\MeasA},\BProj{\MeasB}$ be projectors on a Hilbert space $\RegH$ with corresponding binary projective measurements $\MeasA = \BMeas{\BProj{\MeasA}}$ and $\MeasB = \BMeas{\BProj{\MeasB}}$. Recall that Jordan's lemma~\cite{Jordan75} states that $\RegH$ can be decomposed as a direct sum $\RegH = \bigoplus \RegS_j$ of two-dimensional invariant subspaces $\RegS_j$, where in each $\RegS_j$, the projectors $\BProj{\MeasA}$ and $\BProj{\MeasB}$ act as rank-one projectors $\JorKetBraA{j}{1}$ and $\JorKetBraB{j}{1}$.\footnote{There will also be one-dimensional subspaces, which we ignore in this overview since they can be viewed as ``degenerate'' two-dimensional subspaces.} The vectors $\JorKetA{j}{1}$ and $\JorKetB{j}{1}$ are also left and right singular vectors of $\BProj{\MeasA}\BProj{\MeasB}$ with singular value $\sqrt{p_j}$, where $p_j \eqdef \abs{\JorBraKetAB{j}{1}}^2$. This decomposition allows us to  define on $\RegH$ the projective measurement $\Jor = (\Pi^{\Jor}_j)$ onto the Jordan subspaces $\RegS_j$ (i.e., $\image(\Pi^{\Jor}_j) = \RegS_j$). For an arbitrary state $\ket{\psi}$, we define the \textdef{Jordan spectrum} of $\ket{\psi}$ to be the distribution of $p_j$ induced by $\Jor$.  

We will make use of procedures $\Estimate, \Transform$ satisfying the following properties. 

\begin{itemize}
    \item The Jordan subspaces $\RegS_j$ are invariant\footnote{We allow for decoherence, so we ask that every element of $\RegS_j$ is mapped to a \emph{mixed state} where every component is in $\RegS_j$.} under $\Estimate^{\MeasA,\MeasB}$ and $\Transform^{\MeasA,\MeasB}$. Equivalently, $\Estimate^{\MeasA,\MeasB}$ and $\Transform^{\MeasA,\MeasB}$ should \emph{commute} with $\Jor$.  This property is important for arguing about the output behavior of $\Estimate$ and $\Transform$ on arbitrary states.
    \item $\Estimate^{\MeasA,\MeasB}$: on input $\ket{S_j} \in \RegS_j$, output $p \approx p_j$; the residual state remains in $\RegS_j$.
    
    \item $\Transform^{\MeasA,\MeasB}$ maps each $\JorKetA{j}{1}$ to $\JorKetB{j}{1}$. We have no requirements on any other state in $\RegS_j$ except that it remains in $\RegS_j$.
    
\end{itemize}

    \cite{FOCS:CMSZ21} implement a version of $\Estimate^{\MeasA,\MeasB}$ (following~\cite{MarriottW05}) with $\eps$ accuracy by alternating $\MeasA$ and $\MeasB$ for $t = \poly(\secp)/\eps^2$ steps. The output is $p = d/(t-1)$ where $d$ is the number of occurrences of $b_i = b_{i+1}$ among the outcomes $b_1,b_2,\dots,b_t$. With probability $1-2^\secp$, we have $\abs{p - p_j} \leq \eps$.  They (implicitly) implement $\Transform^{\MeasA,\MeasB}$ by alternating measurements $\MeasA$ and $\MeasB$ back and forth until $\MeasB \rightarrow 1$, with an expected running time of $O(1/p_j)$ on $\RegS_j$.

\paragraph{The~\cite{FOCS:CMSZ21} Extractor.} We now use the abstract procedures $(\Estimate, \Transform)$ to describe (a slightly simplified version of) the~\cite{FOCS:CMSZ21} extractor. Let $\ket{+_R}_{\RegR}$ denote the uniform superposition over challenges on register $\RegR$ and let $\RegH$ denote the register containing the prover's state. Let $\sV_r = \BMeas{\BProj{V, r}}$ denote a binary projective measurement on $\RegH$ that measures whether $P^*$ returns a valid response on $r$. 

The extraction technique makes crucial use of two measurements: the first is $\sU = \BMeas{\BProj{\sU}}$, where $\BProj{\sU} \eqdef \Id_{\RegH} \otimes \ketbra{+_R}_{\RegR}$ is the projective measurement of whether the challenge register $\RegR$ is \emph{uniform}. The second is $\sC = \BMeas{\BProj{\sC}}$, where $\BProj{\sC} \eqdef (\BProj{V,r_i})_{\RegH} \otimes \sum_{r \in R} \ketbra{r}_{\RegR}$ is the projective measurement that runs the prover on the challenge on $\RegR$ and \emph{checks} whether the prover wins. The extraction procedure is described in \cref{fig:tech-overview-cmsz} below.

\begin{mdframed}

\captionsetup{type=figure}
    \captionof{figure}{The \cite{FOCS:CMSZ21} extractor with generic procedures $\Estimate, \Transform$}\label{fig:tech-overview-cmsz}
    
\begin{enumerate}
    \item Generate a first verifier message $\mathsf{vk}$ and run $P^*(\vk) \rightarrow \Co$ to obtain a classical first prover message $\Co$ once and for all. Let $\ket{\psi}$ denote the state of $P^*$ after it returns $\Co$. 
\item\label[step]{step:tech-overview-initial-estimate} Run $\Estimate^{\sU,\sC}$ to accuracy $\gamma/4$ on $\ket{\psi}\ket{+_R}$, which outputs an estimate $p$ of the adversary's success probability and then discard $\RegR$;\footnote{We assume here that when we run $\Estimate^{\sU,\sC}$ on a state $\ket{\phi} \in \image(\BProj{\sU})$, the residual state is in $\image(\BProj{\sU})$. Then we are guaranteed that $\RegR$ is unentangled from $\RegH$, which allows us to discard it. In our eventual construction, this assumption is enforced in \cref{thm:svdisc}.}~abort if $p < \gamma/2$ (this occurs with probability at most $1-\gamma/2$). Subtract $\gamma/4$ from $p$ so that $p$ represents a reasonable lower bound on the success probability. Set an error parameter $\eps = \frac{\gamma^2}{2\secp k}$ for the rest of the procedure and fix $N = \secp k /p$.
    \item We now want to generate $k$ accepting transcripts. For $i$ from $1$ to $N$:
    \begin{enumerate}
        \item Sample a uniformly random challenge $r_i$ and apply $\sV_{r_i}$ to the current state $\ket{\psi_i}$.
        \item If the output is $b_i = 1$, measure the response $\Resp$. This is (computationally) undetectable by the protocol's collapsing property, so we ignore this step \textbf{for now}. 
        \item Let $E$ be a unitary such that applying $E$ to $\RegH \otimes \RegW$ (where $\RegW$ is an appropriate-size ancilla initialized to $\ket{0}_{\RegW}$) and then discarding $\RegW$ is equivalent to running $\Estimate^{\sU,\sC}$ for $\secp p/\eps^2$ steps on $\RegH \otimes \RegR$ (where $\RegR$ is initialized to $\ket{+_R}_{\RegR}$) and then discarding $\RegR$. %Let 

        We \emph{repair} the success probability by initializing $\RegW = \ket{0}_{\RegW}$ and then running $\Transform^{\sD,\sG}$ on $\RegH \otimes \RegW$ where, roughly speaking, $\sD$ is a projective measurement corresponding to the \emph{disturbance} caused by step (a), and $\sG$ is a projective measurement that determines whether the adversary's success probability is \emph{good}, meaning at least $p-\varepsilon$. More precisely:
        \begin{itemize}
            \item $\sG = \BMeas{\BProj{p,\eps}}$ returns $1$ if, after applying $E$, the estimate is at least $p- \varepsilon$.\footnote{In our actual construction/proof, we replace this call to $\Estimate$ (and the additional call at the end of Step 3c) with a weaker primitive that \emph{only} computes the threshold instead of fully estimating $p$. This change makes it easier to instantiate the primitive.}
            \item $\sD = \BMeas{\BProj{r_i,b_i}}$ returns $1$ if $\RegW = \ket{0}_{\RegW}$ \emph{and} applying $\sV_{r_i}$ returns $b_i$.  
        \end{itemize} 
        If $\Transform^{\sD,\sG}$ has not terminated within $T$ calls to $\sD$ and $\sG$, abort (this occurs with probability at most $O(1/T)$). Otherwise, apply $E$, trace out $\RegW$, re-initialize $\RegR$ to $\ket{+_R}$ and then run $\Estimate^{\sU,\sC}$ for $\secp p/\eps^2$ steps to obtain a new probability estimate $p'$. If $p' < p-2\eps$, abort. Finally, discard $\RegR$ and re-define $p := p'$. 
    \end{enumerate}
\end{enumerate}
\end{mdframed}

\subsubsection{Guaranteed extraction, first attempt}
\label{sec:tech-overview-first-attempt}
The \cite{FOCS:CMSZ21} algorithm, interpreted in terms of the abstract procedures $\Estimate, \Transform$, will serve as our initial template for extraction. We now consider whether it can be modified to achieve \emph{guaranteed} extraction. 

\paragraph{Syntactic Changes.}

The first issues with the \cite{FOCS:CMSZ21} extraction procedure are syntactic in nature. Namely, we want an extraction procedure that works for \emph{any} $P^*$, with no a priori lower bound $\gamma$ on the success probability of $P^*$. Of course, an extractor $\Extract^{P^*}$ that extracts with probability close to $1$ given an arbitrary $P^*$ is impossible to achieve (imagine a $P^*$ with negligible success probability), so the game is also changed as described in \cref{def:high-probability-extraction}. In terms of the \cite{FOCS:CMSZ21} template, the change is as follows:
    
\begin{itemize}
        \item After obtaining $(\Co,\ket{\psi})$, measure $\sC$ on $\ket{\psi} \ket{+_R}$ and terminate if the outcome is $0$. 
        \item Otherwise, the state is (re-normalized) $\BProj{\sC} (\ket{\psi} \ket{+_R})$, and the goal is to extract with probability $1-\negl$. 
\end{itemize}

\paragraph{Variable-Runtime Estimation.}
Since we are given no \emph{a priori} lower bound $\gamma$ on the success probability  of $P^*$, there is no fixed additive precision $\epsilon$ for which the initial $\Estimate$ in \cref{step:tech-overview-initial-estimate} guarantees successful extraction --- the initial state $\ket{\psi}\ket{+_R}$ could be concentrated on subspaces $\mathcal S_j$ such that $p_j \ll \eps$, in which case the estimation procedure almost certainly returns $0$.

To remedy this issue, we define a \emph{variable-length} variant of $\Estimate^{\MeasA,\MeasB}$ with the guarantee that for every $j$ and every state in $\mathcal S_j$, $\Estimate^{\MeasA,\MeasB}$ returns $p_j$ to within constant (factor $2$) \emph{multiplicative} accuracy with probability $1-2^{-\secp}$. With regard to instantiation, we note that the \cite{MarriottW05,FOCS:CMSZ21} implementation of $\Estimate^{\MeasA,\MeasB}$ can be modified to be variable-length: simply continue alternating $\Pi_A, \Pi_B$ until sufficiently many ($d = \poly(\secp)$) $b_i = b_{i+1}$ occur, so that the estimate $\frac d {t-1}$ (where $t$ is the number of measurements performed) is reasonably concentrated around its expectation.

Thus, we begin with the natural idea that \cref{step:tech-overview-initial-estimate} should be modified to use this variable-length $\Estimate$. We remark that variable-length $\Estimate$ is \emph{not} required in later steps: the output $p$ of \cref{step:tech-overview-initial-estimate} can be used to set the parameters $(\eps, N)$ for the rest of the procedure.

With this modification, our extractor \emph{never} aborts in \cref{step:tech-overview-initial-estimate}, but it also no longer runs in strict polynomial time. How do we analyze its runtime? First, one can compute that when run on a state in $\mathcal S_j$, the expected running time of this procedure is (up to factors of $\poly(\secp)$) roughly $\frac 1 {p_j}$. This might seem concerning, because this expectation could be large (even superpolynomial) if $p_j$ is very small. However, what we care about is the runtime of $\mathsf{Estimate}^{\sU, \sC}$ on the (re-normalized) state $\BProj{\sC} (\ket{\psi} \ket{+_R}).$ Writing $\ket{\psi} \ket{+_R} = \sum_j \alpha_j \ket{v_{j,1}}$, we see that $\BProj{\sC} (\ket{\psi} \ket{+_R}) = \sum_j \alpha_j \sqrt{p_j} \ket{w_{j,1}}$. 

To calculate the overall expected runtime, we use the fact that $\mathsf{Estimate}^{\sU,\sC}$ commutes with the projective measurement $\Jor$ that outputs $j$ on each subspace $\RegS_j$. This implies that the expected runtime of $\mathsf{Estimate}$ on our state is the weighted linear combination of its expected runtime on the eigenstates $\ket{v_{j,1}}$, namely
\[ \frac 1 {\gamma^*} \sum_j |\alpha_j|^2 p_j \cdot \frac 1 {p_j} = \frac 1 {\gamma^*},
\]
where $\gamma^* = || \BProj{\sC} (\ket{\psi} \ket{+_R}) ||^2$ is the probability that $\sC \rightarrow 1$ in the initial execution.\footnote{One way to see this is to notice that applying $\Jor$ after running $\mathsf{Estimate}^{\sU,\sC}$ clearly cannot affect the runtime of $\mathsf{Estimate}^{\sU,\sC}$. Then $\Jor$ can be commuted to occur before $\mathsf{Estimate}^{\sU,\sC}$.} Thus, the overall expected runtime equals $\gamma^* \cdot \frac 1 {\gamma^*} = 1$, so Step 2 of the procedure is efficient!

\paragraph{Our first attempt.}
With the changes above, Step 2 of the extraction procedure now has zero error and runs in expected polynomial time ($\EQPTM$).

The other source of non-negligible extraction error from \cite{FOCS:CMSZ21} is in the cutoff $T$ imposed on $\Transform^{\sD,\sG}$. By removing this cutoff, we obtain a procedure that is somewhat closer to the goal of guaranteed extraction in expected polynomial time, described in~\cref{fig:tech-overview-first-attempt} below.

\begin{mdframed}
\captionsetup{type=figure}
    \captionof{figure}{Guaranteed extraction (Attempt 1)}\label{fig:tech-overview-first-attempt}
\begin{enumerate}
    \item After obtaining $(\Co, \ket{\psi})$, apply $\sC$ to $\ket{\psi} \ket{+_R}$ and terminate if the measurement returns $0$. Otherwise, let $\ket{\phi}$ denote the resulting state on $\RegH \tensor \RegR$.
    \item Run the variable-length $\Estimate^{\sU, \sC}$ on $\ket{\phi}$, obtaining output $p$. Divide $p$ by $2$ to obtain a lower bound on the resulting success probability. Set $\eps = \frac {p^2} {2\secp k }$ and $N = \secp k /p$.
    \item Run Step 3 of the original \cite{FOCS:CMSZ21} extractor as in \cref{fig:tech-overview-cmsz}, with the parameters $p, \eps, N$. Instead of imposing a time limit $T$, the procedure $\Transform^{\sD, \sG}$ is allowed to run until completion\footnote{To avoid a computation that runs for infinite time, one should at the very least impose an exponential $2^\secp$ time cutoff, which can be shown to incur only a $2^{-\secp}$ correctness error.} $( \mathsf{G}\rightarrow 1)$.
\end{enumerate}

\end{mdframed}

\subsubsection{Problem: Step 3 is not expected poly-time.}
\label{sec:tech-overview-not-eqpt}
    Unfortunately, the ``first attempt'' above does \emph{not} satisfy \cref{def:high-probability-extraction}. The issue lies in its runtime: we argued before that over the randomness of $\Extract^{P^*}$, Step 2 runs in expected polynomial time. However, we did not analyze Step 3, which is the main loop for generating transcripts. Here is a rough estimate for its runtime.
    
    Recall that Step 3 loops the following steps for each $i = 1,\dots,\lambda k/p$:
    \begin{itemize}
        \item Run the prover $P^*$ on a random challenge $r_i$. This takes a fixed $\poly(\secp)$ amount of time.
        \item Then, \emph{regardless} of whether $P^*$ was successful, the residual prover state $\ket{\phi_i}$ must be \emph{repaired} to have success probability $\approx p$.
    \end{itemize}
    
    It turns out that as currently written, the expected runtime of the repair step is (up to $\poly(\secp)$ factors) equal to the runtime of a fixed-length $\Estimate$ procedure with precision $\approx p^2$ (this ensures that after $1/p$ repair steps, the total success probability loss must be at most $p$). Moreover, this runtime is intuitively necessary for \emph{any} possible repair procedure, since repairing the success probability should be at least as hard as computing whether it is above the acceptable threshold. 
    
    In our setting, the \cite{MarriottW05,FOCS:CMSZ21} estimation procedure requires $1/p^3$ time to obtain a $p^2$-accurate estimate in the relevant parameter regime.\footnote{As written in \cite{FOCS:CMSZ21}, the estimation procedure runs in $1/p^4$ time, but a factor of $p$ can be saved because (roughly speaking) the estimate only needs to achieve $p^2$ accuracy when $p_j$ is close to $p$.} Since Step 3 performs this loop $\frac{1}{p}$ times (omitting the $\secp k$ factor), the total runtime will be at least $\frac 1 {p^4}$. This is too long for the ``conditioning'' of Step 1 to save us: if the initial state at the beginning of Step 1 is $\ket{\psi} \ket{+_R} \in \mathcal S_j$, the expected runtime of Step 3 is $p_j \cdot \frac 1 {p_j^4} = \frac 1 {p_j^3}$, which can be arbitrarily large (when $p_j$ is small). 
    
    \paragraph{Idea: Use a faster $\Estimate$?} Given how we have phrased the extractor in terms of abstract $(\Estimate, \Transform)$ algorithms, a natural idea for improving the runtime is to use an implementation of the abstract $\Estimate$ algorithm that is faster than the \cite{MarriottW05}-based one used in \cite{FOCS:CMSZ21}. Indeed, if we use the procedure described in \cite{NagajWZ11} to implement $\Estimate^{\sU,\sC}$, we obtain a quadratic speedup: the runtime of $\Estimate^{\sU,\sC}$ in Step 3c can be improved from $\frac{1}{p^3}$ to $\frac{1}{p^{3/2}}$. %In particular, since 
    
    This speedup will be relevant to our eventual solution, but it does not resolve the problem. The back-of-the envelope calculation now just says that the expected runtime of Step 3 on a state $\ket{\psi}\ket{+_R} \in \RegS_j$ is $p_j \cdot p_j^{-5/2} = p_j^{-3/2}$, which is still unbounded. 
    
    \paragraph{So are we doomed?} Indeed, this runtime calculation seems problematic for the \emph{entire} \cite{FOCS:CMSZ21} template that we abstracted, by the following reasoning:
    
    \begin{itemize}
        \item On a state with initial estimate $p$, each choice of $r_i$ will only produce an accepting transcript with probability $\approx p$, so we must try $\approx k/p$ choices of i.i.d. $r_i$ to obtain $k$ accepting transcripts.
        \item Therefore, as long as the repair step takes \textbf{super-constant time} (as a function of $1/p$), the overall extraction procedure will take too long.
    \end{itemize}
    
    This seems to indicate a dead end for extractors that follow the standard rewinding template of repeatedly running $P^*$ on random $r$ to obtain accepting transcripts. 
    
    \subsubsection{Solution: A New Rewinding Template}
    \label{sec:tech-overview-new-template}
    
    We solve our unbounded runtime issue by abandoning ``classical'' rewinding, in the following sense: unlike prior extraction procedures \cite{EC:Unruh12,FOCS:CMSZ21}, our extractor will \emph{not} follow the standard approach of obtaining transcripts by feeding uniformly random $r_i$ to $P^*$. Instead, we will \emph{generate} accepting transcripts $(r_i, z_i)$ via an inherently quantum procedure so that \emph{every} generated transcript is accepting (as opposed to only a $p$ fraction of them).
    
    We accomplish this by using the procedure $\Transform$, which was previously only used for state repair, to \emph{generate} the transcripts. Consider a prover state $\ket{\psi_i}$ at the beginning of Step 3. By definition, $\ket{\psi_i} \ket{+_R} \in \image(\BProj{\sU})$, so applying $\Transform^{\sU, \sC}$ to $\ket{\psi_i} \ket{+_R}$ produces a state in $\image(\BProj{\sC})$. Now if the challenge register $\RegR$ is \emph{measured} (obtaining a string $r_i$), the residual prover state is \emph{guaranteed} to produce an accepting response on $r_i$!
    
    Moreover, the extraction procedure can afford to run $\Transform^{\sU, \sC}$: since $\ket{\psi_i} \ket{+_R}$ has been constructed to lie almost entirely in subspaces $\mathcal S_j$ such that $p_j < p - \eps$, the expected running time of $\Transform^{\sU, \sC}$ can be shown to be roughly\footnote{For technical reasons, we cut off $\Transform$ after an exponential number of steps so that the component of $\ket{\psi_i} \ket{+_R}$ lying in ``bad'' $\mathcal S_j$ (i.e., where $p_j$ is tiny) does not ruin the expected running time.} $\frac 1 p$.
    
    This gives us a potential \emph{new} template for extraction: we modify the main loop (Step 3) as in \cref{fig:tech-overview-new-template}. 
    
    \begin{mdframed}
    \captionsetup{type=figure}
    \captionof{figure}{Our new extraction template}\label{fig:tech-overview-new-template}
    \begin{enumerate}
    \item After obtaining $(\Co, \ket{\psi})$, apply $\sC$ to $\ket{\psi} \ket{+_R}$ and terminate if the measurement returns $0$. Otherwise, let $\ket{\phi}$ denote the resulting state on $\RegH \tensor \RegR$.
    \item Run the variable-length $\Estimate^{\sU, \sC}$ on $\ket{\phi}$, obtaining output $p$. Divide $p$ by $2$ to obtain a lower bound on the resulting success probability. Set $\eps = \frac {p} {4 k }$.
    \item For $i$ from $1$ to \textcolor{red}{$k$}:
    \begin{enumerate}
        \item Given current prover state $\ket{\psi_i}$, apply $\Transform^{\sU,\sC}$ to $\ket{\psi_i} \ket{+_R}$. Call the resulting state $\ket{\phi_{\sC}}$.
        \item Obtain a \emph{guaranteed accepting} transcript $(r_i, z_i)$ by measuring the $\RegR$ register of $\ket{\phi_{\sC}}$ and then running $P^*$ on $r_i$. As before, measuring $z_i$ is computationally undetectable.
        \item Run the Repair Step (3c) as in \cref{fig:tech-overview-cmsz} by calling $\Transform^{\mathsf{D},\mathsf{G}}$ and re-estimating $p$. 
    \end{enumerate}
    \end{enumerate}
    
    \end{mdframed}
    
We emphasize two crucial efficiency gains from this new extraction template:

\begin{itemize}
    \item As already mentioned, the main loop now has $k$ steps instead of $k/p$, since each transcript is now guaranteed to be accepting. 
    \item Since only $k$ repair operations are now required, the \emph{error} parameter $\eps$ for $\Pi_{p, \eps}$ can be set to $\approx p$ instead of $\approx p^2$. 
\end{itemize}

\paragraph{Correctness Analysis.}
We remark that even the correctness of this new extraction procedure is unclear. In the case of $k$-special sound protocols, we need the extraction procedure to produce $k$ accepting transcripts with distinct $r_i$; previously, this was guaranteed because each $r_i$ was sampled i.i.d., so (w.h.p.) no pair of them coincide. Here, $r_i$ is \emph{not} uniformly random --- it has been sampled by measuring the $\RegR$ register of some state in $\BProj{\sC}$. 

In order to analyze the behavior of this extractor, it is important to understand the state $\ket{\phi_{\sC}}$ obtained after applying $\Transform^{\sU, \sC}$. Of course, we have an explicit representation $\sum_j \alpha_j \sqrt{p_j} \ket{w_{j,1}}$ for it, but it is not clear a priori how this helps.

To prove correctness, we analyze the state $\ket{\phi_{\sC}}$ using what we call the Pseudoinverse Lemma (\cref{lemma:pseudoinverse}), which states that $\ket{\phi_{\sC}}$ can be viewed as a \emph{conditional} state obtained by starting with a state $\ket{\phi_{\sU}} = \ket{\psi_{\sU}} \ket{+_R} \in \image(\BProj{\sU})$ and \emph{post-selecting} (i.e., conditioning) on a $\sC$-measurement of $\ket{\phi_{\sU}}$ outputting $1$. Crucially, this pseudoinverse state has a precisely characterized $(\sU, \sC)$-Jordan spectrum related to the Jordan spectrum of $\ket{\phi_C}$. We emphasize that the state $\ket{\phi_{\sU}}$ does not actually exist in the extraction procedure; it is just a tool for the analysis.

Using the pseudoinverse lemma, one can show that the probability a $\sC$-measurement of $\ket{\phi_{\sU}}$ returns $1$ is $\approx p$, which implies that the joint distribution of $(r_1, \hdots, r_k)$ comes from a ``random enough'' distribution that we formalize as ``admissible'' (\cref{def:admissible-dist}). This is shown by the following reasoning: since measuring $\RegR$ commutes with $\sC$, it is as if we have an initially uniformly random $r_i$ (obtained from measuring $\RegR$ of $\ket{\phi_{\sU}}$) that is ``output'' with probability $\approx p$ (when $\sC$ returns $1$). This is sufficient to argue about correctness properties of the extractor. 

\paragraph{Runtime Analysis Idea.} Analyzing the runtime of $\Transform^{\sD, \sG}$ also turns out to be significantly more subtle than in the \cite{FOCS:CMSZ21} setting. The basic idea is to show that (within a reasonable amount of time) $\Transform^{\sD, \sG}$ \emph{returns} a state on $\RegH \tensor \RegW$ to $\image(\Pi_{p, \eps})$ after it was ``initially'' disturbed by the binary measurement $\sD$. In \cite{FOCS:CMSZ21}, this is literally true: the disturbance is measuring $(\Pi_{V, r}, \Id - \Pi_{V, r})$ for randomly sampled $r$ on the prover state $\ket{\psi_i}$. One can then show that an expected constant number of $(\sD, \sG)$-measurements returns the state to $\sG$ by appealing to the statistics of the $(\sD, \sG)$ Marriott-Watrous distribution.

However, in our setting, the ``disturbance'' is quite different: the amplified state $\ket{\phi_{\sC}} \in \image(\BProj{\sC})$ consists of a prover state \emph{entangled with} the challenge register $\RegR$ in a way that is \emph{guaranteed} to produce an accepting transcript. $\ket{\phi_{\sC}}$ is then disturbed by measuring its $\RegR$ register, and the measurement $\sD$ being applied in $\Transform^{\sD, \sG}$ depends on this $\RegR$ measurement outcome. Since the $\RegR$ measurement can disturb $\ket{\phi_{\sC}}$ by a large amount (unlike $\sD$), it is not a priori clear why $\Transform^{\sD, \sG}$ should return the state to $\image(\Pi_{p, \eps})$. 

At a high level, we show how to bound the runtime of this new procedure by appealing to the pseudoinverse state $\ket{\phi_{\sU}}$, again! In more detail, using the pseudoinverse lemma, the state on $\RegH \tensor \RegW$ obtained after measuring $\RegR$ on $\ket{\phi_{\sC}}$ (along with initializing $\RegW$ to $\ket{0}$) can be alternatively thought of as the state obtained by:

\begin{itemize}
    \item Sampling $r_i$ proportional to the probability $\zeta_{r_i}$ of $\ket{\psi_{\sU}}$ successfully answering $r_i$, and
    \item Outputting (normalized) $\Pi_{r_i}(\ket{\psi_{\sU}} \tensor \ket{0}_{\RegW})$, where $\Pi_{r_i} := \Pi_{r_i, 1}$. 
\end{itemize}
This conditioning argument allows us to appeal to the same ``return to $\Pi_{p, \eps}$'' principle to show that $\Transform^{\sD, \sG}$ indeed ``returns'' the state to $\image(\Pi_{p, \eps})$, as if it had ``started out'' as the state $\ket{\phi_{\sU}}\ket{0}_{\RegW}$, which only exists in the analysis!

\subsubsection{Problem: Step 3 is \emph{still} not expected poly-time.}
\label{sec:tech-overview-no-speedup}

The premise of our new extraction template was to speed up the extraction process by getting rid of excess work from running state repair in situations where no accepting transcript was obtained. Previously, we computed the expected runtime to perform $N \approx k/p$ repair steps in~\cref{fig:tech-overview-cmsz} (conditioned on a successful initial execution and initial estimate $p$) to be $pN/\varepsilon^2 \approx 1/p^4$, since the runtime of each repair step was equivalent (up to a constant factor) to the runtime of $\sG$, which was $p/\varepsilon^2$, and $\varepsilon \approx p^2$. As noted above, with our new template we now only have to perform $N = k$ repair steps, and the error parameter $\epsilon$ can now be $\approx p$. With these improvements alone, one might hope to perform $N$ repair steps in $pN/\varepsilon^2 = p(k)(1/p^2) \approx 1/p$ time. This would result in expected polynomial runtime for the overall extractor when factoring in the conditioning.

Perhaps surprisingly, the above reasoning is incorrect! This new extraction procedure is \emph{still} not expected QPT: the expected runtime of $N$ repair steps will be $\approx \frac 1 {p^{2}}$, not $\frac 1 {p}$.

Why does this happen? It turns out that in this new extraction template, each repair step (which previously made expected $O(1)$ calls to $\sG$) must now make an expected $O(1/p)$ calls to $\sG$, cancelling out the factor-$1/p$ savings in $N$ obtained by using $\Transform^{\sU, \sC}$ to generate transcripts. 

Indeed, the pseudoinverse-based runtime analysis above for $\Transform^{\sD, \sG}$ implies that each repair step must now make
\[ \frac 1 {\zeta_R} \sum_r \zeta_r \cdot \frac 1 {\zeta_r} = \frac 1 {\zeta_R} \approx 1/p
\]
calls to $\sG$ (where $\zeta_R = \sum_r \zeta_r \approx p$ is the normalization factor for the $r_i$-distribution). This results in an overall expected running time of $\frac 1 p$ calls to $\sG$ if $p$ was initially measured. Essentially, this is saying that while obtaining an accepting transcript $(r_i, z_i)$ causes \emph{limited enough} disturbance that repair can work, it causes more disturbance than a binary measurement, resulting in a factor of $1/p$ increase in the repair time.

\subsubsection{Solution: Use faster \texorpdfstring{$\Estimate$}{Estimate}  and  \texorpdfstring{$\Transform$}{Transform}} \label{sec:tech-overview-speedup}
Despite the less-than-expected speedup observed in \cref{sec:tech-overview-no-speedup}, it turns out that we nevertheless made significant progress. The reason is that the bottleneck to obtaining a faster extraction procedure is now in the running times of $\Estimate$ and $\Transform$, so we can hope to obtain an expected polynomial time procedure by using faster algorithms for $\Estimate^{\sU, \sC}$ and $\Transform^{\sD, \sG}$.

As discussed above, speeding up the fixed-length $\mathsf{Estimate}^{\sU, \sC}$ in $\sG$ is relatively straightforward by appealing to \allowbreak \cite{NagajWZ11};\footnote{For technical reasons, we use a different algorithm due to \cite{STOC:GSLW19}, but a variant of \cite{NagajWZ11} would also suffice.} this results in an expected running time of $\frac 1 {\sqrt {p}}$ for $\sG$.

However, implementing a \emph{fast} version of $\Transform^{\sD,\sG}$ achieving $1-\negl(\secp)$ correctness (which is required for our extraction procedure to have negligible error) is less straightforward. Some implementations in the literature (e.g., \cite{STOC:GSLW19}) achieve this correctness guarantee, but only given a known (inverse polynomial) lower bound on the eigenvalue $q_j$ (associated with $(\sD, \sG)$-Jordan subspace $\mathcal T_j$). We have no such lower bound for our state $\frac 1 {\zeta_r} \Pi_r (\ket{\phi_{\sU}} \ket{0}_{\RegW})$. Our resolution is to first apply a variable-length fast phase estimation algorithm (implemented by repeatedly running~\cite{NagajWZ11} to increasing precision, or singular value discrimination \cite{STOC:GSLW19} with decreasing thresholds, until we obtain a multiplicative estimate of the phase) and then run a fixed-length fast $\Transform^{\sD,\sG}$ using the estimated phase to lower bound the eigenvalue. The fixed-length fast $\Transform^{\sD,\sG}$ can be done using \cite{STOC:GSLW19}; it is also possible to use a more elementary algorithm combining fast amplitude amplification \cite{BHMT02} with ideas from \cite{STOC:Watrous06} for achieving $1-\negl(\secp)$ correctness.

To summarize, we obtain a final $1/p$ speedup by combining a $1/\sqrt{p}$ speedup from using a faster $\Estimate^{\sU,\sC}$ with a $1/\sqrt{p}$ speedup from using a faster $\Transform^{\sD,\sG}$. The fact that the latter speedup is actually realized turns out to be subtle to argue.
    
\subsubsection{Last Problem: Measuring $z$ ruins the runtime guarantee}
\label{subsubsec:gk-issue}

Unfortunately, we are \emph{still} not done! There is one subtle issue with our extractor that we have ignored so far: our runtime analysis was only valid ignoring the effect of measuring the prover response $z$. Since all transcripts after running $\Transform^{\sU, \sC}$ are accepting by construction, the collapsing property of the protocol implies that measuring $z$ is computationally undetectable, so one might assume that the runtime analysis extends immediately.

However, the \emph{expected running time} of an algorithm is not an efficiently testable property of the input state. This is not just an issue with our proof strategy: the version of the above extractor where $z$ is measured does not run in expected polynomial time. 

In a nutshell, the issue is that a computationally undetectable measurement can still cause a state's eigenvalues (either $\{p_j\}$, in $\Jor^{\sU, \sC}$, or $\{q_j\}$, in $\Jor^{\sG, \sD}$) to change by a negligible but nonzero amount, affecting the subsequent runtime of $\Transform^{\sD,\sG}$. This negligible change can have an enormous effect on the expected runtime of the extractor, because if the runtime of a procedure is inversely proportional to the disturbed eigenvalue $\tilde p = p-\negl$, an overall expected runtime expression can now contain terms of the form $\frac p {p - \negl}$, which can be unbounded when $p$ is also negligible. Interestingly, such issues have long been known to exist in the \emph{classical} setting: these $\frac p {p-\negl}$ terms are the major technical difficulty in obtaining a classical simulator for the \cite{JC:GolKah96} protocol. This classical analogy inspires our resolution.

\paragraph{Solution: Estimate repair time before measuring $z$.} 
\label{subsubsec:gk-fix}

We modify our extractor so that in each loop iteration, all procedures occurring after the $z$-measurement have a \emph{pre-determined} runtime. Previously, after $z$ was measured, we ran a fast variable-length $\Transform$ by running a the variable-length $\Estimate^{\sD, \sG}$ to determine a time bound $t$, and then running a $t$-time $\Transform^{\sD, \sG}$. Instead of this, we will run $\Estimate^{\sD, \sG}$ \emph{before} $z$ is measured. This allows us to compute a runtime bound for $\Transform^{\sD, \sG}$ before the $z$ measurement disturbs the state, preserving the expected running time of the entire procedure. This results in the final extraction procedure described in \cref{fig:final-extractor-intro} below.

\begin{mdframed}
\captionsetup{type=figure}
    \captionof{figure}{Our final extraction procedure}\label{fig:final-extractor-intro}
    \begin{enumerate}
    \item After obtaining $(\Co, \ket{\psi})$, apply $\sC$ to $\ket{\psi} \ket{+_R}$ and terminate if the measurement returns $0$. Otherwise, let $\ket{\phi}$ denote the resulting state on $\RegH \tensor \RegR$.
    \item Run the variable-length $\Estimate^{\sU, \sC}$ on $\ket{\phi}$, obtaining output $p$. Divide $p$ by $2$ to obtain a lower bound on the resulting success probability. Set $\eps = \frac {p} {4 k }$ and $N = k$.
    \item For $i$ from $1$ to $N$:
    \begin{enumerate}
        \item Given prover state $\ket{\psi_i}$, apply $\Transform^{\sU,\sC} \ket{\psi_i} \ket{+_R}$. Call the resulting state $\ket{\phi_{\sC}}$.
        \item Measure (and discard) the $\RegR$ register of $\ket{\phi_{\sC}}$ to obtain a classical challenge $r_i$.
        \item Initialize $\RegW$ to $\ket{0}_{\RegW}$ and call the variable-length $\Estimate^{\sD, \sG}$, which outputs a value $q$. We require that the output state is in the image of $\Pi_{r_i}$.
        \item Measure the response $z_i$.
        \item We repair the success probability by running $\Transform^{\sD,\sG}$ on $\RegH \otimes \RegW$ for $\frac{\secp}{\sqrt q}$ oracle steps. If the resulting state is not in the image of $\Pi_{p, \eps}$, abort.
        
        Trace out $\RegW$ and run $\Estimate^{\sU,\sC}$ for $\secp \sqrt{p}/\eps$ steps to obtain a new probability estimate $p'$. If $p' < p- 2\eps $, abort. Finally, discard $\RegR$ and re-define $p := p'$.  
    \end{enumerate}
    \end{enumerate}
\end{mdframed}

By making this change, we incur an additional \emph{correctness} error for the extractor, because the collapsing measurement may decrease the probability that $\Transform^{\sD, \sG}$ successfully maps the state to $\Pi_{p, \eps}$. However, this error is negligible because this correctness property is efficiently checkable (unlike the expected runtime). Thus, this procedure achieves both expected polynomial runtime\footnote{It remains to be argued that measuring $z_j$ does not affect the running time of \emph{subsequent} variable-runtime steps. This turns out to hold because the runtime of future loop iterations can be guaranteed by the correctness properties of the re-estimation step, which hold for an \emph{arbitrary} re-estimation input state.}  and the desired correctness guarantees.

\subsubsection{Putting everything together} To summarize, we gave a new extraction template along with a particular instantiation that achieves expected polynomial runtime, by leveraging four different algorithmic improvements:

\begin{enumerate}[noitemsep]
    \item By generating accepting transcripts with $\Transform^{\sU, \sC}$, we now only have to generate $k$ transcripts and repair $k$ prover states (instead of $k/p$). 
    \item (1) allows us to relax the error parameter $\eps$ by a factor of $1/p$ (speeding up $\sG$).
    \item Using a fast algorithm for $\Estimate$ from the literature \cite{NagajWZ11,STOC:GSLW19} saves a factor of $1/\sqrt{p}$ runtime.
    \item Using a new fast, variable-runtime algorithm for $\Transform$ saves another factor of $1/\sqrt{p}$.
\end{enumerate}
Finally, we implement the variable-length $\Transform$ in two phases (variable-length phase estimation followed by fixed-length $\Transform$) and interleave the measurement of the response $z$ between them, so that this $z$-measurement has no effect on the runtime.

We remark that the overall analysis of our extractor is rather involved (as we have omitted additional details in this overview); we refer the reader to \cref{sec:high-probability-extractor} for a full analysis.

\subsection{Post-Quantum ZK for \cite{JC:GolKah96}}
\label{sec:tech-overview-pqzk-gk}

In this section we give an overview of our proof that the Goldreich--Kahan (GK) protocol is post-quantum zero-knowledge (\cref{thm:gk}). Our simulator makes use of some of the techniques described in \cref{sec:tech-overview-hpe}, but the simulation strategy is quite different to our other results. In particular, our simulator does \emph{not} make use of state-preserving extraction.

We first recall the Goldreich--Kahan construction of a constant-round zero-knowledge proof system for $\NP$. Let $(P_{\Sigma},V_{\Sigma})$ be a $\Sigma$-protocol for $\NP$ satisfying special honest verifier zero knowledge (SHVZK)\footnote{Recall that the special honest-verifier zero-knowledge property guarantees the existence of a randomized simulation algorithm $\SHVZK.\Sim(\Ch)$ that takes any $\Sigma$-protocol challenge $\Ch \in R$ as input and outputs a tuple $(\Co,\Resp)$ such that the distribution of $(\Co,\Ch,\Resp)$ is indistinguishable from the distribution of transcripts arising from an honest prover interaction on challenge $\Ch$.} and let $\Com$ be a statistically hiding, computationally binding commitment. \cite{JC:GolKah96} construct a zero knowledge protocol $(P,V)$ as described in \cref{fig:gk-intro}.

\begin{figure}[!ht]
\centering
	\procedure[]{}{
		P (x, w) \> \> V(x) \\ 
		\begin{subprocedure}
			\pseudocode[mode=text]{Sample commitment key $\ck$.}
		\end{subprocedure}
		\> \sendmessageright*[3cm]{\ck} \> 
		\\
		 \> \sendmessageleft*[3cm]{\com} \> \begin{subprocedure}
			\pseudocode[mode=text]{Sample $\Sigma$-protocol challenge $\Ch \gets R$.\\
			Commit to $\Ch$: \\
			$ \com = \Com(\ck, \Ch; \omega)$ for $\omega \gets \{0,1\}^\secp$.}
		\end{subprocedure}\\ 
		\begin{subprocedure}
			\pseudocode[mode=text]{Compute $(\Co, \mathsf{st}) \gets P_\Sigma(x, w)$}
		\end{subprocedure}
		\> \sendmessageright*[3cm]{\Co} \>  \\
		\> \sendmessageleft*[3cm]{\Ch, \omega} \> \begin{subprocedure}
			\pseudocode[mode=text]{}
		\end{subprocedure} \\
		\begin{subprocedure}
			\pseudocode[mode=text]{If $\Com(\ck, \Ch; \omega) \neq  \com$, abort. \\ Compute $\Resp \gets P_\Sigma(\mathsf{st}, \Ch)$
			}
		\end{subprocedure}
			\> \sendmessageright*[3cm]{\Resp} \>  \begin{subprocedure}
			\pseudocode[mode=text]{Accept if $(\Co, \Ch, \Resp)$ is an \\ accepting $\Sigma$-protocol transcript for $x$.}
			\end{subprocedure}
	}
	\caption{The \cite{JC:GolKah96} Zero Knowledge Proof System for $\mathsf{NP}$.}
	\label{fig:gk-intro}
\end{figure}

\iffalse
\begin{enumerate}[noitemsep]
    \item $P$ sends the commitment key of a statistically-hiding, computationally-binding commitment scheme.
    \item $V$ commits to a uniformly random challenge $\Ch$.
    \item $P$ runs $P_{\Sigma}$ to obtain the first message of the sigma protocol $\Co$, and sends it to $V$.
    \item $V$ opens its commitment to $\Ch$.
    \item If the commitment opening was successful, $P$ runs $P_{\Sigma}(\Co,\Ch)$ to obtain the second message of the sigma protocol $\Resp$, and sends it to $V$.
\end{enumerate}
\fi

Soundness of the~\cite{JC:GolKah96} protocol holds against \emph{unbounded} $P^*$ and therefore extends immediately to the quantum setting. 

\paragraph{Recap: the na\"ive classical simulator.}

As observed by~\cite{JC:GolKah96}, there is a natural \emph{na\"ive simulator} for their protocol that, for reasons analogous to~\cref{subsubsec:gk-issue}, turns out to have an unbounded expected runtime. To build intuition for our quantum simulation strategy, we will first recall the na\"ive classical simulator and show how to extend it to a na\"ive quantum simulator (while temporarily ignoring the runtime issue). Then, by using the technique described in~\cref{subsubsec:gk-fix}, we will improve this to a full $\EQPTC$ quantum simulator.

The na\"ive classical simulator does the following:
\begin{enumerate}[noitemsep]
    \item Call $V^*$ on a random commitment key $\ck$ to obtain a commitment $\com$. 
    \item\label[step]{gk-shvzk} Sample $(\Co',\Resp') \gets \SHVZK.\Sim(0)$. 
    \item\label[step]{gk-first-run} Run $V^*$ on $\Co'$ to obtain a challenge-opening pair $(\Ch',\omega')$. If $\omega'$ is not a valid opening of $\com$ to $\Ch'$, terminate the simulation and output the current view of $V^*$.
    \item\label[step]{gk-rewinding-step} \textbf{Rewinding step.} Sample $(\Co,\Resp) \gets \SHVZK.\Sim(\Ch')$ and run $V^*$ on $\Co$. If the output $(r,\omega)$ is not a valid message-opening pair, repeat this step from the beginning. \item Respond with $z$ and output $V^*$'s view.
\end{enumerate}

\noindent To see that this simulator outputs the correct view for $V^*$, consider two hybrid steps:
\begin{itemize}
    \item First, switch to a hybrid simulator in which the sample $(a',z') \gets \SHVZK.\Sim(0)$ is instead computed by running the honest prover $P(x,w)$. The indistinguishability between this hybrid simulator and the real simulator follows from the fact that $a'$ sampled as $(a',z') \gets \SHVZK.\Sim(0)$ is computationally indistinguishable from the honestly generated $a'$.
    \item Next, switch to a second hybrid simulator in which the honest prover is also used in the rewinding step to generate the $(a,z)$ samples rather than $\SHVZK.\Sim(r')$ (where $z$ is generated by running the honest prover on $(a,r')$). This is indistinguishable from the previous hybrid simulator by the $\SHVZK$ property, and moreover, by the computational binding of the commitment, the $r$ obtained in Step 4 must be $r'$ except with $\negl(\secp)$ probability. Moreover, conditioned on $r = r'$, the second hybrid produces the same distribution as the honest interaction.
\end{itemize}

We now show how to extend this simulator to the quantum setting.

\paragraph{Our ``na\"ive'' quantum simulator.} Step 1 of the na\"ive classical simulator will be unchanged in the quantum setting, so we focus on devising quantum verisons of Steps 2,3, and 4 while assuming $\ck,\com$ are fixed throughout.

Let $\ket{\psi}_{\RegV}$ be the state of the malicious verifier immediately after it sends $\com$. We let registers $\RegA,\RegZ$ denote registers containing the messages $a,z$ in the $\Sigma$-protocol and let $\RegM$ be a register that will contain the random coins for $\SHVZK.\Sim$ (or the honest prover later on). Let $\ket{\Sim_r}$ for any $r \in R$ be the state
$\ket{\Sim_r}_{\RegA,\RegZ,\RegM} = \sum  \alpha_{\mu} \ket{\SHVZK.\Sim(r;\mu),\mu}$ obtained by running $\SHVZK.\Sim$ on a uniform superposition of its random coins $\mu$.

We define binary projective measurements analogous to the $\sU$ and $\sC$ measurements used in our state-preserving extractor. However, instead of a single $\sU$ measurement, we will have for each $r \in R$ a measurement $\sS_r = \BMeas{\BProj{\sS,r}}$ on $\RegV \otimes \RegA \otimes \RegZ \otimes \RegM$ where $\BProj{\sS,r} \eqdef \Id_{\RegV} \otimes \ketbra{\Sim_r}_{\RegA,\RegZ,\RegM}$. The idea behind the $\sC = \BMeas{\BProj{\sC}}$ measurement is the same as before: it measures whether the malicious verifier $V^*$ returns a valid opening when run on the challenge $\RegA$. Note that $\sC$ acts as identity on $\RegZ,\RegM$.

The next steps of the quantum simulator are a direct analogue of the corresponding steps in the classical simulator:

\begin{enumerate}
\item[$2^*$.] Initialize $\RegA \otimes \RegZ \otimes \RegM$ to $\ket{\Sim_0}$.
\item[$3^*$.] Measure $\ket{\psi}_{\RegV} \otimes \ket{P}_{\RegA,\RegZ,\RegM}$ with $\sC$. If the outcome of $\sC$ is $0$ (the opening is invalid), terminate the simulation at this step: measure $\RegA$ to obtain $a'$, compute and measure the verifier's response $(r',\omega')$ and return $(\ck,\com,a',(r',\omega'),z = \bot)$ along with $\RegV$. If the outcome of $\sC$ is $1$, we will have to rewind. First, compute the verifier's response and measure it to obtain $r'$. 
\end{enumerate}

When the opening is invalid $(\sC$ outputs $0$), the $\SHVZK$ guarantee informally implies that these steps computationally simulate the view of $V^*$.

The hard case is when the opening is valid ($\sC$ outputs $1$). At this stage of the simulation, the state on $\RegV \otimes \RegA \otimes \RegZ$ is $\BProj{\sC}(\ket{\psi}_{\RegV}\ket{\Sim_0}_{\RegA,\RegZ} )$ (up to normalization). Intuitively, we want to ``swap'' $\ket{\Sim_0}_{\RegA,\RegZ,\RegM}$ for $\ket{\Sim_{r'}}_{\RegA,\RegZ,\RegM}$, but the application of $\BProj{\sC}$ has entangled the $\RegA$ register with $\RegV$. We will therefore apply an operation to disentangle these registers, then swap $\ket{\Sim_0}$ for $\ket{\Sim_{r'}}$, and then ``undo'' the disentangling operation. We do this by defining a unitary $U$ that is the coherent implementation of the following variable-length computation on $\RegV \otimes \RegA \otimes \RegZ \otimes \RegM \otimes \RegR$: measure $\RegR$ to obtain $r$, and then run a variable-length $\Transform^{\sC,\sS_{r}}$ on $\RegV \otimes \RegA \otimes \RegZ \otimes \RegM$.\footnote{The register $\RegR$ is required for the definition of $U$ and should not be confused with the sub-register of $\RegV$ that we measure to obtain the verifier's response.} Recall that implementing a variable-length computation coherently requires additional ancilla registers $\RegW, \RegB, \RegQ$ (see~\cref{sec:intro-creqpt}); we will suppress these registers for this overview, but we emphasize that they must be all be initialized to $\ket{0}$.

The simulator then continues as follows.

\begin{enumerate}
\item[$4^*$.] Run the following steps:
\begin{enumerate}
    \item Initialize $\RegR$ to $\ket{0}$ and apply $U$ to $\BProj{\sC}(\ket{\psi}_{\RegV} \ket{\Sim_0}_{\RegA,\RegZ,\RegM}) \otimes \ket{0}_{\RegR}$. On $\RegV \otimes \RegA \otimes \RegZ \otimes \RegM$, this maps $\image(\BProj{\sC})$ to $\image(\BProj{\sS,0})$, which yields a state of the form $\ket{\psi'}_{\RegV} \ket{\Sim_0}_{\RegA,\RegZ,\RegM}$. Importantly, this (carefully!) breaks the entanglement between $\RegV$ and $\RegA$.
    \item Now the simulator can easily swap $\ket{\Sim_0}$ out for $\ket{\Sim_{r'}}$.
    \item Finally, the simulator changes the $\RegR$ register from $\ket{0}_{\RegR}$ to $\ket{r'}_{\RegR}$, and then applies $U^\dagger$ and traces out $\RegR$. This step maps the state on \emph{back} from $\image(\BProj{\sS,r'})$ to $\image(\BProj{\sC})$. 
\end{enumerate}
\item[$5^*$.] Measure $\RegA$ to obtain $a$, compute and measure the verifier's response $(r,\omega)$, measure $\RegZ$ to obtain $z$, and output $(\ck,\com,a,(r,\omega),z)$ along with $\RegV$.
\end{enumerate}

This simulator can be written as an $\EQPTC$ computation, but we defer the details of this to our full proof (\cref{sec:gk}). For this overview, we will focus on proving the simulation guarantee.

Inspired by classical proof, we prove that our simulator produces the correct view for $V^*$ by considering two hybrid simulators. To describe the hybrid simulators, we define states $\ket{P}$ and $\ket{P_{r}}$ for any $r \in R$ corresponding to responses of the honest prover:
\begin{itemize}
    \item Let $\ket{P}_{\RegA,\RegZ,\RegM}$ be the state from running the honest prover $P_{x,w}$ on a uniform superposition of random coins $\mu$ to generate a first message $P_{x,w}(\mu)$, i.e., $\ket{P}_{\RegA,\RegZ,\RegM} = \sum_{\mu} \ket{P_{x,w}(\mu)}_{\RegA}\ket{0}_{\RegZ} \ket{\mu}_{\RegM}$
    \item For any $r \in R$, let $\ket{P_{r}}$ be the same as $\ket{P}_{\RegA,\RegZ,\RegM}$, except $\RegZ$ additionally contains the honest prover's response to $r$, i.e.,  $\ket{P_{r}}_{\RegA,\RegZ,\RegM} = \sum_{\mu} \ket{P_{x,w}(\mu)}_{\RegA}\ket{P_{x,w}(r;\mu)}_{\RegZ} \ket{\mu}_{\RegM}$
\end{itemize}

\noindent The hybrid simulators are essentially quantum versions of the classical ones:
\begin{itemize}
    \item The first hybrid simulator behaves the same as the original simulator except that everywhere the simulator uses $\ket{\Sim_0}$, the hybrid simulator uses $\ket{P}$ instead. The amplification in Step~$4^* (a)$ is now onto $\image(\ketbra{P})$ rather than $\image(\sS_{0})$. Moreover, in Step~$4^* b$, the simulator swaps $\ket{P}$ out for $\ket{\Sim_{r'}}$.
    
    \item The second hybrid simulator is the same as the first, except every appearance of $\ket{\Sim_{r'}}$ is replaced with $\ket{P_{r'}}$. In particular, in Step~$4^* (b)$, the simulator swaps $\ket{P}$ out for $\ket{P_{r'}}$. The (inverse) amplification in Step~$4^* (c)$ is now from $\image(\ketbra{P_{r'}})$ onto $\image(\BProj{\sC})$.
\end{itemize}

We remark that defining these hybrid simulators also requires extending the definition of the unitary $U$ that performs $\Transform$. In particular, $U$ must now support $\Transform^{\sC,\sP}$ where $\sP = \BMeas{\Id_{\RegV} \otimes \ketbra{P}}$ and $\Transform^{\sC,\sP_r}$ where $\sP_r = \BMeas{\Id_{\RegV} \otimes \ketbra{P_r}}$.

Proving indistinguishability of these hybrids requires some care. Intuitively, we want to invoke the SHVZK property to claim that $\ket{\Sim_0}$ and $\ket{P}$ are indistinguishable given just the reduced density matrices on the $\RegA$ register (for the first hybrid) and that $\ket{\Sim_{r'}}$ and $\ket{P_{r'}}$ are indistinguishable given just the reduced density matrices on  $\RegA \otimes \RegZ$ (for the second hybrid). However, we have to ensure that the application of $\Transform$ --- which makes use of projections onto these states --- does not make this distinguishing task any easier.

We resolve this by proving a general lemma (\cref{lemma:proj-indist}) about quantum computational indistinguishability that may be of independent interest, which we briefly elaborate on here. Consider the states $\ket{\tau_b} \eqdef \sum_{\mu} \ket{\mu}_{\RegX} \ket{D_b(\mu)}_{\RegY}$ where $D_0,D_1$ are computationally indistinguishable classical distributions with randomness $\mu$. If we are only given access to $\RegY$, then distinguishing $\ket{\tau_0}$ from $\ket{\tau_1}$ is clearly hard (since $\Tr_{\RegX}(\ketbra{\tau_b})$ is a random classical sample from $D_b$). \cref{lemma:proj-indist} strengthens this claim: it states that guessing $b$ remains hard even given an oracle implementing the corresponding binary-outcome measurement $\BMeas{\ketbra{\tau_b}_{\RegX,\RegY}}$.

By combining this lemma with the fact that our $\Transform$ procedure can always be truncated (in a further hybrid argument) to have strict $\poly(\secp,1/\varepsilon)$-runtime with $\varepsilon$-accuracy, we can prove the desired indistinguishability claims.

\paragraph{From the na\"ive simulator to the full simulator.} The problem with both the classical and quantum na\"ive simulators presented above is that their expected runtime is not polynomial. The issue is conceptually the same as in~\cref{subsubsec:gk-issue}. Consider a malicious verifier $V^*$ that gives a valid response with negligible probability $p$ when run on $a$ sampled as $(a,z) \gets \SHVZK.\Sim(0)$, and succeeds with probability $p-\negl$ when run on $a$ sampled as $(a,z) \gets \SHVZK.\Sim(r)$. Then the expected running time is $\frac{p}{p-\negl}$, which can be unbounded for small $p$.

The solution described in~\cite{JC:GolKah96} is therefore to \emph{estimate} the running time of the rewinding step before making the computational switch. That is, if the simulator obtains a valid response before the rewinding step, then it keeps running the $V^*$ on samples from $\SHVZK.\Sim(0)$ until it obtains $\secp$ additional valid responses. This gives the simulator an accurate estimate of the success probability of $V^*$, which it uses to bound the running time of the subsequent rewinding step.

We give a quantum simulator in $\EQPTC$ for the~\cite{JC:GolKah96} protocol that implements the analogous quantum version of this estimation trick. As in~\cref{subsubsec:gk-issue}, the idea is to first compute an upper bound on the runtime of the $\Transform$ step (equivalently, a lower bound on the singular values) after measuring $\sC$ in Step $3^*$ \emph{before} measuring $r$. This estimate is computed using a variable-length $\Estimate^{\sS_0,\sC}$ procedure, and since the $\Transform$ step has now been restricted to run in fixed polynomial time, we achieve the desired $p \cdot 1/p = 1$ cancellation in the expected running time.

Implementing this properly requires several tweaks to our simulator. In particular, the simulator no longer measures the verifier's challenge $r'$ directly in Step $3^*$; recording $r'$ is now delegated to $U$, since this step must be performed ``in between'' $\Estimate$ and $\Transform$. That is, we must modify $U$ so that instead of just performing (a coherent implementation of) $\Transform$, it runs the following steps coherently: (1) perform a variable-length $\Estimate$, where $\Estimate$ is parameterized by the same projectors as $\Transform$ (2) compute and measure the verifier's response (3) run $\Transform$ using the time bound computed from $\Estimate$. We defer further details to the full proof (\cref{sec:gk}). We remark that just as in~\cref{subsubsec:gk-issue}, the $\negl(\secp)$ error incurred by the collapsing measurement moves into the \emph{correctness} error of the simulation.

\subsection{Related Work}\label{sec:related-work}

\paragraph{Post-Quantum Zero-Knowledge.} The first construction of a zero-knowledge protocol secure against quantum adversaries is due to Watrous \cite{STOC:Watrous06}. Roughly speaking, \cite{STOC:Watrous06} shows that ``partial simulators'' that succeed with an inverse polynomial probability that is \emph{independent} of the verifier state can be extended to full post-quantum zero-knowledge simulators. This technique handles sequential repetitions of classical $\Sigma$-protocols and has been used as a subroutine in other contexts (e.g., \cite{STOC:BitShm20,C:BCKM21b,C:ChiChuYam21,C:AnaChuLap21}), but its applicability is limited to somewhat special situations. Nevertheless, most prior post-quantum zero-knowledge results have relied crucially on the \cite{STOC:Watrous06} technique.

\cite{STOC:BitShm20,TCC:AnaLap20} recently introduced a beautiful \emph{non-black-box} technique that, in particular, achieves constant-round zero knowledge arguments for $\mathsf{NP}$ with \emph{strict} polynomial time simulation \cite{STOC:BitShm20}. As discussed above, the use of non-black-box techniques is necessary to achieve strict polynomial time simulation in the classical \cite{STOC:BarLin02} and quantum \cite{FOCS:CCLY21} settings (and in the quantum setting this extends to $\EQPTM$ simulation).

Finally, recent work \cite{C:ChiChuYam21} showed that the Goldreich--Kahan protocol achieves post-quantum $\epsilon$-zero knowledge. This is closely related to our \cref{thm:gk}, and so we present a detailed comparison below.

\paragraph{Comparison with \cite{C:ChiChuYam21}.} The post-quantum security of the Goldreich--Kahan protocol was analyzed previously in \cite{C:ChiChuYam21}. Our simulation strategy for \cref{thm:gk} is related to that of \cite{C:ChiChuYam21} in that the two simulators both consider the Jordan decomposition for essentially the same pair of projectors, but the two simulators are otherwise quite different.

At a high level, \cite{C:ChiChuYam21} constructs a (non-trivial) quantum analogue of the following classical simulator: given error parameter $\epsilon$, repeat $\poly(1/\epsilon)$ times: sample $a \gets \Sim(0)$ and run $V^*$ on $a$. If $V^*$ ever opens correctly, record its response $r$. Then, run a single execution of the protocol using $(a, z) \gets \Sim(r)$ and output the result. 

More concretely, the \cite{C:ChiChuYam21} simulator first attempts to extract the verifier's challenge $r$ in $\poly(1/\eps)$ time, and then attempts to generate an accepting transcript in a single final interaction with the verifier. However, if the verifier \emph{aborts} in this final interaction, the simulation fails; this is roughly because successfully extracting $r$ skews the verifier's state towards not aborting. To obtain a full simulator, they use an idea from \cite{STOC:BitShm20}: (1) design a ``partial simulator'' that randomly guesses whether the verifier will abort in its final invocation, then achieves $\varepsilon$-simulation conditioned on a correct guess; (2) apply \cite{STOC:Watrous06}-rewinding to ``amplify'' onto executions where the guess is correct. 

It is natural to ask whether the above simulation strategy would have sufficed to prove \cref{thm:gk} (instead of writing down a new simulator). We remark that this is unlikely; their simulator seems to be tailored to $\eps$-ZK and, moreover, does not address what \cite{JC:GolKah96} describe as the main technical challenge in the classical setting: handling verifiers that abort with all but negligible probability. In more detail:
\begin{itemize}
    \item Their non-aborting simulator (like the classical analogue above) \emph{always} tries to extract $\Ch$. To achieve negligible simulation error, this extraction must succeed with all but negligible probability for any adversary that with inverse polynomial probability does not abort. This would require that the simulator run in superpolynomial time.
    
    Our simulator, as well as essentially all classical black-box ZK simulators, address this issue by first measuring whether the verifier aborts, and then only proceeding with the simulation in the non-aborting case.
    
    \item By Markov's inequality, expected polynomial time simulation implies $\eps$-simulation in time $O(1/\eps)$. As a function of $\eps$, the \cite{FOCS:CCLY21} simulator runs in some large polynomial time (as currently written, they appear to achieve runtime $1/\eps^6$, although it is likely unoptimized). Thus, even a hypothetical variable-runtime version of their simulator would not be expected polynomial time. In particular, the \cite{STOC:Watrous06,STOC:BitShm20} ``guessing'' compiler appears to cause a quadratic blowup in the runtime of their non-aborting simulator (due to a required smaller accuracy parameter).
    
    \item The \cite{STOC:Watrous06,STOC:BitShm20} ``guessing'' compiler adds an additional layer of complexity onto the \cite{C:ChiChuYam21} simulator that is incompatible with the $\EQPTC$ definition in the sense that given an $\EQPTC$ partial simulator, the \cite{STOC:Watrous06,STOC:BitShm20} ``guessing'' compiler would not produce a procedure in $\EQPTC$.
\end{itemize}

\noindent We also achieve some improvements over~\cite{C:ChiChuYam21} unrelated to the simulation accuracy:
\begin{itemize}
    \item \cite{C:ChiChuYam21} require that the underlying sigma protocol satisfies a \emph{delayed witness} property, which is not required in the classical setting. Our ``projector indistinguishability'' lemma (\cref{lemma:proj-indist}; see also~\cref{sec:tech-overview-pqzk-gk}) enables us to handle arbitrary sigma protocols.
    \item \cite{C:ChiChuYam21} require that the verifier commit to the sigma protocol challenge $r$ using a \emph{strong collapse-binding} commitment. Using a new proof technique (see~\cref{sec:unique-message-collapsing}), we show that standard collapse-binding suffices.
\end{itemize}

\paragraph{Post-Quantum Extraction}
As previously discussed, there is a line of prior work \cite{EC:Unruh12,EC:Unruh16,FOCS:CMSZ21} that achieves forms of post-quantum extraction that do \emph{not} preserve the prover state. Below we briefly discuss prior work on state-preserving post-quantum extraction.

\cite{STOC:BitShm20} directly constructs a state-preserving extractable commitment with \emph{non-black-box} extraction in order to achieve their zero-knowledge result. Their construction makes use of post-quantum fully homomorphic encryption (for quantum circuits). Their extractor homomorphically evaluates the adversarial sender.

\cite{STOC:BitShm20} also shows that constant-round zero-knowledge arguments and post-quantum secure function evaluation generically imply constant-round state-preserving extractable commitments. Combining this with \cite{STOC:Watrous06} yields a polynomial-round state-preserving extractable commitment scheme. Since this result also holds in the ``$\epsilon$ setting,'' plugging in \cite{C:ChiChuYam21} implies a constant-round $\epsilon$ state-preserving extractable commitment, although this protocol would have many rounds and is only privately verifiable.

All of the above results achieve computationally state-preserving extraction. \cite{C:AnaChuLap21} constructs a polynomial-round state-preserving extractable commitment scheme with \emph{statistical} state preservation. They use the \cite{STOC:Watrous06} simulation technique as the core of their extraction procedure, applied to a new construction where statistical state preservation is possible. 

%% file: 3-preliminaries.tex
\section{Preliminaries}
\label{sec:preliminaries}

The security parameter is denoted by $\secp$. A function $f \colon \Naturals \rightarrow [0,1]$ is \emph{negligible}, denoted $f(\secp) = \negl(\secp)$, if it decreases faster than the inverse of any polynomial. A probability is \emph{overwhelming} if is at least $1-\negl(\secp)$ for a negligible function $\negl(\secp)$. For any positive integer $n$, let $[n] \coloneqq \{1,2,\dots,n\}$. For a set $R$, we write $r \gets R$ to denote a uniformly random sample $r$ drawn from $R$.

%%%%%%%%%%%%%%%%%%%%%%%%%%%%%%%%%%%%%%%%%%%%%%%%%%%%%%%%%%%%%%%%%%%%%%%%%%%%%%%%
%%%%%%%%%%%%%%%%%%%%%%%%%%%%%%%%%%%%%%%%%%%%%%%%%%%%%%%%%%%%%%%%%%%%%%%%%%%%%%%%
\subsection{Quantum Preliminaries and Notation}
\label{sec:quantum-prelims}

\paragraph{Quantum information.}
A (pure) \emph{quantum state} is a vector $\ket{\psi}$ in a complex Hilbert space $\RegH$ with $\norm{\ket{\psi}} = 1$; in this work, $\RegH$ is finite-dimensional. We denote by $\Hermitians{\RegH}$ the space of Hermitian operators on $\RegH$. A \emph{density matrix} is a positive semi-definite operator $\DMatrix \in \Hermitians{\RegH}$ with $\Tr(\DMatrix) = 1$. A density matrix represents a probabilistic mixture of pure states (a mixed state); the density matrix corresponding to the pure state $\ket{\psi}$ is $\ketbra{\psi}$. Typically we divide a Hilbert space into \emph{registers}, e.g. $\RegH = \RegH_1 \otimes \RegH_2$. We sometimes write, e.g., $\DMatrix^{\RegH_1}$ to specify that $\DMatrix \in \Hermitians{\RegH_1}$.

A unitary operation is a complex square matrix $U$ such that $U \contra{U} = \Id$. The operation $U$ transforms the pure state $\ket{\psi}$ to the pure state $U \ket{\psi}$, and the density matrix $\DMatrix$ to the density matrix $U \DMatrix \contra{U}$.

A \emph{projector} $\Projector$ is a Hermitian operator ($\contra{\Projector} = \Projector$) such that $\Projector^2 = \Projector$. A \emph{projective measurement} is a collection of projectors $\ProjMeasurement = (\Projector_i)_{i \in S}$ such that $\sum_{i \in S} \Projector_i = \Id$. This implies that $\Projector_i \Projector_j = 0$ for distinct $i$ and $j$ in $S$. The application of $\ProjMeasurement$ to a pure state $\ket{\psi}$ yields outcome $i \in S$ with probability $p_i = \norm{\Projector_i \ket{\psi}}^2$; in this case the post-measurement state is $\ket{\psi_i} = \Projector_i \ket{\psi}/\sqrt{p_i}$. We refer to the post-measurement state $\Projector_i \ket{\psi}/\sqrt{p_i}$ as the result of applying $\ProjMeasurement$ to $\ket{\psi}$ and \emph{post-selecting} (conditioning) on outcome $i$. A state $\ket{\psi}$ is an \emph{eigenstate} of $\ProjMeasurement$ if it is an eigenstate of every $\Projector_i$.

A two-outcome projective measurement is called a \emph{binary projective measurement}, and is written as $\ProjMeasurement = \BMeas{\Projector}$, where $\Projector$ is associated with the outcome $1$, and $\Id - \Projector$ with the outcome $0$.

General (non-unitary) evolution of a quantum state can be represented via a \emph{completely-positive trace-preserving (CPTP)} map $T \colon \Hermitians{\RegH} \to \Hermitians{\RegH'}$. We omit the precise definition of these maps in this work; we only use the facts that they are trace-preserving (for every $\DMatrix \in \Hermitians{\RegH}$ it holds that $\Tr(T(\DMatrix)) = \Tr(\DMatrix)$) and linear.

For every CPTP map $T \colon \Hermitians{\RegH} \to \Hermitians{\RegH}$ there exists a \emph{unitary dilation} $U$ that operates on an expanded Hilbert space $\RegH \otimes \RegK$, so that $T(\DMatrix) = \Tr_{\RegK}(U (\rho \otimes \ketbra{0}^{\RegK}) U^{\dagger})$. This is not necessarily unique; however, if $T$ is described as a circuit then there is a dilation $\CPTPUnitary$ represented by a circuit of size $O(|T|)$.

For Hilbert spaces $\RegA,\RegB$ the \emph{partial trace} over $\RegB$ is the unique CPTP map $\Tr_{\RegB} \colon \Hermitians{\RegA \otimes \RegB} \to \Hermitians{\RegA}$ such that $\Tr_{\RegB}(\DMatrix_A \otimes \DMatrix_B) = \Tr(\DMatrix_B) \DMatrix_A$ for every $\DMatrix_A \in \Hermitians{\RegA}$ and $\DMatrix_B \in \Hermitians{\RegB}$.

A \emph{general measurement} is a CPTP map $\Measurement \colon \Hermitians{\RegH} \to \Hermitians{\RegH \otimes \RegO}$, where $\RegO$ is an ancilla register holding a classical outcome. Specifically, given measurement operators $\{ M_{i} \}_{i=1}^{N}$ such that $\sum_{i=1}^{N} M_{i} M_{i}^{\dagger} = \Id$ and a basis $\{ \ket{i} \}_{i=1}^{N}$ for $\RegO$, $\Measurement(\DMatrix) \eqdef \sum_{i=1}^{N} (M_{i} \DMatrix M_{i}^{\dagger} \otimes \ketbra{i}^{\RegO})$. We sometimes implicitly discard the outcome register. A projective measurement is a general measurement where the $M_{i}$ are projectors. A measurement induces a probability distribution over its outcomes given by $\Pr[i] = \Tr(\ketbra{i}^{\RegO} \Measurement(\DMatrix))$; we denote sampling from this distribution by $i \gets \Measurement(\DMatrix)$.

The \emph{trace distance} between states $\DMatrix,\DMatrixW$, denoted $d(\DMatrix,\DMatrixW)$, is defined as $\frac{1}{2}\Tr( \sqrt{(\DMatrix - \DMatrixW)^2})$. The trace distance is contractive under CPTP maps (for any CPTP map $T$, $d(T(\DMatrix),T(\DMatrixW)) \leq d(\DMatrix,\DMatrixW)$). It follows that for any measurement $\Measurement$, the statistical distance between the distributions $\Measurement(\DMatrix)$ and $\Measurement(\DMatrixW)$ is bounded by $d(\DMatrix,\DMatrixW)$. We have the following \emph{gentle measurement lemma}, which bounds how much a state is disturbed by applying a measurement whose outcome is almost certain.

\begin{lemma}[Gentle Measurement~\cite{Winter99}]
\label{lemma:gentle-measurement}
    Let $\DMatrix \in \Hermitians{\RegH}$ and $\ProjMeasurement = \BMeas{\Projector}$ be a binary projective measurement on $\RegH$ such that $\Tr(\Projector \DMatrix) \geq 1-\delta$. Let
    \[\DMatrix' \eqdef \frac{\Projector \DMatrix \Projector}{\Tr(\Projector \DMatrix)} \]
    be the state after applying $\ProjMeasurement$ to $\DMatrix$ and post-selecting on obtaining outcome $1$. Then
\begin{equation*}
d(\DMatrix,\DMatrix') \leq 2\sqrt{\delta}.
\end{equation*}
\end{lemma}

\begin{definition}
\label{def:almost-proj}
A real-valued measurement $\Measurement$ on $\RegH$ is \defemph{$(\varepsilon,\delta)$-almost-projective} if applying $\Measurement$ twice in a row to any state $\DMatrix \in \Hermitians{\RegH}$ produces measurement outcomes $p,p'$ where 
\begin{equation*}
\Pr[\abs{p-p'} \leq \varepsilon] \geq 1-\delta.
\end{equation*}
\end{definition}

\paragraph{Quantum algorithms.}
In this work, a \emph{quantum adversary} is a family of quantum circuits $\{ \Adversary_{\secp} \}_{\secp \in \Naturals}$ represented classically using some standard universal gate set. A quantum adversary is \emph{polynomial-size} if there exists a polynomial $p$ and $\secp_0 \in \Naturals$ such that for all $\secp > \secp_0$ it holds that $|\Adversary_{\secp}| \leq p(\secp)$ (i.e., quantum adversaries have classical non-uniform advice).

\subsection{Black-Box Access to Quantum Algorithms}
\label{subsec:blackbox}

Let $A$ be a polynomial-time quantum algorithm with internal state $\brho \in \mathrm{D}(\RegH)$ whose behavior is specified by a unitary $U$ on $\RegX \otimes \RegH$. A quantum oracle algorithm $S^A$ with \emph{black-box access} to $(A,\brho)$ is restricted to acting on $\RegH$ (which is initially set $\brho$) by applying the unitary $U$ or $U^\dagger$, but can freely manipulate $\RegX$ and an arbitrary external register $\RegY$.

Black-box access models sometimes permit the $U$ and $U^\dagger$ gates to be controlled on any external registers (i.e., any registers other than the registers $\RegZ \otimes \RegH$ to which $U$ is applied). We note that none of the black-box algorithms in this work require controlled access to $U,U^{\dagger}$. This is because our black-box use of $U,U^\dagger$ takes the form $U^{\dagger} (\Id_{\RegH} \otimes V_{\RegX,\RegY_1}) U$ where $V$ is a unitary acting only on $\RegX \otimes \RegY_1$, and we can replace $U,U^\dagger$ controlled on $\RegY_2$, with $V$ controlled on $\RegY_2$.

\paragraph{Algorithms with classical input and output.} We also consider the special case of quantum algorithms that take classical ``challenge'' $r$ and produce classical ``response'' $z$. Writing $\RegX = \RegR \otimes \RegZ$, an algorithm of this form is specified by a unitary $U$ on $\RegR \otimes \RegZ \otimes \RegH$ of the form $\sum_r \ketbra{r}_{\RegR} \otimes U^{(r)}_{\RegZ,\RegH}$. For example, $S^A$ can run $A$ on a superposition of inputs by instantiating $\RegR \otimes \RegZ$ to $\sum_{r} \ket{r}_{\RegR} \otimes \ket{0}_{\RegZ}$ and then applying $U$.

We note that this definition is consistent with the notions of interactive quantum machines and oracle access to an interactive quantum machine used in e.g.~\cite{EC:Unruh12} and other works on post-quantum zero-knowledge.

We remark that our formalism is tailored to the two-message challenge-response setting. While the protocols we analyze in this paper will have more than two messages of interaction, our analysis will typically center around two particular messages in the middle of a longer execution, and $\brho$ will be the intermediate state of the interactive algorithm right before the next challenge is sent. We also point out that the unitary $U$ can be treated as independent of the (classical) protocol transcript before challenge $r$ is sent, since we can assume this transcript is saved in $\brho$. 

%%%%%%%%%%%%%%%%%%%%%%%%%%%%%%%%%%%%%%%%%%%%%%%%%%%%%%%%%%%%%%%%%%%%%%%%%%%%%%%%
%%%%%%%%%%%%%%%%%%%%%%%%%%%%%%%%%%%%%%%%%%%%%%%%%%%%%%%%%%%%%%%%%%%%%%%%%%%%%%%%
\subsection{Jordan's Lemma}
\label{subsec:jordan}

We state Jordan's lemma and its relation to the singular value decomposition.

\begin{lemma}[\cite{Jordan75}]
\label{lemma:jordan}
For any two Hermitian projectors $\BProj{\MeasA}$ and $\BProj{\MeasB}$ on a Hilbert space $\RegH$, there exists an orthogonal decomposition of $\RegH = \bigoplus_j \Subspace_{j}$ into one-dimensional and two-dimensional subspaces $\{\Subspace_{j}\}_{j}$ (the \emph{Jordan subspaces}), where each $\Subspace_{j}$ is invariant under both $\BProj{\MeasA}$ and $\BProj{\MeasB}$. Moreover:
\begin{itemize}[noitemsep]
\item in each one-dimensional space, $\BProj{\MeasA}$ and $\BProj{\MeasB}$ act as identity or rank-zero projectors; and
\item in each two-dimensional subspace $\Subspace_{j}$, $\BProj{\MeasA}$ and $\BProj{\MeasB}$ are rank-one projectors. In particular, there exist distinct orthogonal bases $\{\JorKetA{j}{1},\JorKetA{j}{0}\}$ and $\{\JorKetB{j}{1},\JorKetB{j}{0}\}$ for $\Subspace_{j}$ such that $\BProj{\MeasA}$ projects onto $\JorKetA{j}{1}$ and $\BProj{\MeasB}$ projects onto $\JorKetB{j}{1}$.
\end{itemize}
\end{lemma}

A simple proof of Jordan's lemma can be found in~\cite{Regev06-XXX}.

For each $j$, the vectors $\JorKetA{j}{1}$ and $\JorKetB{j}{1}$ are corresponding left and right singular vectors of the matrix $\BProj{\MeasA} \BProj{\MeasB}$ with singular value $s_j = |\JorBraKetAB{j}{1}|$. The same is true for $\JorKetA{j}{0}$ and $\JorKetB{j}{0}$ with respect to $(\Id-\ProjA)(\Id-\ProjB)$.

\subsection{Commitment Schemes}

A \emph{commitment scheme} consists of a pair of PPT algorithms $\Gen,\Commit$ with the following properties.

\newcommand{\ExpHide}{\mathsf{Exp}^{\Adversary}_{\mathsf{hide}}}
\paragraph{Statistical/computational hiding.}
For an adversary $\Adversary$, define the experiment $\ExpHide(\secp)$ as follows.
\begin{enumerate}[noitemsep]
    \item $\Adversary(1^\secp)$ sends $(\ck,m_0,m_1)$ to the challenger.
    \item The challenger flips a coin $b \in \Bits$ and returns $\com \eqdef \Commit(\ck,m_b)$ to the adversary.
    \item The adversary outputs a bit $b'$. The experiment outputs $1$ if $b = b'$.
\end{enumerate}
We say that $(\Gen,\Commit)$ is statistically (resp. computationally) hiding if for all unbounded (resp. non-uniform QPT) adversaries $\Adversary$,
\begin{equation*}
    | \Pr[\ExpHide(\secp) = 1] - 1/2 | = \negl(\secp) ~.
\end{equation*}

\newcommand{\ExpBind}{\mathsf{Exp}^{\Adversary}_{\mathsf{bind}}}
\paragraph{Statistical/computational binding.}
For an adversary $\Adversary$, define the experiment $\ExpBind(\secp)$ as follows.
\begin{enumerate}[noitemsep]
    \item The challenger generates $\ck \gets \Gen(1^\secp)$.
    \item $\Adversary(\ck)$ sends $(m_0,\omega_0,m_1,\omega_1)$ to the challenger.
    \item The experiment outputs $1$ if $\Commit(\ck,m_0,\omega_0) = \Commit(\ck,m_1,\omega_1)$.
\end{enumerate}
We say that $(\Gen,\Commit)$ is statistically (resp. computationally) binding if for all unbounded (resp. non-uniform QPT) adversaries $\Adversary$,
\begin{equation*}
    \Pr[\ExpBind(\secp) = 1] = \negl(\secp) ~.
\end{equation*}

\newcommand{\ExpCollapse}{\mathsf{Exp}^{\Adversary}_{\mathsf{cl}}}
\newcommand{\MeasValid}{\Meas{\mathsf{valid}}}
\paragraph{Collapse binding.}
For an adversary $\Adversary$, define the experiment $\ExpCollapse(\secp)$ as follows.
\begin{enumerate}[noitemsep]
    \item The challenger generates $\ck \gets \Gen(1^\secp)$.
    \item $\Adversary(\ck)$ sends a commitment $\com$ and a quantum state $\DMatrix$ on registers $\RegM \otimes \RegW$.
    \item The challenger flips a coin $b \in \Bits$. If $b = 0$, the challenger does nothing. Otherwise, the challenger measures $\RegM$ in the computational basis.
    \item The challenger returns registers $\RegM \otimes \RegW$ to the adversary, who outputs a bit $b'$. The experiment outputs $1$ if $b = b'$.
\end{enumerate}
We say that $\Adversary$ is valid if measuring the output of $\Adversary(\ck)$ in the computational basis yields, with probability $1$, $(\com,m,\omega)$ such that $\Commit(\ck,m,\omega) = \com$.

We say that $(\Gen,\Commit)$ is collapse-binding if for all \emph{valid} non-uniform QPT adversaries $\Adversary$,
\begin{equation*}
    | \Pr[\ExpCollapse(\secp) = 1] - 1/2 | = \negl(\secp) ~.
\end{equation*}

\subsection{Preliminaries on Interactive Arguments}\label{subsec:protocol-prelims}

An interactive argument for an $\NP$-language $L$ consists of a pair of interactive algorithms $P, V$:

\begin{itemize}
    \item The prover algorithm $P$ is given as input an $\NP$ statement $x$ and an $\NP$ witness $w$ for $x$. 
    \item The verifier algorithm $V$ is given as input an $\NP$ statement $x$; at the end of the interaction, it outputs a bit $b$ (interpreted as ``accept''/``reject''). 
\end{itemize}

The minimal requirement we ask of such a protocol is \emph{completeness}, which states that when the honest $P, V$ algorithms are executed on a valid instance-witness pair $(x, w)$, the verifier should accept with probability $1-\negl(\secp)$. 

We typically consider interactive arguments consisting of either $3$ or $4$ messages. In many (but not all) settings we assume that the argument system is \emph{public-coin} (in the second-to-last round), meaning that the second-to-last message (or \emph{challenge}) is a uniformly random string $r$ from some domain. We will use the following notation to denote messages in any such protocol:

\begin{itemize}[noitemsep]
    \item For $4$-message public-coin protocols, we use $\mathsf{vk}$ to denote the first verifier message.
    \item We denote the first prover message by $a$. 
    \item We denote the verifier challenge by $r$.
    \item We denote the prover response by $z$. 
    \item We denote the verification predicate by $V(\vk, a, r, z)$. 
\end{itemize}

We consider $3$-message protocols as a special case of $4$-message protocols in which $\vk = \bot$. 

A key property of interactive protocols considered in this work is \emph{collapsing} (and relaxations thereof), defined below.

\begin{definition}[Collapsing Protocol \cite{EC:Unruh16,LiuZ19,DonFMS19}]
\label{def:collapsing-protocol}
An interactive protocol $(P,V)$ is \emph{collapsing} if for every polynomial-size interactive quantum adversary $A$ (where $A$ may have an arbitrary polynomial-size auxiliary input quantum advice state),
\begin{equation*}
\Big| \Pr[\mathsf{CollapseExpt}(0,A) = 1]
- \Pr[\mathsf{CollapseExpt}(1,A) = 1] \Big|
\leq \negl(\secp).
\end{equation*}
For $b \in \{0,1\}$, the experiment $\mathsf{CollapseExpt}(b,A)$ is defined as follows:
\begin{enumerate}[noitemsep]
	\item The challenger runs the interaction $\langle A,V\rangle$ between $A$ (acting as a malicious prover) and the honest verifier $V$, stopping just before the measurement the register $\RegZ$ containing the malicious prover's final message. Let $\tau'$ be the transcript up to this point  excluding the final prover message.
	\item The challenger applies a unitary $U$ to compute the verifier's decision bit $V(\tau',\RegZ)$ onto a fresh ancilla, measures the ancilla, and then applies $U^\dagger$. If the measurement outcome is $0$, the experiment aborts.
	\item If $b = 0$, the challenger does nothing. If $b = 1$, the challenger measures the $\RegZ$ register in the computational basis and discards the result.
	\item The challenger returns the $\RegZ$ register to $A$. Finally $A$ outputs a bit $b'$, which is the output of the experiment.
\end{enumerate}
\end{definition}

\cref{def:collapsing-protocol} captures the collapsing property of Kilian's interactive argument system \cite{STOC:Kilian92} (as well as other $\Sigma$-protocols that make use of ``strongly collapsing commitments'' \cite{C:ChiChuYam21}), but does not accurately capture protocols that make use of commitments satisfying statistical binding but not ``strict binding'' \cite{EC:Unruh12}. To capture these protocols, we introduce a partial-collapsing definition.

For a 3 or 4-message interactive protocol $(P,V)$, let $T$ denote the set of transcript prefixes $\prefix$ (i.e., the first message $a$ in a 3-message protocol or the first two messages $(\vk,a)$ in a 4-message protocol), let $R$ denote the set of challenges $r$ (the second-to-last message) and let $Z$ denotes the set of possible responses $z$ (the final message). Informally, such a protocol is \emph{partially collapsing} with respect to a function $f: T \times R \times Z \rightarrow \{0,1\}^*$ if the prover cannot detect a measurement of $f$.

\begin{definition}[Partially Collapsing Protocol]
\label{def:partial-collapsing-protocol}
Let $f: T \times R \times Z \rightarrow \{0,1\}^*$ be a public efficiently computable function. A 3 or 4-message interactive protocol $(P, V)$ is partially collapsing with respect to $f$ if for every polynomial-size interactive quantum adversary $A$ (where $A$ may have an arbitrary polynomial-size auxiliary input quantum advice state),
\begin{equation*}
\Big| \Pr[\mathsf{PCollapseExpt}(0,f,A) = 1]
- \Pr[\mathsf{PCollapseExpt}(1,f,A) = 1] \Big|
\leq \negl(\secp).
\end{equation*}
For $b \in \{0,1\}$, the experiment $\mathsf{PCollapseExpt}(b,f,A)$ is defined as follows:
\begin{enumerate}[noitemsep]
	\item The challenger runs the interaction $\langle A,V\rangle$ between $A$ (acting as a malicious prover) and the honest verifier $V$, stopping just before the measurement the register $\RegZ$ containing the malicious prover's final message. Let $(\prefix,r)$ be the transcript up to this point (i.e., excluding the final prover message).
	\item The challenger applies a unitary $U$ to compute the verifier's decision bit $V(\tau',\RegZ)$ onto a fresh ancilla, measures the ancilla, and then applies $U^\dagger$. If the measurement outcome is $0$, the experiment aborts.
	\item If $b = 0$, the challenger does nothing. If $b = 1$, the challenger initializes a fresh ancilla $\RegY$ to $\ket{0}_{\RegY}$, applies the unitary $U_f$ (acting on $\RegZ \otimes \RegY$) that computes $f(\prefix,r,\cdot)$ on $\RegZ$ and XORs the output onto $\RegY$, measures $\RegY$ and discards the result, and then applies $U_f^\dagger$.
	\item The challenger returns the $\RegZ$ register to $A$. Finally $A$ outputs a bit $b'$, which is the output of the experiment.
\end{enumerate}
\end{definition}

\cref{def:partial-collapsing-protocol} captures the collapsing property of standard commit-and-open $\Sigma$-protocols \cite{FOCS:GolMicWig86,Blum86} that make use of statistically binding (or, more generally, standard collapse-binding \cite{EC:Unruh16,C:ChiChuYam21}) commitments by setting $f$ to output the part of $z$ corresponding to the committed message (but not the opening). In some other cases (a subroutine of the \cite{FOCS:GolMicWig86} graph non-isomorphism protocol, as well as the \cite{C:LapSha90} ``reverse Hamiltonicity'' $\Sigma$-protocol) we will use more complicated definitions of $f$ that measure different pieces of information depending on the challenge $r$.

Finally, we recall the definition of special honest-verifier zero knowledge.

\begin{definition}[Special honest-verifier zero knowledge]\label{def:shvzk}
    A 3-message sigma protocol $(P_{\Sigma},V_{\Sigma})$ is \emph{special honest verifier zero knowledge} (SHVZK) if there exists an algorithm $\SHVZK.\Sim$ such that for all $(x,w) \in \Relation$ and challenges $\Ch \in R$, the distributions
    \begin{equation*}
        \SHVZK.\Sim(x,\Ch) \quad\text{and}\quad (\Co,\Resp) \gets P_{\Sigma}(x,w,\Ch)
    \end{equation*}
    are computationally indistinguishable.
\end{definition}

%% file: 4-collapse-unique.tex
\section{Standard Collapse-Binding Implies Unique Messages}
\label{sec:unique-message-collapsing}

Recall that the standard collapse-binding security property ensures that if an efficient adversary produces a superposition of valid message-opening pairs $(m,\omega)$ to a commitment $c$, then it cannot detect whether a measurement of $m$ is performed. There is an \emph{apparent} deficiency with this definition as compared to the classical binding definition, which Unruh (implicitly) observes in~\cite{EC:Unruh12,EC:Unruh16}: collapse-binding does not seem to imply that an adversary cannot give valid openings to two different messages \emph{if the openings themselves are not measured}. 

This issue has received relatively little attention, in part because circumventing it turns out to be fairly easy in many cases by either modifying the underlying protocol, or by simply assuming ``strong'' collapse-binding \cite{C:ChiChuYam21} where the measurement of the message \emph{and opening} is undetectable. For example:
\begin{itemize}
\item In~\cite{EC:Unruh12}, Unruh introduces the notion of a \emph{strict-binding} commitment, defined so that for any commitment $c$, there is a unique valid message-opening pair $(m,\omega)$. Unruh shows that standard $\Sigma$-protocols (such as GMW 3-coloring and Blum Hamiltonicity) are sound when instantiated with strict-binding commitments, but due to the issue described above, is unable to prove that these protocols are sound when instantiated with a statistically-binding commitment. 
\item In~\cite{EC:Unruh16}, Unruh gives a generic transformation which converts a classically secure $\Sigma$-protocol into a quantum proof of knowledge by committing to the responses to each challenge in advance. However, in many $\Sigma$-protocols (e.g. \cite{FOCS:GolMicWig86,Blum86}) the response already consists of an opening to a commitment; are these protocols secure if the commitment is collapse-binding?
\item This issue also arises in~\cite{C:ChiChuYam21}, which explicitly asks for a strong collapse-binding commitment to instantiate their $\Sigma$-protocols. (They do note that a statistically binding commitment also suffices via a different argument.)
\end{itemize}

We believe this is an unsatisfying state of affairs. Collapse-binding is widely accepted as the quantum analogue of classical computational binding, but as the above examples illustrate, there are many natural settings where it is unclear whether it can be used as a drop-in replacement for classically binding commitments. Given this issue, a natural suggestion would be to treat strong collapse-binding as the quantum analogue of classical binding. However, we suggest that any definition of quantum computationally binding should at least capture statistically binding commitments. Statistically binding commitments do not generically satisfy strong collapse-binding, but are (standard) collapse-binding. Worse, strong collapse-binding is not a ``robust'' notion: we can make any commitment scheme lose its strong collapse-binding property by adding a single bit to the opening that the receiver ignores.

In this section, we resolve this difficulty and show that standard collapse-binding generically implies that an adversary cannot give two valid openings for two different messages, even when the openings are left unmeasured. This simplifies some of the proofs in this work, and also implies that strong collapse-binding and strict binding are unnecessary in the above examples.

Towards proving this, we first formalize a natural security property that captures the fact that a quantum adversary should only be able to open to a unique message.

Let $\Com = (\Gen,\Commit)$ be a non-interactive commitment scheme. Define the following challenger-adversary interaction $\mathsf{Exp}_{uniq}^{\mathsf{Adv}}(\secp)$ where $\Adv = (\Adv_1,\Adv_2)$ is a two-phase adversary.
\begin{enumerate}
\item\label[step]{step:unique-1} The challenger generates $\ck \gets \Gen(\secp)$.
\item\label[step]{step:unique-2} Run $\mathsf{Adv}_1(\ck)$ to output a classical commitment string $\com$, a classical message $m_1$ and a superposition of openings on register $\RegW$. It also returns its internal state $\RegH$, which is passed onto $\Adv_2$.
\item\label[step]{step:unique-3} The challenger measures whether $\RegW$ contains a valid opening for $m_1$ with respect to $\com$ and aborts (and outputs $0$) if not.
\item\label[step]{step:unique-4} Run $\mathsf{Adv}_2(\ck)$ on $(\RegH,\RegW)$. It outputs another message $m_2$ and a superposition of openings on register $\RegW$. If $m_2 = m_1$ then the expriment aborts and outputs $0$.
\item\label[step]{step:unique-5} The challenger measures whether $\RegW$ contains a valid opening for $m_1$ with respect to $\com$. If so, the experiment outputs $1$, otherwise $0$.
 \end{enumerate}
 
\begin{definition}
We say that a commitment is \emph{unique-message binding} if it can only be opened to a unique message if for all QPT adversaries $\mathsf{Adv}$,
\[ \Pr[\mathsf{Exp}_{uniq}^{\mathsf{Adv}}(\secp) =1] = \negl(\secp).\]
\end{definition}

\begin{lemma}\label{lemma:collapse-binding-unique-message}
Any collapse-binding commitment $\Com$ satisfies unique-message binding.
\end{lemma}

We remark that the unique-message binding definition and this lemma easily extend to interactive collapse-binding commitments. However, we will focus on the non-interactive case for simplicity. Our proof is reminiscent of the ``control qubit'' trick used by Unruh in~\cite{Unruh16-asiacrypt} to prove that collapse-binding implies a notion called sum-binding.

\begin{proof}
Suppose that $\mathsf{Adv} = (\Adv_1,\Adv_2)$ satisfies $\Pr[\mathsf{Exp}_{uniq}^{\mathsf{Adv}}(\secp) =1] = \varepsilon(\secp) = \varepsilon$. Then we construct an adversary $\mathsf{Adv}'$ that obtains advantage $\varepsilon/8$ in the collapsing game for $\Com$ as follows:
\begin{enumerate}
\item Upon receiving $\ck$ from the challenger, $\mathsf{Adv}'$ does the following:
\begin{enumerate}
\item Run $\Adv_1(\ck)$ to obtain a classical commitment $\com$, a classical message $m_1$ (on register $\RegM$), and registers $\RegW, \RegH$.
\item\label[step]{step:collapsing-first-measurement} Measure whether $\RegW$ contains a valid opening for $m_1$ with respect to $\com$; if the opening is invalid, abort and output a random $b'$.\footnote{To match the syntax of the collapsing game, the ``abort'' works as follows: $\mathsf{Adv}'(\ck)$ initializes $\RegM \otimes \RegW$ to some valid commitment, sends it to the challenger, ignores the registers it gets back, and then outputs a random $b'$.}
\item\label[step]{step:apply-U} Next, prepares an ancilla qubit $\RegB$ in the state $\ket{+}_{\RegB}$ and then apply the unitary $U$ defined as 
\[ U = \ketbra{1}_{\RegB} \otimes U^{\Adv_2}_{\RegH,\RegM,\RegW} + \ketbra{0}_{\RegB} \otimes \Id_{\RegH,\RegM,\RegW}.\]
where $U^{\Adv_2}$ is a unitary description of $\Adv_2$ (the action of $\Adv_2$ on $\RegH \otimes \RegM \otimes \RegW$ is unitary without loss of generality). That is, the unitary $U$ has two branches of computation: it does nothing when $\RegB = 0$, and it runs $\Adv_2$ when $\RegB = 1$.
\item\label[step]{step:measure-valid} Next apply the binary projective measurement $\BMeas{\BProj{\ck,\com,m_1}}$ where 
\[\BProj{\ck,\com,m_1} \coloneqq \ketbra{0}_{\RegB} \otimes \Id_{\RegH,\RegM,\RegW} + \ketbra{1}_{\RegB} \otimes \Id_{\RegH} \otimes \sum_{\substack{m,\omega \ : \  m \neq m_1 \wedge \\ \Commit(\ck,m,\omega) = \com}} \ketbra{m,\omega}_{\RegM,\RegW}. \]
This measurement checks that after applying $U$, the output of $\Adv_2$ (when $\RegB = 1$) is a valid message and opening $(m,\omega)$ where $m \neq m_1$. If this measurement rejects, abort and output a random $b'$.
\item Finally, send $\RegM \otimes \RegW$ to the collapsing challenger.
\end{enumerate}
\item\label[step]{step:uncompute-U} When the collapsing challenger returns $\RegM \otimes \RegW$, apply $U^\dagger$.
\item\label[step]{step:distinguish} Perform the binary projective measurement $\BMeas{\BProj{+}}$ where $\BProj{+} \coloneqq \ketbra{+}_{\RegB} \otimes \Id_{\RegH,\RegM,\RegW}$. If the measurement outcome is $1$ (corresponding to $\ket{+}$), then $\mathsf{Adv}'$ outputs $b' = 0$ (i.e., guesses that the collapsing challenger did not measure the message). Otherwise, it outputs $b' = 1$.
\end{enumerate}

We now compute the probability that $\Adv'$ outputs $b' = b$ for each choice of the collapsing challenge bit $b$.

If $b = 1$, then $\Adv'$ guesses correctly (outputs $b' = 1$) with probability exactly $1/2$. This is because if $\Adv'$ aborts, it outputs $1$ with probability $1/2$ by definition, and if it does not abort, then it sends the collapsing challenger $\RegM \otimes \RegW$ where a measurement of the $\RegM$ register will completely determine the $\RegB$ register. In particular, if the outcome of the $\RegM$ measurement is $m_1$, $\RegB$ collapses to $\ket{0}$; otherwise, $\RegB$ collapses to $\ket{1}$. In either case, the probability that the measurement of $\BMeas{\BProj{+}}$ returns $0$ (making $\Adv'$ output $b' = 1$) is exactly $1/2$.

We now consider the case $b = 0$. Let $\DMatrix = \sum_{\ck,\com,m_1} \DMatrix_{\ck,\com,m_1} + \DMatrix_{\bot}$ be the state on $\RegM \otimes \RegW \otimes \RegH$ after \cref{step:collapsing-first-measurement}, where $\DMatrix_{\ck,\com,m_1}$ is the (subnormalized) state corresponding to outcomes $\ck,\com,m_1$ and the outcome ``valid'' in \cref{step:collapsing-first-measurement}, and $\DMatrix_{\bot}$ is the (subnormalized) state corresponding to the outcome ``invalid'' in \cref{step:collapsing-first-measurement}.

Recall that in the case $b = 0$, the collapsing challenger does nothing to $\RegM \otimes \RegW$. Thus the effect of \cref{step:apply-U,step:measure-valid,step:uncompute-U} is to apply a binary projective measurement $\BMeas{\BProj{\ck,\com,m_1}'}$ where $\BProj{\ck,\com,m_1}' \coloneqq U^\dagger \Pi_{\ck,\com,m_1} U$. From the description of the experiment, it holds that
\begin{align*}
    \Pr[b' = 0] &= \frac{1}{2} \Tr(\DMatrix_{\bot}) + \frac{1}{2} \sum_{\ck,\com,m_1} \Tr((\Id-\BProj{\ck,\com,m_1}') (\ketbra{+} \otimes \DMatrix_{\ck,\com,m_1})) \\
    &\quad + \sum_{\ck,\com,m_1} \Tr(\BProj{+} \BProj{\ck,\com,m_1}' (\ketbra{+} \otimes \DMatrix_{\ck,\com,m_1}) \BProj{\ck,\com,m_1}') \\
    & \geq \frac{1}{2} - \frac{1}{2} \sum_{\ck,\com,m_1} \Tr(\BProj{\ck,\com,m_1}' (\ketbra{+} \otimes \DMatrix_{\ck,\com,m_1})) \\
    &\quad + \sum_{\ck,\com,m_1} \Tr(\DMatrix_{\ck,\com,m_1}) \left(\frac{\Tr(\BProj{\ck,\com,m_1}' (\ketbra{+} \otimes \DMatrix_{\ck,\com,m_1}))}{\Tr(\DMatrix_{\ck,\com,m_1})}\right)^2 \\
    &\geq \frac{1}{2} - \frac{1}{2} \sum_{\ck,\com,m_1} \Tr(\BProj{\ck,\com,m_1}' (\ketbra{+} \otimes \DMatrix_{\ck,\com,m_1})) \\
    &\quad + \frac{\left(\sum_{\ck,\com,m_1} \Tr(\BProj{\ck,\com,m_1}' (\ketbra{+} \otimes \DMatrix_{\ck,\com,m_1})) \right)^2}{\sum_{\ck,\com,m_1} \Tr(\DMatrix_{\ck,\com,m_1})}
\end{align*}
where the latter inequality is Jensen, and the former is the following:
\begin{claim}
    If $\ProjA \DMatrix = \DMatrix$ then $\Tr(\ProjA \ProjB \DMatrix \ProjB) \geq \Tr(\ProjB \DMatrix)^2/\Tr(\DMatrix)$.
\end{claim}
\begin{proof}
    $\Tr(\ProjB \DMatrix) = \Tr(\ProjB \DMatrix \ProjA) = \Tr(\ProjA \ProjB \DMatrix) \leq \sqrt{\Tr(\ProjA \ProjB \DMatrix \ProjB \ProjA) \Tr(\DMatrix)}$, where the inequality is by Cauchy-Schwarz.
\end{proof}
Let $\gamma \eqdef \sum_{\ck,\com,m_1} \Tr(\DMatrix_{\ck,\com,m_1})$. Observe that $\sum_{\ck,\com,m_1} \Tr(\BProj{\ck,\com,m_1}' (\ketbra{+} \otimes \DMatrix_{\ck,\com,m_1})) = (\gamma + \varepsilon)/2$. It follows that
\begin{equation*}
    \Pr[b' = 0] = \frac{1}{2} - \frac{1}{4} (\gamma + \eps) + \frac{1}{4} \cdot \frac{(\gamma + \eps)^2}{\gamma} \geq \frac{1}{2} + \frac{\eps}{4} \enspace.
\end{equation*}

Thus the overall probability that $\Adv'$ guesses a random $b$ correctly in the collapsing experiment is at least $1/2 + \varepsilon(\secp)/8$. 
\end{proof}

%% file: 5-special-soundness.tex
\section{Generalized Notions of Special Soundness}
\label{sec:gss}

Let $(P,V)$ denote a 3 or 4-message public-coin interactive proof or argument system. Let $T$ denote the set of transcript prefixes $\prefix$ (i.e., the first message in a 3-message protocol or the first two messages in a 4-message protocol), $R$ denotes the set of challenges $r$ (the second-to-last message) and $Z$ denotes the set of possible responses $z$ (the final message). The instance $x$ is assumed to be part of $\prefix$, which allows us to capture protocols in which the instance is adaptively chosen by the prover in its first message.

We introduce generalizations of the special soundness property to capturing situations where
\begin{enumerate}
    \item the special soundness extractor is able to produce a witness given only a function $f(z)$ of the response $z$, and/or
    \item the extractor is only required to succeed (with some $1-\negl(\secp)$ probability) when the challenges are sampled from an ``admissible distribution.''
\end{enumerate}
The second property is related to the notion of probabilistic special soundness due to~\cite{FOCS:CMSZ21}.\footnote{A similar (but not identical) definition appears in an older version of~\cite{FOCS:CMSZ21}: \url{https://arxiv.org/pdf/2103.08140v1.pdf}.}

Throughout this section, $k$ will be a parameter specifying the number of (partial) transcripts required to extract.

\subsection{Generalized Special Soundness Definitions}

We first recall the standard definition of $k$ special soundness.
\begin{definition}[$k$-special soundness]\label{def:k-ss}
    An interactive protocol $(P, V)$ is $k$-special-sound if there exists an efficient extractor $\SSExtract: T \times (R \times Z)^k \rightarrow \{0,1\}^*$ such that given $\prefix,(r_i,z_i)_{i \in [k]}$ where each $r_i$ is distinct and for each $i$, $(\prefix,r_i,z_i)$ is an accepting transcript, $\SSExtract(\prefix,r_i,z_i)$ outputs a valid witness $w$ for the instance $x$ with probability $1$.
\end{definition}

In order to generalize this definition, we consider interactive protocols $(P,V)$ with a ``consistency'' predicate $g: T \times (R \times \{0,1\}^*)^* \rightarrow \{0,1\}$. The argument $\{0,1\}^*$ corresponds to some partial information $y$ about a response $z$. The consistency predicate should have the property that if $g(\prefix,(r_i,y_i)_{i \in [k]}) =1$, then $g(\prefix,(r_i,y_i)_{i \in G}) = 1$ for all subsets $G \subset [k]$. For any positive integer $k$, we define the set $\Consistent_k$ to be the subset of $T \times (R \times \{0,1\}^*)^k$ on which $g$ outputs $1$. We can extend $k$ special soundness to allow the $\SSExtract$ algorithm to produce a witness given only partial information $y_i$ of the responses $z_i$ provided that the ``partial transcripts'' satisfy consistency.

\begin{definition}[$(k,g)$-special soundness]\label{def:k-g-ss}
    An interactive protocol $(P, V)$ is $(k,g)$-special-sound if there exists an efficient extractor $\SSExtract_g: T \times (R \times \{0,1\}^*)^* \rightarrow \{0,1\}$ such that given $(\prefix,(r_i,y_i)_{i \in [k]}) \in \Consistent_k$ where each $r_i$ is distinct and for each $i$, $\SSExtract_g(\prefix,r_i,y_i)$ outputs a valid witness $w$ with probability $1$.
\end{definition}

Notice that all $k$-special-sound protocols with super-polynomial size challenge space are $(k,g)$-probabilistic-special sound for the ``trivial'' consistency predicate $g$ that simply checks (interpreting $y_i = z_i$ as a full response) whether all the transcripts are accepting.

\begin{claim}
\label{claim:kss-to-kfss}
For any $k$-special-sound protocol $(P,V)$, there exists a consistency predicate $g$ such that $(P,V)$ is $(k,g)$-special-sound.
\end{claim}

\begin{proof}
Define $g$ to output $1$ on input $\prefix,(r_i,y_i)_{i \in [k]}$ if and only if each $(\prefix,r_i,y_i)$ is an accepting transcript. It follows that the original $\SSExtract$ in the special soundness definition satisfies the requirements of the $(k, g)$-special soundness definition.
\end{proof}

When the challenge space $R$ is super-polynomial-size, we can generalize this definition even further so that the extractor need not succeed on worst-case $k$-tuples of distinct challenges, but only on $k$-tuples sampled from an ``admissible distribution.''

\begin{definition}[$Q$-admissible distribution]
\label{def:q-admissible-dist}
    A distribution $D_k$ over $R^k$ is \emph{admissible} if there exists a negligible function $\negl(\secp)$ and a sampling procedure $\mathsf{Samp}$ such that $D_k$ is $\negl(\secp)$-close to the output distribution of the following process:
    \begin{itemize}[nolistsep]
        \item $\mathsf{Samp}$ makes, in expectation, $Q(\secp)$ classical queries to an oracle $O_R$ that outputs a uniformly random challenge $r \gets R$ each time it is queried.
        \item $\mathsf{Samp}$ must produce its outputs as follows. Let $Q_{\mathrm{total}}$ be the total number of queries it makes to $O_R$. $\mathsf{Samp}$ specifies a set $\{i_1,\dots,i_k\} \subseteq [Q_{\mathrm{total}}]$, and its output is defined to be $r_{i_1},\dots,r_{i_k}$ where $r_i$ is the $i$th output of the uniform sampling oracle $O_R$.
        
        We stress that $\mathsf{Samp}$ may use an arbitrary (e.g., even inefficient) process to select the set $\{i_1,\dots,i_k\}$. Moreover, the output challenges $r_{i_1},\dots,r_{i_k}$ do not necessarily have distinct values (this can occur if the sampling oracle $O_R$ outputs the same challenge more than once).
    \end{itemize}
\end{definition}

\begin{definition}[admissible distribution]
\label{def:admissible-dist}
    A distribution $D_k$ over $R^k$ is \emph{admissible} if there exists $Q = \poly(\secp)$ such that $D_k$ is a $Q$-admissible distribution.
\end{definition}

\begin{definition}[$(k,g)$-probabilistic special soundness]\label{def:k-g-pss}
    An interactive protocol $(P, V)$ with consistency predicate $g$ is $(k,g)$-probabilistic-special-sound if there exists an efficient extractor $\SSExtract: T \times (R \times \{0,1\}^*)^k \rightarrow \{0,1\}^*$ such that for any distribution $D$ supported on $\Consistent_k$ whose marginal distribution on $R^k$ is admissible, 
    \[ \Pr_{(\prefix,(r_i,y_i)_{i \in [k]}) \gets D }[\PSSExtract_g(\prefix,(r_i,y_i)_{i \in [k]}) \rightarrow w \wedge w\text{ is a valid witness for }x] = 1-\negl(\secp)\]
\end{definition}

Note that $(k,g)$-probabilistic special soundness (PSS) is only meaningful when the challenge space $R$ has super-polynomial size. When $R$ is polynomial, an admissible distribution $D_k$ can simply output $(r,\dots,r)$ (the same challenge repeated $k$ times) since there exists a $\Samp$ that simply queries $O_R$ until it outputs the same challenge $k$ times.

However, when $|R|$ is superpolynomial, $(k,g)$-PSS is a relaxation of $(k,g)$-special soundness.

\begin{claim}
\label{k-g-ss-implies-k-g-pss}
When $R = 2^{\omega(\log \lambda)}$, any $(k,g)$-special-sound protocol is also $(k,g)$-probabilistic-special-sound.
\end{claim}

\begin{proof}
It suffices to prove that the probability any admissible distribution outputs the same challenge $r$ more than once is $\negl(\secp)$. By the definition of an admissible distribution, its output is $\negl(\secp)$-close to the output of an arbitrary sampling algorithm that makes an expected $\poly(\secp)$ number of queries to a uniform sampling oracle $O_R$ over $R$, and then outputs a size-$k$ subset of the oracle responses. 

Suppose towards a contradiction that there exists constant $c$ such that for infinitely many $\secp \in \mathbb{N}$, the sampling oracle $O_R$ outputs a repeated challenge with probability $1/\secp^c$. Let $d$ be a constant such that the expected number of queries to the uniform sampling oracle $O_R$ is $O(\secp^d)$. If $0 \leq q \leq \lambda^{d+c+1}$ oracle queries have already been made, the probability that the next oracle query allows finding a collision is at most $\lambda^{d+c+1}/|R|$. This implies that finding a collision within $\lambda^{d+c+1}$ queries is at most $\lambda^{2d+2c+2}/|R|$. Thus, to find a collision with probability at least $1/\lambda^c$, the number of oracle queries must be at least $\lambda^{d+c+1}$ with probability at least $1/\lambda^c - \lambda^{2d+2c+2}/|R|$, which implies the expected number of oracle queries is at least $\lambda^{d+c+1}(1/\lambda^c - \lambda^{2d+2c+2}/|R|) = \lambda^{d+1} - \lambda^{3d+3c+3}/|R|$. Since $R = 2^{\omega(\log \lambda)}$, there exists a constant $\lambda_0$ such that for all $\lambda > \lambda_0$, this expectation is $\lambda^{d+1} - \lambda^{3d+3c+3}/|R| > \lambda^{d+1} - 1$. This contradicts our assumption that the expected number of queries to the sampling oracle is $O(\lambda^d)$.
\end{proof}

\subsection{A Special Soundness Parallel Repetition Theorem}

Although it is well-known that $2$-special soundness is preserved under parallel repetition, the situation is more complicated for generalized special soundness notions (and even $k$-special soundness for larger values of $k$). We state and prove a useful theorem about the parallel repetition of special sound protocols.

\begin{lemma}
\label{lemma:parallel-repetition}
If $\Sigma = (P,V)$ is a $(k,g)$-special-sound protocol, then the $t=\Omega(k^2\log^2(\secp))$-fold parallel repetition $\Sigma^t$ is $(k^2,g^t)$-probabilistic special sound where $g^t$ outputs $1$ if and only (1) the arguments $y_i$ consist of $t$ formally separated components, and (2) $g$ outputs $1$ on each of the $t$ components.
\end{lemma}

\begin{proof}

Let $\Consistent_{k^2}$ be the set of $k$-tuples of shared-prefix partial transcripts of $\Sigma^t$ on which $g^t$ outputs $1$. Let $D$ be a distribution supported on $\Consistent_{k^2}$ whose marginal distribution on $(R^t)^{k^2}$ is admissible.

We construct $\PSSExtract_{g^t}$ for $\Sigma^t$ that takes as input
\[(\tau_{\mathsf{pre},j})_{j \in [t]},((r_{j,i})_{j \in [t]},(y_{j,i})_{j \in [t]})_{i \in [k^2]} \gets D\]
and does the following:
\begin{enumerate}
    \item Look for $j \in [t]$ such that $\{r_{j,i}\}_{i \in [k^2]}$ consists of $k$ distinct challenges. If no such $j$ exists, abort and output $\bot$.
    \item If such a $j$ exists, let $H$ be a size-$k$ subset of $[k^2]$ such that $\{r_{j,i}\}_{i \in H}$ consists of $k$ distinct challenges, and let $\SSExtract_g$ be the $(k,g)$-special-soundness extractor for $\Sigma$. Run $\SSExtract_g( \tau_{\mathrm{pre},j},(r_{j,i},y_{j,i})_{i \in H}) \rightarrow w$ and output $w$.
\end{enumerate}

First, we note that \[g(\tau_{\mathrm{pre},j},(r_{j,i},y_{j,i})_{i \in H}) = 1\]
follows from
\[g^t((\tau_{\mathsf{pre},j})_{j \in [t]},((r_{j,i})_{j \in [t]},(y_{j,i})_{j \in [t]})_{i \in [k^2]}) =1.\]

Thus, it suffices to prove that this extractor aborts with probability $\negl(\secp)$. Define $\mathsf{BAD} \subset R^{tk^2}$ to be the set of all $tk^2$-tuples $(r_{j,i})_{j \in [t], i \in [k^2]}$ such that for all $j \in [t]$, the $k^2$-tuple $(r_{j,i})_{i \in [k^2]}$ does not contain $k$ distinct challenges.

Suppose $(r_{j,i})_{j \in [t], i \in [k^2]}$ is sampled uniformly at random from $R^{tk^2}$. Then we have
\[ \Pr_{(r_{j,i})_{j \in [t], i \in [k^2]} \gets R^{tk^2}}[(r_{j,i})_{j \in [t], i \in [k^2]} \in \mathsf{BAD}] \leq \left(\frac{k}{e^k}\right)^t .\]
This follows from the fact that for any fixed $j$, the probability that $(r_{j,i})_{i \in [k^2]}$ does \emph{not} contain $k$ distinct challenges is at most $k((k-1)/k)^{k^2} \leq k/e^k$.

By the definition of an admissible distribution (\cref{def:admissible-dist}), the marginal distribution of $D_{\Sigma^t}$ on $(R^t)^{k^2}$ is the result of the following process (up to $\negl(\secp)$ statistical distance): make an expected $\poly(\secp)$ number of classical queries to a uniform sampling oracle $O_{R^t}$ over $R^t$, receiving a set of challenges $A$, and then (using an arbitrary procedure) output any size-$k^2$ subset $\{(r_{1,1},\dots,r_{t,1}),\dots,(r_{1,k^2},\dots,r_{t,k^2})\}$ of $A$ of $k^2$ challenges. The extractor aborts if $(r_{j,i})_{j \in [t], i \in [k^2]} \in \mathsf{BAD}$. 

Let $d$ be a constant such that the expected number of queries to the uniform sampling oracle $O_{R^t}$ is $O(\secp^d)$. Suppose towards a contradiction that the extractor aborts with non-negligible probability, i.e., there exists a constant $c$ such that for infinitely many $\secp \in \mathbb{N}$, the extractor aborts with probability at least $1/\lambda^c$. If $0 \leq q \leq \lambda^{d+c+1}$ oracle queries have already been made, the probability that the next oracle query allows finding a size-$k^2$ subset of outputs in $\mathsf{BAD}$ is at most
\begin{align*}
    (\lambda^{d+c+1})^{k^2}\left(\frac{k}{e^k}\right)^{k^2\log^2(\lambda)}.
\end{align*}
Moreover, there exists a constant $\lambda_0$ such that for all $\lambda > \lambda_0$, this can be upper bounded as
\begin{align*}
    (\lambda^{d+c+1})^{k^2}\left(\frac{k}{e^k}\right)^{k^2\log^2(\lambda)} < \left(\frac{2^{k^2 + k^2 \log (k)} }{e^{k^3}}\right)^{\log^2(\lambda)} < (1/2)^{\log^2(\lambda)} = 1/\lambda^{\log (\lambda)}.
\end{align*}
Thus for all $\lambda > \lambda_0$, the probability of finding a size-$k^2$ subset of oracle outputs in $\mathsf{BAD}$ within $\lambda^{d+c+1}$ oracle queries is at most $\lambda^{d+c+1} /\lambda^{\log(\lambda)}$; this implies that finding a size-$k^2$ subset of oracle outputs in $\mathsf{BAD}$ with probability $1/\lambda^c$ requires making at least $\lambda^{d+c+1}$ oracle queries with probability at least $1/\lambda^c - \lambda^{d+c+1}/\lambda^{\log(\lambda)}$. Then for $\lambda > \lambda_0$, the expected number of queries is at least $(\lambda^{d+c+1})(1/\lambda^c - \lambda^{d+c+1}/\lambda^{\log(\lambda)}) = \lambda^{d+1} - \lambda^{2d+2c+2}/\lambda^{\log(\lambda)}$. Since $c$ and $d$ are constants, there exists $\lambda_0'$ such that for $\lambda > \lambda_0'$, the expected number of queries is at least $\lambda^{d+1}-1$. This contradicts our assumption that the number expected number of queries to the uniform sampling oracle is $O(\lambda^d)$. \qedhere
\end{proof}

\subsection{Examples of Probabilistic Special Sound Protocols}\label{sec:examples}

We now show that many classical interactive proofs-of-knowledge (or arguments-of-knowledge) satisfy probabilistic special soundness. It was already noted above that (parallel repetitions of) standard special sound protocols satisfy the notion. Here, we highlight three other cases: commit-and-open protocols (where $g$ is only given partial transcripts), Kilian's protocol, and a subroutine of the \cite{FOCS:GolMicWig86} graph non-isomorphism protocol. 

\subsubsection{The ``one-out-of-two'' graph isomorphism subroutine}
In order to prove \cref{thm:szk}, we consider the following proof-of-knowledge subroutine of the \cite{FOCS:GolMicWig86} graph non-isomorphism protocol:

\begin{itemize}
    \item The subroutine instance is three graphs $G_0, G_1, H$. The prover\footnote{The \cite{FOCS:GolMicWig86} verifier acts as the prover in this subroutine.} wants to prove that there exists a bit $b$ such that $G_b$ is isomorphic to $H$. To do so, they execute a parallel repetition of the following protocol.
    \item The prover picks a random permutations $\sigma_0, \sigma_1$, a random bit $c$, and sends $(H_0 = \sigma_0 (G_c),H_1= \sigma_1(G_{1-c}))$ to the verifier.
    \item The verifier sends a random bit $r$.
    \item If $r=0$, the prover sends $(c, \sigma_0, \sigma_1)$ and the verifier checks that $(H_0 = \sigma_0 (G_c),H_1= \sigma_1(G_{1-c}))$ was computed correctly.
    \item If $r=1$, the prover sends $(c\oplus b, \sigma_{c\oplus b}\pi)$, where $\pi$ is an isomorphism mapping $H$ to $G_b$. The verifier then checks that $(\sigma_{c\oplus b}\pi) H = H_{c\oplus b}$.
\end{itemize}

In the classical setting, this is generally viewed as a proof of knowledge of $(b, \pi)$. However, we consider it as a proof of knowledge of the bit $b$, in the situation where $G_0$ and $G_1$ are not isomorphic. We will formalize this in \emph{two} different ways: first by showing that the protocol is $(2, g)$-special sound for a natural consistency predicate $g$, and then by showing that it is $(2, g')$-PSS for a more complicated predicate $g'$ that we have to use to be compatible with the protocol's limited partial collapsing property. 

First, we define an (inefficient) consistency predicate $g$, which is given as input $\prefix$ an arbitrary number of pairs $(\mathbf{r}, \mathbf{c}) \in \{0,1\}^{\secp}\times \{0,1\}^\secp$ (rejecting if the input is not of this form). $g$ outputs $1$ if the following conditions hold for all $\ell \in [\secp]$:
  \begin{itemize}
      \item If $r_\ell = 0$, the graphs $(H_{0,\ell}, H_{1,\ell})$ are isomorphic to $(G_{c_\ell}, G_{1-c_\ell})$.
      \item If $r_\ell = 1$, the graph $H_{{c_\ell}, \ell}$ is isomorphic to $H$.
  \end{itemize}

The following claim then holds immediately by transitivity of graph isomorphism. The extractor, given $(\mathbf{r}, \mathbf{c})$ and $(\mathbf{r}', \mathbf{c}')$ simply chooses an $\ell$ such that $r_\ell \neq r'_\ell$ and outputs $b = c_\ell \oplus c'_\ell$.

\begin{claim}
  If $G_0$ and $G_1$ are not isomorphic, then the \cite{FOCS:GolMicWig86} subroutine satisfies $(2, g)$-special soundness, where the extractor outputs the bit $b$. 
\end{claim}

Finally, we define the predicate $g'$ to be a slight modification of $g$: for the \emph{first} pair $(r^{(1)}, c^{(1)})$, $g'$ ignores\footnote{Alternatively, we could define $g'$ to require inputs with these $c^{(1)}_\ell$ omitted.} the bits $c^{(1)}_\ell$ for $i$ such that $r^{(1)}_\ell = 1$. The protocol will then not be $(2, g')$-special sound (e.g. a first transcript with $r^{(1)} = 1^\secp$ would provide no information), it \emph{will} be $(2, g')$-PSS.

\begin{claim}
\label{claim:gni-pss}
  If $G_0$ and $G_1$ are not isomorphic, then the \cite{FOCS:GolMicWig86} subroutine satisfies $(2, g')$-PSS, where the extractor outputs the bit $b$. 
\end{claim}

\begin{proof}
  This follows from the claim that if $(r^{(1)}, r^{(2)})\in \{0,1\}^\secp \times \{0,1\}^{\secp}$ is sampled according to an admissible distribution, then with all but $\negl(\secp)$ probability, there exists an index $\ell$ such that $r^{(1)}_\ell = 0$ and $r^{(2)}_\ell = 1$. This can be argued using the same reasoning as in the proof of~\cref{k-g-ss-implies-k-g-pss}, since the probability that two uniformly random $\lambda$-bit strings $r^{(1)}$ and $r^{(2)}$ do not have an index $\ell \in [\lambda]$ such that $r^{(1)}_\ell = 0$ and $r^{(2)}_\ell = 1$ is $\negl(\secp)$. 
\end{proof}

\subsubsection{Commit-and-Open Protocols}

The next class of examples we discuss is that of \emph{commit-and-open protocols}. In particular, we are interested in characterizing a special soundness property where the extractor is only given the opened \emph{messages} in the prover's response (and not their openings).

\begin{definition}\label{def:commit-and-open}
   Let $\Com$ denote a (possibly keyed) non-interactive commitment scheme. A commit-and-open protocol is a (3 or 4 message) protocol for an $\NP$ language $L$ of the following form:
   
   \begin{itemize}
       \item (Optional first verifier message) If $\Com$ is keyed, the verifier samples and sends the commitment key $\ck$ for $\Com$.
       \item The prover, given a witness $w$ for some statement $x\in L$, computes a string $y \in \{0,1\}^N$ and sends a bitwise commitment $a = \Com(\ck, y)$ to the verifier.
       \item The verifier samples a string $r$ that encodes a subset $S\subset [N]$ and sends $r$ to the prover.
       \item The prover sends openings to $\{y_i\}_{i\in S}$.
       \item The verifier checks that each opening to $y_i$ (for $i\in S$ is valid and then computes some function $\Check(y_S)$ on the opened bits. 
   \end{itemize}
   
   We say that such a protocol satisfies ``commit-and-open $k$-special soundness'' if there exists an extractor $\mathsf{Extract}(x, y)$ satisfying the following property. For every instance $x$ and every collection of $k$ \emph{distinct} sets $S_1, \hdots, S_k$ (represented by strings $(r_1, \hdots, r_k)$, for \emph{any} string $y$ such that $\Check(y_{S_i}) = 1$ for all $i$, $w= \mathsf{Extract}(x, y)$ is a valid $\NP$-witness for $x$. 
\end{definition}

It is not hard to see that the ``commit-and-open'' $k$-special soundness property, combined with the (computational/statistical) binding of the commitment scheme, implies a standard (computational/statistical) $k$-special soundness property of the $\Sigma$-protocol. However, we consider ``commit-and-open $k$-special soundness'' explicitly in order to satisfy (probabilistic) special soundness with respect to \emph{partial} transcripts. 

This definition captures extremely common $\Sigma$-protocols, such as:

\begin{itemize}
    \item The \cite{FOCS:GolMicWig86} $\Sigma$-protocol for $3$-coloring.
    \item A slight variant of the \cite{Blum86} $\Sigma$-protocol for Hamiltonicity\footnote{In this variant, in addition to committing to a permuted graph $\pi(G)$, the prover commits to the permutation $\pi$ and the permuted cycle $\pi \circ \sigma$. On the $0$ challenge, the prover additionally opens the commitment to $\pi$, and on the $1$ challenge, the prover additionally opens the commitment to $\pi \circ \sigma$.}
    \item Protocols following the ``MPC-in-the-head'' paradigm \cite{STOC:IKOS07}.
\end{itemize}

To view this in terms of generalized $k$-special soundness, define a consistency predicate $g$ as follows: on input $(\prefix,(r_i,\{m_{i,\ell}\}_{\ell \in S_i})_{i \in [k]})$, output $1$ if and only if
\begin{itemize}
    \item For any pair of sets $S_i, S_j$ (corresponding to challenges $r_i, r_j$), for any $\ell \in S_i \cap  S_j$, we have $m_{i,\ell} = m_{j,\ell}$. That is, the ``opened" message subsets are mutually consistent. 
    \item For all $i\in [k]$, $\Check(\{m_{i, \ell}\}_{\ell \in S_i}) = 1$.
\end{itemize}.
With this formalism in place, the following claim is immediate.

\begin{claim}
Any protocol satisfying commit-and-open $k$-special soundness (as described in \cref{def:commit-and-open}) is $(k, g)$-special sound.
\end{claim}

\subsubsection{Kilian's Protocol}\label{sec:kilian}

We briefly recall Kilian's protocol~\cite{STOC:Kilian92} instantiated with a collapsing hash function:

\begin{enumerate}
    \item The verifier samples a collapsing hash function $h \gets H_{\secp}$ and sends $h$ to the prover.
    \item Let $h_{\mathrm{Merkle}}$ be the Merkle hash function corresponding to $h$. The prover uses $w$ to compute a PCP $\pi$, and then sends $\mathsf{rt} = h_{\mathrm{Merkle}}(\pi)$ to the verifier.
    \item The verifier samples random coins $r$ and sends them to the prover.
    \item The prover computes the set of the PCP indices $q_r$ that the PCP verifier with randomness $r$ would check. It sends the corresponding values $\pi[q_r]$ along with the Merkle openings of $\mathsf{rt}$ on the positions $q_r$. 
    \item Finally, the verifier accepts if all the Merkle openings are valid and $V_{\mathrm{PCP},x}(r,\pi[q_r]) =1$, i.e., the PCP verifier with randomness $r$ accepts $\pi[q_r]$.
\end{enumerate}

We will instantiate Kilian's protocol with a PCP of knowledge, defined as follows. Let $\mathsf{WIN}_{\mathsf{PCP},x}(\pi)$ denote the probability that $\pi$ is accepted by the PCP verifier. 

\begin{definition}[PCP of Knowledge]
A PCP has knowledge error $\kappa_{\mathrm{PCP}}(\secp)$ if there is an extractor $\mathsf{E}_{\mathrm{PCP}}$ such that given any PCP $\pi$ where $\mathsf{WIN}_{\mathrm{PCP},x}(\pi) > \kappa_{\mathrm{PCP}}$, the extractor $\mathsf{E}_{\mathrm{PCP}}(\pi) \rightarrow w$ outputs a valid witness $w$ for $x$ with probability $1$.
\end{definition}

The following claim is due to~\cite{FOCS:CMSZ21}, though we have slightly rewritten it to match our definition of $k$-PSS.

\begin{claim}
\label{claim:kilian-pss}
Kilian's protocol instantiated with a PCP with knowledge error $\kappa_{\mathrm{PCP}}(\secp) = \negl(\secp)$ proof length $\ell(\secp)$, and alphabet-size $\Sigma(\secp)$ is $(k,g)$-PSS where $k = \ell \log(|\Sigma|)$ and the consistency function $g$ outputs $1$ on $(\tau_{\mathrm{pre}},(r_i,z_i)_{i \in [k]})$ if (1) for each $i$, the response $z_i$ contains PCP answers $\pi[q_{r_i}]$ such that $V_{\mathrm{PCP}}(x,r_i,\pi[q_{r_i}]) = 1$, and (2) for every $i \neq i'$ the answers $\pi[q_{r_i}]$ and $\pi[q_{r_{i'}}]$ agree on all indices in $q_{r_i} \cap q_{r_{i'}}$.
\end{claim}

\begin{proof}
Our extractor $\PSSExtract_g$ takes as input $(\prefix,(r_i,z_i)_{i \in [k]})$ and generates a witness as follows:
\begin{enumerate}
    \item\label[step]{step:pcp} Generate a PCP string $\pi \in \Sigma^{\ell}$ as follows. For each $t \in [\ell]$, check if $t \in q_{r_i}$ for any $i$. If so, pick such an $i$ arbitrarily and set $\pi[t]$ according to the value specified in $z_i$ (the choice of $i$ does not matter since the input satisfies consistency with respect to $g$). If there is no such $i$, set $\pi[t]$ arbitrarily.
    \item Run $\mathsf{E}_{\mathrm{PCP}}(\pi) \rightarrow w$ and output $w$.
\end{enumerate}
We prove that~\cref{step:pcp} constructs a PCP $\pi$ where $\mathsf{WIN}_{\mathrm{PCP},x}(\pi) > \kappa_{\mathrm{PCP}}$ with $1-\negl(\secp)$ probability whenever $(\tau_{\mathrm{pre}},(r_i,z_i)_{i \in [k]})$ is sampled from a distribution supported on $\Consistent_k$ (i.e., the subset of $T \times (R \times Z)^k$ where $g$ outputs $1$) whose marginal distribution on $R^k$ is admissible.

It suffices to prove that if $(r_1,\dots,r_k)$ are output by $\Samp$ (where $\Samp$ makes an expected $\poly(\secp)$ number of queries to a uniform sampling oracle $O_R$ and then outputs a size-$k$ subset of the outputs of $O_r$) then the probability there \emph{exists} $\pi \in \Sigma^\ell$ such that (1) $\mathsf{WIN}_{\mathrm{PCP},x}(\pi) \leq \kappa_{\mathrm{PCP}}$ and (2) $V_{\mathrm{PCP},x}(r_i,\pi[q_{r_i}]) = 1$ for all $i \in [k]$ is $\negl(\secp)$. This follows by invoking the definition of an admissible distribution, and observing that any $\pi$ resulting from~\cref{step:pcp} satisfies (2) by construction, which means that $\mathsf{WIN}_{\mathrm{PCP},x}(\pi) > \kappa_{\mathrm{PCP}}$ with probability $1-\negl(\secp)$.

Let $d$ be a constant such that for all $\lambda > \lambda_d$, $\Samp$ makes at most $\lambda^d$ queries to the sampling oracle $O_r$. Suppose towards contradiction that there exists a constant $c$ such that for infinitely many $\lambda$, the probability that $\Samp$ outputs $(r_1,\dots,r_k)$ such that with probability at least $1/\lambda^c$, there \emph{exists} $\pi \in \Sigma^\ell$ satisfying conditions (1) and (2) above. Thus, for infinitely many $\lambda$, the probability that $\Samp$ makes $2\lambda^{c+d}$ queries (or more) to its sampling oracle $O_R$ is at most $1/(2\lambda^c)$ by Markov's inequality. This means that even if $\Samp$ makes at most $2\lambda^{c+d}$ queries to its sampling oracle, it still succeeds with probability at least $1/(2\lambda^c)$ for infinitely many $\lambda$.

Consider any fixed PCP $\pi$ such that $\mathsf{WIN}_{\mathrm{PCP},x}(\pi) \leq \kappa_{\mathrm{PCP}}$. The probability that the PCP is accepting on at least $k$ challenges out of $2\lambda^{c+d}$ uniformly random challenges is at most 
\[\kappa_{\mathrm{PCP}}^k \cdot \binom{2\lambda^{c+d}}{k} \leq \kappa_{\mathrm{PCP}}^k (2\lambda^{c+d})^k.\]
By taking a union bound over all $\pi \in \Sigma^\ell$ we conclude that given $2\lambda^{c+d}$ uniformly random challenges, the probability there \emph{exists} a PCP $\pi$ such that $\mathsf{WIN}_{\mathrm{PCP},x}(\pi) \leq \kappa_{\mathrm{PCP}}$ and $\pi$ is accepting on at least $k$ of the $2\lambda^{c+d}$ challenges is at most 
\[ |\Sigma|^\ell \kappa_{\mathrm{PCP}}^k (2\lambda^{c+d})^k  = (|\Sigma|\cdot(2\lambda^{c+d} \kappa_{\mathrm{PCP}})^{\log(|\Sigma|)})^\ell, \]
where we have plugged in $k = \ell \log(|\Sigma|)$. Since $\kappa_{\mathrm{PCP}} = \negl(\secp)$, there exists $\lambda_0$ such that $2\lambda^{c+d} \kappa_{\mathrm{PCP}} < \frac{1}{4}$ for all $\lambda > \lambda_0$. Then for all $\lambda > \lambda_0$, we have 
\[ (|\Sigma|\cdot(2\lambda^{c+d} \kappa_{\mathrm{PCP}})^{\log(|\Sigma|)})^\ell < \frac{1}{|\Sigma|^\ell} \]
Since the PCP alphabet size is at least $|\Sigma| \geq 2$ and the PCP length is at least $\ell \geq \lambda$, the probability that $\Samp$ succeeds when restricted to making at at most $2\lambda^{c+d}$ queries to $O_R$ is at most $O(1/2^\lambda)$, which is a contradiction.
\end{proof}

%% file: 6-svt.tex
\section{Singular Vector Algorithms}

In this section we give algorithms for working with states that are singular vectors of a matrix $\ProjA \ProjB$, where $\ProjA,\ProjB$ are projectors. In \cref{sec:vrsvt} we give an algorithm that transforms left singular vectors to right singular vectors with negligible error. The runtime of the algorithm depends on the corresponding singular value.

\paragraph{Notation.}
Throughout this section we will consider the interaction between two binary projective measurements $\MeasA = \BMeas{\ProjA}, \MeasB = \BMeas{\ProjB}$.

We consider the matrix $\ProjA \ProjB$ and its singular value decomposition $V \Sigma W^{\dagger}$. Recall that $V,W$ are unitary and $\Sigma$ is a diagonal matrix. The columns of $V$ (resp. $W$) are the left (resp. right) singular vectors of $\ProjA \ProjB$, and the entries on the diagonal of $\Sigma$ are the singular values $s_j$. Note that the singular value decomposition is not in general unique; for the purposes of this section we fix one arbitrarily.

We denote left (resp. right) singular vectors of $\ProjA \ProjB$ with $s_j > 0$ by $\JorKetA{j}{1}$ (resp. $\JorKetB{j}{1}$). Define $\Subspace_j \eqdef \spanset(\JorKetA{j}{1}, \JorKetB{j}{1})$. If $s_j < 1$, then $\Subspace_j$ is two-dimensional. The $\Subspace_j$ correspond to the Jordan subspaces of $(\ProjA,\ProjB)$. As such, we also have $\JorKetA{j}{0},\JorKetB{j}{0} \in \Subspace_j$. A straightforward calculation shows that these are left and right singular vectors of $(\Id-\ProjA)(\Id-\ProjB)$ with singular value $s_j$. The Jordan subspace values $p_j$ are the squares of the corresponding singular values. In our setting it is more natural to use the squares (since they correspond to probabilities), and so the guarantees in this section are stated with respect to the squared singular values.

\subsection{Fixed-Runtime Algorithms}

In this section we recall a selection of algorithms for manipulating singular vectors of $\ProjA \ProjB$. All of these algorithms make black-box use of $U_{\MeasA},U_{\MeasB}$; we consider their complexity as circuits with $U_{\MeasA},U_{\MeasB}$ gates. All of these algorithms take as input some threshold $a \in (0,1]$, such that their correctness guarantee will hold for singular vectors of value at least $a$, and their running time is linear in $1/a$.

The first algorithm $\Transform$ implements a fixed-runtime singular vector transformation, taking left singular vectors to their corresponding right singular vectors.

\begin{theorem}[Singular vector transformation \cite{STOC:GSLW19}]
    \label{thm:svt}
    There is a uniform family of circuits $\{ \Transform_{a,\delta} \}_{a,\delta \in (0,1]}$ with $U_{\MeasA},U_{\MeasB}$ gates, of size $O(\log (1/\delta)/\sqrt{a})$, such that the following holds. Let $\JorKetA{j}{1}$ be a left singular vector of $\ProjA \ProjB$ with singular value $s_j$. If $a \leq s_j^2$, $\Transform_{a,\delta}[\MeasA \to \MeasB](\JorKetA{j}{1})$ outputs the state $\JorKetB{j}{1}$ with probability at least $1 - \delta$. Moreover, for all $a$, $\Subspace_j$ is invariant under $\Transform_{a,\delta}$.
\end{theorem}

The second algorithm $\Threshold$ implements a measurement determining, given a threshold $a$ and a singular vector with singular value $s_j$, whether $a \leq s_j$ or $a > 2s_j$ (and otherwise has no guarantee).

\begin{theorem}[Singular value threshold \cite{STOC:GSLW19}]
    \label{thm:svdisc}
    There is an algorithm $\Threshold$ which, for all binary projective measurements $\MeasA,\MeasB$, given black-box access to operators $U_{\MeasA},U_{\MeasB}$, achieves the following guarantee. Given $\delta > 0, b \geq \varepsilon > 0$ and a state $\JorKetA{j}{1}$ which is a left singular vector of $\ProjA \ProjB$ with singular value $s_j$:
    \begin{itemize}
        \item if $s_j^2 \geq b$, then $\Pr[\Threshold^{\MeasA,\MeasB}_{p,\varepsilon,\delta}(\JorKetA{j}{1}) \to 1] \geq 1 - \delta$, and
        \item if $s_j^2 \leq b-\varepsilon$, then $\Pr[\Threshold^{\MeasA,\MeasB}_{p,\eps,\delta}(\JorKetA{j}{1}) \to 1] \leq \delta$.
    \end{itemize}
    Moreover, $\Subspace_j$ is invariant under $\Threshold$, and if the outcome is $1$ the post-measurement state is $\JorKetA{j}{1}$. $\Threshold$ runs in time $O(\log (1/\delta)\sqrt{b}/\varepsilon)$.
\end{theorem}

Next, we describe an algorithm which, with access to $U_{\MeasA},U_{\MeasB}$, can ``flip'' a singular vector state from $\image(\Id - \ProjA)$ to $\image(\ProjA)$ using $\ProjB$, provided that the singular value is sufficiently far from both $0$ and $1$.
\begin{lemma}
    Let $\ProjA,\ProjB$ be projectors. There is an algorithm $\Flip_{\varepsilon}[\ProjA, \ProjB]$ which, on input a state $\JorKetA{j}{0}$ that is a left singular vector of $\ProjA \ProjB$ with $\varepsilon \leq s_j^2 \leq 3/4$, outputs the state $\JorKetA{j}{1}$ with probability $1-\delta$ in time $O(\log(1/\delta)/\sqrt{\varepsilon})$. $\Flip$ is invariant on the subspace spanned by $\{ \JorKetA{j}{1},\JorKetA{j}{0} \}$.
\end{lemma}
\begin{proof}
    The algorithm operates as follows:
    \begin{enumerate}[nolistsep]
        \item \label[step]{step:alternate-to-1} Apply $\MeasA,\MeasB$ in an alternating fashion until either $\MeasA \to 1$, $\MeasB \to 1$ or $3\log(1/\delta)$ measurements have been applied.
        \item If $\MeasA \to 1$, stop.
        \item If $\MeasB \to 1$, apply $\Transform_{\varepsilon,\delta}[\ProjB, \ProjA]$.
    \end{enumerate}
    The lemma follows since the probability that \cref{step:alternate-to-1} takes more than $k$ steps is $(3/4)^k$, and then by the guarantee of $\Transform$.
\end{proof}

\subsection{Variable-Runtime Singular Vector Transformation (vrSVT)}
\label{sec:vrsvt}
In this section we describe our variable-runtime SVT algorithm. In fact, for technical reasons our algorithm consists of two parts: a variable-runtime \emph{singular value estimation} procedure which \emph{preserves} singular vectors, and a \emph{singular vector transformation} procedure which transforms left singular vectors to right singular vectors, whose running time is fixed given a classical input from the estimation procedure.

Below we give a proof of \cref{thm:vrsvt} that makes use of the singular value discrimination and singular vector transformation algorithms of \cite{STOC:GSLW19}. We note that it is possible to prove \cref{thm:vrsvt} via more ``elementary'' means using high-probability phase estimation \cite{NagajWZ11} and amplitude amplification. Indeed, phase estimation for $(2\Id - \ProjA)(2\Id - \ProjB)$ is equivalent to singular value estimation for $\ProjA \ProjB$ and amplitude amplification can be viewed as a (non-coherent) singular vector transformation.

\begin{theorem}[Two-stage variable-runtime singular vector transformation]
    \label{thm:vrsvt}
    Let $\MeasA = \BMeas{\ProjA}$, $\MeasB = \BMeas{\ProjB}$ be projective measurements. There is a pair of algorithms $\VarEstimate[\ProjA \rightleftarrows \ProjB]$ and $\Transform[\ProjB \to \ProjA]$ with $U_{\MeasA}$ and $U_{\MeasB}$ gates with the following properties. Let $\JorKetB{j}{1}$ be a left singular vector of $\ProjA \ProjB$ with singular value $s_j > 0$, and let $\JorKetA{j}{1}$ be the corresponding right singular vector. Then
    \begin{enumerate}
        \item The subspace $\Subspace_j$ is invariant under both $\VarEstimate$ and $\Transform$.
        \item The running time of $\VarEstimate(\JorKetA{j}{1})$ is $O(\log(1/\delta)/s_j)$ with probability $1-\delta$ and $O(\log(1/\delta)/\delta)$ with probability $1$.
        \item The output $(q,\ket{\psi}) \gets \VarEstimate(\JorKetA{j}{1})$ is such that $\ket{\psi} = \JorKetB{j}{1}$ with probability $1-\delta$.
        \item The running time of $\Transform(q,\ket{\psi'})$, where $(\gamma,\ket{\psi}) \gets \VarEstimate(\JorKetA{j}{1})$ and $\ket{\psi'}$ is any state, is $O(\log (1/\delta)/s_j)$ with probability $1-\delta$ and at most $1/\delta$ with probability $1$.
        \item The output state of $\Transform(\VarEstimate(\JorKetA{j}{1}))$ is $\JorKetB{j}{1}$ with probability $1-\delta$.
    \end{enumerate}
\end{theorem}

The $\Transform$ procedure above can be instantiated directly via the singular vector transformation algorithm of \cite{STOC:GSLW19}, see \cref{thm:svt}.

We describe an implementation of $\VarEstimate$ using the singular value discrimination algorithm (\cref{thm:svdisc}). For a binary projective measurement $\MeasA$, let $\bar{\MeasA}$ denote the same measurement with the outcome labels reversed. For $k$ in the procedure below, define $b \eqdef 2^{-k}$ and $\varepsilon \eqdef 2^{-k-1}$.
\begin{enumerate}[noitemsep]
    \item \label[step]{step:est-loop} Set $b \eqdef 0$, $k \eqdef 0$. Repeat the following two steps until $b = 1$ or $k \geq \lceil \log (1/\delta) \rceil$:
        \begin{enumerate}[noitemsep]
            \item Set $k \gets k+1$.
            \item Apply $\MeasB$, obtaining outcome $c$.
            \item If $c = 1$, apply $\Threshold^{\MeasA,\MeasB}(\gamma,\varepsilon,\delta/\log (1/\delta))$ obtaining outcome $b \in \{0,1\}$.
            \item If $c = 0$, apply $\Threshold^{\bar{\MeasA},\bar{\MeasB}}(\gamma,\varepsilon,\delta/\log (1/\delta))$ obtaining outcome $b \in \{0,1\}$.
        \end{enumerate}
    \item \label[step]{step:fix-state} Apply $\MeasB$, obtaining outcome $c$. If $c = 0$, apply $\Flip_{2^{-k-1}}[\MeasA,\MeasB]$.
    \item Output $2^{-k-1}$.
\end{enumerate}

\begin{lemma}[Variable-runtime singular value estimation]
    \label{lemma:varest}
    Let $\JorKetA{j}{1}$ be a left singular vector with singular value $s_j$. Let $\delta > 0$. $\VarEstimate_{\delta}[\MeasA \rightleftarrows \MeasB](\JorKetA{j}{1},\delta)$ runs in time $O(\log(1/\delta)/s_j)$ with probability $1-\delta$ and $O(\log(1/\delta)/\delta)$ with probability $1$. Moreover, $\VarEstimate$ outputs $a$ in the range $\max(\delta,s_j^2)/4 \leq a \leq \max(\delta,s_j^2)$ with probability $1 - \delta$.
\end{lemma}
\begin{proof}
    First, observe that $k$ iterations of \cref{step:est-loop} take time $O(\log(1/\delta) \cdot 2^k)$. Since $\VarEstimate$ terminates within $\lceil \log(1/\delta) \rceil$ iterations of \cref{step:est-loop} with probability $1$, $\VarEstimate$ runs in time $O(\log(1/\delta)/\delta)$ with probability $1$.

    The probability that the singular value discrimination algorithm outputs $1$ when $2^{-k} > 2s_j$ is at most $\delta/(\log (1/\delta))$. Similarly, the probability that it outputs $1$ when $2^{-k} \leq s_j$ is at least $1-\delta/(\log (1/\delta))$. By a union bound, with probability at least $1 - \delta$ the algorithm either stops in the first iteration where $2^{-k} \leq 2 s_j$ (so $s_j < 2^{-k} \leq 2 s_j$) or in the following iteration ($s_j/2 < 2^{-k} \leq s_j$). Thus $2^{-k} \in [s_j/2,2s_j]$, so $2^{-k-1} \in [s_j/4,s_j]$ as required. The running time in this case is $O(\log(1/\delta)/s_j)$.
    
    If $s_j \geq 1/2$ then the algorithm stops after one iteration in state $\JorKetB{j}{1}$ with probability $1-\delta$. Otherwise the probability that $\log(1/\delta)$ alternating measurements $\MeasA,\MeasB$ are applied with only $0$ outcomes is at most $\delta$. If \cref{step:fix-state} terminates with $\MeasB \to 1$, then the resulting state is $\JorKetB{j}{1}$. Otherwise, the resulting state is $\JorKetA{j}{1}$. In this case the $\Transform$ algorithm rotates the state to $\JorKetB{j}{1}$ with probability $1-\delta$.
\end{proof}

The next two claims follow directly from the correctness and subspace invariance guarantees of $\Threshold$ and $\VarEstimate$.
\begin{corollary}
    \label{cor:threshold-est-almost}
    For any state $\DMatrix$, $\delta > 0$, $\varepsilon \colon [0,1] \to [\delta,1]$:
    \[
        \Pr[\Threshold_{p,\varepsilon(p),\delta}(\VarEstimate(\DMatrix)) = 1] \geq 1 - 2\delta,
    \]
    where $p$ is the classical output from $\VarEstimate$.
\end{corollary}

\begin{corollary}
    \label{cor:threshold-almost}
    For any state $\DMatrix$, $\delta > 0$, $\varepsilon \in [\delta,1]$:
    \[
        \Pr\left[ b_1 = 1 \,\wedge\, b_2 = 0  \, \middle\vert \begin{array}{r}
            (b_1,\DMatrix_1) \gets \Threshold_{p,\varepsilon,\delta}(\DMatrix) \\
            (b_2,\DMatrix_2) \gets \Threshold_{p-\varepsilon,\varepsilon,\delta}(\DMatrix_1) \\
        \end{array}
        \right] \leq 2\delta ~.
    \]
    Moreover,
    \[
        \Pr\left[ b_1 = 1 \,\wedge\, p_j < p - \varepsilon  \, \middle\vert \begin{array}{r}
            (b_1,\DMatrix_1) \gets \Threshold_{p,\varepsilon,\delta}(\DMatrix) \\
            j \gets \Meas{\Jor}[\MeasA,\MeasB](\DMatrix_1) \\
        \end{array}
        \right] \leq \delta ~.
    \]
\end{corollary}

%% file: 7-pseudoinverse.tex
\section{Pseudoinverse Lemma}

In this section we show that for binary projective measurements $\MeasA,\MeasB$ any state $\ket{\psi_\MeasA}$ in the image of $\ProjA$, there is a state $\ket{\psi_\MeasB}$ in the image of $\MeasB$ such that $\ket{\psi_\MeasA}$ is (approximately) obtained by applying $\MeasA$ to $\ket{\psi_\MeasB}$ and conditioning on obtaining a $1$. Moreover, if $\ket{\psi_\MeasA}$ has Jordan spectrum that is concentrated around eigenvalue $p$, then $\ket{\psi_\MeasA}$ has the same property. We refer to this as the ``pseudoinverse lemma'' because $\ket{\psi_\MeasB}$ is obtained from $\ket{\psi_\MeasA}$ by applying the pseudoinverse of the matrix $\ProjA\ProjB$.

\begin{lemma}[Pseudoinverse Lemma]
\label{lemma:pseudoinverse}
    Let $\MeasA,\MeasB$ be binary projective measurements, and let $\{ \Subspace_j \}_{j}$ be the induced Jordan decomposition. Let $\SProj[\Jor]{j}$ be the projection on to $\Subspace_j$ and let $p_j$ be the eigenvalue of $\Subspace_j$. Let $\DMatrix$ be a state such that $\Tr(\ProjA \DMatrix) = 1$ and let $\SProj{0} \eqdef \sum_{j, p_j = 0} \SProj[\Jor]{j}$. Let $E \eqdef \sum_{j, p_j > 0} \frac{1}{p_j} \SProj[\Jor]{j}$. There exists a ``pseudoinverse'' state $\DMatrix'$ with $\Tr(\ProjB \DMatrix') = 1$ such that all of the following are true:
    \begin{enumerate}[noitemsep]
        \item $\Tr(\ProjA \DMatrix') = \frac{1-\Tr(\SProj{0} \DMatrix)}{\Tr(E \DMatrix)} $,
        \item $d\left(\DMatrix,\frac{\ProjA \DMatrix' \ProjA}{\Tr(\ProjA \DMatrix')}\right) \leq 2\sqrt{\Tr(\SProj{0} \DMatrix)}$,
        \item for all $j$ such that $p_j > 0$ it holds that $\Tr(\SProj[\Jor]{j} \DMatrix') = \frac{\Tr(\SProj[\Jor]{j} \DMatrix)}{p_j \cdot \Tr(E \DMatrix)}$, and
        \item for all $j$ such that $p_j = 0$ it holds that $\Tr(\SProj[\Jor]{j} \DMatrix') = 0$.
    \end{enumerate}
\end{lemma}
An important consequence of (3) and (4) is that for all $j$, if $\Tr(\SProj[\Jor]{j} \DMatrix) = 0$ then $\Tr(\SProj[\Jor]{j} \DMatrix') = 0$.

\begin{proof}
    Let $C \eqdef \ProjA \ProjB$, and note that $\JorKetA{j}{1},\JorKetB{j}{1}$ are corresponding left and right singular vectors of $C$ with singular value $\sqrt{p_j}$. Hence $C = \sum_{p_j > 0} \sqrt{p_j} \JorKetA{j}{1} \JorBraB{j}{1}$. Let $C^+$ be the pseudoinverse of $C$, i.e., $C^+ = \sum_{p_j>0} \frac{1}{\sqrt{p_j}} \JorKetB{j}{1} \JorBraA{j}{1}$. Define \[\DMatrix' \eqdef \frac{C^+ \DMatrix (C^+)^{\dagger}}{\Tr(C^+ \DMatrix (C^{+})^{\dagger})}.\]
    Since $\Tr(\ProjA\DMatrix)=1$, we have $\Tr(\ProjB \DMatrix') = 1$. We also have
    \begin{align}
        \Tr(C^+ \DMatrix (C^+)^{\dagger}) = \Tr((C C^{\dagger})^{+}\DMatrix) = \sum_j \frac{1}{p_j} \JorBraA{j}{1} \DMatrix \JorKetA{j}{1} = \Tr(E \DMatrix). \label{eq:pseudoinverse-trace}
    \end{align}
    Next, observe that since $C C^+ = \sum_{p_j > 0} \JorKetBraA{j}{1} = \Id - \Pi_0$, we have
    \begin{align}
    \ProjA \DMatrix' \ProjA &= \ProjA \left(\frac{ C^+ \DMatrix (C^+)^{\dagger} }{\Tr(C^+ \DMatrix (C^{+})^{\dagger})}\right) \ProjA \notag \\ 
    &= \ProjA \left(\frac{ C^+ \DMatrix (C^+)^{\dagger} }{\Tr(E \DMatrix)}\right) \ProjA \notag \\
    &= \ProjA \ProjB \left(\frac{ C^+ \DMatrix (C^+)^{\dagger} }{\Tr(E \DMatrix)}\right) \ProjB \ProjA \notag \\
    &= \frac{1}{\Tr(E \DMatrix)} C C^+ \DMatrix (C C^+)^{\dagger} \notag \\ 
    &= \frac{1}{\Tr(E \DMatrix)} (\Id-\Pi_0) \DMatrix (\Id-\Pi_0). \label{eq:pseudoinverse-condition}
    \end{align}
    
    \noindent Given these calculations, we can prove the claimed properties (1-3) in the lemma statement:
    
    \begin{itemize}
        \item \textbf{Proof of (1).} Taking the trace of both sides of \cref{eq:pseudoinverse-condition}, we see that
    \[ \Tr(\ProjA \DMatrix') = \Tr(\ProjA \DMatrix' \ProjA) = \frac 1 {\Tr(E \DMatrix)} \Tr\left((\Id - \Pi_0) \DMatrix (\Id - \Pi_0) \right) = \frac {\Tr\left((\Id - \Pi_0) \DMatrix \right)} {\Tr(E \DMatrix)}  = \frac {1 - \Tr\left(\Pi_0 \DMatrix \right)} {\Tr(E \DMatrix)}.
    \]
    
        \item \textbf{Proof of (2).}
    Given \cref{eq:pseudoinverse-condition} and the trace calculation above, we have that
    
    \[ \frac{\ProjA \DMatrix' \ProjA}{\Tr(\ProjA \DMatrix')} = \frac{1}{1-\Tr(\Pi_0 \DMatrix)} (\Id-\Pi_0) \DMatrix (\Id-\Pi_0)
    \]
    The inequality $d\left(\DMatrix,\frac{\ProjA \DMatrix' \ProjA}{\Tr(\ProjA \DMatrix')}\right) \leq 2\sqrt{\Tr(\SProj{0} \DMatrix)}$ now follows from \cref{lemma:gentle-measurement} (gentle measurement).
    
    \item  \textbf{Proof of (3).} For all $j$ such that $p_j > 0$, making use of the same calculation as \cref{eq:pseudoinverse-trace}, we have
    \[ \Tr(\SProj[\Jor]{j} \DMatrix') = \frac{ \Tr(\SProj[\Jor]{j} C^+ \DMatrix (C^+)^{\dagger})}{\Tr( C^+ \DMatrix (C^+)^{\dagger})}= \frac{\Tr( C^+ \SProj[\Jor]{j} \DMatrix (C^+)^{\dagger})}{\Tr(E \DMatrix)} = \frac{\Tr(E \hspace{.1cm} \SProj[\Jor]{j} \DMatrix)}{\Tr(E \DMatrix)} =  \frac{\Tr(\frac 1 {p_j} \SProj[\Jor]{j} \DMatrix)}{\Tr(E \DMatrix)}.
    \]
    
    \item  \textbf{Proof of (4).} This follows immediately from the fact that $\Pi_0 C^+ = C^+ \Pi_0 = 0$.
    \end{itemize}  
    
    This completes the proof of \cref{lemma:pseudoinverse}.
\end{proof}

We conclude this section by showing that under a mild condition, any state $\DMatrix$ that is close to $\image(\ProjA)$ has a nearby state in $\image(\ProjA)$ with the same Jordan decomposition.

\begin{claim}
    \label{claim:jordan-rotate}
    Let $\DMatrix$ be any state. Let $\SProj[\Jor]{\mathsf{stuck}}$ project on to one-dimensional subspaces $\Subspace_j$ in the image of $\Id - \ProjA$. There exists a state $\DMatrixW$ such that for all $j$, $\Tr(\SProj[\Jor]{j} \DMatrixW) = \Tr(\SProj[\Jor]{j} \DMatrix)$, $\Tr(\ProjA \DMatrixW) = 1 - \Tr(\SProj[\Jor]{\mathsf{stuck}} \cdot \DMatrix)$, and $d(\DMatrix,\DMatrixW) \leq \sqrt{1 - \Tr(\ProjA \DMatrix)}$.
\end{claim}
\begin{proof}
    Define a unitary $U$ which is invariant on the $\Subspace_j$ and, in each two-dimensional $\Subspace_j$, rotates $\JorKetA{j}{0}$ to $\JorKetA{j}{1}$. Formally,
    \begin{equation*}
        U \eqdef \sum_{j,p_j \notin \{0,1\}} (\JorKetA{j}{1}\JorBraA{j}{0} + \JorKetA{j}{0}\JorBraA{j}{1}) + \Id_{\Subspace^{(1)}},
    \end{equation*}
    where $\Subspace^{(1)}$ is the direct sum of the 1D subspaces. Set
    \[
        \DMatrixW \eqdef \ProjA \DMatrix \ProjA + U (I - \ProjA) \DMatrix (I - \ProjA) U^{\dagger}. \qedhere
    \]
\end{proof}

%% file: 8-extractor.tex
\section{Post-Quantum Guaranteed Extraction}\label{sec:high-probability-extractor}

\def \tensor {\otimes}
\def \vk {\mathsf{vk}}

In this section, we give a post-quantum extraction procedure for various 3- and 4-message public-coin interactive protocols. In particular, we will consider interactive protocols satisfying \emph{partial collapsing} (\cref{def:partial-collapsing-protocol}) with respect to some class of efficiently computable functions $F = \{f: T\times R \times Z \rightarrow \{0,1\}^* \}$. Our goal is to establish \emph{guaranteed extraction}, defined below (essentially matching \cref{def:high-probability-extraction}). 

\begin{definition}\label{def:high-probability-extraction-body}
$(P_{\Sigma}, V_{\Sigma})$ is a post-quantum proof of knowledge with \emph{guaranteed extraction} if it has an extractor $\Extract^{P^*}$ of the following form.

\begin{enumerate}[noitemsep]
    \item $\Extract^{P^*}$ first runs the cheating prover $P^*$ to generate a (classical) first message $a$ along with an instance $x$ (in a 4-message protocol, this requires first sampling a random $\vk$ and running $P^*(\vk)$ to obtain $x,a$). 
    \item\label[step]{step:ge-run-coherently} $\Extract^{P^*}$ runs $P^*$ coherently on the superposition $\sum_{r \in R} \ket{r}$ of all challenges to obtain a superposition $\sum_{r,z} \alpha_{r, z} \ket{r, z}$ over challenge-response pairs.\footnote{In general, the response $z$ will be entangled with the prover's state; here we suppress this dependence.} 
    \item $\Extract^{P^*}$ then computes (in superposition) the verifier's decision $V(x, a, r, z)$ and measures it. If the measurement outcome is $0$, the extractor gives up. 
    \item If the measurement outcome is $1$, run some quantum procedure $\FindWitness^{P^*}$ that outputs a string $w$. 
\end{enumerate}

We require that the following two properties hold. 

\begin{itemize}[noitemsep]
    \item \textbf{Correctness (guaranteed extraction):} The probability that the initial measurement returns $1$ but the output witness $w$ is not a valid witness for $x$ is $\negl(\secp)$.
    \item \textbf{Efficiency:} For any QPT $P^*$, the procedure $\Extract^{P^*}$ is in $\EQPTM$. 
\end{itemize}
\end{definition}

We remark that this definition is written to capture (first-message) \emph{adaptive soundness}, where the prover $P^*$ is allowed to choose the instance $x$ when it sends its first message. One could alternatively define a non-adaptive variant of this definition in which the instance $x$ is fixed in advance (and this section's results would hold in this setting as well).~\cref{def:high-probability-extraction-body} suffices for our purposes since none of the 4-message protocols we consider have the first verifier message $\vk$ depend on $x$ (in all cases we consider, $\vk$ is just a commitment key or hash function key), and the protocols all satisfy adaptive soundness. 

\subsubsection{Notation}
\label{sec:ge-notation}

Let $\RegR$ denote a register with the basis $\{\ket{r}\}_{r \in R}$ and let $\ket{+_R}_{\RegR} \eqdef \frac{1}{\sqrt{\abs{R}}} \sum_{r \in R} \ket{r}$. Let $\RegH$ denote the prover's state (including its workspace), and let $U_r$ denote the unitary on $\RegH$ that the prover applies on challenge $r$. Let $\RegZ$ denote the subregister of $\RegH$ that the prover measures to obtain its response $z$ after applying $U_r$. 

Define the projector
\[\Pi_{V,r} = U_{r}^{\dagger} \left(\sum_{z : V(r,z) = 1} \ketbra{z}_{\RegZ} \otimes \Id\right) U_{r}\]
which intuitively projects onto subspace of $\RegH$ where the prover gives an accepting response on challenge $r$.

Define the binary projective measurement $\sC = \BMeas{\BProj{\sC}}$ where
\[ \BProj{\sC} = \sum_r \ketbra{r}_{\RegR} \otimes \Pi_{V,r} ,\]
and $\sU = \BMeas{\BProj{\sU}}$ where 
\[ \BProj{\sU} = \ketbra{+_R}_{\RegR} \otimes \Id_{\RegH}.\]

\subsection{Description of the Extractor}
\label{subsec:description-ext}

We first give a full description of an extraction procedure $\Extract$, defined for any partially collapsing protocol. 

\paragraph{The threshold unitary.} 

Consider the following measurement procedure $\sT_{p,\varepsilon,\delta}$ on $\RegH$, parameterized by threshold $p$, accuracy $\varepsilon$ and error $\delta$.
\begin{itemize}[noitemsep]
    \item Initialize a fresh register $\RegR$ to $\ket{+_R}_{\RegR}$.
    \item Run $\Threshold_{p,\varepsilon,\delta}^{\sU,\sC}$ on $\RegH \tensor \RegR$ , obtaining outcome $b$.
    \item Trace out $\RegR$ and output $b$.
\end{itemize}

We define $U_{p,\varepsilon,\delta}$ to be a \emph{coherent} implementation of $\sT_{p,\eps,\delta}$. $U_{p,\varepsilon,\delta}$ acts on $\RegH \otimes \RegW \otimes \RegB$ where $\RegW \otimes \RegB$ is an ancilla register: $\RegW$ contains the algorithm's workspace and $\RegB$ is a single qubit containing the measurement outcome. In particular, applying $U_{\varepsilon,\delta}$ to $\ket{\psi}_{\RegH} \ket{0}_{\RegW,\RegB}$, measuring $\RegB$, and then tracing out $\RegW$ implements the above measurement.

\paragraph{The repair measurements.} We define the two projective measurements $\sD_r = \BMeas{\BProj{r}},\sG_{p,\varepsilon,\delta} = \BMeas{\BProj{p,\varepsilon,\delta}}$ for our repair step.

For any $p,\varepsilon,\delta > 0$, define the projector $\Pi_{p,\varepsilon,\delta}$ on $\RegH \tensor \RegW \tensor \RegB$ as follows:
\[ \Pi_{p,\varepsilon,\delta} \eqdef U_{p,\varepsilon,\delta}^\dagger (\Id_{\RegH, \RegW} \tensor \ketbra{1}_{\RegB}) U_{p,\varepsilon,\delta}.\]
For any $r \in R$, we define the projector $\Pi_r$ on $\RegH \tensor \RegW$ as
\[\Pi_r \coloneqq (\Pi_{V, r})_{\RegH} \tensor \ketbra{0}_{\RegW}.\]

We describe the extraction procedure $\Extract_{V}^{P^*}(x)$. The procedure is defined with respect to $k$ efficiently computable functions $f_1, \hdots, f_k: T\times R \times Z\rightarrow \{0,1\}^*$. 

\begin{mdframed}

\begin{enumerate}
    \item \label[step]{step:first-measurement} \textbf{Initial Execution.} Use $P^*$ to generate $(\vk, \Co)$, and let $\ket{\psi}$ denote the residual prover state. Apply $\sC = \BMeas{\BProj{\sC}}$ to $ \ket{\psi}_{\RegH}\tensor \ket{+_R}_{\RegR}$. If $0$, terminate (note that we do not consider this an ``abort''.) Otherwise:
    
    \item \label[step]{step:variable-phase-est} \textbf{Estimate success probability.} Run $\VarEstimate[\sC \rightleftarrows \sU]$ (as defined in \cref{sec:vrsvt}) with $\frac 1 2$-multiplicative error and failure probability $\delta = 1/2^\lambda$, outputting a value $p$. Note that since the input state is in $\BProj{\sC}$, the algorithm produces an output state in $\BProj{\sC}$ with probability $1-\delta$.
    
    Abort if $p < \secp k \sqrt{\delta}$. Define $\eps = \frac{p}{4k}$.

    \item \label[step]{step:extract-transcript} \textbf{Main Loop.} Repeat the following ``main loop'' for $i$ from $1$ to $k$:
    \begin{enumerate}
    \item \label[step]{step:loop-reestimate} \textbf{Lower bound success probability.} Run $\Threshold_{p,\varepsilon,\delta}^{\sC,\sU}$ on $\RegH \tensor \RegR $, obtaining outcome $b$. Abort if $b = 0$. Update $p \coloneqq p-\eps$.
    \item \label[step]{step:loop-measure-r} \textbf{Measure the challenge.} Measure the $\RegR$ register, obtaining a particular challenge $r_i \in R$. Discard the $\RegR$ register. 
    \item \label[step]{step:loop-phase-est} \textbf{Estimate the running time of Transform.} Initialize the $\RegW$ register to $\ket{0}_{\RegW}$ and run $\VarEstimate[\sD_r \rightleftarrows \sG_{p, \eps,\delta}]$ with $\frac 1 2$-multiplicative error and failure probability $\delta = 2^{-\secp}$, obtaining classical output $q$. Since the input state is in $\Pi_r$, the algorithm produces an output state in $\Pi_r$ with probability $1-\delta$.
    \item \label[step]{step:loop-collapsing}\textbf{Record part of the accepting response.}  Make a partial measurement of the prover response $z_i$; specifically, measure $y_i = f_i(z_i)$.\footnote{Formally, we (1) apply the prover unitary $U_{r^*}$ to $\RegH$, (2) apply the projective measurement $\big(\Pi_y\big)_y$ for $\Pi_y = \sum_{z: f_i(z) = y} \ketbra{z}_{\RegZ}\tensor \Id_{\RegH'}$ (where $\RegH = \RegZ \tensor \RegH'$), and (3) apply $U_{r^*}^\dagger$ to $\RegH$.} \textbf{If $i=k$, go to Step 4.}
    \item \label[step]{step:loop-svt} \textbf{Transform onto good states.} Apply $\Transform_q[\sD_r \rightarrow \sG_{p, \eps,\delta}]$ with failure probability $\delta = 2^{- \secp}$.
    
    \item \label[step]{step:loop-discard-w} Next, apply $U_{p,\varepsilon,\delta}$ and then discard the $\RegW$ register. Update $p \coloneqq p-\eps$.

    \item \label[step]{step:loop-amplify-C} \textbf{Transform onto accepting executions.} Re-initialize $\RegR$ to $\ket{+_R}$ and then apply $\Transform_{p}[\sU \rightarrow \sC]$; abort if this procedure fails.
    \end{enumerate}
    
    \item Output $(\vk,\Co,\Ch_1,y_1,\cdots,\Ch_k,y_k)$.
\end{enumerate}

The above procedure deterministically terminates and aborts if it has not already stopped after $O(k)/\sqrt \delta$ steps, for $\delta := 2^{-\secp}$.

\end{mdframed}

\subsection{Partial Transcript Extraction Theorem}
Our most general extraction theorem is stated for \emph{any} partially collapsing protocol, but is only guaranteed to output partial transcripts (rather than a witness). In \cref{sec:obtaining-guaranteed-extraction}, we show how this theorem can be used to establish guaranteed extraction of a witness.
\begin{theorem}\label{thm:high-probability-extraction}
   For any $4$-message public-coin interactive argument satisfying partial collapsing (\cref{def:collapsing-protocol}) with respect to the functions $f_1, \hdots, f_{k-1}$ (but \emph{not necessarily} $f_k$), the procedure $\Extract_V$ has the following properties for any instance $x$.
   
   \begin{enumerate}
       \item \textbf{Efficiency:} For any QPT prover $P^*$, $\Extract_V^{P^*}$ runs in expected polynomial time ($\EQPTM$). More formally, the number of calls that $\Extract_V^{P^*}$ makes to $P^*$ is a classical random variable whose expectation is a fixed polynomial in $k, \secp$.
       
       \item \textbf{Correctness:} $\Extract$ aborts with negligible probability.
       
       \item \textbf{Distribution of outputs}: For every choice of $(\vk, a)$, let $\gamma = \gamma_{\vk, a}$ denote the success probability of $P^*$ conditioned on first two messages $(\vk, a)$. Then, if $\gamma > \delta^{1/3}$, the distribution of $(r_1, \hdots, r_k)$ (conditioned on $(\vk, a)$ and a successful first execution) is $O(1/\gamma)$-admissible (\cref{def:admissible-dist}). 
   \end{enumerate}
\end{theorem}

\subsection{Proof of \cref{thm:high-probability-extraction}}

\subsubsection{Intermediate State Notation}
Our extraction procedure and analysis make use of four relevant registers:

\begin{itemize}[noitemsep]
    \item A challenge randomness register $\RegR$,
    \item A prover state register $\RegH$, and
    \item A phase estimation workspace register $\RegW$. 
    \item A one qubit register $\RegB$ that contains a bit $b$ where $b = 1$ indicates that the computation has not aborted during a sub-computation.
\end{itemize}

We now establish some conventions:
\begin{itemize}[noitemsep]
    \item states written using the letter $\dmr$ satisfy $\dmr \in \Hermitians{\RegB \tensor \RegH \tensor \RegR}$ or $\dmr \in \Hermitians{ \RegH \tensor \RegR}$, where we use $\Hermitians{\RegH}$ to denote the space of Hermitian operators on $\RegH$; 
    \item states using the letter $\dmw$ satisfy $\dmw \in  \Hermitians{\RegB \tensor \RegH \otimes \RegW}$ or $\dmw \in  \Hermitians{\RegH \otimes \RegW}$;
    \item states using the letter $\dmh$ satisfy $\dmh \in \Hermitians{\RegB \tensor \RegH}$ or $\dmh \in  \Hermitians{\RegH}$;
    \item states using the letter $\dmwr$ satisfy $\dmwr\in \Hermitians{\RegH\tensor \RegW \tensor \RegR}$
\end{itemize}
 
With these conventions in mind, we define some intermediate states related to the extraction procedure:

\begin{mdframed}

\begin{itemize}[noitemsep]
    \item Let $\psi$ denote the prover state after $(\vk, \Co)$ is generated.
    \item Let $\dmr_{\RegB,\RegH,\RegR}^{(2)}$ denote the state obtained at the end of \cref{step:variable-phase-est}. 
    \item For each iteration of the \cref{step:extract-transcript} loop, we define the following states:
    
    \begin{itemize}[nolistsep]
        \item Let $\dmr_{\RegB,\RegH,\RegR}^{(\mathrm{init})}$ denote the state at the beginning of \cref{step:extract-transcript}. The register $\RegB$ is initialized to $\ketbra{1}$. For the rest of the loop iteration, $\RegB$ is set to $\ketbra{0}$ if the computation aborts. 
        \item Let $\dmr_{\RegB,\RegH,\RegR}^{(\sC)}$ denote the state at the end of \cref{step:loop-reestimate}.
        \item Let $\dmh_{\RegB,\RegH,\RegR}^{(3b)}$ denote the state at the end of \cref{step:loop-measure-r}.
        \item Let $\dmw_{\RegB,\RegH,\RegW}^{(3c)}$ denote the state at the end of \cref{step:loop-phase-est}. 
        \item Let $\dmw_{\RegB,\RegH,\RegW}^{(3e)}$ denote the state immediately before the $\mathcal W$ register is traced out during \cref{step:loop-svt}.
        \item Let $\dmh_{\RegB,\RegH}^{(3f)}$ denote the state at the end of \cref{step:loop-discard-w}.
        \item Let $\dmr_{\RegB,\RegH,\RegR}^{(3g)}$ denote the state at the end of \cref{step:loop-amplify-C}.
    \end{itemize}
\end{itemize}

\end{mdframed}

As in~\cref{lemma:pseudoinverse}, let the Jordan decomposition of $\RegH \otimes \RegR$ corresponding to $\BProj{\sC},\BProj{\sU}$ be $\{\RegS_j\}_j$ where subspace $\RegS_j$ is associated with the eigenvalue/success probability $p_j$. Let $\Pi^{\Jor}_j$ the projection onto $\RegS_j$, i.e., $\image(\Pi^{\Jor}_j) = \RegS_j$. Define the following projections on $\RegH \otimes \RegR$:
\begin{itemize}
    \item $\Pi_0^{\mathsf{Jor}} \eqdef \sum_{j, p_j =0} \Pi^{\Jor}_j$
    \item $\Pi^\Jor_{\geq p}= \sum_{j: p_j \geq  p} \Pi_j^{\Jor}$
    \item $\Pi^{\Jor}_{< p} = \sum_{j: p_j < p} \Pi_j^{\Jor}$
\end{itemize}

\noindent We additionally define the following projectors on $\RegB \otimes \RegH \otimes \RegR$. 
\[ \Pi^{\Jor}_{\mathsf{Bad}} = \ketbra{1}_{\RegB} \tensor \Pi^{\Jor}_{< p} \hspace{10pt} \text{and} \hspace{10pt} \Pi^{\Jor}_{\mathsf{Good}} = \Id_{\RegB,\RegH,\RegR} - \Pi^{\Jor}_{\mathsf{Bad}}.
\]

\begin{claim}\label{claim:reest-good-state}
  For \emph{any} estimate $p$ and \emph{any} state $\dmr_{\RegB,\RegH,\RegR}^{(\mathrm{init})}$ such that $\Tr((\Id_{\RegB} \tensor \Pi_{\sC}) \dmr_{\RegB,\RegH,\RegR}^{(\mathrm{init})}) = 1$, the state $\dmr_{\RegB,\RegH,\RegR}^{(\sC)}$ obtained by running $b\gets \Threshold^{\sC, \sU}_{p, \eps, \delta}$ (and then redefining $p := p-\eps$) and setting $\RegB = \ketbra{b}$ satisfies
  \[ \Tr( \Pi^{\Jor}_{\mathsf{Bad}} \dmr_{\RegB,\RegH,\RegR}^{(\sC)}) \leq \delta.\]
\end{claim}

\begin{proof}
  When $\Threshold^{\sC, \sU}_{p, \eps, \delta}$ returns $0$, the computation aborts. Therefore, the lemma follows immediately from the almost-projectivity of $\Threshold$ (\cref{cor:threshold-almost}).
\end{proof}

\subsubsection{Analysis of \cref{step:first-measurement,step:variable-phase-est}}

We first show that \cref{step:first-measurement,step:variable-phase-est} run in expected polynomial time, and bound the statistic $\mathbb E[1/p \cdot X_1]$, where $p$ is the output of \cref{step:variable-phase-est} and $X_1$ is the indicator for the event ``\cref{step:first-measurement} does not abort''.

\begin{lemma}
    \label{lemma:step-1-2}
    The expected runtime of \cref{step:first-measurement,step:variable-phase-est} is $O(1)$. Moreover, 
    \[\mathbb E[1/p \cdot X_1] = O(1).
    \]
\end{lemma}

\begin{proof}
   Let $\ket{\psi}$ denote the state of $P^*$ after $(\ck, \Co)$ are generated. Then, consider the $(\sU, \sC)$-Jordan decomposition
   \[ \ket{\psi} \tensor \ket{+_R} = \sum_j \alpha_j \ket{v_{j,1}},
   \]
   where each $\ket{v_{j,1}} \in \mathcal S_j \cap \image(\Pi_{\sU})$. Let $\gamma = \sum_j |\alpha_j|^2 p_j$ denote the initial success probability of $\ket{\psi}$. 
   
   \cref{step:first-measurement} runs in a fixed polynomial time and aborts with probability $1-\gamma$. Otherwise, \cref{step:variable-phase-est} is run on the residual state
   \[ \frac 1 {\sqrt \gamma} \sum_j \alpha_j \sqrt{p_j} \ket{w_{j, 1}},
   \]
   where $\ket{w_{j,1}}$ is a basis vector in $\mathcal S_j \cap \image(\Pi_{\sC})$. \cref{lemma:varest} tells us that \emph{both} the runtime of $\VarEstimate^{\sC, \sU}$ on this state (making oracle use of $\sC, \sU$) and the expectation of $1/p$ (where $p$ is the output of \cref{step:variable-phase-est}) are at most a constant times
   \[ \frac 1 {\gamma} \sum_j \alpha_j^2 p_j \cdot \frac 1 {p_j} + \delta \cdot 1/\delta \leq \frac 1 {\gamma} \sum_j \alpha_j^2 + 1 = \frac 1 \gamma + 1,
   \]
   so since $\Pr[\text{\cref{step:first-measurement} does not abort}] = \gamma$, the overall expected value bounds are as claimed. 
\end{proof}

\subsubsection{The Pseudoinverse State}\label{sec:extractor-analysis-pseudoinverse-state}

As defined earlier, let $\dmr_{\RegB,\RegH,\RegR}^{(\sC)} $ denote the state at the end of~\cref{step:loop-reestimate} (for some arbitrary iteration of \cref{step:extract-transcript}). We prove some important properties of the subsequent states in the execution of \cref{step:extract-transcript}.

We begin with \cref{step:loop-measure-r}, which measures the $\RegR$ register, obtaining a challenge $r$ and resulting state $\dmh_{\RegB,\RegH}^{(3b)}$. Let \[\dmr_{\RegB,\RegH,\RegR}'^{(\sC)} = \frac{\Pi^{\Jor}_{\mathsf{Good}} \dmr_{\RegB,\RegH,\RegR}^{(\sC)} \Pi^{\Jor}_{\mathsf{Good}}}{\Tr(\Pi^{\Jor}_{\mathsf{Good}} \dmr_{\RegB,\RegH,\RegR}^{(\sC)})}\]
denote the residual state; we write $\dmr_{\RegB,\RegH,\RegR}'^{(\sC)} = \alpha_0 \ketbra{0}_{\RegB} \tensor \dmr_{\RegH,\RegR}'^{(\sC,0)} + \alpha_1 \ketbra{1}_{\RegB} \tensor \dmr_{\RegH,\RegR}'^{(\sC,1)}$. By \cref{claim:reest-good-state}, we have that:
\begin{claim}\label{claim:analysis-pi-jor-good}
  $\Tr(\Pi^\Jor_{\mathsf{Good}} \dmr_{\RegB,\RegH,\RegR}^{(\sC)}) \geq 1-\delta$
\end{claim}

By gentle measurement, it then follows that
\[ d(\dmr_{\RegB,\RegH,\RegR}^{(\sC)}, \dmr_{\RegB,\RegH,\RegR}'^{(\sC)}) \leq 2\sqrt{\delta}. 
\]

Let $\dmr_{\RegB,\RegH,\RegR}^{(\sU)} = \alpha_0 \ketbra{0}_{\RegB}\tensor  \dmr_{\RegH,\RegR}^{(\sU,0)} + \alpha_1 \ketbra{1}_{\RegB}\tensor \dmr_{\RegH,\RegR}^{(\sU,1)}$ where $\dmr_{\RegH,\RegR}^{(\sU,1)}$ denotes the state guaranteed to exist by applying~\cref{lemma:pseudoinverse} with $\MeasB = \sU $ and $\MeasA = \sC$ on $\dmr_{\RegH,\RegR}'^{(\sC,1)}$.  Recall from~\cref{lemma:pseudoinverse} that $\Tr\left( \BProj{\sU} \dmr_{\RegH,\RegR}^{(\sU,1)}\right) = 1$, since $\Tr( \Pi_0^{\Jor}\dmr_{\RegH,\RegR}'^{(\sC,1)}) = 0$. Moreover, we also have:

\begin{claim}
\label{claim:exact-pseudoinverse}
   \[\dmr_{\RegH,\RegR}'^{(\sC,1)} = \frac{\BProj{\sC}\dmr_{\RegH,\RegR}^{(\sU,1)} \BProj{\sC}}{\Tr(\BProj{\sC}\dmr_{\RegH,\RegR}^{(\sU,1)})}.
   \]
\end{claim}

\begin{proof}
$\dmr_{\RegH,\RegR}^{(\sC,1)}$ is a state satisfying $\Tr(\BProj{\sC}\dmr_{\RegH,\RegR}^{(\sC,1)}) = 1$, and $\dmr_{\RegH,\RegR}'^{(\sC,1)}$ is then a (re-normalized) projection of $\dmr_{\RegH,\RegR}^{(\sC,1)}$ onto $(\sU, \sC)$-Jordan subspaces with bounded Jordan $p_j$-value. Therefore, $\dmr_{\RegH,\RegR}'^{(\sC,1)}$ also satisfies $\Tr(\BProj{\sC}\dmr_{\RegH,\RegR}'^{(\sC,1)}) = 1$. 

From~\cref{lemma:pseudoinverse} (Property 2) we then have \[d(\dmr_{\RegH,\RegR}'^{(\sC,1)}, \frac{\BProj{\sC}\dmr_{\RegH,\RegR}^{(\sU,1)} \BProj{\sC}}{ \Tr(\BProj{\sC}\dmr_{\RegH,\RegR}^{(\sU,1)})}) \leq 2 \sqrt{\Tr(\Pi_0 \dmr_{\RegH,\RegR}'^{(\sC,1)})} = 0,\] which implies $\dmr_{\RegH,\RegR}'^{(\sC,1)} = \BProj{\sC}\dmr_{\RegH,\RegR}^{(\sU,1)} \BProj{\sC}/ \Tr(\BProj{\sC}\dmr_{\RegH,\RegR}^{(\sU,1)})$.
\end{proof}

\noindent Finally, because $\Tr(\Pi^{\Jor}_{\mathsf{Bad}} \dmr_{\RegB,\RegH,\RegR}'^{(\sC)}) = 0$, \cref{lemma:pseudoinverse} (Property 3) tells us that $\Tr(\Pi^{\Jor}_{\mathsf{Bad}} \dmr_{\RegB,\RegH,\RegR}^{(\sU)}) =0$ as well. 

Define $p_{\sU} = \Tr(\BProj{\sC}\dmr_{\RegH,\RegR}^{(\sU,1)})$ to be the normalization factor above, which is equal to the ($\sC$-)success probability of $\dmr_{\RegH,\RegR}^{(\sU,1)}$. 

\begin{claim}\label{claim:pseudoinverse-win-probability}
  $p_{\sU} \geq p$. 
\end{claim}
\begin{proof}
    Since $\Tr(\Pi^{\Jor}_{\geq p} \dmr_{\RegH,\RegR}^{(\sU,1)}) = 1$ and $\BProj{\sC}$ commutes with each $\Pi_j^{\Jor}$, we have:
    \begin{align*} \Tr( \BProj{\sC}\dmr_{\RegH,\RegR}^{(\sU,1)}) &= \Tr(  \BProj{\sC}\Pi^{\Jor}_{\geq p} \dmr_{\RegH,\RegR}^{(\sU, 1)})
    \\ &=  \Tr(\sum_{j: p_j \geq p} \BProj{\sC}\Pi^{\Jor}_{j}  \dmr_{\RegH,\RegR}^{(\sU,1)})
    \\ &\geq p \Tr(\sum_{j: p_j \geq p} \Pi^{\Jor}_{j}  \dmr_{\RegH,\RegR}^{(\sU,1)})
    \\ &= p. \qedhere
    \end{align*}
\end{proof}

Since $\Tr(\BProj{\sU} \dmr_{\RegH,\RegR}^{(\sU,1)}) =1$, it can be written in the form $\dmh_{\RegH}^{(\sU,1)} \tensor \ketbra{+_R} $. For each $r$, we define $\zeta_r := \Tr( \Pi_{V, r} \dmh_{\RegH}^{(\sU,1)})$ to be the success probability of $\dmh_{\RegH}^{(\sU,1)}$ on $r$. Finally, define $\zeta_R = \sum_r \zeta_r$. 

We now proceed to analyze the state $\dmh_{\RegB,\RegH}^{(3b)}$. To do so, we first define $\dmh_{\RegB,\RegH}'^{(3b)}$ to be the state at the end of \cref{step:loop-measure-r} when $\dmr_{\RegB,\RegH,\RegR}'^{(\sC)}$ is used in place of $\dmr_{\RegB,\RegH,\RegR}^{(\sC)}$. We know that $d(\dmh_{\RegB,\RegH}^{(3b)}, \dmh_{\RegB,\RegH}'^{(3b)}) \leq 2\sqrt{\delta}$, so this characterization will suffice.

\begin{claim}
\label{claim:stratify-r}
The state $\dmh_{\RegB,\RegH}'^{(3b)}$ is a mixed state with the following form: with probability $\alpha_0$, it is in the abort state. Otherwise, with conditional probability $\zeta_r/ \zeta_R$, \cref{step:loop-amplify-C} measures challenge $r$ and the resulting state is $\ketbra{1}_{\RegB} \tensor \frac{\Pi_{V, r} \dmh_{\RegH}^{(\sU,1)} \Pi_{V, r}}{ \zeta_r}$.
\end{claim}

\begin{proof}
    By definition of the pseudoinverse state $\dmr_{\RegH,\RegR}^{(\sU,1)} = \dmh_{\RegH}^{(\sU,1)}\tensor \ketbra{+_R}$, we can write $\dmr_{\RegH,\RegR}'^{(\sC,1)}$ as
    \[\dmr_{\RegH,\RegR}'^{(\sC,1)} = \frac{\BProj{\sC}( \dmh_{\RegH}^{(\sU,1)}\tensor \ketbra{+_R} )\BProj{\sC}}{\Tr(\BProj{\sC}( \dmh_{\RegH}^{(\sU,1)} \tensor \ketbra{+_R} ))}.\]
    Since $\BProj{\sC}= \sum_{r \in R}  \Pi_{V, r}\tensor \ketbra{r}$, we can write
    \begin{align*}
         \BProj{\sC}( \dmh_{\RegH}^{(\sU,1)}\tensor \ketbra{+_R})\BProj{\sC}&= (\sum_{r \in R}  \Pi_{V, r} \tensor \ketbra{r})\left( \dmh_{\RegH}^{(\sU,1)} \tensor \ketbra{+_R} \right)(\sum_{r \in R} \Pi_{V, r} \tensor \ketbra{r})\\
         &= \frac{1}{\abs{R}}\sum_{r \in R} \Pi_{V, r} \dmh_{\RegH}^{(\sU,1)} \Pi_{V, r} \tensor \ketbra{r}.
    \end{align*}
    Thus, we can rewrite $\dmr_{\RegH,\RegR}'^{(\sC,1)}$ as
    \begin{align*}
        \frac{\BProj{\sC}( \dmh_{\RegH}^{(\sU,1)} \tensor \ketbra{+_R})\BProj{\sC}}{\Tr(\BProj{\sC} ( \dmh_{\RegH}^{(\sU,1)} \tensor \ketbra{+_R}))} &= \frac{\frac{1}{\abs{R}}\sum_{r \in R}   \Pi_{V, r} \dmh_{\RegH}^{(\sU,1)} \Pi_{V, r}\tensor \ketbra{r}}{\Tr(\frac{1}{\abs{R}}\sum_{r \in R}   \Pi_{V, r} \dmh_{\RegH}^{(\sU,1)} \Pi_{V, r} \tensor \ketbra{r})} \\
        &= \frac{\frac{1}{\abs{R}}\sum_{r \in R}  \Pi_{V, r} \dmh_{\RegH}^{(\sU,1)} \Pi_{V, r} \tensor \ketbra{r}}{\frac{1}{\abs{R}}\sum_{r \in R} \zeta_r} \\
        &= \frac{\sum_{r \in R}  \Pi_{V, r} \dmh_{\RegH}^{(\sU,1)} \Pi_{V, r}\tensor \ketbra{r}}{\zeta_R}.
    \end{align*}
    Therefore, the probability of obtaining $r$ after measuring $\RegR$ of $\dmr_{\RegB,\RegH,\RegR}'^{(\sC)}$ is
    \begin{align*}
        \frac{\Tr((\Id \tensor \ketbra{r}) \sum_{r' \in R}  \Pi_{V, r'} \dmh_{\RegH}^{(\sU,1)} \Pi_{V, r'}\tensor \ketbra{r'})}{\zeta_R} &= \frac{\Tr( \Pi_{V, r} \dmh_{\RegH}^{(\sU,1)}\Pi_{V, r}\tensor \ketbra{r} )}{\zeta_R}\\
        &= \frac{\zeta_r}{\zeta_R},
    \end{align*}
    and the post-measurement state is $ \Pi_{V, r} \dmh_{\RegH}^{(\sU,1)} \Pi_{V, r}$. Thus, the state $\dmh_{\RegB,\RegH}'^{(3b)}$ is as claimed.
\end{proof}

In particular, \cref{claim:stratify-r} tells us that the ratio $\frac{\zeta_R}{\abs{R}}$ is exactly $\mathrm{Tr}(\BProj{\sC}\dmr_{\RegH,\RegR}^{(\sU,1)}) = p_{\sU}$. 

We begin our analysis of the repair step by defining the following states:
    \begin{enumerate}
        \item $\dmw_{\RegH,\RegW}^{(\sU,1)} \eqdef \dmh_{\RegH}^{(\sU,1)} \otimes \ketbra{0}_{\RegW}.$ Here, $\dmh_{\RegH}^{(\sU,1)}$ is the state satisfying $\dmr_{\RegH,\RegR}^{(\sU,1)} = \dmh_{\RegH}^{(\sU,1)} \tensor \ketbra{0}_{\RegR}$.
        \item $\dmw_{\RegH,\RegW}^{(r,1)} \eqdef \Pi_r \dmw_{\RegH,\RegW}^{(\sU,1)} \Pi_r/\zeta_r$.
    \end{enumerate}
By \cref{claim:stratify-r} we can view our variant of~\cref{step:loop-measure-r,step:loop-phase-est,step:loop-collapsing,step:loop-svt} as follows:

\begin{itemize}
    \item  With probability $\alpha_0$, abort. Otherwise:
    \item A challenge is sampled so that each string $r$ occurs with probability $\frac{\zeta_r}{\zeta_R}$
    \item If the string $r$ is sampled, initialize the state to $\ketbra{1}_{\RegB} \tensor \dmw_{\RegH,\RegW}^{(r, 1)} $. 
\end{itemize}

Unfortunately, the state $\dmw_{\RegH,\RegW}^{(\sU,1)}$ only satisfies $ \Tr(\Pi_{p, \eps} \dmw_{\RegH,\RegW}^{(\sU,1)}) \geq 1-\delta$ (it is not \emph{quite} fully in the image of $\Pi_{p, \eps}$). With this in mind, we define two additional states:

    \begin{enumerate}
        \item[3.] $\widetilde{\dmw}_{\RegH,\RegW}^{(\sU,1)} \coloneqq \frac{\Pi_{p,\varepsilon} \dmw_{\RegH,\RegW}^{(\sU,1)} \Pi_{p,\varepsilon}}{\Tr(\Pi_{p,\varepsilon} \dmw_{\RegH,\RegW}^{(\sU,1)})}$. Since $\Tr(\Pi_{p,\varepsilon} \dmw_{\RegH,\RegW}^{(\sU,1)}) = 1-\delta$, we have $d(\dmw_{\RegH,\RegW}^{(\sU,1)},\widetilde{\dmw}_{\RegH,\RegW}^{(\sU,1)}) \leq 2\sqrt{\delta}$ by~\cref{lemma:gentle-measurement}.
        \item[4.] $\widetilde{\dmw}_{\RegH,\RegW}^{(r,1)} \eqdef \Pi_r \widetilde{\dmw}_{\RegH,\RegW}^{(\sU,1)} \Pi_r / \Tr(\Pi_r \widetilde{\dmw}_{\RegH,\RegW}^{(\sU,1)})$. 
    \end{enumerate}

    \noindent Let $\widetilde{\zeta}_r \coloneqq \Tr(\Pi_r \widetilde{\dmw}_{\RegH,\RegW}^{(\sU,1)})$, and observe that $\widetilde{\zeta}_r \in [\zeta_r \pm 2\sqrt{\delta}]$. Define $\widetilde{\zeta}_R \eqdef \sum_{r \in R} \widetilde{\zeta}_r$ and $\widetilde{p}_{\sU} \eqdef \widetilde{\zeta}_R/\abs{R}$. 

\begin{claim}\label{claim:fake-p-inequality} $|\widetilde{p}_{\sU} - p_{\sU}| \leq 2\sqrt{\delta}$
\end{claim}
\begin{proof}
    For every string $r$, we have that 
    \[ |\zeta_r - \widetilde \zeta_r | = \left|\mathrm{Tr}(\Pi_r( \dmw_{\RegH,\RegW}^{(\sU,1)} - \widetilde{\dmw}_{\RegH,\RegW}^{(\sU,1)})\right| \leq 2\sqrt{\delta}
    \]
    since $||\dmw_{\RegH,\RegW}^{(\sU,1)} - \widetilde{\dmw}_{\RegH,\RegW}^{(\sU,1)}|| \leq 2\sqrt{\delta}$. Therefore, we have that
    \[ \left|\frac{\zeta_R}{\abs{R}} - \frac{\widetilde \zeta_R}{\abs{R}}\right| \leq 2\sqrt{\delta}
    \]
    by subadditivity. 
\end{proof}
    
    Consider the following two mixed states 

\[ \dmwr_{\RegH,\RegR,\RegW} := \sum_r \frac{\zeta_r}{\zeta_R} \ketbra r \tensor \frac{\Pi_r \dmw_{\RegH,\RegW}^{(\sU,1)} \Pi_r}{\zeta_r} , \text{ and}
\]

\[\widetilde{\dmwr}_{\RegH,\RegR,\RegW} := \sum_r \frac{\widetilde \zeta_r}{\widetilde \zeta_R} \ketbra r \tensor \frac{\Pi_r \widetilde \dmw_{\RegH,\RegW}^{(\sU,1)} \Pi_r}{\widetilde{\zeta}_r}.
\]
We claim that these two mixed states are close in trace distance.

\begin{claim}\label{claim:tau-trace-distance}
$|| \dmwr - \widetilde{\dmwr}||_1 \leq \frac{4\sqrt{\delta}}{p_{\sU}}$. 
\end{claim}

To see this, we first note that 
\begin{align*} \left|\left|  \widetilde{\dmwr} - \frac {\widetilde \zeta_R}{\zeta_R}\widetilde{\dmwr} \right|\right|_1 &= |1 -\frac{\widetilde \zeta_R}{\zeta_R}|\cdot  \left|\left| \widetilde{\dmwr}\right|\right|_1
\\ &= \left|1 -\frac{\widetilde \zeta_R}{\zeta_R}\right|
\\ &= \left|1 -\frac{\widetilde p_{\sU}}{p_{\sU}}\right|
\\ &\leq \frac{2\sqrt{\delta}}{p_{\sU}}
\end{align*}

\noindent Moreover, we have that 
\begin{align*}
    \left|\left|  \dmwr - \frac {\widetilde \zeta_R}{\zeta_R}\widetilde{\dmwr} \right|\right|_1 
    &= \frac 1 {\zeta_R} \left|\left|\sum_r   \ketbra{r} \tensor \Pi_r (\dmw_{\RegH,\RegW}^{(\sU,1)} - \widetilde{\dmw}_{\RegH,\RegW}^{(\sU,1)}) \Pi_r \right|\right|_1 
    \\ &\leq \frac {|R|}{\zeta_R} \cdot ||\dmw_{\RegH,\RegW}^{(\sU,1)} - \widetilde{\dmw}_{\RegH,\RegW}^{(\sU,1)}||_1
    \\ &\leq \frac {|R|}{\zeta_R} \cdot 2\sqrt{\delta}
    \\ &= \frac{2\sqrt{\delta}}{p_{\sU}}. 
\end{align*}

\noindent Thus, we conclude that $|| \dmwr - \widetilde{\dmwr}||_1 \leq \frac{2\sqrt{\delta}}{p_{\sU}} + \frac {2\sqrt{\delta}}{p_{\sU}} \leq \frac{4\sqrt{\delta}}{p_{\sU}}$ by the triangle inequality.

This trace bound will allow us to analyze correctness and bound the expected runtime of the extractor by appealing to properties of the state $\widetilde{\dmwr}$.

\subsubsection{Runtime Analysis}

In this section, we bound the expected running time of $\mathsf{Ext}$ (proving property (1) of \cref{thm:high-probability-extraction}). 

\begin{theorem}\label{thm:expected-qpt}
For any QPT $P^*$, $\mathsf{Ext}^{P^*}$ runs in $\EQPTM$. 
\end{theorem}

We note that \cref{lemma:step-1-2} already showed that the expected running time of \cref{step:first-measurement,step:variable-phase-est} is $O(1)$ calls to $(\sU, \sC)$. 

%%%%%%%%%%%%
%%%%%%%%%%%%

Next, we show that the expected runtime of the main loop (\cref{step:extract-transcript}) is also $\poly(\secp)$. To prove this, we make use of the syntactic property (enforced by the definition of \cref{step:loop-amplify-C}) that for every $i\in \{0,1,\dots,k\}$, the state $\dmr_{\RegB,\RegH,\RegR}^{(\sC,\mathrm{init})}$ at the beginning of the $i$th iteration of \cref{step:extract-transcript} is in $\image(\Pi_{\sC})$ (provided that the computation has not aborted). We then show

\begin{lemma}\label{lemma:running-time-p}
    Let $p$ be an arbitrary real number output by \cref{step:variable-phase-est}, and let $\dmr_{\RegB,\RegH,\RegR}^{(\sC,\mathrm{init})}$ be an \emph{arbitrary} non-aborted state (i.e. $\RegB$ is initialized to $\ketbra{1}_{\RegB}$) that is in the image of $\Pi_{\sC}$. 
    
    Then, the expected runtime of one iteration of \cref{step:extract-transcript} on $\dmr_{\RegB,\RegH,\RegR}^{(\sC,\mathrm{init})}$ is $\poly(\secp)/p$. 
    
\end{lemma}

\begin{proof}

We analyze the running time assuming that the collapsing measurement of~\cref{step:loop-collapsing} is not performed. This is without loss of generality; the steps following the (partial) collapsing measurement have a fixed runtime (as a function of previously computed parameters in the execution), so the collapsing measurement cannot affect the overall expected running time.\footnote{We only remove the collapsing measurement of the \emph{current} loop iteration; previous collapsing measurements are baked into the (arbitrary) state $\dmr_{\RegB,\RegH,\RegR}^{(\sC,\mathrm{init})}$.}

First, note that \cref{step:loop-reestimate,step:loop-amplify-C} run in a fixed $\poly(\secp)/\sqrt{p}$ time by \cref{thm:svt}. Thus, we focus on \cref{step:loop-phase-est,step:loop-svt}.

We bound the expected runtime of \cref{step:loop-phase-est,step:loop-svt} via the following hybrid argument.

\begin{itemize}
    \item $\mathsf{Hyb}_0$: This is the real procedure, assuming that \cref{step:loop-collapsing} is not performed.
    \item $\mathsf{Hyb}_1$: In this hybrid, the $\RegR$-measurement outcome and residual state $\dmh_{\RegB,\RegH}^{(3b)}$ is prepared differently:
    \begin{itemize}
        \item With probability $\alpha_0$, abort. Otherwise:
        \item The challenge $r$ is sampled with probability equal to $\zeta_r/\zeta_R$.
        \item If the string $r$ is sampled, initialize the state to $\ketbra{1}_{\RegB} \tensor \frac{\Pi_{V, r} \dmh_{\RegH}^{(\sU,1)} \Pi_{V, r}}{ \zeta_r}$.
    \end{itemize}
    This is an alternate description of the state $\dmr_{\RegB,\RegH,\RegR}'^{(\sC)}$.
    \item $\mathsf{Hyb}_2$: In this hybrid, the state at the beginning of \cref{step:loop-phase-est} (which is usually $\dmh_{\RegB,\RegH}^{(3b)}\tensor \ketbra{0}_{\RegW}$) is prepared differently:
    
    \begin{itemize}
    \item With probability $\alpha_0$, abort. Otherwise:
    \item A challenge is sampled so that each string $r$ occurs with probability $\widetilde \zeta_r/\widetilde \zeta_R$
    \item If the string $r$ is sampled, initialize the state to $\widetilde{\dmw}_{\RegH,\RegW}^{(r,1)}$. 
\end{itemize}
\end{itemize}

\begin{claim}
   The expected running times of $\mathsf{Hyb}_0$, $\mathsf{Hyb}_1$ differ by at most $O(k)$, and the expected running times of $\mathsf{Hyb}_1$ and $\mathsf{Hyb}_2$ differ by at most $O(k/p)$. 
\end{claim}

\begin{proof}
The \emph{worst-case} running time of $\mathsf{Ext}$ is bounded to be $k/\sqrt{\delta}$ by definition. We will combine this with trace distance bounds to prove the claim.

For $\mathsf{Hyb}_0$ and $\mathsf{Hyb}_1$, we note that the running time of \cref{step:loop-phase-est,step:loop-svt} can be viewed as a classical distribution over integers obtained via applying a CPTP map to the input state, which is either $\dmr_{\RegB,\RegH,\RegR}^{(\sC)}$ or $\dmr_{\RegB,\RegH,\RegR}'^{(\sC)}$. Since trace distance is contractive under CPTP maps, we conclude that these integer distributions are $2\sqrt{\delta}$-close in statistical distance. Since (as integers) they are $(k/\sqrt{\delta})$-bounded, we conclude that their expectations differ by $O(k)$.

The argument is similar for $\mathsf{Hyb}_1$ and $\mathsf{Hyb}_2$, except that the running time of \cref{step:loop-phase-est,step:loop-svt} can instead be viewed as a classical distribution obtained via a CPTP map from either $\dmwr$ or $\widetilde{\dmwr}$, which have trace distance at most $\frac{4\sqrt{\delta}}{p_{\mathsf{U}}} \leq \frac{4\sqrt{\delta}}{p}$. 
\end{proof}

\noindent Thus, it suffices to bound the expected runtime in the procedure $\mathsf{Hyb}_2$. 

Without~\cref{step:loop-collapsing}, we can view~\cref{step:loop-phase-est,step:loop-svt} as a variable-runtime $\mathsf{Transform}^{\sD_r\rightarrow \sG_{p, \eps, \delta}}$ with respect to the projectors $(\Pi_r, \Pi_{p, \epsilon})$, where $r$ is sampled from the above distribution. We first analyze the runtime of this procedure for a fixed value of $r$. 

Let $\dan = (\Pi_j^{\dan})_j$ denote the Jordan measurement corresponding to projections $(\Pi_r, \Pi_{p, \eps})$, and let $q_j$ denote the eigenvalue associated with $\Pi_j^{\dan}$. Define the \emph{Jordan weights} of $\widetilde{\dmw}_{\RegB,\RegH,\RegW}^{(\sU)}$ as the vector $(y^{\dan}_j)_j$ where
    \[ y^{\dan}_j \eqdef \Tr(\Pi^{\dan}_j \widetilde{\dmw}_{\RegB,\RegH,\RegW}^{(\sU)}).\]
    Then, the Jordan weights of $\widetilde{\dmw}_{\RegB,\RegH,\RegW}^{(r)}$ are $(z^{\dan}_j)_j$ where
    \[ z^{\dan}_j \eqdef q_j y^{\dan}_j/\widetilde{\zeta}_r.\]

\begin{claim}\label{lemma:svt-runtime-r}
    Given the string $r$ and state $\widetilde{\dmw}_{\RegB,\RegH,\RegW}^{(r)}$ as input, \cref{step:loop-phase-est,step:loop-collapsing,step:loop-svt} make an expected $\poly(\secp) \cdot  1 /\sqrt{\widetilde \zeta_r}$ calls to $\Pi_{p,\varepsilon}$ and $\Pi_r$.  
\end{claim}

\begin{proof}
    
    By \cref{thm:vrsvt}, the expected running time (in number of calls to $\Pi_r,\Pi_{p,\varepsilon}$) of \cref{step:loop-phase-est,step:loop-collapsing,step:loop-svt} on a state with Jordan weights $(q_j y^{\dan}_j/\widetilde{\zeta}_r)_j$ is
    \begin{align*} \sum_j \frac{q_j y^{\dan}_j}{\widetilde{\zeta}_r} \cdot \frac{\poly(\secp)}{\sqrt{q_j}} &\leq \poly(\secp) \sqrt{\sum_j \frac{q_j y^{\dan}_j}{\widetilde{\zeta}_r q_j} }\\
    &= \poly(\secp) \sqrt{ \frac{1}{\widetilde{\zeta}_r } }. 
    \end{align*}
    
    \noindent This completes the proof of \cref{lemma:svt-runtime-r}. 
\end{proof}

\noindent By \cref{lemma:svt-runtime-r}, along with the fact that $\Pi_{p, \eps}$ is implemented in a fixed $\poly(\secp)/\sqrt{p}$ time, the expected running time of~\cref{step:loop-phase-est,step:loop-collapsing,step:loop-svt} in $\mathsf{Hyb}_2$ is:
\begin{align}
    \sum_{r \in R} \frac{\widetilde{\zeta}_r}{\widetilde{\zeta}_R} \cdot  \frac{\poly(\secp)}{\sqrt{\widetilde{\zeta}_r  p}} &= \frac{\poly(\lambda)}{\sqrt{ p}} \sum_{r \in R} \frac{\widetilde{\zeta}_r}{\widetilde{\zeta}_R} \cdot  \frac 1 {\sqrt{\widetilde{\zeta}_r}} \nonumber \\
    &\leq \frac{\poly(\lambda)}{\sqrt{ p}} \sqrt{\sum_{r \in R} \frac{\widetilde{\zeta}_r}{\widetilde{\zeta}_R} \cdot  \frac 1 {\widetilde{\zeta}_r}} \nonumber\\
    &= \frac{\poly(\lambda)}{\sqrt{ p}} \sqrt{ \frac{\abs{R}}{\widetilde{\zeta}_R}} \nonumber\\
    &= \frac{\poly(\lambda)}{\sqrt{ p}} \sqrt{\frac 1 { \widetilde{p}_{\sU}}} \nonumber \\
    &\leq \frac{\poly(\lambda)}{ \sqrt{p (p - 2\sqrt{\delta})}} \nonumber \\ &\leq \frac{\poly(\secp)}{p}
\end{align} 
where the first inequality is an application of Jensen's inequality, the second inequality holds by \cref{claim:fake-p-inequality}, and the last inequality holds by the abort condition in \cref{step:variable-phase-est} ($p$ drops by a factor of at most $2$ in the entire process). 

This completes the proof of \cref{lemma:running-time-p}.

\end{proof}

Finally, combining \cref{lemma:running-time-p} with \cref{lemma:step-1-2} along with the fact that throughout the extraction procedure, the updated value of $p$ is at most a factor of $2$ smaller than the initial output of \cref{step:variable-phase-est}, we conclude that the overall expected running time of \cref{step:extract-transcript} is at most $\mathbb E[\poly(\secp)/p] \leq \poly(\secp)$, completing the proof of \cref{thm:expected-qpt}.

\subsubsection{Correctness of the repair step}
\label{subsec:repair-runtime}
In this section, we prove that $\Extract$ aborts with negligible probability (property (2) of \cref{thm:high-probability-extraction}). 

\begin{lemma}\label{lemma:negl-abort}
    The probability that the procedure aborts is negligible.
\end{lemma}
\begin{proof}
    By \cref{thm:expected-qpt}, the probability that the procedure aborts because it ran for too long is $O(\poly(\secp)/\sqrt{\delta}) = \negl(\secp)$.

    By \cref{lemma:step-1-2}, $\mathbb{E}[1/p \cdot X] = O(1)$, where $X$ is the indicator for whether \cref{step:first-measurement} outputs $1$. Hence by \cref{claim:abort-chosen-p} (proven below) and \cref{cor:threshold-est-almost} (which implies that the first iteration of \cref{step:loop-reestimate} only aborts with negligible probability), the probability that any iteration of the loop aborts when we remove \cref{step:loop-collapsing} is at most
    \begin{equation*}
        k \cdot O(\sqrt{\delta}) \mathbb{E}[1/p \cdot X] = O(\sqrt{\delta}).
    \end{equation*}
    Then by the collapsing guarantee (applied to the measurements of $y_1, \hdots, y_{k-1}$; it is not necessary for $y_k$), and by \cref{thm:expected-qpt}, the probability that any iteration of the loop aborts is negligible.
\end{proof}

\begin{claim}
    \label{claim:abort-chosen-p}
    Let $\dmr \in \Hermitians{\RegH \otimes \RegR}$ be a state such that $\Tr(\BProj{\sC} \dmr_{\RegH,\RegR}) = 1$, and consider running two iterations of \cref{step:extract-transcript} in sequence on $\dmr_{\RegH,\RegR}$, with the following modifications:
    
    \begin{itemize}
        \item \cref{step:loop-collapsing} is not applied, and
        \item The $O(k)/\sqrt{\delta}$ runtime cutoff has been removed.
    \end{itemize}
    Then, for any choice of $p \in [0,1]$, the probability that the first iteration (with initial value $p$) does not abort in \cref{step:loop-reestimate} and the second iteration aborts in \cref{step:loop-reestimate} is at most $O(\sqrt{\delta}/p)$.
    
    Also, the probability that an iteration of \cref{step:extract-transcript} (where \cref{step:loop-collapsing} is not applied) does not abort in \cref{step:loop-reestimate} but does abort in \cref{step:loop-amplify-C} is at most $O(\sqrt{\delta}/p)$.
\end{claim}
\begin{proof}
    For a projector $\Pi$, we write $\Pi^{\RegB}$ to denote the projection
    \[
        \ketbra{0}_{\RegB} \otimes \Id + \ketbra{1}_{\RegB} \otimes \Pi.
    \]
    Let $\RegB$ store the output of $\Threshold$ in the first application of \cref{step:loop-reestimate}. Recall that in \cref{sec:extractor-analysis-pseudoinverse-state} we have defined the following states:
    \begin{itemize}
        \item $\dmr_{\RegB,\RegH,\RegR}^{(\sC)}$ denotes the state after applying \cref{step:loop-reestimate} (i.e., coherently applying $\Threshold$ to $\dmr_{\RegH,\RegR}$ where the output is stored on $\RegB$). The extraction procedure now re-defines/updates $p := p - \epsilon$. Note that $\Tr((\SProj[\Jor]{\geq p})^{\RegB} \dmr_{\RegB,\RegH,\RegR}^{(\sC)}) \geq 1 - \delta$ (\cref{claim:analysis-pi-jor-good}).
        \item $\dmr_{\RegH,\RegR}^{(\sC,b)} \eqdef \bra{b}_{\RegB} \dmr_{\RegB,\RegH,\RegR}^{(\sC)} \ket{b}_{\RegB}/q$ where $q \eqdef \Tr(\bra{b}_{\RegB} \dmr_{\RegB,\RegH,\RegR}^{(\sC)} \ket{b}_{\RegB})$. We may assume $q > 0$ or else the claim holds trivially.
        \item $\dmr_{\RegB,\RegH,\RegR}'^{(\sC)} \eqdef \frac{\left(\SProj[\Jor]{\geq p}\right)^{\RegB} \dmr_{\RegB,\RegH,\RegR}^{(\sC)} \left(\SProj[\Jor]{\geq p}\right)^{\RegB}}{\Tr(\left(\SProj[\Jor]{\geq p}\right)^{\RegB} \dmr_{\RegB,\RegH,\RegR}^{(\sC)} )}$ is the result of projecting $\dmr_{\RegB,\RegH,\RegR}^{(\sC)}$ onto eigenvalues $\geq p$ (for the Jordan decomposition corresponding to $\BProj{\sC},\BProj{\sU}$) when $\RegB = 1$.
        \item $\dmr_{\RegH,\RegR}'^{(\sC, 1)} \eqdef \bra{1}_{\RegB} \dmr_{\RegB,\RegH,\RegR}'^{(\sC)} \ket{1}_{\RegB}/q'$ where $q' \eqdef \Tr( \bra{1}_{\RegB} \dmr_{\RegB, \RegH,\RegR}'^{(\sC,1)} \ket{1}_{\RegB})$. Note that $\Tr(\BProj{\sC}\dmr_{\RegH,\RegR}'^{(\sC, 1)}) =1$.
        \item $\dmr_{\RegH,\RegR}^{(\sU,1)}$ denotes the pseudoinverse of $\dmr_{\RegH,\RegR}'^{(\sC, 1)}$ with respect to $(\sU,\sC)$ as guaranteed by the pseudoinverse lemma (\cref{lemma:pseudoinverse}); by definition, $\Tr(\BProj{\sU}\dmr_{\RegH,\RegR}^{(\sU, 1)}) =1$. Moreover $\Tr(\BProj{\sC}\dmr_{\RegH,\RegR}^{(\sU, 1)}) \geq p$ (\cref{claim:pseudoinverse-win-probability}) since all the $(\BProj{\sC},\BProj{\sU})$-Jordan-eigenvalues of $\dmr_{\RegH,\RegR}'^{(\sC, 1)}$ are at least $p$, which implies the same property holds for the pseudoinverse state.
        \item $\dmh_{\RegH}^{(\sU,1)} \eqdef \Tr_{\RegR}(\dmr_{\RegH,\RegR}^{(\sU,1)})$. Note that since $\Tr(\BProj{\sU}\dmr_{\RegH,\RegR}^{(\sU, 1)}) =1$, we have $\dmr_{\RegH,\RegR}^{(\sU, 1)} = \dmh_{\RegH}^{(\sU,1)} \otimes \ketbra{+}_{\RegR}$.
        \item $\dmr_{\RegH,\RegR}'^{(3b,1)} \eqdef \frac{\sum_{r} \BProj{V,r} \dmh_{\RegH}^{(\sU,1)} \BProj{V,r} \otimes \ketbra{r}}{|R| \cdot \Tr(\BProj{\sC} \dmr_{\RegH,\RegR}^{(\sU,1)})} $ is within trace distance $2\sqrt{\delta}$ of the state after measuring the $\RegR$ register in~\cref{step:loop-measure-r} but \emph{before} discarding $\RegR$ (\cref{claim:stratify-r}). 
        \item $\dmw_{\RegH,\RegW}^{(\sU, 1)} \eqdef \dmh_{\RegH}^{(\sU,1)} \tensor \ketbra{0}_{\RegW}$. We have that $\Tr(\Pi_{p, \eps, \delta}\dmw_{\RegH,\RegW}^{(\sU, 1)}) \geq 1-\delta $ because $\Threshold_{p,\eps,\delta}$ outputs $1$ on $\dmr_{\RegH,\RegR}^{(\sU,1)}$ with probability $1-\delta$ by the Jordan spectrum guarantee of \cref{lemma:pseudoinverse}.
        \item $\widetilde{\dmw}_{\RegH,\RegW}^{(\sU,1)} \eqdef \frac{\Pi_{p, \eps, \delta}\dmw^{(\sU, 1)}\Pi_{p, \eps, \delta}}{\Tr(\Pi_{p, \eps, \delta}\dmw^{(\sU, 1)} )}$. By the gentle measurement lemma (\cref{lemma:gentle-measurement}), we have $d(\dmh_{\RegH}^{(\sU,1)} \otimes \ketbra{0}_{\RegW}, \widetilde{\dmw}_{\RegH,\RegW}^{(\sU,1)}) \leq 2\sqrt{\delta}$.
        \item $\widetilde{\dmw}_{\RegH,\RegW}^{(r,1)} \eqdef \frac{\Pi_r \widetilde{\dmw}_{\RegH,\RegW}^{(\sU,1)} \Pi_r }{\widetilde \zeta_r}$ where $\widetilde \zeta_r \eqdef \Tr(\Pi_r \widetilde{\dmw}_{\RegH,\RegW}^{(\sU,1)})$ for all $r \in R$.
        \item $\widetilde{\dmwr}_{\RegH,\RegW,\RegR} \eqdef \sum_r \frac{\widetilde \zeta_r}{\widetilde \zeta_R} \widetilde{\dmw}_{\RegH,\RegW}^{(r,1)}\otimes \ketbra{r}$ for $\widetilde{\zeta}_R = \sum_r \widetilde{\zeta}_r$. By \cref{claim:tau-trace-distance}, we have that $\widetilde{\dmwr}_{\RegH,\RegW,\RegR}$ is within trace distance $\frac{4\sqrt{\delta}}{p}$ of the state $\dmwr = \dmr_{\RegH,\RegR}'^{(3b,1)} \otimes \ketbra{0}_{\RegW}$.
    \end{itemize}
    
    We now consider the application of the variable-runtime singular vector transform performed across \cref{step:loop-phase-est,step:loop-svt} (recall that we omit \cref{step:loop-collapsing} for this analysis). We consider applying these steps to the state $\widetilde{\dmw}_{\RegH,\RegW}^{(r,1)}$. Note that \cref{step:loop-phase-est,step:loop-svt} commute with $\Meas{\Jor}[\sD_r,\sG_{p,\varepsilon,\delta}]$. Hence writing $\widetilde{\dmw}_{\RegH,\RegW}^{(3e,r,1)}$ for the state after applying \cref{step:loop-phase-est,step:loop-svt} to $\widetilde{\dmw}_{\RegH,\RegW}^{(r,1)}$, we have
    \[ \Tr(\SProj[\dan]{j} \widetilde{\dmw}_{\RegH,\RegW}^{(3e,r,1)}) = \Tr(\SProj[\dan]{j} \widetilde{\dmw}_{\RegH,\RegW}^{(r,1)}) = \frac{q_j \Tr(\SProj[\dan]{j}\widetilde{\dmw}_{\RegH,\RegW}^{(\sU,1)})}{\widetilde{\zeta}_r}~ \]
    where $\SProj[\dan]{j}$ is the $j$-th element of $\Meas{\Jor}[\sD_r,\sG_{p,\varepsilon,\delta}]$. Let $\SProj[\dan]{\sG_{p,\varepsilon,\delta},\mathsf{stuck}}$ be defined (analogous to $\SProj[\Jor]{\mathsf{stuck}}$  in~\cref{claim:jordan-rotate}) as $\SProj[\Jor]{\mathsf{stuck}} \eqdef \sum_{j \in S} \SProj[\dan]{j}$ where $S$ is the set of all $j$ where $\RegS_j$ is a one-dimensional Jordan subspace $\RegS_j \in \image(\Id-\BProj{p,\varepsilon,\delta})$. We now invoke~\cref{claim:jordan-rotate} to ``rotate'' the state $\widetilde{\dmw}_{\RegH,\RegW}^{(3e,r,1)}$ into $\image(\BProj{p,\varepsilon,\delta})$ while preserving the Jordan spectrum, which is possible as long as the component of the state in $\SProj[\dan]{\sG_{p,\varepsilon,\delta},\mathsf{stuck}}$ is $0$. This is satisfied here because $\Tr(\SProj[\dan]{\mathsf{stuck}} \widetilde{\dmw}_{\RegH,\RegW}^{(3e,r,1)}) = \Tr(\SProj[\dan]{\mathsf{stuck}} \widetilde{\dmw}_{\RegH,\RegW}^{(\sU,1)}) = 0$ since $\widetilde{\dmw}_{\RegH,\RegW}^{(\sU,1)}$ was defined so that $\Tr(\BProj{p,\varepsilon,\delta}\widetilde{\dmw}_{\RegH,\RegW}^{(\sU,1)}) = 1$.
    
    Additionally, by the guarantee of \cref{thm:vrsvt}, $\Tr(\BProj{p,\varepsilon,\delta} \widetilde{\dmw}_{\RegH,\RegW}^{(3e,r,1)}) \geq 1 - \delta$. Hence by \cref{claim:jordan-rotate}, there exists a state $\widetilde{\dmw}_{\RegH,\RegW}^{(\sG,r,1)}$ with the same Jordan spectrum as $\widetilde{\dmw}_{\RegH,\RegW}^{(3e,r,1)}$, $\Tr(\ProjP \widetilde{\dmw}_{\RegH,\RegW}^{(\sG,r,1)}) = 1$ and $d(\widetilde{\dmw}_{\RegH,\RegW}^{(\sG,r,1)},\widetilde{\dmw}_{\RegH,\RegW}^{(3e,r,1)}) \leq \sqrt{\delta}$. Note that $\widetilde{\dmw}_{\RegH,\RegW}^{(\sG,r,1)}$ also has the same Jordan spectrum of $\widetilde{\dmw}_{\RegH,\RegW}^{(r,1)}$

    Consider the pseudoinverse state $\widetilde{\dmw}_{\RegH,\RegW}^{(\sD,r,1)}$ under $(\sD_r,\sG_{p,\varepsilon,\delta})$ of $\widetilde{\dmw}_{\RegH,\RegW}^{(\sG,r,1)}$. Since the Jordan spectrum of $\widetilde{\dmw}_{\RegH,\RegW}^{(\sG,r,1)} \in \image(\ProjP)$ is identical to the Jordan spectrum of $\widetilde{\dmw}_{\RegH,\RegW}^{(r,1)} \in \image(\BProj{r})$, and $\widetilde{\dmw}_{\RegH,\RegW}^{(\sU,1)}$ is the pseudoinverse under $(\sG_{p,\varepsilon,\delta},\sD_r)$ of $\widetilde{\dmw}_{\RegH,\RegW}^{(r,1)}$, it follows that 
    \[\Tr( \ProjP  \widetilde{\dmw}_{\RegH,\RegW}^{(\sD,r,1)}) = \Tr( \BProj{r} \widetilde{\dmw}_{\RegH,\RegW}^{(\sU,1)}) = \widetilde{\zeta}_r,\]
    and moreover $\widetilde{\dmw}_{\RegH,\RegW}^{(\sG,r,1)} = \ProjP \widetilde{\dmw}_{\RegH,\RegW}^{(\sD,r,1)} \ProjP/\widetilde{\zeta}_r$.
    
    Hence, if the state before \cref{step:loop-phase-est} is $\sum_r \frac{\widetilde{\zeta}_r}{\widetilde{\zeta}_R} \widetilde{\dmw}^{(r,1)}_{\RegH,\RegW}$ (which is $\frac{4\sqrt{\delta}}p$ close to the actual state before \cref{step:loop-phase-est}) then the state after \cref{step:loop-svt} is $O(\sqrt{\delta})$-close to the following state:
    \begin{equation*}
        \widetilde{\dmw}_{\RegH,\RegW}^{(3e,1)} \eqdef \frac{1}{\widetilde{\zeta}_R} \sum_{r} \ProjP \left( \widetilde{\dmw}_{\RegH,\RegW}^{(\sD,r,1)}\right) \ProjP.
    \end{equation*}
    Therefore, writing $\widetilde{\dmw}_{\RegH,\RegW}^{(\sD,r,1)} = \widetilde{\dmh}_{\RegH}^{(\sD,r,1)}\tensor \ketbra{0}_{\RegW} $ ($\widetilde{\dmw}_{\RegH,\RegW}^{(\sD,r,1)}$ has this form since $\widetilde{\dmw}_{\RegH,\RegW}^{(\sD,r,1)} \in \image(\BProj{r})$), the state at the end of \cref{step:loop-discard-w} is $O(\sqrt{\delta})$-close to
    \begin{equation*}
        \widetilde{\dmh}_{\RegH}^{(3f,1)} \eqdef  \frac{1}{\widetilde{\zeta}_R} \sum_{r\in R} M_{p,\varepsilon,\delta} \left(  \widetilde{\dmh}_{\RegH}^{(\sD,r,1)}\right) M_{p,\varepsilon,\delta}^{\dagger},
    \end{equation*}
    where $M_{p,\eps,\delta}$ is the measurement element of $\sT_{p,\eps,\delta}$ that corresponds to a $1$ outcome.
    
    By the guarantee of $\Threshold$ (\cref{thm:svdisc}), it holds that for all states $\dmh \in \Hermitians{\RegH}$,
    \begin{equation*}
        \Tr(\SProj[\Jor]{< p-\eps} M_{p,\eps,\delta} \dmh M_{p,\eps,\delta}^{\dagger}) \leq \delta,
    \end{equation*}
    and so by linearity,
    \begin{align}
        \Tr(\SProj[\Jor]{< p-\eps} \widetilde{\dmh}_{\RegH}^{(3f,1)}) &= \frac{1}{\widetilde{\zeta}_R} \sum_{r\in R} \Tr(\SProj[\Jor]{< p-\eps}M_{p,\varepsilon,\delta} \left(  \widetilde{\dmh}_{\RegH}^{(\sD,r,1)}\right) M_{p,\varepsilon,\delta}^{\dagger}) \nonumber \\ &\leq \frac 1 {\widetilde{\zeta}_R} \cdot |R| \cdot \delta \nonumber \\
        &\leq \frac{\delta}{p-2\sqrt{\delta}} \label{eq:fake-p}\\ 
        &= O(\delta/p) \nonumber,
    \end{align}
    where \cref{eq:fake-p} holds by \cref{claim:fake-p-inequality}. It follows from the guarantees of the fixed-runtime singular vector transform (\cref{thm:svt}) that the state at the end of \cref{step:loop-amplify-C} is $O(\delta)$-close to the state $\widetilde{\dmr}_{\RegH,\RegR}^{(3g,1)} = \Transform_{p-\eps}[\sU \to \sC](\widetilde{\dmh}_{\RegH}^{(3f,1)}\tensor \ketbra{+_R}_{\RegR})$, which has the property that
    \[ \Tr((\Id-\BProj{\sC}) \SProj[\Jor]{\geq p-\eps} (\widetilde{\dmh}_{\RegH}^{(3f,1)}\otimes \ketbra{+_R}_{\RegR})) \leq \delta.\]
    
    Combining this with $\Tr(\SProj[\Jor]{< p-\eps} \widetilde{\dmh}_{\RegH}^{(3f,1)}) = O(\delta/p)$, we conclude that if the state before \cref{step:loop-phase-est} is $\sum_r \frac{\widetilde{\zeta}_r}{\widetilde{\zeta}_R} \widetilde{\dmw}^{(r,1)}_{\RegH,\RegW}$, then the probability that \cref{step:loop-amplify-C} aborts is at most $O(\delta/p)$. Additionally, the guarantee of $\Threshold$ (\cref{thm:svdisc}) implies that in the next iteration of \cref{step:extract-transcript}, the probability that \cref{step:loop-reestimate} aborts on $\widetilde{\dmr}_{\RegH,\RegR}^{(3g,1)}$ is also at most $O(\delta/p)$. By a trace distance argument, the probability that \cref{step:loop-amplify-C} or the subsequent \cref{step:loop-reestimate} aborts in a \emph{real} execution of \cref{step:extract-transcript} (with the modifications as in the statement of~\cref{claim:abort-chosen-p}) when the first \cref{step:loop-reestimate} did not abort is at most $O(\sqrt{\delta}/p)$. This completes the proof of \cref{claim:abort-chosen-p}. 
\end{proof}

\subsubsection{Correctness of Transcript Generation}

Finally, we prove property (3) of \cref{thm:high-probability-extraction}. 

\begin{lemma}
    For every $\tau_{\mathrm{pre}} = (\vk, a)$, let $\gamma = \gamma_{\vk, a}$ denote the initial success probability of $P^*$ conditioned on $\tau_{\mathrm{pre}}$. Then, if $\gamma > \delta^{1/3}$, the distribution $D_k$ on  $(r_1, \hdots, r_k)$ (conditioned on $(\vk, a)$ and a successful first execution) is $O(1/\gamma)$-admissible (\cref{def:admissible-dist}).
\end{lemma}
This follows by appealing to the following claim in each round, making use of the fact that the expectation of $1/p$ conditioned on an accepting initial execution is equal\footnote{Here (and elsewhere) we informally make use of the fact that the ``current'' value of $p$ in any iteration of \cref{step:extract-transcript} is always at least $p_0/2$, where $p_0$ is the \emph{initial} estimated $p$.} to $1/\gamma$; the $O(\sqrt{\delta})$-closeness from the claim also degrades to $O(\sqrt{\delta}/\gamma)$ when conditioning on an accepting initial execution.

\begin{claim}
    Consider the distribution $D$ supported on $R \cup \{\bot\}$ obtained running a single iteration of \cref{step:extract-transcript} with parameter $p$ on an arbitrary state $\dmr \in \Hermitians{\RegH \otimes \RegR}$ with $\Tr(\BProj{\sC} \dmr) = 1$ (where $r \eqdef \bot$ if $\Extract$ aborts). There exists a procedure $\Samp$ that makes expected $O(1/p)$ queries to uniform sampling oracle $O_R$ (but can otherwise behave arbitrarily and inefficiently) that outputs a distribution $O(\sqrt{\delta})$-close to $D$, and if the output of $\Samp$ is not $\bot$ then is one of the responses to its oracle queries.
\end{claim}
\begin{proof}
    $\Samp$ initially behaves similarly to $\Extract$: apply $\Threshold_{p,\eps,\delta}$ to $\dmr$; if $\Threshold$ outputs $0$ then output $\bot$. Let $\dmr_{\RegB,\RegH,\RegR}^{(\sC)}$ be the state after applying $\Threshold$, and (as in $\Extract$) re-set $p := p-\eps$.
    
    As before, let $\dmr_{\RegB,\RegH,\RegR}'^{(\sC)} \eqdef \frac{\left(\SProj[\Jor]{\geq p}\right)^{\RegB} \dmr_{\RegB,\RegH,\RegR}^{(\sC)} \left(\SProj[\Jor]{\geq p}\right)^{\RegB}}{\Tr(\left(\SProj[\Jor]{\geq p}\right)^{\RegB} \dmr_{\RegB,\RegH,\RegR}^{(\sC)} )}$. Let  $\dmr_{\RegH,\RegR}^{(\sU,1)}$ be the pseudoinverse of $\dmr_{\RegH,\RegR}'^{(\sC, 1)} = \bra{1}_{\RegB} \dmr_{\RegB,\RegH,\RegR}'^{(\sC)} \ket{1}_{\RegB}/\Tr( \bra{1}_{\RegB} \dmr_{\RegB, \RegH,\RegR}'^{(\sC,1)} \ket{1}_{\RegB})$ as guaranteed by the pseudoinverse lemma (\cref{lemma:pseudoinverse}).
    
    We have by \cref{lemma:gentle-measurement,lemma:pseudoinverse} that $d(\dmr_{\RegB,\RegH,\RegR}^{\sC},\dmr_{\RegB,\RegH,\RegR}'^{\sC}) \leq 2\sqrt{\delta}$ and $\dmr_{\RegH,\RegR}'^{\sC,1} = \frac{\Pi_C \dmr_{\RegH,\RegR}^{(\sU,1)} \Pi_C }{\Tr(\BProj{\sC} \dmr_{\RegH,\RegR}^{(\sU,1)})}$. Finally, write $\dmr_{\RegH,\RegR}^{(\sU,1)} = \dmh_{\RegH}^{(\sU,1)} \otimes \ketbra{+}_{\RegR}$.
    
    $\Samp$ now behaves differently than $\Extract$. $\Samp$ ``clones'' $\dmh_{\RegH}^{(\sU,1)}$ (recall that $\Samp$ can be an arbitrary function) and repeats the following until $b = 1$: query $O_R$, obtaining $r \in R$; on a fresh copy of $\dmh_{\RegH}^{(\sU,1)}$, measure whether the verifier accepts on challenge $r$ (i.e., $\BMeas{\BProj{V,r}}$), obtaining bit $b$. Output $r$ if $b = 1$.
    
    Let $p_{\sU} \eqdef \Tr(\BProj{\sC} \dmh_{\RegH}^{(\sU,1)})$; we have that $p_{\sU} \geq p $ by \cref{claim:pseudoinverse-win-probability}. Hence, the expected number of queries $\Samp$ makes is $1/p$. Observe that $p_{\sU}$ is the probability that a uniform $r$ is accepted. Let $\zeta_r \eqdef \Tr(\BProj{V,r} \dmh_{\RegH}^{(\sU, 1)})$; $\zeta_r$ is the probability that $r$ is accepted. Then, for every $r^*\in R$,
    \begin{equation*}
        \Pr_{\mathsf{Samp}}[r = r^*] = \sum_{n=0}^{\infty} \Pr[r_1,\ldots,r_{n} \text{ rejected}] \Pr[r_{n+1} = r^*, r^* \text{ accepted}] = \sum_{n = 0}^{\infty} (1-p_{\sU})^n \cdot \frac{\zeta_{r^*}}{|R|} = \frac{\zeta_{r^*}}{p_{\sU} \cdot |R|}~.
    \end{equation*}
    
    Consider now the distribution on $r$ obtained by measuring $\RegR$ on state $\dmr_{\RegH,\RegR}'^{(\sC,1)}$: for every $r^*$,
    \begin{equation*}
        \Pr[r = r^*] = \Tr(\ketbra{r^*} \dmr_{\RegH,\RegR}'^{(\sC,1)}) = \frac{\Tr(\ketbra{r^*} \BProj{\sC} \dmr_{\RegH,\RegR}^{(\sU,1)} \BProj{\sC})}{p_{\sU}} = \frac{\Tr(\BProj{V,r^*} \dmh_{\RegH}^{(\sU,1)})}{p_{\sU}\cdot|R|} = \frac{\zeta_{r^*}}{p_{\sU} \cdot |R|}
    \end{equation*}
    since $(\Id_{\RegH}\tensor \ketbra{r^*} ) \BProj{\sC} = \BProj{V,r^*}\tensor \ketbra{r^*} $ and $\ketbra{r^*} \dmr_{\RegH,\RegR}^{(\sU,1)} \ketbra{r^*} = \frac{1}{|R|} \dmh_{\RegH}^{(\sU,1)} \tensor \ketbra{r^*}_{\RegR} $.
    
    Overall, $D$ is obtained by measuring $(\RegB,\RegR)$ on the state $\dmr_{\RegB,\RegH,\RegR}^{(\sC)}$, which is $O(\sqrt{\delta})$-close to $\dmr_{\RegB,\RegH,\RegR}'^{(\sC)}$; the claim follows by contractivity of trace distance.
\end{proof}

Having established properties (1), (2), and (3), we have proved \cref{thm:high-probability-extraction}!

\subsection{Obtaining Guaranteed Extraction}\label{sec:obtaining-guaranteed-extraction}
In this section, we combine the guarantees of \cref{thm:high-probability-extraction} with additional analysis to prove that all of the example protocols from \cref{sec:examples} have guaranteed extractors (additionally assuming partial collapsing where necessary). We remark that handling the graph isomorphism subroutine requires a slight modification of the \cref{thm:high-probability-extraction}, which we detail below.

We begin with a general-purpose corollary of \cref{thm:high-probability-extraction} for the case of protocols satisfying $(k, g)$-PSS (\cref{def:k-g-pss}) in addition to $f_1, \hdots, f_{k-1}$-partial collapsing (which was assumed in \cref{thm:high-probability-extraction}).

\begin{corollary}\label{cor:guaranteed-extractor-or-inconsistent}
  Let $(P_{\Sigma}, V_{\Sigma})$ be a 3- or 4- message public coin interactive argument with a consistency function $g: T \times (R \times \{0,1\}^*)^*\rightarrow \{0,1\}$, and let $f_1, \hdots, f_k$ be functions. Suppose that:
  \begin{itemize}
      \item The protocol is partially collapsing with respect to $f_1, \hdots, f_{k-1}$, and
      \item The protocol is $(k, g)$-PSS for some $k = \poly(\secp)$.
  \end{itemize}
  Then, one of the two following conclusions holds:
  \begin{enumerate}
      \item The extractor from \cref{thm:high-probability-extraction} composed with the PSS extractor $\PSSExtract$ satisfies \emph{guaranteed extraction}, OR
      \item The extractor from \cref{thm:high-probability-extraction} outputs a $k$-tuple of partial transcripts $(r_1, y_1, \hdots, r_k, y_k)$ such that $g(\prefix, r_1, y_1, \hdots, r_k, y_k) = 0$ (the transcripts are \emph{inconsistent}) with non-negligible probability. 
  \end{enumerate}
\end{corollary}

\begin{proof}
  Suppose that conclusion (1) is false, meaning that there exist infinitely many $\secp$ and a constant $c$ such that the extractor from \cref{thm:high-probability-extraction} has an accepting initial execution but the call to $\PSSExtract$ fails to produce a witness with probability at least $1/\secp^c$. We know that the \cref{thm:high-probability-extraction} extractor aborts with negligible probability, so we also assume that the extractor does not abort here. Then, by an averaging argument, with probability at least $\frac 1 {2\secp^c}$ over the distribution of $(\vk, a)$, the above event conditioned on $(\vk, a)$ holds with probability at least $\frac 1 {2\secp^c}$. This in particular implies that $\gamma_{\vk, a}$ (as defined in \cref{thm:high-probability-extraction}) is at least $\frac 1 {2\secp^c}$ for these choices of $(\vk, a)$. Then, property (3) of \cref{thm:high-probability-extraction} implies that the distribution of $(r_1, \hdots, r_k)$ is \emph{admissible} for these choices of $(\vk, a)$ (and choices of $\secp$). Thus, the $(k, g)$-PSS property of $(P_{\Sigma}, V_{\Sigma})$ implies that for every such $(\vk, a)$, the $k$-tuple of partial transcripts must be inconsistent with probability at least $\frac 1 {2 \secp^c}$ (as otherwise $\PSSExtract$ would succeed with $1-\negl$ probability). Therefore, assuming that conclusion (1) is false, the probability that the $k$-tuple of transcripts output by the \cref{thm:high-probability-extraction} extractor are inconsistent is at least $\frac 1 {4 \secp^{2c}}$ for infinitely many $\secp$, implying conclusion (2). \qedhere
\end{proof}

Finally, we apply \cref{cor:guaranteed-extractor-or-inconsistent} to obtain guaranteed extractors for all of the \cref{sec:examples} example protocols (along with a general result for $k$-special sound protocols).

\begin{corollary}\label{cor:collapsing-ss-extract}
  If $(P_{\Sigma}, V_{\Sigma})$ is (fully) collapsing and $k$-special sound, and $|R|= 2^{\omega(\log \secp)}$, then the protocol has guaranteed extraction.
\end{corollary}
\begin{proof}
  Since $(P_{\Sigma}, V_{\Sigma})$ is $k$-special sound and $|R|= 2^{\omega(\log \secp)}$, we know that the protocol is $(k, g)$-PSS for the ``trivial'' transcript consistency predicate $g$. Therefore, \cref{cor:guaranteed-extractor-or-inconsistent} applies to this protocol (where the extractor sets $f_1 = \hdots = f_k = \mathsf{Id}$). However, conclusion (2) of \cref{cor:guaranteed-extractor-or-inconsistent} cannot happen because the consistency predicate of $\PSSExtract$ in this case simply checks that the transcripts are accepting, which is guaranteed by the fact that $(r_i, y_i = z_i)$ was a measurement outcome of a state in $\Pi_C$.
\end{proof}

\begin{corollary}\label{cor:commit-and-open}
If $(P_{\Sigma}, V_{\Sigma})$ is a commit-and-open protocol (\cref{def:commit-and-open}) satisfying commit-and-open $k$-special soundness and $R = 2^{\omega(\log \secp)}$ (either natively or enforced by parallel repetition), and the commitment scheme is instantiated using a collapse-binding commitment \cite{EC:Unruh16}, then the protocol has a guaranteed extractor. 
\end{corollary}
\begin{proof}
Under the hypotheses of the corollary (along with \cref{claim:kss-to-kfss,k-g-ss-implies-k-g-pss}), the protocol satisfies either $(k, g)$-PSS (if it has a natively superpolynomial challenge space) or $(k^2\log^2(\secp), g)$-PSS (if parallel repeated; see~\cref{lemma:parallel-repetition}), where $g$ is a predicate that enforces the constraint that all opened messages are consistent with each other. We set $f_1 = \hdots = f_k = f$ where $f(z)$ outputs the substring of $z$ corresponding to the opened messages (and not the openings). Then, the \cref{thm:high-probability-extraction} extraction procedure does not violate $g$-consistency by the unique-message binding of the commitment scheme (shown in \cref{lemma:collapse-binding-unique-message}). Thus, \cref{cor:guaranteed-extractor-or-inconsistent} implies that $(P_{\Sigma}, V_{\Sigma})$ has a guaranteed extraction procedure. \qedhere

\end{proof}

\begin{corollary}
\label{corollary:kilian-guaranteed}
Kilian's succinct argument system \cite{STOC:Kilian92}, when instantiated using a collapsing hash function and a PCP of knowledge, has a guaranteed extraction procedure.
\end{corollary}

\begin{proof}
  We know from~\cref{claim:kilian-pss} the \cite{STOC:Kilian92} succinct argument system is (1) (fully) collapsing, and (2) $(k, g)$-PSS for $k = \poly(n, \secp)$ and $g$ defined so that when $z_i$ and $z_j$ contain overlapping leaves of the Merkle tree, the leaf values are equal. We set $f_1 = \hdots = f_k = \mathsf{Id}$, and observe that the \cref{thm:high-probability-extraction} extractor does not violate $g$-consistency, because if it output two transcripts $(r_1, z_1), (r_2, z_2)$ with inconsistent leaf values, since the transcripts are accepting (they were obtained by measuring a state in $\Pi_C$), this would violate the collision-resistance (implied by collapsing) of the hash family. Thus, by \cref{cor:guaranteed-extractor-or-inconsistent}, the protocol has a guaranteed extractor. 
\end{proof}

\begin{corollary}\label{cor:gni-guaranteed-extractor}
The one-out-of-two graph isomorphism subroutine has a guaranteed extraction procedure that extracts the bit $b$ (when $G_0$ and $G_1$ are not isomorphic). 
\end{corollary}
\begin{proof}
By~\cref{claim:gni-pss}, this protocol is $(2,g')$-PSS where $g'$ is the following asymmetric function:

\begin{itemize}
    \item For the first partial transcript $(\prefix, r^{(1)}, c^{(1)})$, $g'$ checks that for all $i$ such that $r_i = 0$, $(H_{0,i}, H_{1,i})$ are isomorphic to $(G_{c^{(1)}_i}, G_{1-c^{(1)}_i})$.
    \item For the second partial transcript $(\prefix, r^{(2)}, c^{(2)})$, $g'$ \emph{additionally} checks that for all $i$ such that $r_i = 1$, $H_{c^{(2)}_i, i}$ is isomorphic to $H$.
\end{itemize}
We define the following pair of functions $f_1, f_2$:
\begin{itemize}
    \item $f_1(\prefix, r, z)$ outputs the following substring of $z$. For every $i$ such that $r_i = 0$, the substring includes the bit $c_i$ (where $z_i = (c_i, \sigma_{0,i}, \sigma_{1,i})$).
    \item $f_2(\prefix, r, z)$ outputs the substring $c$ (the distinguished single bit of each $z_i$). 
\end{itemize}
We note that the graph isomorphism subprotocol is $f_1$-collapsing; this follows from the fact that for any \emph{accepting} transcript $(\prefix, r, z)$, the bits $c_i$ (for $r_i = 0$) are information-theoretically determined as a function of $(G_0, G_1, H_{0,i}, H_{1,i})$. 

Thus, if we instantiate the \cref{thm:high-probability-extraction} extractor using $(f_1, f_2)$ (note that we require no properties of $f_2$) we have that \cref{cor:guaranteed-extractor-or-inconsistent} applies. Moreover, $g'$-consistency of the transcripts output by the extractor is not violated, because it is formally implied by the fact that they were obtained by partially measuring a state in $\Pi_C$ (any accepting partial transcript $(r_i, c_i)$ satisfies the condition checked by $g'$). Thus, we conclude that the protocol has a guaranteed extractor by \cref{cor:guaranteed-extractor-or-inconsistent}. 
\end{proof}

%% file: 9-eqpt.tex
% !TeX root = 0-main.tex

\newcommand{\Tape}{\mathbb{T}}

\section{Expected Polynomial Time for Quantum Simulators}\label{sec:eqpt}

We introduce a notion of efficient computation we call coherent-runtime expected quantum polynomial time ($\EQPTC$). We then formalize a new definition of post-quantum zero-knowledge with $\EQPTC$ simulation.

\subsection{Quantum Turing Machines}
\label{sec:qtms}

We recall the definition of a quantum Turing machine (QTM) of Deutsch \cite{Deutsch85}. A QTM is a tuple $(\Alphabet,\TStates,\Tx,\TStateInit,\TStateFinal)$ where $\Alphabet$ is a finite set of symbols, $\TStates$ is a finite set of states, $\Tx \colon \TStates \times \Alphabet \to \C^{\TStates \times \Alphabet \times \{-1,0,1\}}$ is a transition function, and $\TStateInit,\TStateFinal$ are the initial and final (halting) states respectively.

We fix registers $\RegState$ containing the state, $\RegHead$ containing the position of the tape head, and $\RegTape$ containing the tape. A configuration state of a Turing machine is a vector $\ket{\TState,i,\Tape} \in \RegState \otimes \RegHead \otimes \RegTape$ where $\TState \in \TStates$ is the current state, $i \in \Naturals$ is the location of the tape head, and $\Tape \in \Alphabet^*$ is the (finite) contents of the tape.

A transition is given by the map $U_{\delta}$, which acts on basis states as follows:
\begin{equation*}
    \ket{q,i,T} \mapsto \sum_{\TState' \in \TStates} \sum_{a \in \Alphabet} \sum_{d \in \{-1,0,1\}} \alpha_{\TState',a,d,b} \ket{q',i + d,\Tape_{i \to a}}
\end{equation*}
where $\delta(q,\Tape_i) = \sum_{q',a,d} \alpha_{q',a,d} \ket{q',a,d}$. $\delta$ is a valid transition function if and only if $U_{\delta}$ is unitary. The definition of QTMs generalises to multiple tapes in the natural way. We will consider QTMs having a separate input/output tape on register $\RegInput$ (with head position in $\RegInpHead$).

The execution of a $T$-bounded QTM proceeds as follows.
\begin{enumerate}[noitemsep]
    \item Initialize register $\RegState$ to $\ket{\TStateInit}$, $\RegHead,\RegInpHead$ to $\ket{0}$, and $\RegTape$ to the empty tape state $\ket{\varnothing}$.
    \item \label[step]{step:qtm-main-loop} Repeat the following for at most $T$ steps:
    \begin{enumerate}
        \item Apply the measurement $\BProj{f} = \BMeas{\ketbra{\TStateFinal}}$ to $\RegState$. If the outcome is $1$, halt and discard all registers except $\RegInput$.
        \item Apply $U_{\delta}$.
    \end{enumerate}
\end{enumerate}
The \defemph{output} $M(\DMatrix)$ of a QTM $M$ on input $\DMatrix \in \Hermitians{\RegInput}$ is the state on $\RegInput$ when the machine halts. The \defemph{running time} $t_M(\DMatrix)$ of $M$ on input $\DMatrix$ is the number of iterations of \cref{step:qtm-main-loop}. Note that both of these quantities are random variables. 

\begin{definition}
	The expected running time $E_M(n)$ of a QTM $M$ is the maximum over all $n$-qubit states $\DMatrix$ of $\Expectation[t_M(\DMatrix)]$. A $T$-bounded QTM $M$ (for some $T \leq \exp(n)$) is $\EQPTM$ and there exists a polynomial $p$ such that $E_M(n) \leq p(n)$ for all $n$.
\end{definition}

\subsection{Coherent-Runtime EQPT}

\newcommand{\CSymb}{\mathsf{D}}

\begin{definition}
	A $\CSymb$-circuit is a quantum circuit $C$ with special gates $\{ G_i,G_i^{-1} \})_{i=1}^{k}$ with the following restriction: for each $i$, there is a single $G_i$ gate and a single $G_i^{-1}$ gate acting on a designated register $\RegX_i$, where $G_i$ acts before $G_i^{-1}$. All other gates may act arbitrarily on $\RegY \otimes \bigotimes_{i=1}^{k} \RegX_i$, for some register $\RegY$. For any CPTP maps $\Phi_i \colon \Hermitians{\RegX_i} \to \Hermitians{\RegX_i}$, $C[\Phi_1,\ldots,\Phi_k] \colon \Hermitians{\RegY \otimes \bigotimes_{i=1}^k \RegX_i} \to \Hermitians{\RegY \otimes \bigotimes_{i=1}^k \RegX_i}$ is the superoperator defined as follows:
	\begin{enumerate}[noitemsep]
		\item For each $i$, $U_i$ be a unitary dilation of $\Phi_i$. That is, let $\RegZ_i$ be an ancilla Hilbert space and $U_{\Phi}$ unitary on $\RegX_i \otimes \RegZ_i$ such that $\Phi(\DMatrixW) = \Tr_{\RegZ_i}(U_i (\DMatrixW \otimes \ketbra{0}_{\RegZ_i}) U_i^{\dagger})$ for all $\DMatrixW \in \Hermitians{\RegX_i}$.
		\item Construct a circuit $C'$ on $\RegY \otimes \bigotimes_{i=1}^{k} (\RegX_i \otimes \RegZ_i)$ from $C$ by replacing $G_i$ with $U_i$ and $G_i^{-1}$ with $U_i^{\dagger}$ for each $i$.
		\item Let $C$ be the superoperator $\DMatrix \mapsto \Tr_{\RegZ}(C'(\DMatrix \otimes \bigotimes_{i=1}^{k} \ketbra{0}_{\RegZ_i}))$.
	\end{enumerate}
	Since all choices of $U_i$ are equivalent up to a local isometry on $\RegZ_i$, the map $C[\Phi_1,\ldots,\Phi_k]$ is well-defined.
\end{definition}

We are now ready to define our notion of \emph{coherent-runtime expected quantum polynomial time}.
\begin{definition}
\label{def:cr-eqpt}
	A sequence of CPTP maps $\{\Phi_{n}\}_{n \in \Naturals}$ is a \defemph{$\EQPTC$ computation} if there exist a uniform family of $\CSymb$-circuits $\{ C_n \}_{n \in \Naturals}$ and $\EQPTM$ computations $M_1,\ldots,M_k$ such that $C_n[M_1,\ldots,M_k] = \Phi_n$ for all $n$.
\end{definition}

\newcommand{\InitState}{\ket{\mathsf{init}}}

We show that any $\EQPTC$ computation can be approximated to any desired precision by a polynomial-size quantum circuit. We first show the following claim. Let $\InitState \eqdef \ket{\TStateInit}_{\RegState} \ket{0,0}_{\RegHead,\RegInpHead} \ket{\varnothing}_{\RegTape}$.

\begin{claim}\label{lemma:truncation}
    \label{claim:qtm-truncation}
    Let $M$ be a $T$-bounded QTM running in expected time $t$, and let $U$ be the unitary dilation of $M$ as in \cref{fig:coherent-qtm}. For all $\gamma \colon \Naturals \to (0,1]$, there is a uniform sequence of unitary circuits $\{ V_{n} \}_n$ of size $\poly(n)/\gamma(n)^2$ such that for every unitary $A$ on $\RegInput$ and state $\ket{\psi} \in \RegInput$:
    \begin{equation*}
        \norm{(U^{\dagger} (\Id \otimes A) U - V_n^{\dagger} (\Id \otimes A) V_n) \ket{\psi} \InitState \ket{0^T}_{\RegB}} \leq \gamma(n).
    \end{equation*}
\end{claim}
\begin{proof}
Let $V$ be the unitary given by truncating $U$ to just after the $\tau$-th iteration of $U_{\delta}$, where $\tau \eqdef \lceil t/4\gamma^2 \rceil$. Let $\Pi \eqdef \ketbra{1^{T-\tau}}_{\RegB_{\tau+1},\ldots,\RegB_T}$. Observe that for every state $\ket{\psi}$,
\[ \Pi U_i \ket{\psi} \InitState \ket{0^T}_{\RegB} = \Pi_f V_i \ket{\psi} \InitState \ket{0^\tau 1^{T-\tau}}_{\RegB}, \]
because $\Pi$ projects on to computations that finish in at most $\tau$ steps, and once the computation finishes, the remaining $\mathsf{CNOT}_{\Pi_f}$ gates flip the corresponding $\RegB_i$ from $0$ to $1$.

Moreover, for every state $\ket{\phi} \in \RegA \otimes \RegState \otimes \RegW$ and $z \in \Bits^\tau$,
\[ U_i^{\dagger} \ket{\phi} \ket{z 1^{T-\tau}} = V_i^{\dagger} \ket{\phi} \ket{z 0^{T-\tau}}, \]
since the applications of $U_{\delta}$ controlled on $\RegB_{\tau+1},\ldots,\RegB_T$ act as the identity and the $\mathsf{CNOT}_{\Pi_f}$ gates acting on $\RegB_{\tau+1},\ldots,\RegB_T$ flip the corresponding register from $1$ to $0$. Hence
\[
    U^{\dagger}(I \otimes A) \Pi U \ket{\psi} \InitState \ket{0^T}_{\RegB} = V^{\dagger}(I \otimes A) \Pi_f V \ket{\psi}\InitState\ket{0^T}_{\RegB}.
\]

The claim follows since, by Markov's inequality,
\[ \norm{\Pi U_i \ket{\psi} \InitState \ket{0^T}_{\RegB}} \leq \sqrt{t/\tau} \leq \varepsilon(n)/2. \qedhere \]
\end{proof}

\begin{lemma}
	For any $\EQPTC$ computation $\{ \Phi_{n} \}_n$ and $\varepsilon \colon \Naturals \to (0,1]$, there is a uniform sequence of (standard) quantum circuits $\{ C_{n} \}_n$ of size $\poly(n,1/\varepsilon(n))$ such that $d(\Phi_{n}(\DMatrix),C_n(\DMatrix)) \leq \varepsilon(n)$ for all $\DMatrix$.
\end{lemma}
\begin{proof}
    Let $D_{n}$ be a $\CSymb$-circuit and $M_1,\ldots,M_k$ such that $\Phi_n = D_n[M_1,\ldots,M_k]$, and let $U_n$ be the unitary circuit obtained by replacing each $G_i,G_i^{-1}$ with the corresponding coherent implementation of $M_i$ as in \cref{fig:coherent-qtm}. Let $U'_n$ be as $U_n$, but where the $G_i$-gates are replaced with unitaries $V_i$ as guaranteed by \cref{claim:qtm-truncation}, with $\gamma(n) \eqdef \varepsilon(n)/k$. The circuit $C_n$ is obtained by initializing the ancillas to $\InitState \ket{0^\tau}_{\RegB}$, applying $U'_n$, and then tracing out the ancillas.

    We make use of the fidelity distance $d_F$, defined in \cite{STOC:Watrous06} to be
    \begin{equation*}
        d_F(\DMatrix,\DMatrixW) \eqdef \inf \{ \norm{\ket{\psi} - \ket{\phi}} : \text{$\ket{\psi},\ket{\phi}$ purify $\DMatrix,\DMatrixW$, respectively} \}.
    \end{equation*}
    \cite{STOC:Watrous06} shows that $d_F(\DMatrix,\DMatrixW) \geq d(\DMatrix,\DMatrixW)$. We can choose the purifications $U_n \ket{\psi} \InitState \ket{0^T}$ of $\Phi_n$ and $U'_n \ket{\psi} \ket{0}$ of $C_n$. By \cref{claim:qtm-truncation}, and the triangle inequality, the distance between these states is at most $\varepsilon(n)$.
\end{proof}

\subsection{Zero Knowledge with $\normalfont{\textsf{EQPT}}_c$ Simulation}
\newcommand{\out}{\mathsf{out}}
Given our definition of $\EQPTC$ above, we now formally define zero-knowledge with $\EQPTC$ simulation for interactive protocols. 

For an interactive protocol $(P, V)$, let $\out_{V^*}\langle P,V^* \rangle$ denote the output of $V^*$ after interacting with $P$.
\begin{definition}
    An interactive argument is \emph{black-box} statistical (resp. computational) post-quantum zero knowledge if there exists an $\EQPTC$ simulator $\Sim$ such that for all polynomial-size quantum malicious verifiers $V^*$ and all $(x,w) \in R_L$, the distributions
    \begin{equation*}
        \out_{V^*}\langle P(x,w),V^* \rangle \quad \text{and} \quad \Sim^{V^*}(x)
    \end{equation*}
    are statistically (resp. quantum computationally) indistinguishable.
\end{definition}

%% file: 10-state-preserving-ext.tex
\section{State-Preserving Extraction}\label{sec:state-preserving}

So far, we have constructed $\EQPTM$ \emph{guaranteed extractors} for various protocols of interest (\cref{sec:high-probability-extractor}) and established the $\EQPTC$ model that allows for \emph{state-preserving} extraction (\cref{sec:eqpt}). In this section, we prove a generalization of \cref{lemma:tech-overview-state-preserving-high-probability}, showing how to convert a $\EQPTM$ guaranteed extractor into a state-preserving $\EQPTC$ extractor. 

In \cref{sec:hp-to-sp}, we write down an explicit reduction from state-preserving extraction to guaranteed extraction and prove \cref{lemma:state-preserving-high-probability}, which gives a condition (\cref{def:witness-binding}) under which the reduction is valid (intuitively capturing ``computational uniqueness'' of the witness given the first message of the protocol). Then, in \cref{sec:state-preserving-examples}, we show examples to which \cref{lemma:state-preserving-high-probability} applies; namely, protocols for languages with unique (partial) witnesses and general commit-and-prove protocols. Finally, in \cref{sec:state-preserving-main-theorems}, we conclude \cref{thm:succinct-state-preserving,thm:state-preserving-wi}.

\subsection{From Guaranteed Extraction to State-Preserving Extraction}
\label{sec:hp-to-sp}

We first recall our definition of state-preserving proofs of knowledge (\cref{def:state-preserving-extraction}).

\begin{definition}
   An interactive protocol $\Pi$ is defined to be a \textdef{state-preserving argument (resp. proof) of knowledge} if there exists an extractor $\mathsf{Ext}^{(\cdot)}$ with the following properties:
   
   \begin{itemize}
        \item \textbf{Syntax}: For any quantum algorithm $P^*$ and auxiliary state $\ket{\psi}$, $\mathsf{Ext}^{P^*, \ket{\psi}}$ outputs a protocol transcript $\tau$, prover state $\ket{\psi'}$, and witness $w$. 
        \item \textbf{Extraction Efficiency}: If $P^*$ is a QPT algorithm, $E^{P^*, \ket{\psi}}$ runs in expected quantum polynomial time ($\EQPTC$).
        \item \textbf{Extraction Correctness}: the probability that $\tau$ is an accepting transcript but $w$ is an invalid $\mathsf{NP}$ witness is negligible.
        \item \textbf{State-Preserving}: the pair $(\tau, \ket{\psi'})$ is computationally (resp. statistically) indistinguishable from a transcript-state pair $(\tau^*, \ket{\psi^*})$ obtained through an honest one-time interaction with $P^*(\cdot, \ket{\psi})$ (where $\ket{\psi^*}$ is the prover's residual state). 
   \end{itemize}
\end{definition}

We now introduce the notion of ``witness-binding'' protocols, i.e., protocols that are collapse-binding to functions of the witness $w$. For an adversary $\mathsf{Adv}$ and an interactive protocols $(P,V)$ we define a witness-binding experiment $\mathsf{Exp}^{\mathsf{Adv}}_{\mathsf{wb}}(b,\mathsf{Pred},f,\lambda)$ parameterized by a challenge bit $b$, a predicate $\mathsf{Pred}$ and a function $f$.
\begin{enumerate}
    \item The challenger generates the first verifier message $\vk$ and sends it to $\Adv$; skip this step if the protocol is a 3-message protocol.
    \item $\Adv$ replies with a classical instance $x$, classical first prover message $a$, and a quantum state on registers $\RegW_{\mathrm{witness}} \otimes \RegY_{\mathrm{aux}}$.
    \item The challenger performs a binary-outcome projective measurement to learn the output of $\mathsf{Pred}(x,\vk,a,\cdot,\cdot)$ on $\RegW_{\mathrm{witness}}\otimes \RegY_{\mathrm{aux}}$. If the output is $0$, the experiment aborts.
    \item If $b = 0$, the challenger does nothing. If $b = 1$, the challenger initializes a fresh ancilla $\RegK$ to $\ket{0}_{\RegK}$, applies the unitary $U_f$ (acting on $\RegW_{\mathrm{witness}} \otimes \RegK$) that computes $f(\cdot)$ on $\RegW_{\mathrm{witness}}$ and XORs the output onto $\RegK$, measures $\RegK$, and then applies $U_f^\dagger$.
    \item The challenger returns the $\RegW_{\mathrm{witness}} \otimes \RegY_{\mathrm{aux}}$ registers to $\Adv$. Finally, $\Adv$ outputs a bit $b'$, which is the output of the experiment (if the experiment has not aborted).
\end{enumerate}
\begin{definition}[$(\mathsf{Pred},f)$-binding to the witness]\label{def:witness-binding}
A 3 or 4-message protocol is witness binding with respect to predicate $\mathsf{Pred}$ and function $f$ if for any computationally bounded quantum adversary $\mathsf{Adv}$,
\[
\Big| \Pr[\mathsf{Exp}^{\mathsf{Adv}}_{\mathsf{wb}}(0,\mathsf{Pred},f,\lambda) = 1]
- \Pr[\mathsf{Exp}^{\mathsf{Adv}}_{\mathsf{wb}}(1,\mathsf{Pred},f,\lambda) = 1] \Big|
\leq \negl(\secp).\] 
\end{definition}

Next, we write down a general-purpose reduction from state-preserving extraction to guaranteed extraction and show (\cref{lemma:state-preserving-high-probability}) that the reduction is valid under an appropriate witness-binding assumption.

\begin{lemma} \label{lemma:state-preserving-high-probability}
  Suppose that $(P_{\Sigma}, V_{\Sigma})$ is a post-quantum proof/argument of knowledge with guaranteed extraction. We optionally assume that the extractor $\Extract^{P^*}$ outputs some auxiliary information $y$ in addition to the witness $w$. We then make the following additional assumptions with respect to a predicate $\mathsf{Pred}$:
  
  \begin{itemize}
      \item The protocol $(P_{\Sigma}, V_{\Sigma})$ is $(\mathsf{Pred}, f = \mathsf{Id})$-witness binding, and
      \item The tuple $(w,y)$ output by the guaranteed extractor $\Extract^{P^*}$ satisfies $\mathsf{Pred}(\vk, x, a, w, y) = 1$ with $1-\negl$ probability. 
  \end{itemize}
  Then, $(P_{\Sigma}, V_{\Sigma})$ is a state-preserving proof/argument of knowledge with $\EQPTC$ extraction.
\end{lemma}
\begin{remark}
This lemma is stated with respect to $f = \mathsf{Id}$ to match the state-preserving proof of knowledge abstraction; however, we also consider (\cref{cor:state-preserving-partial-witness}) versions of this reduction where $f\neq \mathsf{Id}$. 
\end{remark}

\begin{proof}

We want to show that $(P_{\Sigma},V_{\Sigma})$ is a state-preserving proof/argument of knowledge. We begin by describing our candidate state-preserving extractor $\overline{\Extract}^{P^*}$. 

\begin{construction}\label{construction:state-preserving-reduction}
Let $\Extract^{P^*}$ be a post-quantum guaranteed extractor (\cref{def:high-probability-extraction-body}). We present an $\EQPTC$ extractor $\overline{\Extract}^{P^*}$ that has the form of an $\EQPTC$ computation (see~\cref{fig:cr-eqpt-simple}) where the unitary $U$ is a coherent implementation of the following $\EQPTM$ computation on input register $\RegH \otimes \RegR \otimes \RegS$:
\begin{enumerate}
    \item Measure $\RegR \otimes \RegS$ with the projective measurement 
    \[\BMeas{\ketbra{+_R}_{\RegR} \otimes \ketbra{0}_{\RegS}}.\]
    If the output is $0$, abort. 
    \item If the output is $1$, we are guaranteed that $\RegR \otimes \RegS$ is $\ket{+_R}_{\RegR} \otimes \ket{0}_{\RegS}$. Run $\Extract^{P^*}$ on prover state $\RegH$ using $\RegR$ as the superposition of challenges (in~\cref{step:ge-run-coherently} of~\cref{def:high-probability-extraction-body}). We assume that the randomness $\Extract^{P^*}$ uses to sample a classical random $\vk$ is generated by applying a Hadamard to a subregister of $\RegS$.
    
    Write everything that is measured/obtained during the execution of $\Extract^{P^*}$ onto subregisters of $\RegS$. This includes the instance $x$, the first two messages of the 4-message protocol $(\vk,a)$, the bit $b$ indicating the verifier's decision (i.e., whether the prover succeeds when run on the uniform superposition of challenges), and the extracted output $(w,y)$ (if $b = 1$, $w$ is a valid witness for $x$ and $\mathsf{Pred}(x, \vk, a, w, y)=1$ with $1-\negl(\secp)$ probability).
\end{enumerate}
The fact that the above computation is in $\EQPTM$ follows from the fact that $\Extract^{P^*}$ is $\EQPTM$. Let $U$ denote its coherent implementation (as in~\cref{sec:eqpt}, $U$ is a unitary on $\RegH \otimes \RegR \otimes \RegS$ and an exponential-size ancilla register).

Our state-preserving $\EQPTC$ extractor $\overline{\Extract}^{P^*}$ takes as input a prover state on $\RegH$ and does the following.

$\overline{\Extract}^{P^*}:$
\begin{enumerate}[noitemsep]
    \item Initialize additional registers $\RegR \otimes \RegS$ to $\ket{+_R}_{\RegR} \ket{0}_{\RegS}$.
    \item Apply $U$.
    \item\label[step]{step:measure-extracted-w} Measure the subregister of $\RegS$ containing $(x,\vk,a,b,w)$ where $w = 0$ is interpreted as $\bot$. Note that $\RegS$ contains a subregister corresponding to $y$, but $y$ is not measured here.
    \item\label[step]{step:apply-U-dagger} Apply $U^\dagger$.
    \item\label[step]{step:run-P-again} Run the prover $P^*$ on first message $\vk$ to obtain $x,a$ (again). Then run $P^*$ on challenge $\RegR$. Measure $\RegR$ to obtain $r$, and measure the register of $\RegH$ corresponding to its output to obtain $z$. Output $(x,\vk,a,r,z,w)$ and $\RegH$.
\end{enumerate}

\end{construction}

First, we note that the above procedure is $\EQPTC$ by construction. To prove the extraction correctness guarantee, it suffices to show that when $b = 1$, the witness $w$ is valid with $1-\negl(\secp)$ probability, and that when $b = 0$, the extractor outputs a rejecting transcript. The former statement follows immediately from the assumption that $\Extract^{P^*}$ is a guaranteed extractor. For the latter, observe (using the definition of $\Extract^{P^*}$ and the fact that $U$ is a coherent implementation of $\Extract^{P^*}$) that when $b = 0$, the state on $\RegH \otimes \RegR$ after running $P^*$ to obtain $a$ in~\cref{step:run-P-again} corresponds to a rejecting execution, so the transcript measured in~\cref{step:run-P-again} will be rejecting.

It remains to argue that the state-preserving extractor satisfies the indistinguishability property. Observe that $\overline{\Extract}^{P^*}$ can be rewritten so that $\vk,x,a,b$ are no longer obtained by running $\Extract^{P^*}$ coherently as $U$ and then measuring those values afterwards, but instead by running those steps accroding to the standard $\EQPTM$ implementation of $\Extract^{P^*}$. Thus the only part of $\Extract^{P^*}$ that is written as a coherent implementation of a variable runtime procedure is the $\FindWitness^{P^*}$ subroutine; let $U_{\mathsf{FW}}$ denote the coherent implementation of $\FindWitness^{P^*}$. Note that while $\FindWitness^{P^*}$ is technically not $\EQPTM$ on its own (i.e., there exist inputs that could make it run for too long), the fact that $\Extract^{P^*}$ is $\EQPTM$ ensures that $U_{\mathsf{FW}}$ is only applied on inputs where it runs for expected polynomial time.

Given the above definitions, the output of $\overline{\Extract}^{P^*}$ is perfectly equivalent to the following:
\begin{enumerate}[noitemsep]
    \item Sample a random $\vk$, and run the prover $P^*$ to obtain $x,a$.
    \item Initialize $\RegR$ to $\ket{+_R}_{\RegR}$ and measure $\sC$ (this is the binary projective measurement on $\RegH \otimes \RegR$ defined in \cref{sec:ge-notation} that measures whether the verifier accepts when the prover with state $\RegH$ is run on the challenge $\RegR$).
    \item If $\sC =1$, apply $U_{\mathsf{FW}}$. Otherwise if $\sC = 0$, set $w = \bot$ and skip to~\cref{step:measure-r-z}.
    \item Measure the subregister corresponding to the part of the output of $U_{\mathsf{FW}}$ containing $w$. Note that there is also a subregister corresponding to $y$, but $y$ is not measured.
    \item Apply $U_{\mathsf{FW}}^\dagger$.
    \item\label[step]{step:measure-r-z} Measure $\RegR$ to obtain $r$ and run the prover $P^*$ on $r$ to obtain its response $z$.
    \item\label[step]{step:extractor-output-w} Output $(x,\vk,a,r,z,w)$ and $\RegH$.
\end{enumerate}

Let $\mathsf{Hybrid}_0$ be identical to $\overline{\Extract}^{P^*}$ except that~\cref{step:extractor-output-w} is modified to output $(x,\vk,a,r,z)$ and $\RegH$ (i.e., omitting $w$). To show computational indistinguishability, it suffices to show that the output of $\mathsf{Hybrid}_0$ is computational indistinguishable from $\mathsf{Hybrid}_1$ defined as follows:

\begin{enumerate}[noitemsep]
    \item Sample a random $\vk$, and run the prover $P^*$ to obtain $x,a$.
    \item Initialize $\RegR$ to $\ket{+_R}_{\RegR}$ and measure $\sC$ (this is the binary projective measurement on $\RegH \otimes \RegR$ defined in \cref{sec:ge-notation} that measures whether the verifier accepts when the prover with state $\RegH$ is run on the challenge $\RegR$).
    \item Measure $\RegR$ to obtain $r$ and run the prover $P^*$ on $r$ to obtain its response $z$.
    \item Output $(x,\vk,a,r,z)$ and $\RegH$.
\end{enumerate}

$\mathsf{Hybrid}_1$ corresponds to an honest execution of $P^*$ since the measurement of $\sC$ commutes with the measurement of $\RegR$. 

By assumption, in $\mathsf{Hybrid}_0$, the reduced density $\dm_{\RegS}$ of $\RegS$ satisfies $\Tr(\Pi_{\mathrm{Valid}} \dm_{\RegS}) = 1-\negl(\secp)$, where $\Pi_{\mathrm{Valid}}$ checks that either $b=0$ or (1) $w$ is a valid witness for $x$ and (2) $\mathsf{Pred}(\vk, x, a, w,y) = 1$. Therefore, the indistinguishability of $\mathsf{Hybrid}_0$ and $\mathsf{Hybrid}_1$ should intuitively follow from the witness-binding property, since if the measurement of $w$ is skipped, then $U_{\mathsf{FW}}$ cancels out with $U_{\mathsf{FW}}^\dagger$. However, to appeal to the guarantee that measuring $w$ is undetectable, we need to ensure that $U_{\mathsf{FW}}$ corresponds to an efficient operation.

We handle this by considering a fixed polynomial-time truncation of $U_{\mathsf{FW}}$. Suppose that a distinguisher can distinguish $\mathsf{Hybrid}_0$ from $\mathsf{Hybrid}_1$ with non-negligible advantage $\varepsilon(\secp)$. Then we can modify $\mathsf{Hybrid}_0$ to use $U_{\mathsf{FW},\varepsilon}$, a coherent implementation of a strict $\poly(\secp,1/\varepsilon)$-runtime algorithm that approximates $\FindWitness^{P^*}$ to precision $\varepsilon/2$. Now the same distinguisher must distinguish between $\mathsf{Hybrid}_{0,\varepsilon}$ and $\mathsf{Hybrid}_{1}$ with advantage $\varepsilon/2$, where $\mathsf{Hybrid}_{0,\varepsilon}$ is the following:

\begin{enumerate}[noitemsep]
    \item Sample a random $\vk$, and run the prover $P^*$ to obtain $x,a$.
    \item Initialize $\RegR$ to $\ket{+_R}_{\RegR}$ and measure $\sC$.
    \item If $\sC =1$, apply $U_{\mathsf{FW},\varepsilon}$. Otherwise if $\sC = 0$, set $w = \bot$ and skip to~\cref{step:measure-r-z-eps}.
    \item Measure a subregister of the output register of $U_{\mathsf{FW},\varepsilon}$ to obtain $w$.
    \item Apply $U_{\mathsf{FW},\varepsilon}^\dagger$.
    \item\label[step]{step:measure-r-z-eps} Measure $\RegR$ to obtain $r$ and run the prover $P^*$ on $r$ to obtain its response $z$.
    \item Output $(x,\vk,a,r,z)$ and $\RegH$.
\end{enumerate}

Since $\varepsilon(\lambda)$ is at least $1/\lambda^c$ for some constant $c$ for infinitely many $\lambda$, it follows that $U_{\mathsf{FW},\varepsilon}$ and $U_{\mathsf{FW},\varepsilon}^\dagger$ are $\poly(\secp)$-runtime algorithms for infinitely many $\lambda$. Then a distinguisher that distinguishes between $\mathsf{Hybrid}_{0,\varepsilon}$ and $\mathsf{Hybrid}_1$ contradicts the witness-binding property of $(P, V)$. \qedhere
 
\end{proof}

\subsection{Applying \cref{lemma:state-preserving-high-probability}}\label{sec:state-preserving-examples}

We now show that the witness-binding hypotheses in \cref{lemma:state-preserving-high-probability} are satisfied in two cases of interest: protocols for unique-witness (or partial witness) languages (\cref{cor:state-preserving-unique-witness}), and commit-and-prove protocols (\cref{cor:commit-and-prove}).

\begin{corollary}\label{cor:state-preserving-unique-witness}
   Let $L\in \mathsf{UP}$ be a language with unique $\mathsf{NP}$ witnesses. Then, if $L$ has a post-quantum proof of knowledge with guaranteed extraction, it also has a post-quantum state-preserving proof of knowledge.
\end{corollary}

\begin{proof}
  This follows immediately from the fact that any protocol for a $\mathsf{UP}$ language is $(\mathsf{Pred}, f)$-witness binding for $\mathsf{Pred} = 1$ (the trivial predicate) and $f = \mathsf{Id}$ (because there is a unique valid witness). Since $\mathsf{Pred} = 1$, any guaranteed extractor also satisfies the $\mathsf{Pred}$-hypothesis of \cref{lemma:state-preserving-high-probability}, so we are done. \qedhere
\end{proof}

We briefly state how \cref{cor:state-preserving-unique-witness} can be extended to languages $L$ with unique \emph{partial} witnesses, provided that the extractor only measures a function $f(w)$ that is a deterministic function of the instance $x$.

\begin{corollary}\label{cor:state-preserving-partial-witness}
  Let $L\in \mathsf{NP}$, and let $f$ be an efficient function such that for all instances $x\in L$ and all witnesses $w\in R_x$, $f(x,w) = g(x)$ is equal to some fixed (possibly inefficient) function of $x$. 
  
  Suppose that $L$ has a proof/argument of knowledge $(P_{\Sigma}, V_{\Sigma})$ with guaranteed extraction. Then, a modified variant of $\overline{\Extract}$ (\cref{construction:state-preserving-reduction}), in which only $f(x,w)$ is measured instead of $w$, is a state-preserving proof/argument of knowledge extractor for $(P_{\Sigma}, V_{\Sigma})$ that outputs $g(x)$. \qedhere
\end{corollary}

This holds by the same reasoning as \cref{cor:state-preserving-unique-witness}: the hypothesis of \cref{cor:state-preserving-partial-witness} implies that any protocol for $L$ is $(\mathsf{Pred} = 1, f)$-witness binding, and so the reduction from \cref{lemma:state-preserving-high-probability} applies (when $f(x,w)$ is measured rather than $w$). \qedhere

\subsubsection{Commit-and-Prove Protocols}

Let $(P_\Sigma, V_\Sigma)$ denote a post-quantum proof/argument of knowledge with guaranteed extraction (\cref{def:high-probability-extraction-body}). Recall that \cref{def:high-probability-extraction-body} has been designed to capture (first-message) adaptive soundness, in which the prover $P^*$ can adaptively choose the instance $x$ as it sends its first message.

Then, we consider a \emph{commit-and-prove} compiled protocol $(P_{\Com}, V_{\Com})$ using $(P_\Sigma, V_\Sigma)$ and a commitment scheme $\Com$. $(P_{\Com}, V_{\Com})$ is executed as follows:

\begin{itemize}
    \item $V_{\Com}$ sends a first message for $(P_{\Sigma}, V_{\Sigma})$ (if the protocol has four messages). Moreover, if $\Com$ is a two-message commitment scheme, $V_{\Com}$ sends a commitment key $\ck$. 
    \item $P_{\Com}$ then sends:
    \begin{itemize}
        \item A commitment $\com = \Com(\ck, w)$ to a witness $w$ for the underlying language $L$, and
        \item A first prover message for an execution of $(P_{\Sigma}, V_{\Sigma})$ for the statement ``$\exists w, r$ such that $\com = \Com(\ck, w; r)$ and $w$ is an $\mathsf{NP}$-witness for $x\in L$. 
    \end{itemize}
    \item $P_{\Com}$ and $V_{\Com}$ then complete the execution of $(P_{\Sigma}, V_{\Sigma})$. 
\end{itemize}

\begin{corollary}\label{cor:commit-and-prove}
 If $(P_{\Sigma}, V_{\Sigma})$ is a post-quantum proof/argument of knowledge with guaranteed extraction for all $\mathsf{NP}$ languages and $\Com$ is a collapse-binding commitment scheme, then the commit-and-prove compiled protocol is a state-preserving proof/argument of knowledge. 
\end{corollary}

\begin{proof}

We first remark that since $(P_{\Sigma}, V_{\Sigma})$ is a post-quantum proof/argument of knowledge with guaranteed extraction, the commit-and-prove composed protocol is also immediately a post-quantum proof/argument of knowledge with guaranteed extraction. Namely, $\Extract^{P^*}$ interprets the cheating prover as an adaptive-input cheating prover for $(P_{\Sigma}, V_{\Sigma})$ with respect to the language
\[ L_{\ck, \com} = \big\{ (w, \omega): w \in R_x \text{ and } \Com(\ck, w; \omega) = \com\big\}
\]
and runs the guaranteed extractor for $(P_{\Sigma}, V_{\Sigma})$. Morevoer, this extraction procedure outputs \emph{both} an $\mathsf{NP}$-witness $w$ \emph{and} commitment randomness $\omega$ such that $\com = \Com(\ck, w; \omega)$; we treat $\omega$ as auxiliary information $y$.

We then define $\mathsf{Pred}(x,(\ck, \vk), (\com, a), w, \omega)$ to output $1$ if and only if $\Com(\ck, w; \omega) = \com$. Then, we observe that the commit-and-prove protocol is $(\mathsf{Pred}, \mathsf{Id})$-witness binding (for the language $L$) by the collapse-binding of the commitment scheme $\Com$. Moreover, the correctness property of $\Extract^{P^*}$ further guarantees that $\mathsf{Pred}(x, (\ck, \vk), (\com, a), w, \omega) = 1$ with probability $1-\negl(\secp)$.

Thus, we conclude that \cref{lemma:state-preserving-high-probability} applies, and so the commit-and-prove protocol has a state-preserving extractor. \qedhere

\end{proof}

\subsection{Concluding \cref{thm:succinct-state-preserving,thm:state-preserving-wi}}\label{sec:state-preserving-main-theorems}

Finally, we describe how to conclude the results of \cref{thm:succinct-state-preserving,thm:state-preserving-wi}. We begin with \cref{thm:succinct-state-preserving}, re-stated below.

\begin{theorem}[\cref{thm:succinct-state-preserving}]
Assuming collapsing hash functions exist, there exists a 4-message public-coin state-preserving succinct argument of knowledge for $\mathsf{NP}$.
\end{theorem}

\begin{proof}
  Given a collapsing hash function family $\mathsf{H}$, we construct a state-preserving succinct argument of knowledge for $\mathsf{NP}$ as follows:
  
  \begin{itemize}
      \item First, we define Kilian's succinct argument system (see \cref{sec:kilian}) with respect to $\mathsf H$. By \cref{corollary:kilian-guaranteed}, this argument system is a post-quantum argument of knowledge with guaranteed extraction.
      \item Next, we apply the commit-and-prove compiler (\cref{cor:commit-and-prove}) using the collapse-binding commitment scheme $\Com(\ck = h, m) = h(m)$. This commitment scheme does not formally satisfy any hiding property, but it is \emph{succinct}, which is what is relevant for \cref{thm:succinct-state-preserving}.
  \end{itemize}
  
  \cref{cor:commit-and-prove} tells us that the resulting composed protocol is a state-preserving argument of knowledge for $\mathsf{NP}$. Moreover, it satisfies all of the properties (4-message, public-coin, succinct) claimed in the theorem statement. \qedhere
\end{proof}

Next, we prove \cref{thm:state-preserving-wi}, re-stated below.

\begin{theorem}\label{thm:state-preserving-wi-body}
Assuming collapsing hash functions or \emph{super-polynomially secure} one-way functions, there exists a 4-message public-coin state-preserving witness-indistinguishable argument (in the case of collapsing) or proof (in the case of OWFs) of knowledge. Assuming \emph{super-polynomially secure} non-interactive commitments, there exists a 3-message PoK achieving the same properties.
\end{theorem}

\begin{proof}
All three variants of this theorem are proved via the same approach: combining commit-and-prove with a (strong) witness-indistinguishable $\Sigma$-protocol.

Formally, let $\Com$ denote a (possibly keyed) non-interactive commitment scheme. We use $\Com$ to instantiate a commit-and-open $\Sigma$-protocol (\cref{def:commit-and-open}) such as the \cite{C:GolMicWig86} protocol for graph $3$-coloring or the (potentially modified) \cite{Blum86} protocol for Hamiltonicity. We do a sufficient parallel repetition of the commit-and-open protocol so that its challenge space satisfies $|R| = 2^t$ for $t\leq \poly(\secp)$\footnote{Using \cite{Blum86}, one can set $t = \poly(\log \secp)$.} and it achieves $\negl(\secp)$ soundness error. Then, \cref{cor:commit-and-open} tells us that this protocol is a post-quantum proof/argument of knowledge (depending on whether $\Com$ is statistically or collapse-binding) with guaranteed extraction.

Next, we additionally assume (as is the case for \cite{C:GolMicWig86,Blum86}) that the $\Sigma$-protocol satisfies special honest-verifier zero knowledge (\cref{def:shvzk}). In fact, we assume that it satisfies SHVZK against quantum adversaries that run in time $2^t \cdot \poly(\secp)$, which holds (for these examples) provided that $\Com$ is computationally hiding against $2^t\cdot \poly(\secp)$-time adversaries.

Under this assumption, Watrous' rewinding lemma \cite{STOC:Watrous06} implies that the $\Sigma$-protocol has a time $2^t\cdot\poly(\secp)$ malicious verifier post-quantum simulator.

We now plug this $\Sigma$-protocol into the commit-and-prove compiler (\cref{cor:commit-and-prove}), again making use of the commitment scheme $\Com$ (for simplicity of the proof, we assume here that a different commitment key is used, although this is not necessary). \cref{cor:commit-and-prove} tells us that the resulting protocol is a state-preserving proof/argument of knowledge (again depending on whether $\Com$ is statistically binding).

It remains to show WI of the commit-and-prove protocol. That is, we want to show that for every malicious verifier $V^*$ (and maliciously chosen commitment key $\ck$), a commitment $\com = \Com(\ck, w_1)$ and the view of $V^*$ in an execution of the $\Sigma$-protocol is computationally indistinguishable from the analogous state when a second witness $w_2$ is instead used. This is argued via the usual hybrid argument:

\begin{itemize}
    \item Define $\mathsf{Hybrid}_{0,b}$ to be $\Com(\ck, w_b)$ along with the actual $\Sigma$-protocol view of $V^*$.
    \item Define $\mathsf{Hybrid}_{1,b}$ to consist of $\com = \Com(\ck, w_b)$ along with a $2^t\cdot \poly(\secp)$-time \emph{simulated} view of $V^*$ on input $(\ck, \com)$. We have that $\mathsf{Hybrid}_{1,b}\approx_c \mathsf{Hybrid}_{0,b}$ by the super-polynomial time simulatability of the $\Sigma$-protocol (as discussed above).
    \item Finally, we have that $\mathsf{Hybrid}_{1,0}\approx_c \mathsf{Hybrid}_{1,1}$ by the (already assumed) $2^t\cdot \poly(\secp)$-hiding of $\Com$.
\end{itemize}

\noindent To conclude the theorem statement, it suffices to instantiate $\Com$ in three ways: 

\begin{itemize}
    \item Assuming $2^t\cdot \poly(\secp)$-secure non-interactive commitments (e.g. \cite{C:BarOngVad03,TCC:GHKW17,ePrint:LomSch19}), one obtains the claimed $3$-message protocol.
    \item Assuming $2^t\cdot \poly(\secp)$-secure one-way functions, one obtains the OWF-based $4$-message protocol.
    \item Assuming polynomially-secure collapsing hash functions, one obtains the collapsing-based $4$-message protocol by defining $\Com(h, m; r, s) = (h(r), s, \langle r, s\rangle \oplus m)$. This commitment scheme is \emph{statistically} hiding (i.e. hiding against unbounded adversaries), and so WI of the commit-and-prove protocol holds unconditionally, while the AoK property relies on collapsing. 
\end{itemize}

\noindent This completes the proof of \cref{thm:state-preserving-wi}.\qedhere 
\end{proof}

%% file: 11-gni.tex
\newcommand{\GIComm}{\mathsf{GIComm}}
\newcommand{\IDC}{\mathsf{IDC}}
\newcommand{\PoK}{\mathsf{PoK}}

\section{The~\cite{FOCS:GolMicWig86} GNI Protocol is Post-Quantum Zero Knowledge}

In this section, we show that our state-preserving extraction results imply the post-quantum ZK of the graph non-isomorphism protocol, proving \cref{thm:szk}. We begin by giving a description of the GNI protocol in \cref{fig:gni}. Our description achieves soundness error $1/2$ (as does the original~\cite{FOCS:GolMicWig86}), but can be extended to the negligible soundness case (without increasing the number of rounds) with essentially the same proof of (post-quantum) ZK.

\begin{figure}[!ht]
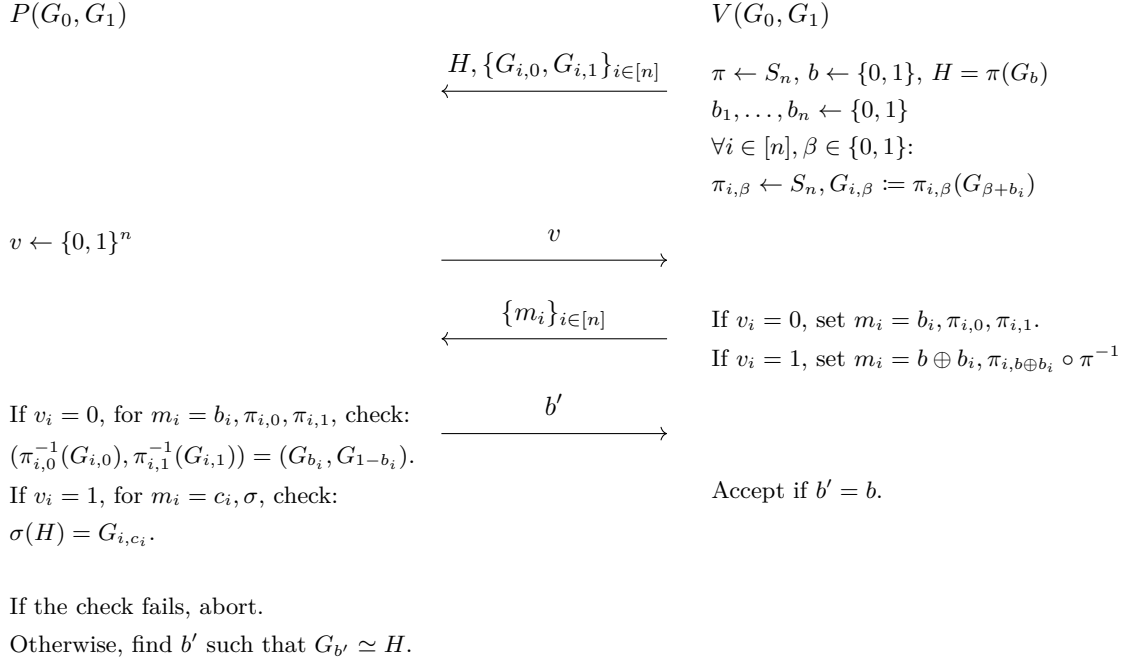

\centering
	\procedure[]{}{
		P (G_0,G_1) \> \> V(G_0,G_1) \\ 
		 \> \sendmessageleft*[3cm]{H,\{G_{i,0},G_{i,1}\}_{i \in [n]}} \> \begin{subprocedure}
			\pseudocode[mode=text]{$\pi \gets S_n$, $b \gets \{0,1\}$, $H = \pi(G_b)$ \\ $b_1, \hdots, b_n \gets \{0,1\}$ \\ 
			$\forall i \in [n], \beta \in \{0,1\}$:\\ $\pi_{i,\beta} \gets S_n, G_{i,\beta} \coloneqq \pi_{i,\beta}(G_{\beta + b_i})$}
		\end{subprocedure}\\ 
		\begin{subprocedure}
			\pseudocode[mode=text]{$v \gets \{0,1\}^n$}
		\end{subprocedure}
		\> \sendmessageright*[3cm]{v} \>  \\
		\> \sendmessageleft*[3cm]{\{m_i\}_{i \in [n]}} \> \begin{subprocedure}
			\pseudocode[mode=text]{If $v_i = 0$, set $m_i = b_i, \pi_{i,0},\pi_{i,1}$.\\
			If $v_i = 1$, set $m_i = b\oplus b_i, \pi_{i,b \oplus b_i} \circ \pi^{-1}$}
		\end{subprocedure} \\
		\begin{subprocedure}
			\pseudocode[mode=text]{If $v_i = 0$, for $m_i = b_i, \pi_{i,0},\pi_{i,1}$, check: \\ $(\pi_{i, 0}^{-1}(G_{i,0}), \pi_{i, 1}^{-1}(G_{i,1})) = ( G_{b_i}, G_{1-b_i})$.\\
			If $v_i = 1$, for $m_i = c_i, \sigma$, check:\\
			$\sigma(H) = G_{i, c_i}$.
			\\ 
			\\ If the check fails, abort. 
			\\ Otherwise, find $b'$ such that $G_{b'} \simeq H$.
			}
		\end{subprocedure}
			\> \sendmessageright*[3cm]{b'} \>  \begin{subprocedure}
			\pseudocode[mode=text]{\\ \\ Accept if $b' = b$.}
			\end{subprocedure}
	}
	\caption{The Zero Knowledge Proof System for Graph Non-Isomorphism.}
	\label{fig:gni}
\end{figure}

Next, we give a slightly more abstract description of the protocol using instance-dependent commitments \cite{BelMicOst90,JC:ItoOhtShi97,C:MicVad03}. 

\begin{construction}
    Fix a language $L$, let $\IDC$ be a non-interactive instance-dependent commitment\footnote{That is, when $x\in L$, a commitment $\Com(x, m)$ statistically hides the message $m$. When $x\not\in L$, a commitment $\Com(x, m)$ statistically binds the committer to $m$.} \cite{BelMicOst90,JC:ItoOhtShi97,C:MicVad03} for $L$, and let $\PoK$ be a statistically witness-indistinguishable proof of knowledge of the committed bit for $\IDC$. Then, we define the following interactive proof system for the complement language $\overline{L}$.
    \begin{enumerate}
        \item The verifier commits to a bit $b \in \{0,1\}$ using $\IDC$ and sends it to the prover.
        \item The prover and verifier engage in $\PoK$ where the verifier proves knowledge of $b$.
        \item If the prover accepts in $\PoK$, then it sends $b'$ as determined by the verifier's commitment.
        \item The verifier accepts if $b' = b$.
    \end{enumerate}
\end{construction}

\cite{FOCS:GolMicWig86} instantiates this framework for the language $L$ consisting of pairs of isomorphic graphs (and so $\overline{L}$ consists of pairs of non-isomorphic graphs, up to well-formedness of the string $x$). 

Let $\GIComm$ be the following instance-dependent commitment scheme: $\GIComm((G_0,G_1),b;\pi) = \pi(G_b):= H$. Observe that if $G_0,G_1$ are isomorphic then this commitment is perfectly hiding, and if they are not then it is perfectly binding. Moreover, this commitment scheme admits a proof of knowledge of the committed bit as follows.

\begin{enumerate}
    \item The prover chooses $b_1,\ldots,b_{\secp} \in \Bits$ uniformly at random and sends commitments $C_{i,0} \eqdef \GIComm((G_0,G_1),b_i;\sigma_{i,0})$ and $C_{i,1} \eqdef \GIComm((G_0,G_1),b_i \oplus 1;\sigma_{i,1})$.
    \item The verifier sends a random string $v \in \Bits^{\secp}$.
    \item The prover sends $b_i$ and opens $C_{i,0},C_{i,1}$ for all $i$ such that $v_i = 0$. The prover sends $c_i \eqdef b \oplus b_i$, $\tau_i \eqdef \sigma_{i,c_i} \circ \pi^{-1}$ for all $i$ such that $v_i = 1$.
    \item The verifier accepts if the received openings are valid when $v_i = 0$, and $C_{i,c_i} = \tau_i(H)$ when $v_i = 1$, where $H$ is the commitment graph.
\end{enumerate}

Classically, we obtain $b$ by rewinding to find two accepting transcripts with $v_i \neq v_i'$; then $b = c_i \oplus b_i$.

\begin{lemma}\label{lemma:gi-state-preserving}
    If $(G_0,G_1)$ are non-isomorphic then the above protocol is a statistically state-preserving proof of knowledge of the committed message for $\GIComm$.
\end{lemma}
\begin{proof}
    We have already shown (\cref{cor:gni-guaranteed-extractor}) that this protocol has a guaranteed extractor, because when $G_0$ and $G_1$ are not isomorphic, this protocol is collapsing onto the $b_i$ part of $0$-challenge responses (as $b_i$ is fixed by the commitments $C_{i,0}, C_{i,1}$) and the protocol is $(2, g)$-probabilistically special sound (where $g$ checks (for the first challenge-partial response pair) the correctness of the $0$ challenge response bits $b_i$ for $v_i = 0$ and (for the second challenge-partial response pair) the correctness of all $b_i$ ($v_i = 0$) and $c_i$ ($v_i = 1$)). 
    
    Moreover, the language $L_{G_0, G_1} = \{ H: \exists (b, \pi) \text{ such that } \pi G_b \simeq H\}$ has partial unique witnesses: for any $H\in L_{G_0, G_1}$, the bit $b$ is uniquely determined (given that $G_0$ and $G_1$ are not isomorphic). Thus, the state-preserving reduction of \cref{lemma:state-preserving-high-probability} applies (see \cref{cor:state-preserving-partial-witness}), so this protocol has a state-preserving extractor.
\end{proof}

Finally, we note that \cref{lemma:gi-state-preserving} immediately implies that the GNI protocol is post-quantum (statistical) zero knowledge. We assume without loss of generality that the cheating verifier $V^*$ has a ``classical'' first message by replacing $V^*$ (with auxiliary state $\ket{\psi}$) with $(V^*, \dm)$ for the mixed state $\dm$ obtained by running $U_{V^*}$ on $\ket{\psi}$ to generate a first message, measuring it, and running $U_{V^*}^\dagger$. 

The simulator is then described as follows:

\begin{itemize}
    \item Given cheating verifier $V^*$ with classical first message $(\mathsf{com}, \mathsf{pok}_1)$, run the state-preserving $\mathsf{PoK}$ extractor on $V^*$ (which now acts as a $\mathsf{PoK}$ cheating prover).
    \item If the transcript generated by the state-preserving extractor is accepting, then output the bit $b$ in the ``partial witness'' slot of the extractor's output. Otherwise, send an aborting message.
\end{itemize}

The (statistical) zero knowledge property of this simulator follows immediately from the state-preserving property of the extractor. Moreover, the simulator inherits the $\EQPTC$ structure directly from the extractor (with additional fixed polynomial-time pre- and post-processing). This completes the proof of \cref{thm:szk}.

%% file: 12-feige-shamir.tex
\section{The~\cite{STOC:FeiSha90} Protocol is Post-Quantum Zero Knowledge}\label{sec:feige-shamir}
We recall the Feige-Shamir 4-message zero knowledge argument system for $\mathsf{NP}$. This protocol uses three primitives as building blocks:

\begin{itemize}
    \item A non-interactive commitment scheme $\mathsf{Com}$.
    \item The 3-message WI argument of knowledge $\mathsf{AoK}$  constructed in \cref{sec:state-preserving-main-theorems}. We note that $\mathsf{AoK}$ is public-coin. 
    \item A 3-message delayed-witness WI argument of knowledge $\mathsf{dAoK}$. 
\end{itemize}

We will argue security using the particular instantiations of $\mathsf{AoK},\mathsf{dAoK}$ due to subtleties arising from the concurrent composition. Unlike $\mathsf{AoK}$, we do \emph{not} require that $\mathsf{dAoK}$ is state-preserving.

\noindent The protocol is executed as follows. 

\begin{itemize}
    \item The verifier sends the following strings as its first message:
    \begin{itemize}
        \item Two commitments $c_0, c_1$ generated as $c_i = \Com(0; r_i)$ for i.i.d. random strings $r_i$. For the post-quantum variant, following \cite{EC:Unruh12,EC:Unruh16}, we additionally include commitments $c'_i = \Com(r_i; \rho_i)$ to the two random strings $r_0, r_1$. 
        \item A first (prover) message of $\mathsf{AoK}$ corresponding to the statement ``$\exists i, r_i, \rho_i$ such that $c_i = \Com(0; r_i)$ and $c'_i = \Com(r_i; \rho_i)$.'' By default, the verifier uses $(b, r_b, \rho_b)$ as its witness for a randomly chosen bit $b$. 
    \end{itemize}
    \item The prover sends two strings as its first message:
    \begin{itemize}
        \item A second (verifier) message of $\mathsf{AoK}$ (which is a uniformly random string).
        \item A first (prover) message of $\mathsf{dAoK}$ corresponding to the statement ``$x\in L$ or $\exists i, r_i, \rho_i$ such that $c_i = \Com(0; r_i)$ and $c'_i = \Com(r_i; \rho_i)$.'' No witness is required.
    \end{itemize}
    \item The verifier sends two strings as its second message:
    \begin{itemize}
        \item A third (prover) message of $\mathsf{AoK}$, computed using $(b, r_b, \rho_b)$.
        \item A second (verifier) message of $\mathsf{dAoK}$ (which is a uniformly random string).
    \end{itemize}
    \item Finally, the prover sends the third message of $\mathsf{dAoK}$.  The prover uses a witness $w$ for $x\in L$ to generate this message. 
\end{itemize}

\subsection{Building Block: Delayed-Witness Proofs of Knowledge}

In order to instantiate the Feige-Shamir protocol, we need a post-quantum instantiation of $\mathsf{dAoK}$. In particular, we need:

\begin{lemma}\label{lemma:delayed-witness-wi}
   Assume that post-quantum non-interactive commitments exist. Then, there exists a delayed-witness $\Sigma$-protocol for $\mathsf{NP}$ that is witness indistinguishable against quantum verifiers and is a post-quantum proof of knowledge with negligible knowledge error.
\end{lemma}

\cref{lemma:delayed-witness-wi} does not immediately follow from extraction techniques such as~\cite[Lemma 7]{EC:Unruh12} or \cite{FOCS:CMSZ21} because the canonical delayed-witness $\Sigma$-protocol \cite{C:LapSha90} is not collapsing, and these works only give results for collapsing protocols. Nonetheless, we show that (similar to the one-out-of-two graph isomorphism subprotocol of \cite{FOCS:GolMicWig86}) making use of a variant $(2, g)$-PSS (\cref{def:k-g-pss}), a simple modification of Unruh's rewinding technique~\cite{EC:Unruh12} suffices to prove \cref{lemma:delayed-witness-wi}. 

\subsubsection{The \cite{C:LapSha90} Protocol}

We begin by recalling the \cite{C:LapSha90} $\Sigma$-protocol for graph Hamiltonicity.  The protocol uses a non-interactive commitment scheme $\Com$ as a building block, and is executed as follows. 

\begin{itemize}
    \item The prover, given as input the security parameter $1^\secp$ and an input length $1^n$,\footnote{Note that the prover does not even need to know the instance $x$ to compute this message; however, we consider an a priori fixed statement $x$ to make sense of the proof-of-knowledge property.} sends $\secp$ commitments $\mathsf{com}_i$ to adjacency matrices of i.i.d. random cycle graphs on $n$ vertices (i.e., graphs $H_i = \sigma_i C_n$ that are random permutations of a fixed cycle graph on $n$ vertices). 
    \item The verifier sends a uniformly random string $r \gets \{0,1\}^\secp$.
    \item For the third round, the prover is given a graph $G$ and a fixed $n$-cycle represented by a permutation $\pi$ mapping $C_n$ to $G$. The prover then sends the following messages.
    
   \begin{itemize}
    \item For each $i$ such that $r_i = 0$, the prover sends a full opening of the $i$th commitment $\mathsf{com}_i$. 
    
    \item For each $i$ such that $r_i = 1$, the prover sends $\sigma_i \pi^{-1}$ and opens the substring of $\mathsf{com}_i$ consisting of commitments to each non-edge of $\sigma_i \pi^{-1}(G)$. 
   \end{itemize}
    
    \item For each $i$ such that $r_i = 0$, the verifier checks that $\mathsf{com}_i$ was correctly opened to the adjacency matrix of a cycle graph. For each $i$ such that $r_i = 1$, the verifier checks that every matrix entry opened is a valid decommitment to $0$.
\end{itemize}
By the perfect binding of $\mathsf{Com}$, we know that this protocol satisfies $2$-special soundness. In fact, it is the parallel repetition of a protocol satisfying $2$-special soundness: for any index $i$, a commitment string $a_i$ along with a valid response $z_0$ to $r_i = 0$ and a valid response $z_1$ to $r_i = 1$ can be used to compute a Hamiltonian cycle in $G$. Indeed, it satisfies a variant special soundness (implicitly related to $(2, g')$-PSS) described here: 

\begin{claim}
 There exists an extractor $\SSExtract(a, r_1, z_{1,i}^{(1)}, r_2, z_{2,i})$ for the \cite{C:LapSha90} protocol such that $\SSExtract$ outputs a valid $\mathsf{NP}$ witness under the following conditions:
  \begin{itemize}
      \item $r_{1,i} = 0, r_{2,i} = 1$.
      \item $(a_i, r_{2,i}, z_{2,i})$ is an accepting transcript.
      \item There \emph{exists} a response $z_{1,i}$ xwith prefix $z_{1, i}^{(1)}$ such that $(a_i, r_{1,i}, z_{1,i})$ is an accepting transcript.
  \end{itemize}
  Here, $z^{(1)}$ denotes the part of a response $z$ consisting of the \emph{messages} opened (but not the commitment randomness).
\end{claim}

Moreover, we note that the protocol is partially collapsing on $0$-challenges: given a tuple $(x, a, r)$ and a state $\ket{\phi} = \sum_z \alpha_z \ket{z}$, any accepting response $z_i$ such that $r_i = 0$ can be \emph{partially} measured --- namely, the committed bits (but not the openings) can be measured --- without disturbing $\ket{\phi}$. This is sufficient to prove \cref{lemma:delayed-witness-wi}.

\subsubsection{Proof of \cref{lemma:delayed-witness-wi}}

The fact that this protocol is witness indistinguishable follows from the fact that it is a parallel repetition of a post-quantum ZK protocol \cite{STOC:Watrous06}. What remains is to establish the proof-of-knowledge property.

We consider the following variant of Unruh's approach to knowledge extraction~\cite{EC:Unruh12}:

\begin{enumerate}
    \item Given a cheating prover $P^*$, first generate a (classical) first message $a$ from $P^*$. Let $\ket{\psi}$ denote the internal state of $P^*$ at this point.
    \item Sample a uniformly random challenge $r$, compute the $P^*$ unitary $U_r \ket{\psi}$, which writes its response onto some register $\RegZ$. Apply the one-bit measurement $(\Pi_{V, r}, \Id - \Pi_{V, r})$ that checks whether $V(x, a, r, z) = 1$.
    \item If the measurement returns $1$, additionally measure every register $\RegZ^{(1)}_i$ (the opened messages, but not the commitment randomness) corresponding to $r_i = 0$.
    \item Apply $U_r^\dagger$ to the prover state.
    \item Sample an independent random challenge $r'$ and apply $U_{r'}$. Apply the one-bit measurement $(\Pi_{V, r'}, \Id - \Pi_{V, r'})$.
    \item If the measurement returns $1$, additionally measure the \emph{entire} response $\RegZ$. 
    \item If both measurements returned $1$, and there exists an index $i$ such that $r_i = 0$ and $r'_i = 1$, compute $\SSExtract(x, \mathsf{com}_i, 0, z_i^{(1)}, 1, z'_i)$ where $z_i^{(1)}$ is the first partially measured response in location $i$ and $z'_i$ is the second measured response in location $i$. Otherwise, abort.
\end{enumerate}

To show that this extraction procedure works, we first consider the variant in which no response measurements are applied (Step 3 and Step 6 are omitted). Then, by Unruh's rewinding lemma \cite[Lemma 7]{EC:Unruh12}, if $U_r\ket{\psi}$ produces an accepting response with probability at least $\epsilon$ (over the randomness of $r$), then the two binary measurements applied above will \emph{both} return $1$ with probability at least $\epsilon^3$. Then, by the fact that the protocol is partially collapsing on $0$-challenges, this continues to hold even if the measurement in Step 3 is applied. 

Finally, since the probability that i.i.d. uniform strings $r, r'$ do not have an index $i$ such that $r_i = 0$ and $r'_i = 1$ is $(3/4)^\secp = \negl(\secp)$, we conclude that with probability $\epsilon^3 - \negl(\secp)$, the above extractor produces partial accepting response $z_i^{(1)}$ and accepting response $z'_i$ for some $i$ such that $r_i = 0$ and $r'_i = 1$, and so $\SSExtract$ successfully outputs a witness. If $P^*$ is convincing with initial non-negligible probability $\gamma$, then with probability at least $\frac \gamma 2$, $\ket{\psi}$ is at least $\frac \gamma 2$-convincing, and so $\SSExtract$ outputs a valid witness with probability at least $\Omega(\gamma^3)$. This completes the proof of \cref{lemma:delayed-witness-wi}. 

\subsection{Proof of Security for the \cite{STOC:FeiSha90} protocol}

We now prove the security of the Feige-Shamir protocol using suitable building blocks $(\Com, \mathsf{AoK}, \mathsf{dAoK})$.

\begin{theorem}\label{thm:fs-body}
Suppose that:

\begin{itemize}
    \item $\Com$ is a post-quantum non-interactive commitment scheme,
    \item $\mathsf{AoK}$ is the $3$-message state-preserving WI proof of knowledge for $\mathsf{NP}$ (with $\EQPTC$ extraction) from \cref{sec:state-preserving-main-theorems}.
    \item $\mathsf{dAoK}$ is the argument system from \cref{lemma:delayed-witness-wi}. 
\end{itemize}
Then, the Feige-Shamir protocol is both sound and zero-knowledge against QPT adversaries. The zero-knowledge simulator is $\EQPTC$. 
\end{theorem}

\noindent Combining \cref{thm:fs-body} with the results of \cref{sec:state-preserving} implies \cref{thm:feige-shamir}.

We remark that the theorem is non-generic with respect to $\mathsf{AoK},\mathsf{dAoK}$ due to complications in the security proof coming from the fact that $\mathsf{AoK}$ and $\mathsf{dAoK}$ are executed simultaneously. 

\begin{proof} We first prove soundness, followed by ZK.

\vspace{10pt} \noindent \textbf{Proof of Soundness.}
Suppose that $x\not\in L$ and $P^*$ is a QPT prover that convinces $V$ with non-negligible probability. Given such a $P^*$, we define a cheating prover $ P^*_{\mathsf{dAoK}}$ for the underlying $\mathsf{dAoK}$ that is given as additional auxiliary input strings $(c_0, c'_0, c_1, c'_1, b, r_b, \rho_b)$ such that $c_b = \Com(0; r_b)$ and $c'_b = \Com(r_b; \rho_b)$. $ P^*_{\mathsf{dAoK}}$ simply emulates $P^*$ while generating $\mathsf{AoK}$ messages using its auxiliary input. That is:

\begin{itemize}
    \item $ P^*_{\mathsf{dAoK}}$ generates a message $\mathsf{aok}_1$ using its auxiliary input and calls $P^*$ on $(c_0, c'_0, c_1, c'_1, \mathsf{aok}_1)$. This results in a $P^*$-message $(\mathsf{aok}_2, \mathsf{daok}_1)$. $ P^*_{\mathsf{dAoK}}$ returns $\mathsf{daok}_1$.
    \item Upon receiving a verifier challenge $r$, $ P^*_{\mathsf{dAoK}}$ computes an honestly generated message $\mathsf{aok}_3$ (deterministic\footnote{If randomness is required to generate this message, let it be fixed in advance in $P^*_{\mathsf{dAoK}}$'s internal state.} and independent of $r$) using its auxiliary input and calls $P^*$ on $(\mathsf{aok}_3, r)$. This results in a $P^*$-message $\mathsf{daok}_3$, which $ P^*_{\mathsf{dAoK}}$ outputs. 
\end{itemize}

If the auxiliary input $(c_0, c'_0, c_1, c'_1, b, r_b, \rho_b)$ is sampled from the correct distribution, $ P^*_{\mathsf{dAoK}}$ perfectly emulates the interaction of $P^*$ and the honest Feige-Shamir verifier, so $ P^*_{\mathsf{dAoK}}$ is convincing with non-negligible probability $\eps$ by assumption. Thus, the $\mathsf{dAoK}$ knowledge extractor from \cref{lemma:delayed-witness-wi} outputs a valid witness for the statement ``$\exists i, r_i, \rho_i$'' such that $c_i = \Com(0; r_i)$ and $c'_i = \Com(r_i; \rho_i)$'' with probability at least $\Omega(\eps^3)$. 

\begin{claim}
  The probability that the $\mathsf{dAoK}$ extractor succeeds and $i \neq b$ is also $\Omega(\eps^3)$.
\end{claim}

\begin{proof}
  If this is not the case, then we obtain an algorithm breaking the WI property of $\mathsf{AoK}$. For a fixed statement $(c_0, c_0', c_1, c_1')$, the algorithm $V^*_{\mathsf{dAoK}}$, given an honestly generated message $\mathsf{aok}_1$, calls $(\mathsf{aok}_2, \mathsf{dAoK}_1) \gets P^*(c_0, c_0', c_1, c_1', \mathsf{aok}_1)$ and returns the message $\mathsf{aok}_2$. Given a \emph{fixed} response $\mathsf{aok}_3$, $V^*_{\mathsf{dAoK}}$ emulates the  $\mathsf{dAoK}$ extractor from \cref{lemma:delayed-witness-wi} by sampling i.i.d. strings $r, r'$ for $\mathsf{dAoK}$ and re-using the message $\mathsf{aok}_3$. Then, if the extractor returns a valid witness $(i, r_i, \rho_i)$, $V^*_{\mathsf{dAoK}}$ returns the bit $i$. If not, $V^*_{\mathsf{dAoK}}$ guesses at random.
  
  Since this faithfully emulates the execution of the $\mathsf{dAoK}$ extractor on $P^*_{\mathsf{dAoK}}$ and we assumed that it succeeds with probability $\Omega(\eps^3)$, we conclude that the WI property of $\mathsf{AoK}$ with respect to $V^*_{\mathsf{dAoK}}$ implies the claim. 
\end{proof}

However, this implies that the $\mathsf{dAoK}$ extractor breaks the computational hiding property of $\Com$. This is because if $c_{1-b}$ were instead sampled as $\Com(1; r_{1-b})$ and $c_{1-b}'$ sampled as $\Com(r_{1-b}; \rho_{1-b})$, it is information theoretically impossible for the $\mathsf{dAoK}$ extractor to output a witness such that $i \neq b$. This concludes the proof of soundness.

\vspace{10pt} \noindent \textbf{Proof of ZK.} We assume without loss of generality that the cheating verifier $V^*$ has a ``classical'' first message $(c_0, c_0', c_1, c_1', \mathsf{aok}_1)$ by replacing $V^*$ (with auxiliary state $\ket{\psi}$) with $(V^*, \dm)$ for the mixed state $\dm$ obtained by running $U_{V^*}$ on $\ket{\psi}$ to generate a first message, measuring it, and running $U_{V^*}^\dagger$. 

By the construction of $\mathsf{AoK}$ (see \cref{cor:commit-and-prove,thm:state-preserving-wi-body}) we know that the tuple $(c_0, c_0', c_1, c_1', \allowbreak \mathsf{aok}_1)$ \emph{uniquely} determines a witness $\mathsf{td} = (b, r_b, \rho_b)$ that the $\mathsf{AoK}$ extractor can ever output (if such a witness exists; otherwise, we define $\mathsf{td}$ to be $\bot$). We non-uniformly include $\mathsf{td}$ in the description of the $V^*$ state $\dm$ without loss of generality (this does not affect the simulator, only the analysis). 

Our black-box zero-knowledge simulator is defined in \cref{fig:feige-shamir-simulator}:

\begin{figure}
\begin{mdframed}

\begin{itemize}
   \item Construct a first message $\mathsf{daok}_1$ using the honest $\mathsf{dAoK}$ prover algorithm.
   \item For fixed classical strings $(c_0, c_0', c_1, c_1', \mathsf{aok}_1, \mathsf{daok}_1)$, define an $\mathsf{AoK}$ cheating prover $P^*_{\mathsf{AoK}}$ with the following description:
   \begin{itemize}
       \item Send $\mathsf{aok}_1$
       \item On challenge $s$, call $V^*$ on $(s, \mathsf{daok}_1)$. Upon receiving $(\mathsf{aok}_3, r)$, return $\mathsf{aok}_3$.
   \end{itemize}
   \item Run the state-preserving extractor $\Extract^{P^*_{\mathsf{AoK}}, \mathsf{daok}_1, \dm}$, outputting the (unique possible) witness $\mathsf{td}$ along with a $P^*_{\mathsf{AoK}}$-view (which includes a $V^*$-view in it).
   \item If the output witness is $\bot$, send an aborting final message. Otherwise, compute $\mathsf{daok}_3$ using $\mathsf{td}$. 
\end{itemize}

\end{mdframed}
\caption{The Feige-Shamir protocol simulator}
\label{fig:feige-shamir-simulator}
\end{figure}

We claim that this achieves negligible simulation accuracy. We prove this via a hybrid argument:

\begin{itemize}
    \item $\mathsf{Hyb}_0$: This is the simulated view of $V^*$.
    \item $\mathsf{Hyb}_1$: This is the same as $\mathsf{Hyb}_0$,  \emph{except} that $\mathsf{daok}_3$ is computed using an $\mathsf{NP}$-witness $w$ for $x$.
    \item $\mathsf{Hyb}_2$: This is the real view of $V^*$.
\end{itemize}

The indistinguishability of $\mathsf{Hyb}_2$ and $\mathsf{Hyb}_1$ follows immediately from the state-preserving property of $\mathsf{AoK}$, as the view of $P^*_{\mathsf{AoK}}$ contains an entire correctly emulated view of $V^*$.

The indistinguishability of $\mathsf{Hyb}_1$ and $\mathsf{Hyb}_0$ follows from the witness indistinguishability of $\mathsf{dAoK}$. To prove this, we assume for the sake of contradiction that $\mathsf{Hyb}_1$ and $\mathsf{Hyb}_0$ are distinguishable by a polynomial-time distinguisher $D$ with non-negligible advantage $\eps$. Then, we construct the following two additional hybrids: 

\begin{itemize}
    \item $\mathsf{Hyb}'_0$: This is simulated view of $V^*$, \emph{except} that $\Extract$ is replaced by a $\poly(\secp, 1/\epsilon)$-size oracle algorithm that achieves accuracy $\frac{\eps}4$. 
    \item $\mathsf{Hyb}'_1$: This is the same as $\mathsf{Hyb}_0'$ \emph{except} that $\mathsf{daok}_3$ is computed using an $\mathsf{NP}$-witness $w$ for $x$.

\end{itemize}

By a hybrid argument, we conclude that $D$ also distinguishes $\mathsf{Hyb}'_0$ and $\mathsf{Hyb}'_1$ with advantage $\eps/2$. We claim that this breaks the witness indistinguishability of $\mathsf{dAoK}$. Define a $\mathsf{dAoK}$ verifier $V^*_{\mathsf{dAoK}}$ operating as follows

\begin{itemize}
    \item $V^*_{\mathsf{dAoK}}$ has the state $\dm$ as auxiliary input (including $c_0, c_0', c_1, c_1', \mathsf{aok}_1, \mathsf{td})$. $V^*_{\mathsf{dAoK}}$ wants to distinguish between proofs using witness $w$ and proofs using witness $\mathsf{td}$.
    \item $V^*_{\mathsf{dAoK}}$ receives $\mathsf{daok}_1$ from the prover. It then calls (the $\epsilon/4$-truncated) $\Extract^{P^*_{\mathsf{AoK}}, \mathsf{daok}_1, \dm}$, which returns a $P^*_{\mathsf{AoK}}$-view. $V^*_{\mathsf{dAoK}}$ sends the challenge $r$ from the $P^*_{\mathsf{AoK}}$-view to the prover.
    \item Finally, upon receiving $\mathsf{daok}_3$ from the prover, $V^*_{\mathsf{dAoK}}$ outputs the emulated $V^*$ view.
\end{itemize}

$V^*_{\mathsf{dAoK}}$ has been constructed to be (aux-input) QPT, and (along with the distinguisher $ D$) violates the WI property of $\mathsf{dAoK}$, giving the claimed contradiction.

We conclude that the Feige-Shamir protocol is ZK, as desired. We note that the zero-knowledge simulator inherits the $\EQPTC$ structure of the $\mathsf{AoK}$ state-preserving extractor (with some additional fixed poly-time pre- and post-processing). \qedhere 

\end{proof}

%% file: 13-goldreich-kahan.tex
\section{The~\cite{JC:GolKah96} Protocol is Post-Quantum Zero Knowledge}
\label{sec:gk}

In this section we show that the Goldreich--Kahan constant-round proof system for $\NP$ is post-quantum zero knowledge by giving an $\EQPTC$ simulator. In \cref{sec:distinguishability} we give a technical lemma about the distinguishability of certain purifications that will be of central importance in the proof. In \cref{sec:gk-simulator} we describe our quantum simulator.

\subsection{Indistinguishability of Projections onto Indistinguishable States}
\label{sec:distinguishability}

Consider the states $\ket{\tau_b} \eqdef \sum_{x} \ket{x}_{\RegX} \ket{D_b(x)}_{\RegY}$ where $D_0,D_1$ are computationally indistinguishable (w.r.t. quantum adversaries) efficiently sampleable \emph{classical} distributions with random coins $x$ (in a slight abuse of notation, $D_b$ denotes both the distribution and the sample). If we are only given access to $\RegY$, then distinguishing $\ket{\tau_0}$ from $\ket{\tau_1}$ is clearly hard since $\Tr_{\RegX}(\ketbra{\tau_b})$ is equivalent to a random classical sample from $D_b$.

In this subsection, we show that this indistinguishability generically extends to the setting where \emph{we additionally give the distinguisher access to the projection $\ketbra{D_b}$ on $\RegX \otimes \RegY$}. This is formalized by giving the distinguisher an additional one-qubit register $\RegO$ and black-box access (see \cref{subsec:blackbox}) to the unitary $U_{b}$ and its inverse acting on $\RegX \otimes \RegY \otimes \RegO$ defined as
\[ U_b \eqdef \ketbra{D_b}_{\RegX,\RegY} \otimes \mathbf{X}_{\RegB} + (\Id_{\RegX,\RegY} - \ketbra{D_b}_{\RegX,\RegY}) \otimes \Id_{\RegB},\]
where $\mathbf{X}_{\RegB}$ denotes the bit-flip operator on $\RegB$. In particular, it is no longer the case that access to $\ket{\tau_b}$ is equivalent to a random classical sample from $D_b$, since the distinguisher's access to $U_b$ means the $\RegX$ is no longer independent of its view. Nevertheless, we prove the following.

\begin{lemma}
    \label{lemma:proj-indist}
    If there exists a polynomial-time quantum oracle distinguisher $S^{U_b}$ without direct access to $\RegX$ achieving
    \begin{equation*}
        \left|\Pr[S^{U_0}(\ket{D_0}_{\RegX,\RegY}) = 1] - \Pr[S^{U_1}(\ket{D_1}_{\RegX,\RegY}) =1]\right| \geq 1/\poly(\secp),
    \end{equation*}
    then there exists a polynomial-time quantum algorithm $S$ that distinguishes classical samples from the distributions $D_0$ and $D_1$
    
\end{lemma}

Our proof will make use of two results by Zhandry~\cite{C:Zhandry12,Zhandry15}, which we restate here for convenience. In the following, quantum oracle access to a function $f: X \rightarrow Y$ refers to black-box access to the unitary that maps $\ket{x}\ket{y} \rightarrow \ket{x} \ket{f(x) \oplus y}$ for all $x,y$.

\begin{theorem}[Theorem 1.1 of \cite{C:Zhandry12}]
Let $D_0$ and $D_1$ be efficiently sampleable distributions on a set $Y$, and let $X$ be some other set. Let $O_0$ and $O_1$ be the distributions of functions from $X$ to $Y$ where for each $x \in X$, $O_b(x)$ is chosen independently according to $D_b$. Then if $A$ is an efficient quantum algorithm that can distinguish between quantum access to the oracle $O_0$ from quantum access to the oracle $O_1$, we can construct an efficient quantum algorithm $B$ that distinguishes classical samples from $D_0$ and $D_1$.
\end{theorem}

\begin{theorem}[\cite{Zhandry15}]
An efficient quantum algorithm cannot distinguish between quantum access to an oracle $f$ implementing a random function $X \rightarrow X$ and an oracle $\pi$ implementing a random permutation $X \rightarrow X$.
\end{theorem}

\begin{proof}
    By \cite[Theorem 1.1]{C:Zhandry12}, it suffices for us to show that if there exists a distinguisher $S^{U_b}$ that distinguishes $\ket{D_0}_{\RegX,\RegY}$ from $\ket{D_1}_{\RegX,\RegY}$ can without directly accessing the $\RegX$ register, then there is an algorithm to distinguish between quantum oracle access to $D_0 \circ f$ and $D_1 \circ f$ (where $D_b \circ f$ is the composed function $D_b(f(\cdot))$) where $f: X \rightarrow X$ is a random function.
    
    By~\cite{Zhandry15}, we observe that it suffices to show that $S^{U_b}$ implies an algorithm to distinguish between quantum oracle access to $D_0 \circ \pi$ and $D_1 \circ \pi$ for a random \emph{permutation} $\pi: X \rightarrow X$.
    
    Given quantum oracle access to $D_b \circ \pi$, we can implement a unitary $V_{b,\pi}$ that maps $\ket{0}_{\RegX,\RegY}$ to the state $\ket{D_{b,\pi}} \coloneqq \sum_x \ket{x} \ket{D_b(\pi(x))}$ as follows: apply a Hadamard to $\RegX$, then apply the $D_b \circ \pi$ oracle to $\RegX \otimes \RegY$.
    
    We can hence use $S$ to distinguish $b = 0$ from $b = 1$ as follows. We prepare the state $\ket{D_{b,\pi}}_{\RegX,\RegY}$ using $V_{b,\pi}$. Using $V_{b,\pi}$ we can also implement the operation
    \[
        U_{b,\pi} \eqdef \ketbra{D_{b,\pi}} \otimes \mathbf{X}_{\RegB} + (\Id - \ketbra{D_{b,\pi}}) \otimes \Id_{\RegB}
    \]
    as follows: apply $V_{b,\pi}^{\dagger}$ to $\RegX \otimes \RegY$, apply $\ketbra{0}_{\RegX,\RegY} \otimes \mathbf{X}_{\RegB} + (\Id-\ketbra{0}_{\RegX,\RegY}) \otimes \Id_{\RegB}$, then apply $U_{b,\pi}$. We can therefore run $S^{U_{b,\pi}} \ket{D_{b,\pi}}$.
    
    Since $S^{U_b}$ does not act on $\RegX$ except via its oracle, and $\ket{D_b}$ is related to $\ket{D_{b,\pi}}$ by a unitary acting on $\RegX$ only, it holds that
    \[
        \Tr_{\RegX}(S^{U_{b,\pi}}(\ketbra{D_{b,\pi}})) = \Tr_{\RegX}(S^{U_{b}}(\ketbra{D_{b}})),
    \]
    which completes the proof.
\end{proof}

\subsection{Quantum Simulator}
\label{sec:gk-simulator}

We begin by describing a variable-runtime $\EQPTM$ estimation procedure that will be a useful subroutine in our quantum zero-knowledge simulator for~\cite{JC:GolKah96}. Following~\cref{thm:vrsvt}, let $\VarEstimate$ and $\Transform$ be the first and second stages of the variable-runtime singular vector transform (vrSVT). For binary projective measurements $\MeasA,\MeasB,\MeasC$ on $\RegA$, we define a ``estimate-disturb-transform'' procedure $\mathsf{EDT}[\MeasA,\MeasB,\MeasC]$. Intuitively, this procedure first uses $\VarEstimate$ to compute an upper bound on the running time of $\Transform[\MeasA \rightarrow \MeasB]$, but then disturbs the state with the measurement $\sC$ before running $\Transform[\MeasA \rightarrow \MeasB]$. However, to ensure that $\VarEstimate$ does not run for unbounded time, the input is first ``conditioned'' by applying $\MeasB$ followed by $\MeasA$, and only proceeding if both measurements return $1$.

Formally, the procedure takes an input state on $\RegA$ and does the following:

\medskip \begin{minipage}{0.9\textwidth}
\noindent $\mathsf{EDT}[\MeasA,\MeasB,\MeasC]$:
\begin{enumerate}[nolistsep]
    \item Apply $\MeasB$ to $\RegA$, obtaining outcome $b_1$.
    \item Apply $\MeasA$ to $\RegA$, obtaining outcome $b_2$.
    \item If $b_1 = 0$ or $b_2 = 0$, stop and output $(0,\bot)$.
    \item Otherwise, run $\VarEstimate_{\delta}[\MeasA \rightleftarrows \MeasB]$ on $\RegA$, obtaining classical output $y$.
    \item Apply $\MeasC$ to $\RegA$, obtaining outcome $c$.
    \item Run $\Transform_y[\MeasA \rightarrow \MeasB]$ on $\RegA$.
    \item Output $\RegA$ and $(1,c)$.
\end{enumerate}
\end{minipage}\medskip

\noindent Let $\widehat{\mathsf{EDT}}[\MeasA,\MeasB,\MeasC]$ denote a coherent implementation of this procedure.

\begin{claim}
    \label{claim:vrsvt-mreqpt}
    For any efficient measurements $\MeasA,\MeasB,\MeasC$, $\mathsf{EDT}[\MeasA,\MeasB,\MeasC]$ is $\EQPTM$.
\end{claim}
\begin{proof}
Since $\mathsf{EDT}[\MeasA,\MeasB,\MeasC]$ commutes with $\MeasJor[\MeasA,\MeasB]$, it suffices to analyze its running time for states contained within a single Jordan subspace. Let $\ket{\psi_j} \eqdef \alpha \JorKetB{j}{1} + \beta \JorKetB{j}{0}$. Then
    \begin{equation*}
        \Pr[b_1 = b_2 = 1] = |\alpha|^2 \Pr[\MeasA(\JorKetB{j}{1}) \to 1] \leq p_j.
    \end{equation*}
    Note that $\MeasC$ does not affect the running time of $\Transform_y$. Hence the expected running time of this procedure on $\ket{\psi_j}$ is
    \begin{equation*}
        O((p_j \cdot \log(1/\delta)/\sqrt{p_j} + 1) \cdot (t_{\MeasA} + t_{\MeasB})) = O(\log (1/\delta) \cdot (t_{\MeasA} + t_{\MeasB})).
    \end{equation*}
    It follows that this procedure is $\EQPTM$.
\end{proof}

We define the states and measurements used in the simulator.
\begin{itemize}[noitemsep]
\item For $\Ch \in R$, let $\ket{\Sim_\Ch} \eqdef \frac{1}{\sqrt{2^\secp}} \sum_{\mu} \ket{\mu} \ket{\SHVZK.\Sim(\Ch;\mu)}$.
\item Let $\Meas{\Sim} \eqdef \BMeas{\BProj{\Sim}}$, where $\BProj{\Sim} \eqdef \sum_{\Ch} \ketbra{\Ch} \otimes \ketbra{\Sim_\Ch} \otimes \Id$.
\item Let $\Meas{\RegR} \eqdef (\BProj{r})_{r \in R}$, where $\BProj{r} \eqdef U_{V^*}^{\dagger} \ketbra{r}_{\RegR} U_{V^*}$.
\item Let $\Meas{\com} \eqdef \BMeas{\BProj{\com}}$, where \[\BProj{\com} \eqdef \sum_{\substack{r,\omega \\ \Commit(\ck,r,\omega) = \com}} \ketbra{r,\omega}~.\]
\end{itemize}

\begin{minipage}{0.9\textwidth}
\noindent $\Sim^{V^*}$:
\begin{enumerate}[nolistsep]
    \item Run $V^*(\ck)$ for $\ck \gets \Gen(1^\secp)$ to obtain a commitment $\com$.
    \item Generate the state $\ket{0}_{\RegR'} \ket{\Sim_0}_{\RegM,\RegA,\RegZ}$.
    \item \label[step]{step:gk-svt-forward} Apply $\widehat{\mathsf{EDT}}[\Meas{\com},\Meas{\Sim},\Meas{\RegR}]$, obtaining outcome $(b,\Ch)$ (in superposition). Measure $b$.
    \item \label[step]{step:gk-replace-state} If $b = 1$, measure $r$ and replace the state on $\RegR',\RegM,\RegA,\RegZ$ with $\ket{\Ch}_{\RegR'}\ket{\Sim_\Ch}_{\RegM,\RegA,\RegZ}$.
    \item \label[step]{step:gk-svt-reverse} Apply $\widehat{\mathsf{EDT}}[\Meas{\com},\Meas{\Sim},\Meas{\RegR}]^{\dagger}$.
    \item \label[step]{step:gk-final-meas} Measure register $\RegA$, obtaining outcome $\Co$. Apply $U_{V^*}$ and measure $\RegR,\RegW$ to obtain $(\Ch',\omega)$; if $\Commit(\Ch',\omega) \neq \com$, stop and output the view of $V^*$. Otherwise, measure $\RegZ$, obtaining outcome $\Resp$. Send $z$ to $V^*$ and output the view of $V^*$.
\end{enumerate}
\end{minipage}

\begin{lemma}
    If $\Com$ is a collapse-binding commitment then $\Sim^{V^*}(\DMatrix)$ is computationally indistinguishable from $\out_{V^*} \langle P,V^* \rangle$. $\Sim^{V^*}$ is an $\EQPTC$ algorithm.
\end{lemma}
\begin{proof}
    By \cref{claim:vrsvt-mreqpt}, $P[\Meas{\com},\Meas{\Sim},\Meas{\RegR}]$ is $\EQPTM$, and so $\Sim^{V^*}$ is $\EQPTC$.
    
    We consider three hybrid simulators $H_1,H_2,H_3$, as follows. All three are provided with some witness $w$ such that $(x,w) \in \Relation$. We first define $H_1$.
    
    \medskip \begin{minipage}{0.9\textwidth}
        $H_1^{V^*}(x,w)$:
        \begin{enumerate}[nolistsep]
            \item Run $V^*(\ck)$ for $\ck \gets \Gen(1^\secp)$ to obtain a commitment $\com$.
            \item Generate the state $\ket{P} \eqdef \sum_{\mu} \ket{\mu}_{\RegM} \ket{P_{\Sigma}(x,w;\mu)}_{\RegA}$.
            \item \label[step]{step:gk-hybrid-1-svt-forward} Let $\Meas{P} \eqdef \BMeas{\ketbra{P}}$. Apply $\widehat{\mathsf{EDT}}[\Meas{\com},\Meas{P},\Meas{\RegR}]$, obtaining outcome $(b,\Ch)$ (in superposition). Measure $b$.
            \item[4-6.] As in $\Sim$.
        \end{enumerate}
    \end{minipage} \medskip
    
    $H_1$ is indistinguishable from $\Sim$ by \cref{lemma:proj-indist}: $H_1$ is obtained from $\Sim$ by replacing $\ket{\Sim_0}$ and $\Meas{\Sim}$ with $\ket{P}$ and $\Meas{P}$ and interacts only with the $\RegA$ register, and the distributions on $a$ induced by $(a,z) \gets \SHVZK.\Sim(0;\mu)$ and $a \gets P_{\Sigma}(x,w;\mu')$ are computationally indistinguishable.
    
    \medskip \begin{minipage}{0.9\textwidth}
        $H_2^{V^*}(x,w)$:
        \begin{enumerate}[nolistsep]
            \item[1-3.] As in $H_1$.
            \setcounter{enumi}{3}
            \item \label[step]{step:gk-hybrid-2-replace} If $b = 1$, measure $r$ and replace the state on $\RegM,\RegA,\RegZ$ with
            \[ \ket{P_\Ch} = \sum_{\mu} \ket{\mu}_{\RegM} \ket{P_{\Sigma}(x,w;\mu)}_{\RegA} \ket{P_{\Sigma}(x,w,\Ch;\mu)}_{\RegZ}. \]
            \item Let $\Meas{P,\Ch} \eqdef \BMeas{\ketbra{P_\Ch}}$. Apply $\widehat{\mathsf{EDT}}[\Meas{\com},\Meas{P,\Ch},\Meas{\RegR}]^{\dagger}$.
            \item As in $\Sim$.
        \end{enumerate}
    \end{minipage} \medskip
    
    By the SHVZK guarantee, the distributions on $(a,z)$ given by $a \gets P_{\Sigma}(x,w;\mu), z \gets P_{\Sigma}(x,w,\Ch;\mu)$ and $(a,z) \gets \SHVZK.\Sim(\Ch;\mu')$ are computationally indistinguishable.

    Hence by \cref{lemma:proj-indist}, $H_1$ and $H_2$ are computationally indistinguishable.
    
    By the correctness guarantee of $Q$, if $b = 1$ then the state at the beginning of \cref{step:gk-hybrid-2-replace} has $\Tr(\ketbra{P} \DMatrix) \geq 1 - \delta$. Note that $\ket{P}$ and $\ket{P_r}$ are related by an efficient local isometry $T_r \colon \RegM \to \RegM \otimes \RegZ$. Hence \cref{step:gk-hybrid-2-replace} is $\sqrt{\delta}$-close in trace distance to an application of this isometry. Switching to this state, we can commute the isometry through $\widehat{\mathsf{EDT}}[\Meas{\com},\Meas{P,r},\Meas{\RegR}]^{\dagger}$, which conjugates it to $\widehat{\mathsf{EDT}}[\Meas{\com},\Meas{P},\Meas{\RegR}]^{\dagger}$. This leads to the third hybrid, below.
    
    \medskip \begin{minipage}{0.9\textwidth}
        $H_3^{V^*}(x,w)$:
        \begin{enumerate}[nolistsep]
            \item[1-3.] As in $H_2$.
            \setcounter{enumi}{3}
            \item \label[step]{step:gk-hybrid-3-collapse} If $b = 1$, measure $\Ch$.
            \item Apply $\widehat{\mathsf{EDT}}[\Meas{\com},\Meas{P},\Meas{\RegR}]^{\dagger}$.
            \item Apply $U_{V^*}$ and measure $\RegR,\RegW$ to obtain $(r',\omega)$. If $\Commit(r',\omega) \neq \com$, stop and output the view of $V^*$. Otherwise, apply $T_{r'}$ to $\RegM$ and measure $\RegZ$, obtaining outcome $z$. Send $z$ to $V^*$ and output the view of $V^*$.
        \end{enumerate}
    \end{minipage} \medskip
    
    $H_3$ is statistically close to $H_2$ provided that $\Pr[r = r'] = 1 - \negl(\secp)$. Moreover, the collapsing property of the commitment implies that \cref{step:gk-hybrid-3-collapse} is computationally undetectable. If this step is removed then the effect of \cref{step:gk-svt-forward,step:gk-svt-reverse} is simply to apply $\Meas{\com}$; the output is then precisely the view of $V^*$ in a real execution.
    
    Finally, we have that by $\Ch = \Ch'$ with all but negligible probability by the unique message-binding of the commitment scheme~(\cref{lemma:collapse-binding-unique-message}). \qedhere

\end{proof}

%% file: 14-acknowledgments.tex
\section*{Acknowledgments}

We thank Zvika Brakerski, Ran Canetti, Yael Kalai, Vinod Vaikuntanathan, and Mark Zhandry for helpful discussions. NS is supported by DARPA under Agreement No. HR00112020023.

%% file: 0-main.bbl
\newcommand{\etalchar}[1]{$^{#1}$}
\begin{thebibliography}{GHKW17}

\bibitem[AAB{\etalchar{+}}19]{Google}
Frank Arute, Kunal Arya, Ryan Babbush, Dave Bacon, Joseph Bardin, Rami Barends,
  Rupak Biswas, Sergio Boixo, Fernando Brandao, David Buell, Brian Burkett,
  Yu~Chen, Jimmy Chen, Ben Chiaro, Roberto Collins, William Courtney, Andrew
  Dunsworth, Edward Farhi, Brooks Foxen, Austin Fowler, Craig~Michael Gidney,
  Marissa Giustina, Rob Graff, Keith Guerin, Steve Habegger, Matthew Harrigan,
  Michael Hartmann, Alan Ho, Markus~Rudolf Hoffmann, Trent Huang, Travis
  Humble, Sergei Isakov, Evan Jeffrey, Zhang Jiang, Dvir Kafri, Kostyantyn
  Kechedzhi, Julian Kelly, Paul Klimov, Sergey Knysh, Alexander Korotkov, Fedor
  Kostritsa, Dave Landhuis, Mike Lindmark, Erik Lucero, Dmitry Lyakh, Salvatore
  Mandrà, Jarrod~Ryan McClean, Matthew McEwen, Anthony Megrant, Xiao Mi,
  Kristel Michielsen, Masoud Mohseni, Josh Mutus, Ofer Naaman, Matthew Neeley,
  Charles Neill, Murphy~Yuezhen Niu, Eric Ostby, Andre Petukhov, John Platt,
  Chris Quintana, Eleanor~G. Rieffel, Pedram Roushan, Nicholas Rubin, Daniel
  Sank, Kevin~J. Satzinger, Vadim Smelyanskiy, Kevin~Jeffery Sung, Matt
  Trevithick, Amit Vainsencher, Benjamin Villalonga, Ted White, Z.~Jamie Yao,
  Ping Yeh, Adam Zalcman, Hartmut Neven, and John Martinis.
\newblock Quantum supremacy using a programmable superconducting processor.
\newblock {\em Nature}, 574:505–510, 2019.

\bibitem[ACL21]{C:AnaChuLap21}
Prabhanjan Ananth, Kai-Min Chung, and Rolando~L. {La Placa}.
\newblock On the concurrent composition of quantum zero-knowledge.
\newblock In Tal Malkin and Chris Peikert, editors, {\em CRYPTO~2021, Part~I},
  volume 12825 of {\em {LNCS}}, pages 346--374, Virtual Event, August 2021.
  Springer, Heidelberg.

\bibitem[AL20]{TCC:AnaLap20}
Prabhanjan Ananth and Rolando~L. {La Placa}.
\newblock Secure quantum extraction protocols.
\newblock In Rafael Pass and Krzysztof Pietrzak, editors, {\em TCC~2020,
  Part~III}, volume 12552 of {\em {LNCS}}, pages 123--152. Springer,
  Heidelberg, November 2020.

\bibitem[ARU14]{FOCS:AmbRosUnr14}
Andris Ambainis, Ansis Rosmanis, and Dominique Unruh.
\newblock Quantum attacks on classical proof systems: The hardness of quantum
  rewinding.
\newblock In {\em 55th FOCS}, pages 474--483. {IEEE} Computer Society Press,
  October 2014.

\bibitem[BBBV97]{BBBV97}
Charles~H Bennett, Ethan Bernstein, Gilles Brassard, and Umesh Vazirani.
\newblock Strengths and weaknesses of quantum computing.
\newblock {\em SIAM journal on Computing}, 26(5):1510--1523, 1997.

\bibitem[BCKM21]{C:BCKM21b}
James Bartusek, Andrea Coladangelo, Dakshita Khurana, and Fermi Ma.
\newblock One-way functions imply secure computation in a quantum world.
\newblock In Tal Malkin and Chris Peikert, editors, {\em CRYPTO~2021, Part~I},
  volume 12825 of {\em {LNCS}}, pages 467--496, Virtual Event, August 2021.
  Springer, Heidelberg.

\bibitem[BG93]{C:BelGol92}
Mihir Bellare and Oded Goldreich.
\newblock On defining proofs of knowledge.
\newblock In Ernest~F. Brickell, editor, {\em CRYPTO'92}, volume 740 of {\em
  {LNCS}}, pages 390--420. Springer, Heidelberg, August 1993.

\bibitem[BG02]{CCC:BarGol02}
Boaz Barak and Oded Goldreich.
\newblock Universal arguments and their applications.
\newblock In {\em Proceedings 17th IEEE Annual Conference on Computational
  Complexity}, pages 194--203. IEEE, 2002.

\bibitem[BHMT02]{BHMT02}
Gilles Brassard, Peter Høyer, Michele Mosca, and Alain Tapp.
\newblock Quantum amplitude amplification and estimation.
\newblock {\em Quantum Computation and Information}, page 53–74, 2002.

\bibitem[BL02]{STOC:BarLin02}
Boaz Barak and Yehuda Lindell.
\newblock Strict polynomial-time in simulation and extraction.
\newblock In {\em 34th ACM STOC}, pages 484--493. {ACM} Press, May 2002.

\bibitem[Blu86]{Blum86}
Manuel Blum.
\newblock How to prove a theorem so no one else can claim it.
\newblock In {\em Proceedings of the International Congress of Mathematicians},
  volume~1, page~2. Citeseer, 1986.

\bibitem[BMO90]{BelMicOst90}
Mihir Bellare, Silvio Micali, and Rafail Ostrovsky.
\newblock Perfect zero-knowledge in constant rounds.
\newblock In {\em Proceedings of the twenty-second annual ACM symposium on
  Theory of Computing}, pages 482--493, 1990.

\bibitem[BOV03]{C:BarOngVad03}
Boaz Barak, Shien~Jin Ong, and Salil~P. Vadhan.
\newblock Derandomization in cryptography.
\newblock In Dan Boneh, editor, {\em CRYPTO~2003}, volume 2729 of {\em {LNCS}},
  pages 299--315. Springer, Heidelberg, August 2003.

\bibitem[BS20]{STOC:BitShm20}
Nir Bitansky and Omri Shmueli.
\newblock Post-quantum zero knowledge in constant rounds.
\newblock In Konstantin Makarychev, Yury Makarychev, Madhur Tulsiani, Gautam
  Kamath, and Julia Chuzhoy, editors, {\em 52nd ACM STOC}, pages 269--279.
  {ACM} Press, June 2020.

\bibitem[BV97]{BV97}
Ethan Bernstein and Umesh Vazirani.
\newblock Quantum complexity theory.
\newblock {\em SIAM Journal on computing}, 26(5):1411--1473, 1997.

\bibitem[CCLY21]{FOCS:CCLY21}
Nai-Hui Chia, Kai-Min Chung, Qipeng Liu, and Takashi Yamakawa.
\newblock On the impossibility of post-quantum black-box zero-knowledge in
  constant rounds.
\newblock FOCS~'21, 2021.

\bibitem[CCY21]{C:ChiChuYam21}
Nai-Hui Chia, Kai-Min Chung, and Takashi Yamakawa.
\newblock A black-box approach to post-quantum zero-knowledge in constant
  rounds.
\newblock In Tal Malkin and Chris Peikert, editors, {\em CRYPTO~2021, Part~I},
  volume 12825 of {\em {LNCS}}, pages 315--345, Virtual Event, August 2021.
  Springer, Heidelberg.

\bibitem[CMSZ21]{FOCS:CMSZ21}
Alessandro Chiesa, Fermi Ma, Nicholas Spooner, and Mark Zhandry.
\newblock Post-quantum succinct arguments: breaking the quantum rewinding
  barrier.
\newblock FOCS~'21, 2021.

\bibitem[Deu85]{Deutsch85}
David Deutsch.
\newblock Quantum theory, the {Church}--{Turing} principle and the universal
  quantum computer.
\newblock {\em Proceedings of the Royal Society of London. A. Mathematical and
  Physical Sciences}, 400(1818):97--117, 1985.

\bibitem[DFMS19]{DonFMS19}
Jelle Don, Serge Fehr, Christian Majenz, and Christian Schaffner.
\newblock Security of the {F}iat--{S}hamir transformation in the quantum
  random-oracle model.
\newblock In {\em Proceedings of the 39th Annual International Cryptology
  Conference}, CRYPTO~'19, pages 356--383, 2019.

\bibitem[FS90]{STOC:FeiSha90}
Uriel Feige and Adi Shamir.
\newblock Witness indistinguishable and witness hiding protocols.
\newblock In {\em 22nd ACM STOC}, pages 416--426. {ACM} Press, May 1990.

\bibitem[GHKW17]{TCC:GHKW17}
Rishab Goyal, Susan Hohenberger, Venkata Koppula, and Brent Waters.
\newblock A generic approach to constructing and proving verifiable random
  functions.
\newblock In Yael Kalai and Leonid Reyzin, editors, {\em TCC~2017, Part~II},
  volume 10678 of {\em {LNCS}}, pages 537--566. Springer, Heidelberg, November
  2017.

\bibitem[GK96]{JC:GolKah96}
Oded Goldreich and Ariel Kahan.
\newblock How to construct constant-round zero-knowledge proof systems for
  {NP}.
\newblock {\em Journal of Cryptology}, 9(3):167--190, June 1996.

\bibitem[GLSV21]{EC:GLSV21}
Alex~B. Grilo, Huijia Lin, Fang Song, and Vinod Vaikuntanathan.
\newblock Oblivious transfer is in {MiniQCrypt}.
\newblock In Anne Canteaut and Fran\c{c}ois-Xavier Standaert, editors, {\em
  EUROCRYPT~2021, Part~II}, volume 12697 of {\em {LNCS}}, pages 531--561.
  Springer, Heidelberg, October 2021.

\bibitem[GMR85]{STOC:GolMicRac85}
Shafi Goldwasser, Silvio Micali, and Charles Rackoff.
\newblock The knowledge complexity of interactive proof-systems (extended
  abstract).
\newblock In {\em 17th ACM STOC}, pages 291--304. {ACM} Press, May 1985.

\bibitem[GMW86]{FOCS:GolMicWig86}
Oded Goldreich, Silvio Micali, and Avi Wigderson.
\newblock Proofs that yield nothing but their validity and a methodology of
  cryptographic protocol design (extended abstract).
\newblock In {\em 27th FOCS}, pages 174--187. {IEEE} Computer Society Press,
  October 1986.

\bibitem[GMW87]{C:GolMicWig86}
Oded Goldreich, Silvio Micali, and Avi Wigderson.
\newblock How to prove all {NP}-statements in zero-knowledge, and a methodology
  of cryptographic protocol design.
\newblock In Andrew~M. Odlyzko, editor, {\em CRYPTO'86}, volume 263 of {\em
  {LNCS}}, pages 171--185. Springer, Heidelberg, August 1987.

\bibitem[GSLW19]{STOC:GSLW19}
Andr{\'a}s Gily{\'e}n, Yuan Su, Guang~Hao Low, and Nathan Wiebe.
\newblock Quantum singular value transformation and beyond: exponential
  improvements for quantum matrix arithmetics.
\newblock In Moses Charikar and Edith Cohen, editors, {\em 51st ACM STOC},
  pages 193--204. {ACM} Press, June 2019.

\bibitem[IKOS07]{STOC:IKOS07}
Yuval Ishai, Eyal Kushilevitz, Rafail Ostrovsky, and Amit Sahai.
\newblock Zero-knowledge from secure multiparty computation.
\newblock In David~S. Johnson and Uriel Feige, editors, {\em 39th ACM STOC},
  pages 21--30. {ACM} Press, June 2007.

\bibitem[IOS97]{JC:ItoOhtShi97}
Toshiya Itoh, Yuji Ohta, and Hiroki Shizuya.
\newblock A language-dependent cryptographic primitive.
\newblock {\em Journal of Cryptology}, 10(1):37--50, December 1997.

\bibitem[Jor75]{Jordan75}
Camille Jordan.
\newblock Essai sur la g{\'e}om{\'e}trie {\`a} $ n $ dimensions.
\newblock {\em Bulletin de la Soci{\'e}t{\'e} math{\'e}matique de France},
  3:103--174, 1875.

\bibitem[Kil92]{STOC:Kilian92}
Joe Kilian.
\newblock A note on efficient zero-knowledge proofs and arguments (extended
  abstract).
\newblock In {\em 24th ACM STOC}, pages 723--732. {ACM} Press, May 1992.

\bibitem[Klo20]{eprint:Klooss20}
Michael Kloo{\ss}.
\newblock On (expected polynomial) runtime in cryptography.
\newblock Cryptology ePrint Archive, Report 2020/809, 2020.
\newblock \url{https://eprint.iacr.org/2020/809}.

\bibitem[Lev86]{Levin}
Leonid~A Levin.
\newblock Average case complete problems.
\newblock {\em SIAM Journal on Computing}, 15(1):285--286, 1986.

\bibitem[LP98]{LindenP98}
Noah Linden and Sandu Popescu.
\newblock The halting problem for quantum computers.
\newblock {\em arXiv preprint quant-ph/9806054}, 1998.

\bibitem[LS91]{C:LapSha90}
Dror Lapidot and Adi Shamir.
\newblock Publicly verifiable non-interactive zero-knowledge proofs.
\newblock In Alfred~J. Menezes and Scott~A. Vanstone, editors, {\em CRYPTO'90},
  volume 537 of {\em {LNCS}}, pages 353--365. Springer, Heidelberg, August
  1991.

\bibitem[LS19]{ePrint:LomSch19}
Alex Lombardi and Luke Schaeffer.
\newblock A note on key agreement and non-interactive commitments.
\newblock Cryptology ePrint Archive, Report 2019/279, 2019.
\newblock \url{https://eprint.iacr.org/2019/279}.

\bibitem[LZ19]{LiuZ19}
Qipeng Liu and Mark Zhandry.
\newblock Revisiting post-quantum {F}iat--{S}hamir.
\newblock In {\em Proceedings of the 39th Annual International Cryptology
  Conference}, CRYPTO~'19, pages 326--355, 2019.

\bibitem[MV03]{C:MicVad03}
Daniele Micciancio and Salil~P. Vadhan.
\newblock Statistical zero-knowledge proofs with efficient provers: Lattice
  problems and more.
\newblock In Dan Boneh, editor, {\em CRYPTO~2003}, volume 2729 of {\em {LNCS}},
  pages 282--298. Springer, Heidelberg, August 2003.

\bibitem[MW05]{MarriottW05}
Chris Marriott and John Watrous.
\newblock Quantum {A}rthur--{M}erlin games.
\newblock {\em Computational Complexity}, 14(2):122--152, 2005.

\bibitem[Mye97]{Myers97}
John~M Myers.
\newblock Can a universal quantum computer be fully quantum?
\newblock {\em Physical Review Letters}, 78(9):1823, 1997.

\bibitem[NWZ09]{NagajWZ11}
Daniel Nagaj, Pawel Wocjan, and Yong Zhang.
\newblock Fast amplification of {QMA}.
\newblock {\em Quantum Information \& Computation}, 9(11{\&}12):1053--1068,
  2009.

\bibitem[Oza98a]{Ozawa98a}
Masanao Ozawa.
\newblock Quantum nondemolition monitoring of universal quantum computers.
\newblock {\em Physical Review Letters}, 80(3):631, 1998.

\bibitem[Oza98b]{Ozawa98b}
Masanao Ozawa.
\newblock Quantum {Turing} machines: local transition, preparation,
  measurement, and halting.
\newblock {\em arXiv preprint quant-ph/9809038}, 1998.

\bibitem[PRS02]{FOCS:PraRosSah02}
Manoj Prabhakaran, Alon Rosen, and Amit Sahai.
\newblock Concurrent zero knowledge with logarithmic round-complexity.
\newblock In {\em 43rd FOCS}, pages 366--375. {IEEE} Computer Society Press,
  November 2002.

\bibitem[PW09]{TCC:PasWee09}
Rafael Pass and Hoeteck Wee.
\newblock Black-box constructions of two-party protocols from one-way
  functions.
\newblock In Omer Reingold, editor, {\em TCC~2009}, volume 5444 of {\em
  {LNCS}}, pages 403--418. Springer, Heidelberg, March 2009.

\bibitem[Reg06]{Regev06-XXX}
Oded Regev.
\newblock Fast amplification of {QMA} (lecture notes), Spring 2006.

\bibitem[RK99]{EC:RicKil99}
Ransom Richardson and Joe Kilian.
\newblock On the concurrent composition of zero-knowledge proofs.
\newblock In Jacques Stern, editor, {\em EUROCRYPT'99}, volume 1592 of {\em
  {LNCS}}, pages 415--431. Springer, Heidelberg, May 1999.

\bibitem[Ros04]{TCC:Rosen04}
Alon Rosen.
\newblock A note on constant-round zero-knowledge proofs for {NP}.
\newblock In Moni Naor, editor, {\em TCC~2004}, volume 2951 of {\em {LNCS}},
  pages 191--202. Springer, Heidelberg, February 2004.

\bibitem[Sho94]{FOCS:Shor94}
Peter~W. Shor.
\newblock Algorithms for quantum computation: Discrete logarithms and
  factoring.
\newblock In {\em 35th FOCS}, pages 124--134. {IEEE} Computer Society Press,
  November 1994.

\bibitem[Unr12]{EC:Unruh12}
Dominique Unruh.
\newblock Quantum proofs of knowledge.
\newblock In David Pointcheval and Thomas Johansson, editors, {\em
  EUROCRYPT~2012}, volume 7237 of {\em {LNCS}}, pages 135--152. Springer,
  Heidelberg, April 2012.

\bibitem[Unr16a]{Unruh16-asiacrypt}
Dominique Unruh.
\newblock Collapse-binding quantum commitments without random oracles.
\newblock In {\em Proceedings of the 22nd International Conference on the
  Theory and Applications of Cryptology and Information Security},
  ASIACRYPT~'16, pages 166--195, 2016.

\bibitem[Unr16b]{EC:Unruh16}
Dominique Unruh.
\newblock Computationally binding quantum commitments.
\newblock In Marc Fischlin and Jean-S{\'{e}}bastien Coron, editors, {\em
  EUROCRYPT~2016, Part~II}, volume 9666 of {\em {LNCS}}, pages 497--527.
  Springer, Heidelberg, May 2016.

\bibitem[vdG97]{vandeGraaf97}
Jeroen van~de Graaf.
\newblock {\em Towards a formal definition of security for quantum protocols}.
\newblock PhD thesis, University of Montreal, 1997.

\bibitem[Wat06]{STOC:Watrous06}
John Watrous.
\newblock Zero-knowledge against quantum attacks.
\newblock In Jon~M. Kleinberg, editor, {\em 38th ACM STOC}, pages 296--305.
  {ACM} Press, May 2006.

\bibitem[Win99]{Winter99}
Andreas~J. Winter.
\newblock Coding theorem and strong converse for quantum channels.
\newblock {\em IEEE Transactions on Information Theory}, 45(7):2481--2485,
  1999.

\bibitem[Zha12]{C:Zhandry12}
Mark Zhandry.
\newblock Secure identity-based encryption in the quantum random oracle model.
\newblock In Reihaneh Safavi-Naini and Ran Canetti, editors, {\em CRYPTO~2012},
  volume 7417 of {\em {LNCS}}, pages 758--775. Springer, Heidelberg, August
  2012.

\bibitem[Zha15]{Zhandry15}
Mark Zhandry.
\newblock A note on the quantum collision and set equality problems.
\newblock {\em Quantum Inf. Comput.}, 15(7{\&}8):557--567, 2015.

\end{thebibliography}
